\documentclass[11pt]{article}
\usepackage[left=1in,top=1in,right=1in,bottom=1in,head=.1in,nofoot]{geometry}

\usepackage{setspace,url,bm,amsmath} %

\usepackage{titlesec} %
\titlelabel{\thetitle.\quad} %

\titleformat*{\section}{\bf\Large\center}

\usepackage{graphicx} %
\usepackage{bbm}
\usepackage{latexsym}
\usepackage{caption}
\usepackage[margin=20pt]{subcaption}
\usepackage{hyperref}
\usepackage{booktabs}
\usepackage{multirow}

\usepackage{enumitem}

\usepackage[table]{xcolor}

\newcommand{\GG}[1]{}

\usepackage{amsthm}
\usepackage{amssymb}
\usepackage{amsmath}
\usepackage{color}

\usepackage{comment}
\theoremstyle{definition}

\newtheorem*{theorem*}{Theorem}
\newtheorem{theorem}{Theorem}
\newtheorem*{rmk*}{Remark}
\newtheorem{rmk}{Remark}
\newtheorem{proposition}{Proposition}
\newtheorem{lemma}{Lemma}

\newtheorem{condition}{Condition}

\newtheorem{corollary}{Corollary}
\newtheorem*{corollary*}{Corollary}

\def\TV{\text{TV}}

\usepackage{natbib} %
\bibpunct{(}{)}{;}{a}{}{,} %

\usepackage{etoolbox} %
\apptocmd{\sloppy}{\hbadness 10000\relax}{}{} %

\usepackage{color}
\usepackage{listings}

\DeclareMathOperator*{\argmin}{arg\,min}

\def\Pr{\mathbb{P}}

\def\Var{\text{Var}}
\def\Cov{\text{Cov}}

\def\converge{\stackrel{}{\longrightarrow}}
\def\convergeas{\stackrel{\text{a.s.}}{\longrightarrow}}
\def\converged{\stackrel{d}{\longrightarrow}}
\def\convergep{\stackrel{\Pr}{\longrightarrow}}

\def\I{\mathbbm{1}}
\def\E{\mathbb{E}}
\def\a{\text{a}}

\def\sresem{\widetilde{\text{ReSEM}}}

\def\res{\text{res}}

\def\bs{\boldsymbol}

\allowdisplaybreaks

\usepackage{textcomp}

\usepackage{soul}

\def\Unif{\text{Unif}}

\def\limsup{\overline{\lim}}
\def\liminf{\underline{\lim}}

\RequirePackage[normalem]{ulem}

\usepackage{tikz}
\usetikzlibrary{shapes.geometric, arrows}
\tikzstyle{io} = [trapezium, trapezium left angle=70, trapezium right angle=110, minimum width=1cm, minimum height=1cm, text centered, draw=black,  trapezium stretches=true, thick]
\tikzstyle{process} = [rectangle, minimum width=3cm, minimum height=1cm, text centered, draw=black, thick]
\tikzstyle{decision} = [diamond, minimum width=1cm, minimum height=1cm, text centered, draw=black, thick]
\tikzstyle{arrow} = [thick,->,>=stealth]

\usepackage{subcaption}
\usepackage[textsize=tiny, textwidth = 2cm, shadow]{todonotes}

\numberwithin{equation}{section}

\begin{document}

\singlespacing

\title{\bf 
Rejective Sampling, Rerandomization and Regression Adjustment in Survey Experiments
}
\author{
	Zihao Yang, Tianyi Qu and Xinran Li
\footnote{
Zihao Yang and Tianyi Qu are Doctoral Candidates, Department of Statistics, University of Illinois at Urbana-Champaign, Champaign, IL 61820 (E-mail: \href{mailto:zihaoy3@illinois.edu}{zihaoy3@illinois.edu} and \href{mailto:tianyiq3@illinois.edu}{tianyiq3@illinois.edu}). 
Xinran Li is Assistant Professor, Department of Statistics, University of Illinois at Urbana-Champaign, Champaign, IL 61820 (E-mail: \href{mailto:xinranli@illinois.edu}{xinranli@illinois.edu}). 
}
}
\date{}
\maketitle

\begin{abstract}
	Classical randomized experiments, equipped with randomization-based inference, 
    provide assumption-free inference for treatment effects. 
	They have been the gold standard for drawing causal inference and provide excellent internal validity. 
	However, they have also been criticized for questionable external validity, in the sense that the conclusion may not generalize well to a larger population. 
	The randomized survey experiment is a design tool that can help mitigate this concern, by randomly selecting the experimental units from the target population of interest. 
    However, as pointed out by \citet{Morgan2012},
    chance imbalances often exist in covariate distributions between different treatment groups even under completely randomized experiments.
	Not surprisingly, such covariate imbalances also occur in randomized survey experiments. 
    Furthermore, 
    the covariate imbalances happen not only between different treatment groups, but also between the sampled experimental units and the overall population of interest. 
	In this paper, we propose a two-stage rerandomization design that can actively avoid undesirable covariate imbalances at both  the sampling 
	and treatment assignment stages.
	We further develop asymptotic theory for rerandomized survey experiments, demonstrating that  rerandomization provides better covariate balance, more precise treatment effect estimators, and shorter large-sample confidence intervals.
	We also propose covariate adjustment to deal with remaining covariate imbalances after rerandomization, showing that it can further improve both the sampling and estimated precision. 
	Our work allows general relationship among covariates at the sampling, treatment assignment and analysis stages, 
	and 
	generalizes both rerandomization in classical randomized experiments \citep{Morgan2012} and rejective sampling in survey sampling \citep{Fuller2009}.
\end{abstract}


\section{Introduction}\label{sec:intro}

Understanding the causal effect of some intervention 
or policy has received a lot of attention in 
social, medical and industrial studies \citep{ATHEY201773, rosenberger2015randomization, box2005statistics}. 
Causal conclusions drawn from observational data generally require some untestable assumptions and can be vulnerable to biases from unmeasured confounding. 
Since \citet{Fisher:1935}'s seminal work, 
randomized experiments have become the gold standard for drawing causal inference. 
They balance all confounding factors (no matter observed or unobserved) on average and justifies simple comparisons among different treatment groups. 
However, randomized experiments have been suffering from critiques 
for their generalizability 
to larger populations. 
In other words, randomized experiments 
provide excellent internal validity, 
in the sense that the inference is valid for units within the experiment, but can be questioned on their external validity \citep{campbell2015experimental}, 
in the sense that the inference based on the experimental units may not be valid for 
a larger population of interest.
For example, we can often obtain unbiased estimators for the average treatment effect across units in the experiment; however, these estimators may be biased for the average effect of the target population that the treatment is intended to be applied to.
As pointed out by \citet{ROTHWELL2005},  
lack of consideration of external validity is the most frequent criticism by clinicians of randomized controlled trials, 
and is one explanation for the widespread underuse in routine practice of treatments that were beneficial in trials and that are recommended in guidelines. 
Nevertheless, 
some researchers regard internal validity as a prerequisite to external validity \citep{Lucas2003, jimenez2010trade}, 
and think the questionable external validity is worth the trade-off for unrivaled gains in internal validity \citep{mize2019survey}; see also \citet{ATHEY201773} for related discussion.

There has been growing literature in studying how to generalize the causal conclusions from randomized experiments to larger populations of interest, see, e.g., \citet{Stuart2011} and \citet{Tipton2013}. 
However, as pointed by \citet{Tipton2013} and \citet{Tipton2014}, the problem of generalization is rarely addressed in the experimental design literature, and very little research has focused on how units should be selected into an experiment to facilitate generalization; for exception, see \citet{cochran1992experimental}, \citet{kish2004statistical} and recent work by \citet{Tipton2014} and \citet{Stuart2018}. 
The survey experiment\footnote{
The term``survey experiment'' 
is referred to 
a survey containing some randomized conditions in the questionnaires, e.g., different wordings 
or 
different background information 
for a question. 
In this paper, we use the term ``survey experiment'' more generally, referring to a general experiment involving two steps: random sampling of a subset of units into the experiment and random assignment of the sampled units into different treatment arms.}, 
which involves random sampling of experimental units from the target population of interest,  
has been viewed as the gold standard for estimating treatment effects for the target population \citep{Imai2008, Stuart2018}.
It helps not only ensure internal validity as usual randomized experiments but also guarantee external validity moving beyond usual randomized experiments; see, e.g., \citet{mutz2011population} for a comprehensive review on population-based survey experiments. 
The use of survey experiment was 
often overlooked in practice, 
and frequently the experimental units are convenient samples that may not be representative of the target population.
This can be due to unknown or hard-to-access target population, costly random sampling of units, and nonparticipation of selected units \citep{Stuart2015}. 
However, as pointed out by \citet{Stuart2015, Stuart2018}, 
there is a handful of studies that actually enrolled experimental units randomly, such as the U.S. federal government programs including 
Upward Bound \citep{seftor2009impacts}, Job Corps \citep{burghardt1999national}, and Head Start \citep{puma2010head}, 
and the possibilities for such design are likely to increase with more and more large-scale (administrative) data becoming available, e.g., 
the publicly available state academic excellence indicator system for educational research \citep{Tipton2013}, 
the electronic health record data for clinical and medical research, and 
the dramatically increasing amount of Internet data in many technology companies (where A/B testing is frequently conducted).  
Besides, even for study with smaller population size (e.g., students 
in a university), it may not be possible to assign the interventions of interest to all units due to some resource constraint, under which  survey experiments may be preferred.

In a standard randomized survey experiment, the first step is a simple random sampling (SRS) and the second step is a completely randomized experiment (CRE); see, e.g., \citet[][Chapter 6]{imbens2015causal} and \citet{Branson:2019aa}. For descriptive convenience, we call it a completely randomized survey experiment (CRSE).  
The CRSE balances all observed and unobserved factors \textit{on average}, making the experimental units comparable to the target population of interest and the treated units comparable to the control units.
Intuitively, the randomized experiment in the second step helps draw causal inference while the survey sampling in the first step helps generalize causal conclusions. 
Recently, 
\citet{Morgan2012} observed that the covariate distributions between treated and control groups are likely to be unbalanced for a realized treatment assignment, and proposed rerandomization to 
improve covariate balance in randomized experiments, 
which further results in more precise treatment effect estimation \citep{li2019rerandomization} and thus improves the internal validity of the randomized experiment. 
However, the use of rerandomization for improving external validity of randomized experiments has not been explored in the literature. Indeed, the covariate distributions between sampled units and the whole population of interest are also likely to be unbalanced for a realized sampling. Similar in spirit to rerandomization, we can avoid such unlucky 
sampling by discarding realizations with bad covariate balance and keeping resampling until we get an acceptable one satisfying some covariate balance criterion. 
This is closely related to rejective sampling proposed by \citet{Fuller2009} in the context of survey sampling.

Our study on the design and analysis of survey experiments 
contributes to the literature in the following way. 
First, we propose a two-stage rerandomization design for survey experiments, which combines rejective sampling from survey sampling and rerandomization from usual treatment-control experiments. 
Intuitively, 
rerandomization 
improves the internal validity, while rejective sampling 
improves the external validity. 
In particular, 
the latter is achieved by balancing covariates between experimental units and the target population, 
which is often overlooked in practice and is crucial for 
the generalizability of inferred causal conclusions. 
Second, we demonstrate that the proposed design can improve the precision of the 
usual
difference-in-means estimator, and we further quantify the improvement from rerandomization at both stages. 
Third, we consider covariate adjustment for rerandomized survey experiments and study optimal adjustment in terms of both sampling and estimated precision. 
Compared to usual treatment-control experiments, 
our asymptotic analysis of the survey experiments involves more delicate analysis due to the two-stage feature. 
We relegate all the technical details to the Supplementary Material.

\section{Framework, Notation and Assumption}\label{sec:framework}

\subsection{Potential outcomes, sampling and treatment assignment}

We consider an experiment for a finite population of $N$ units. 
Due to, say, some resource constraint, 
only $n$ of them will enter the experiment, among which $n_1$ will receive an active treatment and the remaining $n_0$ will receive control, where $n_1 + n_0 = n \le N$. 
Let $f = n/N$ be the proportion of sampled units, 
and $r_1 = n_1/n$ and $r_0 = n_0/n$ be the proportions of treated and control units, respectively, where $r_0+r_1=1$. 
We introduce the potential outcome framework \citep{Neyman:1923, Rubin:1974} to define treatment effects. 
For each unit $1\le i \le N$, 
let $Y_i(1)$ and $Y_i(0)$ be the potential outcomes under treatment and control, respectively, 
and $\tau_i = Y_i(1) - Y_i(0)$ be the individual treatment effect. 
The population average potential outcomes across all $N$ units under treatment and control are, respectively, 
$\bar{Y}(1) = N^{-1} \sum_{i=1}^{N} Y_i(1)$ and $\bar{Y}(0) = N^{-1} \sum_{i=1}^{N} Y_i(0)$, 
and the population average treatment effect is $\tau = N^{-1} \sum_{i=1}^N \tau_i =  \bar{Y}(1) - \bar{Y}(0)$.

For each unit $i$, we introduce $Z_i$ to denote the sampling indicator, and $T_i$ to denote the treatment assignment  indicator. 
Specifically, 
$Z_i$ equals 1 if unit $i$ is sampled to enroll the experiment, and 0 otherwise.  
For sampled unit $i$ with $Z_i=1$, 
$T_i$ equals $1$ if the unit is assigned to treatment, and 0 otherwise. 
Let $\mathcal{S} = \{i: Z_i = 1,  1\le i \le N\}$ denote the set of sampled units. 
For each sampled unit $i$ in $\mathcal{S}$, its observed outcome is one of its two potential outcomes depending on the treatment assignment, i.e., 
$Y_i = T_i Y_i(1) + (1-T_i) Y_i(0)$.

Define $\bs{Z} = (Z_1,Z_2, \ldots, Z_N)$ and $\bs{T} = (T_1, T_2, \ldots, T_N)$ as the sampling and treatment assignment vectors for all units. 
Obviously, $T_i$ is well-defined if and only if $Z_i=1$. 
Let $\bs{T}_{\mathcal{S}}$ denote the subvector of  $\bs{T}$ with indices in $\mathcal{S}$. 
Under a SRS,  
the probability that
$\bs{Z}$ takes a particular value $\bs{z} = (z_1, z_2, \ldots, z_N)$ is $\binom{N}{n}^{-1}$, if $z_i \in \{0,1\}$ for all $i$ and $\sum_{i=1}^{N} z_i = n$. 
Given the sampled units $\mathcal{S}$, 
under a CRE, 
the probability that $\bs{T}_{\mathcal{S}}$ takes a particular value $\bs{t}_{\mathcal{S}} = (t_i: i \in \mathcal{S})$ is $\binom{n}{n_1}^{-1}$, if $t_i\in \{0,1\}$ for all $i \in \mathcal{S}$ and $\sum_{i \in \mathcal{S}} t_i = n_1$. 
Under the CRSE, 
the joint distribution of $\bs{Z}$ and $\bs{T}$ then has the following equivalent forms: 
\begin{align}\label{eq:srs_cre}
    \Pr\left(
	\bs{Z} = \bs{z}, \ 
	\bs{Z} \circ \bs{T} = \bs{z} \circ \bs{t} 
	\right)
	= 
	\Pr\left(
	\bs{Z} = \bs{z}, \ 
    \bs{T}_{\mathcal{S}} = \bs{t}_{\mathcal{S}}
	\right)
	= 
	\binom{N}{n}^{-1} \binom{n}{n_1}^{-1} = 
	\frac{(N-n)! n_1! n_0!}{N!}, 
\end{align}
if $z_i\in \{0,1\}$ and $z_i t_i \in \{0,1\}$ for all $i$, 
$\sum_{i=1}^{N} z_i = n$ 
and $\sum_{i=1}^{N} z_it_i = n_1$; 
and  zero otherwise. 
Therefore, the CRSE
is mathematically equivalent to randomly partitioning all $N$ units into three groups of sizes $n_1$, $n_0$ and $N-n$, respectively. 

\subsection{Covariate imbalance and rerandomization in sampling and assignment}\label{sec:rerand_general}

Under the SRS, the distributions of any observed or unobserved covariate for the sampled $n$ units and the total $N$ units are the same \textit{on average}; 
under the CRE, the distributions of any observed or unobserved covariate for the treatment and control groups are the same \textit{on average}.
The former 
enables us 
to infer properties of the finite population of all $N$ units using only the sampled $n$ units, 
and the latter 
enables us 
to infer treatment effects for the sampled $n$ units by comparing the treatment and control groups. 
However, as pointed by \citet{Fuller2009} and \cite{Morgan2012},  
a realized simple random sample 
may appear undesirable with respect to the available auxiliary (i.e., covariate) information, 
and a realized complete randomization can often be unbalanced in terms of the covariate distributions in treatment and control groups. 
To avoid the covariate imbalance, 
\citet{Fuller2009} proposed rejective sampling to avoid unlucky samples 
in survey sampling, 
and 
\citet{Morgan2012} proposed rerandomization to avoid unlucky treatment assignments in randomized experiments, both of which share similar spirit. 
In the remaining discussion, we will also view rejective sampling as 
rerandomization in the sense that it rerandomizes the selected samples.  

In survey experiments, 
we want to actively avoid covariate imbalance in both sampling and treatment assignment. 
Inspired by \citet{Fuller2009} and \citet{Morgan2012}, 
we propose a general two-stage rerandomization design  
for a survey experiment as follows: 
\begin{enumerate}[label={(S\arabic*)}, topsep=1ex,itemsep=-0.3ex,partopsep=1ex,parsep=1ex
	]
	\item Collect covariate data of the $N$ units, and specify a covariate balance criterion for sampling. 
	
	\item 
	Randomly sample $n$ units from the population of $N$ units. 
	
	\item 
	Check the covariate balance based on the criterion for sampling specified in (S1). If the balance criterion is satisfied, continue to (S4); otherwise, return to (S2).
	
	\item 
	Select the $n$ units using the accepted sample from (S3), collect more covariate data for the $n$ units if possible, 
	and specify a covariate balance criterion for treatment assignment. 
	 
	\item 
	Randomly assign the $n$ units into treatment and control groups. 
	
	\item 
	Check the covariate balance based on the criterion for treatment assignment specified in (S4). 
	If the balance criterion is satisfied, continue to (S7); otherwise, return to (S5). 
	
	\item Conduct the experiment using the accepted assignment from (S6). 
\end{enumerate}
Importantly, after the experiment, we need to analyze the observed data taking into account the rerandomization used in (S3) and (S6).

For each unit $1\le i\le N$, we use $\bs{W}_i \in \mathbb{R}^J$ to denote the available $J$ dimensional covariate vector at the sampling stage, 
and $\bs{X}_i \in \mathbb{R}^K$ to denote the available $K$ dimensional covariate vector at the treatment assignment stage.  
We emphasize that the covariate $\bs{W}$ is  observed for all $N$ units in the population of interest, 
while  the covariate $\bs{X}$  may only be observed for sampled units in $\mathcal{S}$.  
Oftentimes, $\bs{X}$ 
contains a richer set of covariates than $\bs{W}$, 
in the sense that $\bs{X} \supset \bs{W}$, since we may be able to collect more covariates after sampling the $n$ units.
However, if the experimenter who conducts random assignment does not have access to the covariate information used 
for random sampling, then it is possible that $\bs{X}$ and $\bs{W}$ overlap with each other and both of them contain additional covariate information. 
Throughout the paper, 
we will consider a general scenario without any constraint on the relationship between $\bs{X}$ and $\bs{W}$, unless otherwise stated.

\subsection{Finite population inference and asymptotics}

We conduct finite population inference, sometimes also called randomization-based or design-based inference, that relies solely on the randomness in the sampling and treatment assignment. 
Specifically, 
all the potential outcomes $Y_i(1)$'s and $Y_i(0)$'s and covariates $\bs{W}_i$'s and $\bs{X}_i$'s for the $N$ units of interest are viewed as fixed constants, or equivalently being conditioned on as conducting conditional inference. 
Consequently, 
we do not impose any 
distributional assumption on the outcomes or covariates, as well as the dependence of the outcomes on the covariates. 
The randomness in the observed data comes solely from the random sampling indicators $Z_i$'s and the random treatment assignments $T_i$'s. 
Thus, the distribution of $(\bm{Z}, \bm{T})$, such as that in \eqref{eq:srs_cre}, plays an important role in governing data generating process as well as statistical inference.

We then introduce some fixed finite population quantities. 
Let $\bar{\bs{W}} = N^{-1} \sum_{i=1}^N \bs{W}_i$ and 
$\bar{\bs{X}} = N^{-1} \sum_{i=1}^N \bs{X}_i$ be the finite population averages of covariates. 
For $t=0,1$, 
let 
$
S^2_t = (N-1)^{-1}\sum_{i=1}^N \{ Y_i(t)- \bar{Y}(t) \}^2$ 
and $S^2_{\tau} = (N-1)^{-1}\sum_{i=1}^N (\tau_i-\tau)^2$
be the finite population variances of potential outcomes and individual effects. 
We define analogously 
$
\bs{S}_{\bs{W}}^2
$
and 
$
\bs{S}_{\bs{X}}^2
$
as the finite population covariance matrices of covariates, 
and 
$\bs{S}_{t, \bs{W}}$, 
$\bs{S}_{t, \bs{X}}$,   
$\bs{S}_{\tau, \bs{W}}$ 
and 
$\bs{S}_{\tau, \bs{X}}$
as the finite population covariances between potential outcomes, individual effects and covariates. 

Finally, we introduce the finite population asymptotics that will be utilized throughout the paper. 
Specifically, 
we embed the finite population into a sequence of finite populations with increasing sizes, 
and study the limiting distributions of certain estimators as the size of population goes to infinity; 
see, e.g., \citet{hajek1960limiting} and \citet{fpclt2017}. 
We impose the following regularity condition along the sequence of finite populations. 
\begin{condition} \label{cond:fp}
	As $N \rightarrow \infty$,
	the sequence of finite populations satisfies 
	\begin{enumerate}[label=(\roman*), topsep=1ex,itemsep=-0.3ex,partopsep=1ex,parsep=1ex]
		\item the proportion $f$ of sampled units has a limit in $[0, 1)$;  

		\item the proportions $r_1$ and $r_0$ of units assigned to treatment and control have positive limits;  

		\item the finite population variances $S^2_1, S^2_0, S^2_\tau$ 
		and covariances 
		$
		\bs{S}_{\bs{W}}^2, \bs{S}_{1, \bs{W}}, \bs{S}_{0, \bs{W}}
		$
		$
		\bs{S}_{\bs{X}}^2, \bs{S}_{1, \bs{X}}, \bs{S}_{0, \bs{X}}
		$
		have limiting values, 
		and the limits of $\bs{S}_{\bs{W}}^2$ 
		and $\bs{S}_{\bs{X}}^2$ are
		nonsingular;
		
		\item 
		for $t\in \{0, 1\}$, 
		\begin{equation*}
		    n^{-1} \max_{1 \le i\le N} \{ Y_i(t) -\bar{Y}(t) \}^2 \rightarrow 0, \ \ 
		    n^{-1} \max_{1 \le i\le N}\| \bs{W}_i - \bar{\bs{W}} \|_2^2
		    \rightarrow 0, \ \ 
		    n^{-1} \max_{1 \le i\le N}\| \bs{X}_i - \bar{\bs{X}} \|_2^2
		    \rightarrow 0. 
		\end{equation*}
	\end{enumerate}
\end{condition}

Below we intuitively explain the regularity conditions in Condition \ref{cond:fp}. 
Condition \ref{cond:fp}(i) and (ii) are natural requirements. More importantly, we allow the proportion $f$ to have zero limit. This can be a more reasonable asymptotic approximation when the population size $N$ is much larger than the sample size $n$, which is common in classical survey sampling and population-based survey experiments. 
For (iii) and (iv),
we consider a special case where the potential outcomes and covariates are independent and identically distributed (i.i.d.) from a distribution, whose covariances for covariates are nonsingular.  
If $f$ has a positive limit and the distribution has more than two moments, 
then (iii) and (iv) will hold with probability one. 
If $f$ has a zero limit, 
the distribution is sub-Gaussian, and the sample size $n \gg \log N$, 
then (iii) and (iv) will still hold with probability one. 
Besides, if the potential outcomes and covariates are bounded, then (iii) and (iv) hold with probability one as long as $n\rightarrow \infty$. 
Therefore, Condition \ref{cond:fp} imposes reasonable regularity conditions. 
For descriptive convenience, we further introduce the notation $\dot\sim$ to denote two sequences of random vectors or distributions converging weakly to the same distribution.

\section{Rerandomized Survey Experiments using the Mahalanobis Distances}\label{sec:resem}

\subsection{Covariate balance criteria using the Mahalanobis distances}\label{sec:balance}
As discussed before, 
to conduct
rerandomization in survey experiments, 
we first need to specify the covariate balance criteria for both the sampling and treatment assignment stages at Steps (S3) and 
(S6). 
Recall that 
$\bar{\bs{W}}$ is the average covariate vector for the whole population. 
Let 
$\bar{\bs{W}}_{\mathcal{S}} = n^{-1} \sum_{i=1}^N Z_i \bs{W}_i$ 
be the average covariate vector for the sampled units. 
Following \citet{Fuller2009}, 
in the sampling stage, 
we consider the difference in covariate means between sampled and all units, i.e., 
$\hat{\bs{\delta}}_{\bs{W}} = \bar{\bs{W}}_{\mathcal{S}} - \bar{\bs{W}}$, 
and measure the covariate imbalance using the corresponding Mahalanobis distance, which has the advantage of being affinely invariant \citep[see, e.g.,][]{Morgan2012}: 
\begin{align}\label{eq:M_S}
	M_S \equiv
	\hat{\bs{\delta}}_{\bs{W}}^\top 
	\left\{
	\Cov\left(
	\hat{\bs{\delta}}_{\bs{W}}
	\right)
	\right\}^{-1}
	\hat{\bs{\delta}}_{\bs{W}}
	= 
	\left(\bar{\bs{W}}_{\mathcal{S}} -\bar{\bs{W}} \right)^\top 
	\left\{
	\left( \frac{1}{n} - \frac{1}{N} \right)
	\bs{S}^2_{\bs{W}}
	\right\}^{-1}
	\left(\bar{\bs{W}}_{\mathcal{S}} -\bar{\bs{W}}\right), 
\end{align}
where 
$\Cov(
\hat{\bs{\delta}}_{\bs{W}}
)$
refers to the covariance matrix of the difference-in-means of covariates  $\hat{\delta}_{\bs{W}}$ under the CRSE (or equivalently the SRS).
Let 
$\bar{\bs{X}}_1 = n_1^{-1} \sum_{i=1}^N Z_i T_i \bs{X}_i$ 
and 
$\bar{\bs{X}}_0 = n_0^{-1} \sum_{i=1}^N Z_i (1-T_i) \bs{X}_i$ 
be the average covariates in the treatment and control groups, 
and 
$\bar{\bs{X}}_{\mathcal{S}} = n^{-1}\sum_{i=1}^N Z_i\bs{X}_i$ 
and 
$\bs{s}^2_{\bs{X}} = (n-1)^{-1} \sum_{i: Z_i=1} (\bs{X}_i - \bar{\bs{X}}_{\mathcal{S}})(\bs{X}_i - \bar{\bs{X}}_{\mathcal{S}})^\top$ 
be the sample average and covariance matrix of covariates for sampled units. 
Following \citet{Morgan2012}, 
we consider the difference in covariate means between treated and control units, i.e., 
$
\hat{\bs{\tau}}_{\bs{X}} = \bar{\bs{X}}_1 - \bar{\bs{X }}_0, 
$
and 
measure the covariate imbalance using the corresponding Mahalanobis distance but conditional on the sampled units $\mathcal{S}$: 
\begin{align}\label{eq:M_T}
	M_T & \equiv 
	\hat{\bs{\tau}}_{\bs{X}}^\top 
	\left\{
	\Cov\left(
	\hat{\bs{\tau}}_{\bs{X}} \mid  \mathcal{S}
	\right)
	\right\}^{-1}
	\hat{\bs{\tau}}_{\bs{X}}
	= 
	\left(\bar{\bs{X}}_1 - \bar{\bs{X}}_0\right)^\top 
	\left(
	\frac{n}{n_1 n_0} 
	\bs{s}^2_{\bs{X}}
	\right)^{-1}
	\left(\bar{\bs{X}}_1 - \bar{\bs{X}}_0\right),
\end{align}
where 
$\Cov(
\hat{\bs{\tau}}_{\bs{X}} \mid  \mathcal{S}
)$
refers to the conditional covariance matrix of the difference-in-means of covariates $\hat{\bs{\tau}}_{\bs{X}}$ given sampled units $\mathcal{S}$ under the CRSE (or equivalently the CRE given $\mathcal{S}$). 
We emphasize that here we use the conditional Mahalanobis distance instead of the marginal one, which will replace the sample covariance  $\bs{s}^2_{\bs{X}}$ in \eqref{eq:M_T} by  the finite population covariance 
$\bs{S}^2_{\bs{X}}$. 
This is because $\bs{S}^2_{\bs{X}}$ depends on the covariates $\bs{X}_i$'s for all $N$ units and may thus be unknown. 
Nevertheless, 
$\bs{s}^2_{\bs{X}}$ is actually a consistent estimator for $\bs{S}^2_{\bs{X}}$, and using either of them will lead to rerandomization with the same asymptotic property; see the Supplementary Material for details. 

Based on the Mahalanobis distances in \eqref{eq:M_S} and \eqref{eq:M_T}, 
we propose the following rerandomization scheme for survey experiments. 
Let $a_S$ and $a_T$ be two predetermined positive thresholds. 
Under rerandomized survey experiments using Mahalanobis distances (ReSEM), 
at the sampling stage, 
a simple random sample is acceptable if and only if $M_S \le a_S$, 
and 
at the treatment assignment stage, a complete randomization is acceptable if and only if $M_T \le a_T$. 
The detailed procedure for ReSEM is illustrated in Figure \ref{figure:flowchart}, in parallel with the general procedure discussed in Section \ref{sec:rerand_general}. 

\usetikzlibrary{positioning}
\begin{figure}[htbp] 
\small
\begin{center}
\begin{tikzpicture}[node distance=8mm]
\node (S1) [io, text width=2cm] {Collect $\bs{W_i}$ for all $i$};
\node (S2) [process, right= of S1, text width=4.5cm] { Randomly sample $n$ units, and compute $M_S$};
\node (S3) [decision, aspect=2, right= of S2] {$M_S \le a_S$?};
\node (S4) [io, right= of S3, text width=2cm] { Collect $\bs{X_i}$ for $i \in \mathcal{S}$ };
\node (S5) [process, below= of S4.south west, anchor=north east, text width=6cm] { Randomly assign the selected $n$ units into treatment and control, and compute $M_T$};
\node (S6) [decision, aspect=2, left= of S5] {$M_T \le a_T$?};
\node (S7) [process, left= of S6, text width=2.5cm] {Conduct the  experiment };
\draw [arrow] (S1) -- (S2);
\draw [arrow] (S2) -- (S3);
\draw [arrow] (S3) -- node [above] {yes}  (S4);
\draw [arrow] (S3) -- node [right] {no} ++(0cm, 1.2cm) -|   (S2);
\draw [arrow] (S4) |- (S5);
\draw [arrow] (S5) -- (S6);
\draw [arrow] (S6) -- node [above] {yes}  (S7);
\draw [arrow] (S6) -- node [left] {no} ++(0cm, -1.2cm) -|   (S5);
\end{tikzpicture}
\end{center}
\caption{Procedure for conducting ReSEM.}\label{figure:flowchart} 
\end{figure}
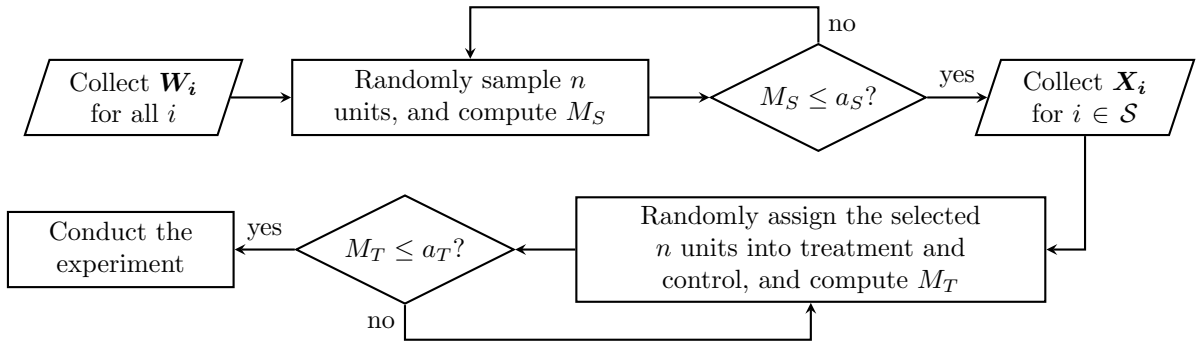

Note that $M_S$ measures the covariate balance between sampled units (which consist of both treated and control units) and the whole population, and $M_T$ measures the covariate balance between treated and control units. 
Thus, ReSEM can improve the covariate balance among the treated, control and all units. 
Intuitively, we can also consider a single-stage rerandomization to improve the covariate balance among them. 
However, we prefer to separate it into two stages as depicted in Figure \ref{figure:flowchart}, due mainly to the following reasons. 
First, the two-stage separation can reduce the computation cost, 
because unlucky realizations of sampling will be excluded without entering the treatment assignment stage. 
This is particularly useful when the covariate balance criteria are stringent, e.g., the thresholds $a_T$ and $a_S$ are small.  
For example, with the same criteria of $M_S\le a_S$ and $M_T \le a_T$, the numbers of randomizations needed to get an acceptable assignment for the single- and two-stage rerandomization are approximately $(p_S p_T)^{-1}$ and $p_S^{-1} + p_T^{-1}$, 
where $p_S = \Pr(\chi^2_J \le a_S)$ and $p_T = \Pr(\chi^2_K \le a_T)$ are the asymptotic acceptance probabilities for the sampling and treatment assignment stages; see the Supplementary Material for details. 
Second, 
the implementation of a two-stage experiment can be more flexible, where
the sampling and treatment assignment can be conducted by different designers at different times. 
For example, after the sampling stage, it is possible to collect richer covariates for the sampled units (which generally have much smaller size than the whole population) that can be used for rerandomization at the treatment assignment stage. On the contrary, 
a single-stage rerandomization needs to collect 
covariates  
for every realization of sampled units until we get an acceptable assignment.

We now discuss how to choose the thresholds $a_S$ and $a_T$ 
for ReSEM. 
Some researchers \citep{Kasy2016, Kallus2018} have suggested to use thresholds as small as possible, leading to certain ``optimal'' designs. 
However, this may result in only one (or two when $n_1=n_0$) acceptable assignment, making the randomization-based asymptotic approximation inaccurate and randomization test powerless. 
Indeed, \citet{Harmonizing2020} showed that such ``optimal'' designs can endanger treatment effect estimation since some unobserved covariates can be highly imbalanced, 
and \citet{Banerjee2020} studied experimental design from an ambiguity-averse decision-making perspective and suggested that targeting a fixed quantile of balance is safer than targeting an absolute balance; 
see also \citet{johansson2019optimal} for comparison between rerandomization and the ``optimal'' design. 
In this paper, we follow the recommendation from \citet{Morgan2012} and fix the thresholds at some small values, 
which not only reduces the burden for optimizing certain covariate balance as in ``optimal'' designs but also provides asymptotically valid randomization-based inference. 
Moreover, 
as commented in \citet{rerand2018} and \citet{schultzberg2019asymptotic}, 
asymptotically, 
the additional gain from further decreasing the threshold becomes smaller as the threshold decreases, and, with a small enough threshold, the precision of the difference-in-means estimator is close to the optimal one that we expect even if we set the threshold close to zero. 
In practice with finite samples, 
\citet{kapelner2019optimal}
considered optimizing the threshold to minimize a certain quantile of treatment effect estimation error under some model assumptions on the potential outcomes. 
Recently, \citet{harshaw2019balancing} proposed a Gram–Schmidt Walk design involving an explicit parameter, which plays a similar role as the rerandomization threshold,  for controlling the trade-off between covariate balance and robustness.

\subsection{Squared multiple correlations between potential outcomes and covariates}\label{sec:R2}

We will introduce two $R^2$-type measures for the association between potential outcomes and covariates, 
which play an important role in studying distributions of treatment effect estimators under ReSEM. 
In particular, we will use the squared multiple correlations under the CRSE between the difference-in-means of covariates, 
$\hat{\bs{\delta}}_{\bs{W}}$ and $\hat{\bs{\tau}}_{\bs{X}}$ in \eqref{eq:M_S} and \eqref{eq:M_T}, 
and
the difference-in-means estimator, 
which 
is an intuitive treatment effect estimator of the following form: 
\begin{align}\label{eq:diff_outcome}
	\hat{\tau} & = \frac{1}{n_1} \sum_{i=1}^N Z_i T_i Y_i  - \frac{1}{n_0} \sum_{i=1}^N Z_i(1-T_i) Y_i. 
\end{align}
As demonstrated in the following proposition, we can conveniently represent the squared multiple correlations between difference-in-means of outcome and covariates using finite population linear projections.
Specifically, 
for 
$t=0,1$, 
the linear projection of $Y(t)$ on $\bs{X}$ 
for unit $i$ is 
$\bar{Y}(t) + \bs{S}_{t, \bm{X}}(\bs{S}_{\bs{X}}^2)^{-1} (\bs{X}_i - \bar{\bs{X}})$, 
whose finite population variance simplifies to 
$S^2_{t \mid  \bs{X}} \equiv \bs{S}_{t,\bs{X}}(\bs{S}_{\bs{X}}^2)^{-1} \bs{S}_{\bs{X},t}$. 
Similarly, 
the finite population variances for the linear projections of individual effect on covariates $\bs{X}$ and $\bs{W}$ are, respectively, 
$
S^2_{\tau \mid  \bs{X}} = 
\bs{S}_{\tau, \bs{X}} (\bs{S}^2_{\bs{X}})^{-1} \bs{S}_{\bs{X},\tau}$
and 
$
S^2_{\tau \mid  \bs{W}} = 
\bs{S}_{\tau, \bs{W}} (\bs{S}^2_{\bs{W}})^{-1} \bs{S}_{\bs{W},\tau}$.

\begin{proposition} \label{prop:R2}
	Under the CRSE, 
	the squared multiple correlations between the difference-in-means of outcome $\hat{\tau}$ and that of the covariates at the sampling stage $\hat{\bs{\delta}}_{\bs{W}}$ and treatment assignment stage $\hat{\bs{\tau}}_{\bs{X}}$ have the following equivalent forms, respectively, 
	\begin{align}\label{eq:R2}
	R_S^2 & = 
	\frac{(1-f)S^2_{\tau \mid  \bs{W}}}{r_1^{-1}S^2_1 +r_0^{-1}S^2_0 -fS^2_\tau},
	\quad 
	\text{and}
	\quad
	R_T^2 = 
	\frac{r_1^{-1}S^2_{1 \mid  \bs{X}}+r_0^{-1}S^2_{0 \mid  \bs{X}} -S^2_{\tau \mid  \bs{X}}}{r_1^{-1}S^2_1 +r_0^{-1}S^2_0 -fS^2_\tau}. 
	\end{align}
\end{proposition}

The form of $R_T^2$ is similar to the $R^2$-type measure introduced in \citet{rerand2018} for studying rerandomization in treatment-control experiments. 
However, $R_S^2$ has a rather different form and depends explicitly on how covariates explain the individual treatment effect. 
Below we consider two special cases that can help further simplify the 
expressions 
in \eqref{eq:R2} and give us more intuition on these two $R^2$-type measures for the association between potential outcomes and covariates. 
First, 
if the treatment effects are additive, 
in the sense that 
$\tau_i$ is constant across all units, 
then the finite population variances of the individual effects and their projections
$S^2_\tau, S^2_{\tau \mid  \bs{W}}$ 
and 
$S^2_{\tau \mid  \bs{X}}$
all reduce to zero. 
From Proposition \ref{prop:R2}, 
the squared multiple correlation for sampling stage reduces to $R^2_S = 0$, 
while that for treatment assignment stage reduces to  $R_T^2 = S^2_{0 \mid  \bs{X}}/S^2_0$, 
the squared multiple correlation between $Y (0)$ and $\bs{X}$.  
Second, if the control potential outcomes are constant across all units, 
then $S_0^2 = S^2_{0 \mid  \bs{X}} =  0$, $S_1^2 = S_{\tau}^2$,  $S^2_{1\mid \bs{W}} = S^2_{\tau\mid \bs{W}}$ and $S^2_{1\mid \bs{X}} = S^2_{\tau\mid \bs{X}}$. 
From Proposition \ref{prop:R2}, the squared multiple correlations reduce to 
$R^2_S = (1-f)/(r_1^{-1}-f) \cdot S^2_{\tau \mid \bs{W}}/S^2_{\tau}$ and 
$R^2_T = (r_1^{-1} - 1)/(r_1^{-1} - f) \cdot S^2_{\tau \mid \bs{X}}/S^2_{\tau}$. 
When $f\approx 0$, i.e., the sample size $n \ll N$, they further reduce to 
$R^2_S \approx r_1 S^2_{\tau \mid \bs{W}}/S^2_{\tau}$ and 
$R^2_T \approx r_0 S^2_{\tau \mid \bs{X}}/S^2_{\tau}$, 
both of which are certain proportions of the squared multiple correlations between the individual treatment effect and the covariates.

\subsection{Asymptotic distribution under rerandomization}\label{sec:asym_diff}
We study the asymptotic distribution of the difference-in-means estimator $\hat{\tau}$ in \eqref{eq:diff_outcome} under ReSEM. 
Let  $\varepsilon \sim \mathcal{N}(0,1)$ be a standard Gaussian random variable, 
and 
$L_{J, a_S} \sim D_1 \mid  \bs{D}^\top \bs{D}\le a_S$ 
be a constrained Gaussian random variable with 
$\bs{D} = (D_1, \ldots, D_J)^\top \sim \mathcal{N}(\bs{0}, \bs{I}_J)$.
Similarly, define $L_{K, a_T}\sim \tilde{D}_1 \mid  \tilde{\bs{D}}^\top \tilde{\bs{D}}\le a_T$ with 
$\tilde{\bs{D}} 
= (\tilde{D}_1, \ldots, \tilde{D}_K)^\top 
\sim \mathcal{N}(\bs{0}, \bs{I}_K)$. 
We assume that $\varepsilon, L_{J,a_S}$ and $L_{K,a_T}$ are mutually independent throughout the paper.
Recall that 
$R^2_S$ and $R^2_T$
measure 
the outcome-covariate associations 
at the sampling and assignment stages, respectively, 
and define
\begin{align}\label{eq:V_tautau}
	V_{\tau \tau} = r_1^{-1}S_1^2 + r_0^{-1}S_0^2 -f S_{\tau}^2,
\end{align}
which is actually the variance of $\sqrt{n} (\hat\tau - \tau)$ under the CRSE \citep{imbens2015causal, Branson:2019aa}. 
Throughout the paper, we will also use the explicit conditioning on ReSEM to emphasize that we are studying the repeated sampling properties under ReSEM. 

\begin{theorem} \label{thm:dist}
Under Condition \ref{cond:fp} and ReSEM, 
\begin{align}\label{eq:dist}
\sqrt{n} (\hat\tau - \tau) \mid \text{ReSEM}  \ & \dot\sim\  V_{\tau\tau}^{1/2} \left(  \sqrt{1-R_S^2 - R_T^2} \cdot \varepsilon
+ \sqrt{R_S^2} \cdot L_{J,a_S}
+
\sqrt{R_T^2} \cdot L_{K,a_T} 
\right). 
\end{align}
\end{theorem}

The derivation of Theorem \ref{thm:dist} relies crucially on two facts. 
First, 
the two-stage ReSEM is asymptotically equivalent to the single-stage rerandomization with the same covariate balance criteria. 
Second, the distribution of $\hat{\tau}$ under the single-stage rerandomization is the same as its conditional distribution under the CRSE given that sampling and treatment assignment satisfy the balance criteria. 
From Theorem \ref{thm:dist}, 
$\hat{\tau}$ under ReSEM follows the same asymptotic distribution as
the summation of three independent Gaussian or constrained Gaussian random variables, 
with coefficients depending on the two squared multiple correlations $R_S^2$ and $R_T^2$. 
Intuitively, 
$\varepsilon$ is the part of $\hat{\tau}$ that is unexplained by the covariates at either the sampling or treatment assignment stages, 
while 
$L_{J,a_S}$ and $L_{K,a_T}$ are the parts that are explained by the covariates at the sampling and treatment assignment stages
and are thus affected by the corresponding covariate balance criteria.

Note that $L_{J, a_S}$ or $L_{K, a_T}$ reduces to a standard Gaussian random variable if $a_S=\infty$ or $a_T = \infty$, 
in the sense that all samplings or assignments are acceptable. 
Thus, 
Theorem \ref{thm:dist} immediately implies 
the asymptotic distribution of $\hat{\tau}$ when there is no rerandomization at the sampling or assignment stages. 
As discussed shortly, 
the design without rerandomization at either
stage 
generally provides less efficient difference-in-means estimator than that with rerandomization at both stages. 

\begin{rmk}\label{rmk:special_dim}
Note that the CRSE, rerandomized treatment-control experiment using the Mahalanobis distance \citep[ReM;][]{Morgan2012, rerand2018}
and rejective sampling \citep{Fuller2009} can all be viewed as special cases of ReSEM. 
Thus, 
Theorem \ref{thm:dist} can also imply the asymptotic distributions of the difference-in-means estimator under the CRSE and ReM as well as the sample average estimator under the rejective sampling; 
see the Supplementary Material for details. 
For example, under Condition \ref{cond:fp} and the CRSE, 
\begin{align}\label{eq:crse}
\sqrt{n} (\hat\tau - \tau)  \  \dot\sim\  V_{\tau\tau}^{1/2} \cdot \varepsilon.
\end{align}
\end{rmk}

\subsection{Improvement on sampling precision from rerandomization}

Theorem \ref{thm:dist} and 
Remark \ref{rmk:special_dim} 
characterize the asymptotic properties of the difference-in-means estimator under ReSEM and the CRSE. 
Below we compare the asymptotic distributions \eqref{eq:dist} and \eqref{eq:crse} of $\sqrt{n}(\hat{\tau} - \tau)$ under these two designs, showing the advantage of rerandomization at both the sampling and treatment assignment stages.

First, 
because the Gaussian random variable $\varepsilon$ and the constrained Gaussian random variables $L_{J,a_S}$ and $L_{K,a_T}$ are all symmetric around zero, 
the asymptotic distribution of $\sqrt{n}(\hat{\tau} - \tau)$ is also symmetric around zero. 
These imply that $\hat{\tau}$ is asymptotically unbiased and consistent for the average treatment effect $\tau$ under ReSEM. 
Furthermore, any covariates, no matter observed or unobserved, are asymptotically balanced between two treatment groups, as well as between the sampled units and the overall population; 
see the Supplementary Material for details.

Second, we compare the asymptotic variance of $\hat{\tau}$ under the CRSE and ReSEM. 
Define $v_{k,a} = P(\chi^2_{k+2}\le a)/P(\chi^2_k\le a)$ for any $a > 0$ and positive integer $k$. 
Then the variances of the constrained Gaussian random variables $L_{J, a_S}$ and $L_{K, a_T}$
are, respectively, 
$
v_{J, a_S}
$
and 
$
v_{K, a_T}
$
\citep{Morgan2012}.  
From \eqref{eq:crse} and Theorem \ref{thm:dist}, the asymptotic variances of $\sqrt{n}( \hat{\tau} - \tau)$ under the CRSE and ReSEM are, respectively,  
$V_{\tau\tau}$ and
$
V_{\tau \tau}\{ 1 -(1-v_{J,a_S})R_S^2-(1- v_{K,a_T})R_T^2 \}. 
$
The following corollary summarizes the gain from ReSEM on asymptotic variance of the difference-in-means estimator.

\begin{corollary} \label{corollary:PRIASV}
Under Condition \ref{cond:fp} and ReSEM, 
compared to the CRSE, 
the percentage reduction in asymptotic variance (PRIAV) of $\sqrt{n}(\hat{\tau} - \tau)$ is 
$(1-v_{J, a_S}) R_S^2 + (1-v_{K, a_T})R_T^2 \ge 0$. 
\end{corollary}

As discussed in Section \ref{sec:R2}, $R_S^2$ and $R_T^2$ measure the associations between the potential outcomes and covariates at the sampling and treatment assignment stages, respectively. 
From Corollary \ref{corollary:PRIASV}, 
rerandomization at both sampling and treatment assignment stages can help reduce the variability of the treatment effect estimator, and the amount of percentage reduction is nondecreasing and additive in the outcome-covariate associations $R^2_S$ and $R^2_T$. 
By the fact that $v_{J, \infty} = v_{K,\infty} =1$, 
the PRIAV is $(1-v_{K, a_T})R_T^2$ or $(1-v_{J, a_S}) R_S^2$ if there is no rerandomization at the sampling or treatment assignment stage. 
Therefore, the gain from the two stages of ReSEM is additive. 

\begin{rmk}\label{rmk:loss_comp_optimal}
From Corollary \ref{corollary:PRIASV}, the PRIAV is nonincreasing in the thresholds $a_S$ and $a_T$.
However, this does not mean that we should use as small thresholds as possible. 
This is because
the asymptotics in Corollary \ref{corollary:PRIASV} requires fixed positive thresholds that do not vary with the sample size. 
As discussed in Section \ref{sec:balance}, too small thresholds may result in few acceptable sampling and treatment assignment, making the asymptotic approximation inaccurate in finite samples. 
Note that 
the difference between the PRIAV from a particular choice of $(a_S, a_T)$ and the ideal optimal PRIAV that we expect with almost zero thresholds is $v_{J, a_S} R_S^2 + v_{K, a_T}R_T^2 \le \max\{ v_{J, a_S}, v_{K, a_T}\}$. 
Thus, if we choose $a_S$ and $a_T$ to be small enough such that $\max\{ v_{J, a_S}, v_{K, a_T}\}$ is less than, say, $5\%$, then the additional gain we can achieve with more balanced covariates will be small, at the cost of additional computation burden and less precise asymptotic approximation. 
Therefore, in practice, we suggest small, but not overly small, thresholds for the covariate balance criteria. 
\end{rmk}

Third, because the asymptotic distribution \eqref{eq:dist} for $\hat{\tau}$ under ReSEM is non-Gaussian in general, 
the sampling variance does not characterize the full sampling distribution. 
We further study quantile ranges of $\hat{\tau}$ under ReSEM, 
due to their close connection to the confidence intervals. 
Note that the asymptotic distribution  \eqref{eq:dist} is symmetric and unimodal around zero \citep{rerand2018}. 
We focus only on the symmetric quantile ranges, because they always have the shortest lengths given any coverage level \citep[][Theorem 9.3.2]{casella2002statistical}. 
For any $\xi \in (0,1)$, 
let $\nu_{\xi}(R_S^2,R_T^2)$ be the $\xi$th quantile of 
$
(1-R_S^2-R_T^2)^{1/2}\cdot \varepsilon  + R_S \cdot L_{J,a_S} + R_T \cdot
L_{K,a_T}, 
$
and $z_\xi = \nu_{\xi}(0,0)$ be the $\xi$th quantile of a standard Gaussian distribution. 
The following corollary demonstrates the improvement from ReSEM on reducing the lengths of quantiles ranges of the treatment effect estimator. 

\begin{corollary} \label{corollary:QR}
Under Condition \ref{cond:fp}, the $1-\alpha$ symmetric quantile range of the asymptotic distribution of $\sqrt{n}(\hat\tau - \tau)$ under ReSEM is narrower than or equal to that under the CRSE, 
and the percentage reduction in the length of asymptotic $1-\alpha$ symmetric quantile range is $1 - \nu_{1-\alpha/2}(R_S^2,R_T^2)/z_{1-\alpha/2}$, which is 
nondecreasing in both $R_S^2$ and $R_T^2$.
\end{corollary}

From Corollary \ref{corollary:QR},
the stronger the association between potential outcomes and covariates,  
the more reduction in quantile ranges we will have when using ReSEM rather than the CRSE.
Moreover, ignoring rerandomization at either the sampling or assignment stages 
will lead to efficiency loss.

\section{Covariate Adjustment for ReSEM}\label{sec:adj_resem}

\subsection{Covariate imbalance in analysis and regression adjustment}

After conducting the actual experiment, there may still remain some covariate imbalance,  especially when we are able to observe more covariates than that before the experiment. 
Carefully adjusting the covariate imbalance can further improve the efficiency for treatment effect estimation. 
Let $\bs{E}$ denote the available covariate vector for all $N$ units and $\bs{C}$ denote the available covariate vector for the sampled units in $\mathcal{S}$, at the analysis stage after conducting the experiment.
Obviously, $\bs{E} \subset \bs{C}$, but they may not be equal. 
Similar to 
Section \ref{sec:balance}, 
the covariate imbalance for $\bs{E}$ can be characterized by the difference-in-means between sampled units and the whole population,
and that for $\bs{C}$ can be characterized by the difference-in-means between treatment and control groups. 
Let 
$\bar{\bs{E}}$
and $\bar{\bs{C}}$
be the average covariates for the whole population, 
$\bar{ \bs{E}}_{\mathcal{S}}$ 
and 
$\bar{ \bs{C}}_{\mathcal{S}}$
be the average covariates for the sampled units,
and 
$\bar{\bs{C}}_1$
and 
$\bar{\bs{C}}_0$
be the average covariates for treatment and control groups. 
The covariate imbalance with respect to $\bs{E}$ and $\bs{C}$ can then be characterized by the following two difference-in-means:
$\hat{\bs{\delta}}_{\bs{E}} = \bar{ \bs{E}}_{\mathcal{S}} - \bar{\bs{E}}$ and
$\hat{\bs{\tau}}_{\bs{C}} = \bar{\bs{C}}_1 - \bar{\bs{C}}_0$, 
and a general linearly regression-adjusted estimator has the following form: 
\begin{align}\label{eq:reg}
    \hat{\tau}(\bs{\beta}, \bs{\gamma})
    & = 
    \hat{\tau} - \bs{\beta}^\top \hat{\bs{\tau}}_{\bs{C}} 
    -
    \bs{\gamma}^\top \hat{\bs{\delta}}_{\bs{E}}, 
\end{align}
where 
$\bs{\beta}$ and $\bs{\gamma}$ are the adjustment coefficients for the covariate imbalance in analysis. 
From \eqref{eq:reg}, the regression-adjusted estimator is essentially the difference-in-means estimator $\hat{\tau}$ in \eqref{eq:diff_outcome} linearly adjusted by the difference-in-means of covariates $\hat{\bs{\tau}}_{\bs{C}}$ and $\hat{\bs{\delta}}_{\bs{E}}$. 
Moreover, 
\eqref{eq:reg} generalizes the usual covariate adjustment in classical randomized experiments \citep[see, e.g.,][]{lin2013, li2019rerandomization} that focuses mainly on adjusting covariate imbalance between two treatment groups, and considers explicitly the covariate imbalance between experimental units and the whole population of interest.
When both $\bs{\beta}$ and $\bs{\gamma}$ are zero vectors, the regression-adjusted estimator $\hat{\tau}(\bs{0}, \bs{0})$ reduces to the difference-in-means $\hat{\tau}$. 
Then a natural question to ask is:  what is the optimal choice of the regression adjustment coefficients $\bs{\beta}$ and $\bs{\gamma}$? 
We will answer this question in the remaining of this section. 

\begin{rmk}
To actually implement the regression adjustment and construct 
variance estimators and confidence intervals as discussed in Section \ref{sec:large_sample_CI}, it suffices to know the population mean $\bar{\bs{E}}$ and the covariate values for sampled units.  
This can be particularly useful when we have, say, some census information on population averages, while lack exact covariate information for each individual.
\end{rmk}

\subsection{Asymptotic distribution of the regression-adjusted estimator}\label{sec:asym_reg_adj}

We study the sampling distribution of a general regression-adjusted estimator in \eqref{eq:reg} under ReSEM, 
allowing general relationship among the covariates in design and analysis (i.e., $\bs{W}, \bs{X}, \bs{E}$ and $\bs{C}$).
We first extend Condition \ref{cond:fp} to include the covariates $\bs{E}$ and $\bs{C}$ in analysis. 

\begin{condition}\label{cond:fp_analysis}
	Condition \ref{cond:fp} holds, 
	and it still holds with covariates $(\bs{W}, \bs{X})$ replaced by $(\bs{E},\bs{C})$. 
\end{condition}

For any adjustment coefficients $\bs{\beta}$ and $\bs{\gamma}$, 
we introduce adjusted potential outcomes $Y_i(1; \bs{\beta}, \bs{\gamma})$ and $Y_i(0; \bs{\beta}, \bs{\gamma})$ for all $N$ units as follows: 
\begin{align}\label{eq:adj_potential_outcome}
    Y_i(t; \bs{\beta}, \bs{\gamma}) 
    & = 
    Y_i(t) - \bs{\beta}^\top \bs{C}_i - (-1)^{t-1} r_t \bs{\gamma}^\top  (\bs{E}_i - \bar{\bs{E}}), 
    \qquad ( t=0, 1;\  1\le i \le N)
\end{align}
recalling that $r_1$ and $r_0$ are proportions of sampled units assigned to treatment and control. 
The corresponding observed adjusted outcome for sampled unit $i \in \mathcal{S}$ then simplifies to 
\begin{align}\label{eq:adj_obs_outcome}
    Y_i(\bs{\beta}, \bs{\gamma}) = T_i Y_i(1; \bs{\beta}, \bs{\gamma}) + (1-T_i) Y_i(0; \bs{\beta}, \bs{\gamma}) = 
    \begin{cases}
    Y_i - \bs{\beta}^\top \bs{C}_i - r_1 \bs{\gamma}^\top  (\bs{E}_i - \bar{\bs{E}}), & \text{if } T_i = 1, \\
    Y_i - \bs{\beta}^\top \bs{C}_i + r_0 \bs{\gamma}^\top (\bs{E}_i - \bar{\bs{E}}), & \text{if } T_i = 0.
    \end{cases}
\end{align}
We can verify that
(i) the average treatment effect for adjusted potential outcomes is the same as that for original potential outcomes,  
and 
(ii)
the difference-in-means estimator based on adjusted observed outcomes
is equivalently the regression-adjusted estimator in \eqref{eq:reg}. 
Thus, 
from Theorem \ref{thm:dist}, we can immediately derive the asymptotic distribution of the regression-adjusted estimator. 
By the same logic as \eqref{eq:R2} and \eqref{eq:V_tautau}, 
the variance of $\sqrt{n} \{ \hat{\tau}(\bs{\beta}, \bs{\gamma}) - \tau \} $ under the CRSE is
\begin{align}\label{eq:V_tau_reg}
    V_{\tau \tau}(\bs{\beta}, \bs{\gamma}) = r_1^{-1}S_1^2(\bs{\beta}, \bs{\gamma}) + r_0^{-1}S_0^2(\bs{\beta}, \bs{\gamma}) -f S_{\tau}^2(\bs{\beta}, \bs{\gamma}),
\end{align}
where $S_1^2(\bs{\beta}, \bs{\gamma}), S_0^2(\bs{\beta}, \bs{\gamma})$ and $S_{\tau}^2(\bs{\beta}, \bs{\gamma})$ denote the finite population variances of adjusted potential outcomes and individual treatment effects, 
and the squared multiple correlations between the regression-adjusted estimator $\hat{\tau}(\bs{\beta}, \bs{\gamma})$ and the difference-in-means of covariates $\hat{\bs{\delta}}_{\bs{E}}$ and $\hat{\bs{\tau}}_{\bs{C}}$ 
are 
\begin{equation}\label{eq:R2_reg}
	R_S^2(\bs{\beta}, \bs{\gamma}) = 
	\frac{(1-f)S^2_{\tau \mid  \bs{W}}(\bs{\beta}, \bs{\gamma})}{V_{\tau \tau}(\bs{\beta}, \bs{\gamma})}, 
	\ \ 
	R_T^2(\bs{\beta}, \bs{\gamma}) = 
	\frac{r_1^{-1}S^2_{1 \mid  \bs{X}}(\bs{\beta}, \bs{\gamma})+r_0^{-1}S^2_{0 \mid  \bs{X}}(\bs{\beta}, \bs{\gamma}) -S^2_{\tau \mid  \bs{X}}(\bs{\beta}, \bs{\gamma})}{V_{\tau \tau}(\bs{\beta}, \bs{\gamma})}, 
\end{equation}
where $S^2_{1 \mid  \bs{X}}(\bs{\beta}, \bs{\gamma})$, $S^2_{0 \mid  \bs{X}}(\bs{\beta}, \bs{\gamma})$,  $S^2_{\tau \mid  \bs{X}}(\bs{\beta}, \bs{\gamma})$ and $S^2_{\tau \mid  \bs{W}}(\bs{\beta}, \bs{\gamma})$ are the finite population variances of the linear projections of adjusted potential outcomes and individual effect on covariates. 
Recall that $\varepsilon\sim \mathcal{N}(0,1)$, and $L_{J,a_S}$ and $L_{K,a_T}$ are two constrained Gaussian random variables 
as in 
Theorem \ref{thm:dist}.

\begin{theorem}\label{thm:dist_reg_general}
Under Condition \ref{cond:fp_analysis} and ReSEM, 
\begin{align}\label{eq:dist_reg_general}
& \quad \ \sqrt{n} \left\{ \hat{\tau}(\bs{\beta}, \bs{\gamma}) - \tau \right\} \mid  \text{ReSEM}
\nonumber
\\
& \ \dot\sim\  
V_{\tau\tau}^{1/2}(\bs{\beta}, \bs{\gamma}) \Big(  \sqrt{1-R_S^2(\bs{\beta}, \bs{\gamma}) - R_T^2(\bs{\beta}, \bs{\gamma})} \cdot \varepsilon
+ \sqrt{R_S^2(\bs{\beta}, \bs{\gamma})} \cdot L_{J,a_S}
+
\sqrt{R_T^2(\bs{\beta}, \bs{\gamma})} \cdot L_{K,a_T} 
\Big). 
\end{align}
\end{theorem}

Theorem \ref{thm:dist_reg_general}, although provides the asymptotic distribution of the regression-adjusted estimator under ReSEM, does not characterize the role of adjustment coefficients in an obvious way. 
The following corollary gives an equivalent form of the asymptotic distribution in \eqref{eq:dist_reg_general}, 
which indicates that the asymptotic distribution of $\hat{\tau}(\bs{\beta}, \bs{\gamma})$ depends on the adjustment coefficients $\bs{\beta}$ and $\bs{\gamma}$ in an additive way. 
In other words, the dependence on $\bs{\beta}$ and $\bs{\gamma}$ is separable.

\begin{corollary}\label{cor:dist_reg_general_equ}
Under Condition \ref{cond:fp_analysis} and ReSEM, 
\begin{align}\label{eq:dist_reg_general_equ}
\sqrt{n} \left\{ \hat{\tau}(\bs{\beta}, \bs{\gamma}) - \tau \right\} \mid  \text{ReSEM}   
& 
\ \dot\sim \  
\sqrt{
V_{\tau\tau}(\bs{0}, \bs{\gamma})\{1 - R_S^2(\bs{0}, \bs{\gamma})\}
+ 
V_{\tau\tau}(\bs{\beta}, \bs{0})\{1 - R_T^2(\bs{\beta}, \bs{0}) \}
-V_{\tau\tau}
} \cdot \varepsilon
\nonumber
\\
& \quad \ + \sqrt{ V_{\tau\tau}(\bs{0}, \bs{\gamma}) R_S^2(\bs{0}, \bs{\gamma})} \cdot L_{J,a_S}
+
\sqrt{V_{\tau\tau}(\bs{\beta}, \bs{0})R_T^2(\bs{\beta}, \bs{0})} \cdot L_{K,a_T}. 
\end{align}
\end{corollary}

Corollary \ref{cor:dist_reg_general_equ} is intuitive given that both $\bs{\gamma}$ and 
$L_{J, a_S}$ relate to 
covariate balance between sampled and all units, 
while both $\bs{\beta}$ and 
$L_{K,a_T}$ relate to 
covariate balance between treated and control units. 
The equivalent expression in \eqref{eq:dist_reg_general_equ} also shows that the study of optimal choice of $\bs{\beta}$ and $\bs{\gamma}$ can be separated. 
However, in general scenario where there lacks communication between designer and analyzer, 
studying the optimal regression adjustment can be quite challenging \citep{li2019rerandomization}.
In particular, it 
may not be achievable based on the analyzer's observed data.  
Therefore, 
in the next few subsections, 
we will focus on the special case where the analyzer can observe the covariate information at the sampling stage (i.e., $\bs{W} \subset \bs{E}$) or treatment assignment stage (i.e., $\bs{X} \subset \bs{C}$).

\subsection{Squared multiple correlations and finite population least squares coefficients}\label{sec:R2_proj_coef}

Before going to details for optimal regression adjustment, we first introduce some finite population quantities.  
Recall that $R_S^2$ and $R_T^2$
in \eqref{eq:R2} characterize the associations between potential outcomes and covariates at the sampling and assignment stages, respectively. 
We analogously define $R_E^2$ and $R_C^2$ to denote the associations between potential outcomes and covariates at the analysis stage: 
\begin{align}\label{eq:R_A2}
    R_E^2 
    = \frac{(1-f)S^2_{\tau \mid  \bs{E}}}{r_1^{-1}S^2_1 +r_0^{-1}S^2_0 -fS^2_\tau},
    \quad 
	\text{and}
	\quad
	R_C^2  = 
	\frac{r_1^{-1}S^2_{1 \mid  \bs{C}}+r_0^{-1}S^2_{0 \mid  \bs{C}} -S^2_{\tau \mid  \bs{C}}}{r_1^{-1}S^2_1 +r_0^{-1}S^2_0 -fS^2_\tau},
\end{align}
where $S^2_{1 \mid  \bs{C}}, S^2_{0 \mid  \bs{C}}$, $S^2_{\tau \mid  \bs{C}}$ and $S^2_{\tau \mid  \bs{E}}$ are the finite population variances of linear projections of potential outcomes and individual treatment effect on covariates $\bs{C}$ and $\bs{E}$. 
For $t=0,1$, define 
$\tilde{\bs{\beta}}_t$ and $\tilde{\bs{\gamma}}_t$ as the finite population 
linear projection 
coefficients of potential outcome $Y(t)$ on covariates $\bs{C}$ and $\bs{E}$, respectively, i.e., 
$
\tilde{\bs{\beta}}_t = (  \bs{S}^2_{\bs{C}} )^{-1} \bs{S}_{\bs{C}, t}
$
and 
$\tilde{\bs{\gamma}}_t = (  \bs{S}^2_{\bs{E}} )^{-1} \bs{S}_{\bs{E}, t}$, 
where 
$\bs{S}_{\bs{C}}^2$ and $\bs{S}_{\bs{E}}^2$ are the finite population covariance matrices of the covariates $\bs{C}$ and $\bs{E}$, 
and 
$\bs{S}_{\bs{C},t}$ and $\bs{S}_{\bs{E},t}$ are the finite population covariances between the covariates $\bs{C}$ and $\bs{E}$ and the potential outcome $Y(t)$. 
We further define 
$
\tilde{\bs{\beta}} \equiv r_0 \tilde{\bs{\beta}}_1 + r_1 \tilde{\bs{\beta}}_0
$
and 
$
\tilde{\bs{\gamma}} \equiv \tilde{\bs{\gamma}}_1 - \tilde{\bs{\gamma}}_0, 
$
both of which are linear combinations of the finite population 
linear projection 
coefficients. 
As demonstrated in the Supplementary Material, 
$\tilde{\bs{\beta}}$ 
and 
$\tilde{\bs{\gamma}}$ are indeed the 
linear projection 
coefficients of $\hat{\tau}$ on $\hat{\bs{\tau}}_{\bs{C}}$ and $\hat{\bs{\delta}}_{\bs{E}}$, respectively, under the CRSE.

\subsection{Optimal adjustment when there is more covariate information in analysis}\label{sec:optimal_ana_more}

In this subsection, we focus on the case in which the covariate information in analysis contains those in sampling or treatment assignment, in the sense that 
$\bs{W} \subset \bs{E}$ or $\bs{X} \subset \bs{C}$.
In this case, 
the asymptotic distribution \eqref{eq:dist_reg_general_equ} for a general regression-adjusted estimator $\hat{\tau}(\bs{\beta}, \bs{\gamma})$ can be further simplified. 
Define $\bs{S}^2_{\bs{E}\setminus\bs{W}} \equiv \bs{S}^2_{\bs{E}} - \bs{S}^2_{\bs{E}\mid\bs{W}}$ as the finite population variance of the residual from the linear projection of $\bs{E}$ on $\bs{W}$, and analogously $\bs{S}^2_{\bs{C}\setminus\bs{X}} \equiv \bs{S}^2_{\bs{C}} - \bs{S}^2_{\bs{C}\mid\bs{X}}$. 
When $\bs{W} \subset \bs{E}$, the two terms involving $\bs{\gamma}$ in the coefficients of $\varepsilon$ and $L_{J, a_S}$ in  \eqref{eq:dist_reg_general_equ} simplify to, respectively, 
\begin{align}\label{eq:simp_exp_reg_sampling}
    V_{\tau\tau}(\bs{0}, \bs{\gamma})\{1 - R_S^2(\bs{0}, \bs{\gamma})\}
    & = 
    V_{\tau\tau} (1-R_E^2) + 
    (1-f) (\bs{\gamma} - \tilde{\bs{\gamma}})^\top \bs{S}^2_{\bs{E} \setminus \bs{W}} (\bs{\gamma} - \tilde{\bs{\gamma}}), 
    \nonumber
    \\
    V_{\tau\tau}(\bs{0}, \bs{\gamma})R_S^2(\bs{0}, \bs{\gamma})
    & = (1-f) (\bs{\gamma} - \tilde{\bs{\gamma}})^\top \bs{S}^2_{\bs{E} \mid \bs{W}} (\bs{\gamma} - \tilde{\bs{\gamma}}), 
\end{align}
both of which are minimized at $\bs{\gamma} = \tilde{\bs{\gamma}}$. 
When $\bs{X} \subset \bs{C}$, the two terms involving $\bs{\beta}$ in the coefficients of $\varepsilon$ and $L_{K, a_T}$ in \eqref{eq:dist_reg_general_equ} simplify to, respectively, 
\begin{align}\label{eq:simp_exp_reg_assignment}
    V_{\tau\tau}(\bs{\beta}, \bs{0})\{1 - R_T^2(\bs{\beta}, \bs{0}) \}
    & = 
    V_{\tau\tau} (1-R_C^2) +  (r_1 r_0)^{-1} (\bs{\beta} - \tilde{\bs{\beta}})^\top \bs{S}^2_{\bs{C} \setminus \bs{X}} (\bs{\beta} - \tilde{\bs{\beta}}), 
    \nonumber
    \\
    V_{\tau\tau}(\bs{\beta}, \bs{0}) R_T^2(\bs{\beta}, \bs{0})
    & = 
    (r_1 r_0)^{-1} (\bs{\beta} - \tilde{\bs{\beta}})^\top \bs{S}^2_{\bs{C} \mid \bs{X}} (\bs{\beta} - \tilde{\bs{\beta}}),
\end{align}
both of which are minimized at $\bs{\beta} = \tilde{\bs{\beta}}$. 
These 
immediately 
imply the optimal choice of $\bs{\gamma}$ or $\bs{\beta}$ when $\bs{W} \subset \bs{E}$ or  $\bs{X} \subset \bs{C}$, 
as summarized in the theorem below. 
Specifically, we say 
an estimator is 
$\mathcal{S}$-optimal 
among a class of estimators 
if it has the shortest asymptotic $1-\alpha$ symmetric quantile range for all $\alpha\in (0, 1)$, where we use $\mathcal{S}$ to refer to sampling precision. 

\begin{theorem}\label{thm:optimal}
    Under ReSEM and Condition \ref{cond:fp_analysis}, 
    among all regression-adjusted estimators of form \eqref{eq:reg}, 
    (i) if $\bs{W} \subset \bs{E}$, then 
    $\hat{\tau}(\bs{\beta}, \tilde{\bs{\gamma}})$ is $\mathcal{S}$-optimal 
    for any given $\bs{\beta}$; 
    (ii) if $\bs{X} \subset \bs{C}$, then 
    $\hat{\tau}(\tilde{\bs{\beta}}, \bs{\gamma})$ is $\mathcal{S}$-optimal 
    for any given $\bs{\gamma}$; 
    (iii) if $\bs{W} \subset \bs{E}$ and $\bs{X} \subset \bs{C}$, then $\hat{\tau}(\tilde{\bs{\beta}}, \tilde{\bs{\gamma}})$ is $\mathcal{S}$-optimal 
    with the following asymptotic distribution: 
	\begin{align}\label{eq:reg_opt}
		\sqrt{n} \big\{ \hat{\tau}( \tilde{\bs{\beta}}, \tilde{\bs{\gamma}}) - \tau \big\} \mid \text{ReSEM} 
		\ \dot\sim \ 
		V_{\tau\tau}^{1/2} 
		\sqrt{ 1 - R_E^2 - R_C^2 }
		\cdot \varepsilon.
	\end{align}
\end{theorem}

For conciseness, we relegate the expressions for the asymptotic distributions of 
$\hat{\tau}(\bs{\beta}, \tilde{\bs{\gamma}})$ and $\hat{\tau}(\tilde{\bs{\beta}}, \bs{\gamma})$ to the Supplementary Material. 
From Theorem \ref{thm:optimal}, 
when we have more covariate information for measuring imbalance between treatment and control groups or between sampled units and the whole population, 
the corresponding optimal adjustment is achieved at the linear projection coefficient $\tilde{\bs{\beta}}$ or $\tilde{\bs{\gamma}}$. 
When the analyzer has all the covariate information in design, 
$\hat{\tau}( \tilde{\bs{\beta}}, \tilde{\bs{\gamma}})$ achieves the optimal precision.  
Note that both $\tilde{\bs{\beta}}$ and $\tilde{\bs{\gamma}}$ depend on all the potential outcomes and are thus generally unknown in practice. 
As discussed shortly in Section \ref{sec:est_beta_gamma_tilde}, 
we can consistently estimate both coefficients and still obtain the optimal efficiency in \eqref{eq:reg_opt}. 

\begin{rmk}\label{rmk:special_covadj}
By the same logic as Remark \ref{rmk:special_dim}, 
Theorem \ref{thm:optimal} can imply the optimal regression adjustment for the CRSE, ReM and rejective sampling, with details relegated to the Supplementary Material. 
For example, under the CRSE, 
the $\mathcal{S}$-optimal regression-adjusted estimator is attainable at $(\bs{\beta}, \bs{\gamma}) = (\tilde{\bs{\beta}}, \tilde{\bs{\gamma}})$, with the 
same 
asymptotic distribution
as \eqref{eq:reg_opt}. 
\end{rmk}

\subsection{Improvements from design and  analysis}\label{sec:improve_reg}

In general scenarios without any restriction on the relation between covariates in design and analysis, 
the analyzer may not be able to conduct optimal adjustment, 
and 
the regression-adjusted estimator $\hat{\tau}( \tilde{\bs{\beta}}, \tilde{\bs{\gamma}})$ may be less precise than the unadjusted $\hat{\tau}$ \citep[][]{li2019rerandomization}. 
However, ReSEM never hurts and can generally improve the precision of the regression-adjusted estimator, regardless of the choice of adjustment coefficients.
Thus, ReSEM should always be preferred in practice.
In the following discussion, 
we will focus on the special case where the analyzer has more covariate information, i.e, $\bs{X} \subset \bs{C}$ and $\bs{W} \subset \bs{E}$, 
and investigate the additional gain in sampling precision from the design and analysis, separately. 
Specifically, 
we first measure the additional gain from rerandomization given that we use optimal regression adjustment in analysis,
and then measure the additional gain from optimal regression adjustment given that we use rerandomization in design. 

First, to measure the additional gain from design, 
we compare the asymptotic distribution of the optimal regression-adjusted estimator under ReSEM to that under the CRSE. 
From Theorem \ref{thm:optimal} and Remark \ref{rmk:special_covadj}, 
the asymptotic distributions 
for the $\mathcal{S}$-optimal regression-adjusted estimators under ReSEM and the CRSE are the same. 
Thus, there is no additional gain from rerandomization when the analyzer uses optimal regression adjustment. 
This is not surprising given that more covariates have been adjusted at the analysis stage. 
Moreover, this also implies that rerandomization will not hurt the precision of the $\mathcal{S}$-optimal adjusted estimator, and the covariates in analysis can be used in the same way as that under the complete randomization.

Second, to measure the additional gain from analysis, we compare the asymptotic distribution of the $\mathcal{S}$-optimal regression-adjusted estimator $\hat{\tau}(\tilde{\bs{\beta}}, \tilde{\bs{\gamma}})$ under ReSEM to that of the unadjusted estimator $\hat{\tau}$. 
From Theorems \ref{thm:dist} and \ref{thm:optimal}, both asymptotic distributions \eqref{eq:dist} for $\hat{\tau}$ and \eqref{eq:reg_opt} for $\hat{\tau}(\tilde{\bs{\beta}}, \tilde{\bs{\gamma}})$ can be viewed as linear combinations of Gaussian and constrained Gaussian random variables. 
Moreover, all the coefficients in \eqref{eq:reg_opt} for the  $\mathcal{S}$-optimal adjusted  estimator are less than or  equal to 
that 
in \eqref{eq:dist} for the unadjusted estimator, and their differences are nondecreasing in $R_E^2$ and $R_C^2$, which measure the associations between potential outcomes and covariates in analysis. 
Therefore, the optimal adjustment improves the precision of the treatment effect estimation, with the improvement being nondecreasing in $R_E^2$ and $R_C^2$, as summarized in the following corollary. 

\begin{corollary}\label{cor:gain_analysis}
	Under Condition \ref{cond:fp_analysis} and ReSEM with $\bs{W} \subset \bs{E}$ and $\bs{X} \subset \bs{C}$, 
	compared to the unadjusted $\hat{\tau}$, 
	for 
	any 
	$\alpha\in (0,1)$, 
	the percentage reductions in asymptotic variance and length of asymptotic $1-\alpha$ symmetric quantile range of the $\mathcal{S}$-optimal adjusted estimator $\hat{\tau}(\tilde{\bs{\beta}}, \tilde{\bs{\gamma}})$ are, respectively, 
	\begin{align*}
		\frac{(R_E^2 - R_S^2) + (R_C^2 - R_T^2) + v_{J,a_S}R_S^2 + v_{K, a_T}R_T^2 }{1 - (1-v_{J,a_S}) R_S^2 - (1-v_{K, a_T}) R_T^2} 
		\quad
		\text{and}
		\quad
		1 - 
		\sqrt{1-R_E^2-R_C^2} \cdot \frac{z_{1-\alpha/2}}
		{
			\nu_{1-\alpha/2}(R_S^2, R_T^2)}. 
	\end{align*}
\end{corollary}

In Corollary \ref{cor:gain_analysis}, both percentage reductions are nonnegative and nondecreasing in $R_E^2$ and $R_C^2$. 
In particular, when $R_E^2 + R_C^2$ increases to $1$, 
they both 
become close to 1. 
Thus, with more covariate information in analysis, regression adjustment can provide a substantial gain in estimation precision. 

Below we consider a special case where there is 
no additional covariates in analysis (i.e., $\bs{E} = \bs{W}$ and $\bs{C} = \bs{X}$), 
and the thresholds $a_S$ and $a_T$ for rerandomization are small. 
In this case, 
both percentage reductions in Corollary \ref{cor:gain_analysis} become close to zero, which implies that 
the $\mathcal{S}$-optimal adjusted estimator and the unadjusted one have almost the same precision. 
Therefore, rerandomization and regression adjustment are essentially dual of each other, both of which try to adjust the covariate imbalance, while one is at the design stage and the other is at the analysis stage. 
Our asymptotic comparison between them sheds some light on the ongoing debate for whether achieving covariate balance in design is preferable to ex post adjustment \citep[see, e.g.,][]{Morgan2012, Miratrix2013}. 
It will be interesting to further investigate their properties in finite samples. In general, 
rerandomization 
with the unadjusted estimator 
provides more transparent analysis and helps avoid data snooping  \citep{cox2007, Freedman2008chance, Rosenbaum:2010, lin2013}.

\section{Large-sample Confidence Intervals for ReSEM}\label{sec:large_sample_CI}

\subsection{Estimation of the finite population least squares coefficients}\label{sec:est_beta_gamma_tilde}

As discussed in Section \ref{sec:optimal_ana_more}, 
when there is more covariate information in analysis, 
the optimal regression adjustment coefficients are $(\tilde{\bs{\beta}}, \tilde{\bs{\gamma}})$, which depend on the  potential outcomes and covariates of all units and are generally unknown in practice. 
Thus, we propose to estimate them using their sample analogues. 
Specifically, 
for $t=0,1$, 
let $\hat{\bs{\beta}}_t$ and $\hat{\bs{\gamma}}_t$ be the least squares coefficients from the linear projections of the observed outcome on covariates $\bs{C}$ and $\bs{E}$ in the treatment group $t$.  
Define 
$\hat{\bs{\beta}} \equiv r_0 \hat{\bs{\beta}}_1 + r_1 \hat{\bs{\beta}}_0$
and 
$\hat{\bs{\gamma}} \equiv \hat{\bs{\gamma}}_1 - \hat{\bs{\gamma}}_0$. 
The following theorem shows that $\hat{\bs{\beta}}_t$ and $\hat{\bs{\gamma}}_t$ are consistent for their population analogues $\tilde{\bs{\beta}}_t$ and $\tilde{\bs{\gamma}}_t$, 
and the regression-adjusted estimator using the estimated coefficients $\hat{\tau}(\hat{\bs{\beta}}, \hat{\bs{\gamma}})$ has the same asymptotic distribution as $\hat{\tau}(\tilde{\bs{\beta}}, \tilde{\bs{\gamma}})$. 

\begin{theorem}\label{thm:plug_in}
	Under Condition \ref{cond:fp_analysis} and ReSEM, 
	$\hat{\bs{\beta}}_t - \tilde{\bs{\beta}}_t = o_{\Pr}(1)$ 
	and 
	$\hat{\bs{\gamma}}_t - \tilde{\bs{\gamma}}_t = o_{\Pr}(1)$ 
	for $t=0,1$, 
	and 
	$
	\hat{\tau}(\hat{\bs{\beta}}, \hat{\bs{\gamma}})
	-
	\hat{\tau}(\tilde{\bs{\beta}}, \tilde{\bs{\gamma}})
	= o_{\Pr}(n^{-1/2}).
	$
	Consequently, under ReSEM, 
	$
	\sqrt{n}\{
	\hat{\tau}(\hat{\bs{\beta}}, \hat{\bs{\gamma}}) - \tau
	\}
	$
    follows the same asymptotic distribution as 
    $
	\sqrt{n}\{
	\hat{\tau}(\tilde{\bs{\beta}}, \tilde{\bs{\gamma}}) - \tau
	\}. 
	$
\end{theorem}

From Theorems \ref{thm:optimal} and \ref{thm:plug_in}, 
when the analyzer has access to all covariate information in design, 
we are able to achieve the optimal efficiency among estimators in \eqref{eq:reg} using the observed data. 
As a side note, similar to \citet{lin2013}, the regression-adjusted estimator $\hat{\tau}(\hat{\bs{\beta}}, \hat{\bs{\gamma}})$ is closely related to the least squares estimator from a linear regression model of observed outcome on treatment indicator, covariates and their interaction. 
We relegate the detailed discussion to the Supplementary Material. 

\subsection{Variance estimation and confidence intervals} \label{sec:estimate_and_CI}

Below we study the large-sample inference for the average treatment effect $\tau$ based on a general regression-adjusted estimator $\hat{\tau}(\bs{\beta}, \bs{\gamma})$, which includes the difference-in-means $\hat{\tau} = \hat{\tau}(\bs{0}, \bs{0})$ as a special case. 
From Theorem \ref{thm:dist_reg_general}, to construct variance estimate for $\hat{\tau}(\bs{\beta}, \bs{\gamma})$ and confidence intervals for $\tau$, 
it suffices to estimate 
$V_{\tau \tau}(\bs{\beta}, \bs{\gamma})$, 
$R_S^2(\bs{\beta}, \bs{\gamma})$ and $R_T^2(\bs{\beta}, \bs{\gamma})$.
From \eqref{eq:R2} and \eqref{eq:V_tautau}, these quantities depend on the finite population variances of the adjusted potential outcomes and their projections on covariates, 
which can be estimated by 
their sample analogues.

For units in treatment group $t\in \{0,1\}$,  
let $s^2_t(\bs{\beta}, \bs{\gamma})$ and  $\bs{s}^2_{\bs{X}}(t)$ be the sample variance and covariance matrix of the observed adjusted outcomes and covariates, 
$\bs{s}_{t,\bs{X}}(\bs{\beta}, \bs{\gamma})$ 
be the sample covariance between them, 
and $s^2_{t\mid \bs{X}}(\bs{\beta}, \bs{\gamma})$ be the sample variance of the linear projection of the observed adjusted outcome on covariates. 
Let $\bs{s}_{\bs{X}}(t)$ be the 
positive-definite square root of the sample covariance matrix $\bs{s}^2_{\bs{X}}(t)$
and $\bs{s}^{-1}_{\bs{X}}(t)$ be its inverse.
Then an intuitive estimator for the variance of the linear projection of the adjusted individual effect on covariate $S_{\tau\mid \bs{X}}^2(\bs{\beta}, \bs{\gamma})$ is 
\begin{align*}
    s_{\tau\mid \bs{X}}^2(\bs{\beta}, \bs{\gamma}) = 
    \big\| \bs{s}_{1,\bs{X}}(\bs{\beta}, \bs{\gamma}) \cdot \bs{s}^{-1}_{\bs{X}}(1) - \bs{s}_{0,\bs{X}}(\bs{\beta}, \bs{\gamma}) \cdot \bs{s}^{-1}_{\bs{X}}(0) \big\|_2^2, 
\end{align*}
where $\|\cdot\|_2$ denotes the Euclidean norm.
We further define analogously $s_{\tau\mid \bs{C}}^2(\bs{\beta}, \bs{\gamma})$ and $s_{\tau\mid \bs{W}}^2(\bs{\beta}, \bs{\gamma})$. 
As demonstrated in the Supplementary Material, these sample quantities are consistent for their finite population analogues. 
Consequently, 
if all design information is known in analysis, 
we can estimate 
$V_{\tau \tau}(\bs{\beta}, \bs{\gamma})$,  ${R}^2_S(\bs{\beta}, \bs{\gamma})$ and ${R}^2_T(\bs{\beta}, \bs{\gamma})$, respectively, by 
\begin{align}\label{eq:VR2_estimate}
\hat V_{\tau \tau}(\bs{\beta}, \bs{\gamma}) &=
r_1^{-1}s^2_{1}(\bs{\beta}, \bs{\gamma}) + r_0^{-1}s^2_{0}(\bs{\beta}, \bs{\gamma}) - f s_{\tau\mid \bs{C}}^2(\bs{\beta}, \bs{\gamma}), 
\quad \ \ 
\hat{R}^2_S(\bs{\beta}, \bs{\gamma}) =
(1-f) \hat V_{\tau \tau}^{-1} (\bs{\beta}, \bs{\gamma})
s_{\tau\mid \bs{W}}^2(\bs{\beta}, \bs{\gamma}), 
\nonumber
\\
\hat{R}^2_T(\bs{\beta}, \bs{\gamma}) &=
\hat V_{\tau \tau}^{-1}(\bs{\beta}, \bs{\gamma}) \big\{ r_1^{-1}s^2_{1\mid  \bs{X}}(\bs{\beta}, \bs{\gamma}) +r_0^{-1}s^2_{0\mid  \bs{X}}(\bs{\beta}, \bs{\gamma}) -
s_{\tau\mid \bs{X}}^2(\bs{\beta}, \bs{\gamma})
\big\}. 
\end{align}
However, if we do not have access to the design information 
at the sampling or assignment stages, 
we can conservatively underestimate $R_S^2(\bs{\beta}, \bs{\gamma})$ or $R_T^2(\bs{\beta}, \bs{\gamma})$ by 
zero, 
pretending that $\bs{W} = \emptyset$ and $a_S = \infty$ or $\bs{X} = \emptyset$ and $a_T = \infty$. 
This worst-case consideration guarantees that the resulting estimated distribution for the regression-adjusted estimator $\hat{\tau}(\bs{\beta}, \bs{\gamma})$ will be asymptotically conservative.

Finally, we can estimate the sampling variance of $\hat\tau(\bs{\beta}, \bs{\gamma})$ under ReSEM  by 
$n^{-1} \hat V_{\tau \tau} (\bs{\beta}, \bs{\gamma})\{ 1 -(1-v_{J,a_S})\hat R_S^2(\bs{\beta}, \bs{\gamma}) -(1- v_{K,a_T})\hat R_T^2(\bs{\beta}, \bs{\gamma})  \}$, 
and 
use $\hat \tau(\bs{\beta}, \bs{\gamma}) \pm  n^{-1/2} 
\hat{V}_{\tau\tau}^{1/2}(\bs{\beta}, \bs{\gamma})\nu_{1-\alpha/2}(\hat{R}_S^2(\bs{\beta}, \bs{\gamma}), \hat{R}_T^2(\bs{\beta}, \bs{\gamma}))
$
as a $1-\alpha$ confidence interval for $\tau$. 
Besides, we can 
construct variance estimator and confidence intervals for the regression-adjusted estimator $\hat{\tau}(\hat{\bs{\beta}}, \hat{\bs{\gamma}})$ with estimated coefficients
in the same way as in \eqref{eq:VR2_estimate}. 
Both the variance estimator and confidence intervals are asymptotically  conservative, in the sense that the probability limit of the variance estimator is larger than or equal to the true variance, 
and the limit of the coverage probability of the confidence interval is larger than or equal to its nominal level. 
If all the design information is available in analysis, then both the variance estimator and confidence intervals become asymptotically exact when $f\rightarrow 0$ or the individual treatment effects adjusted by covariates are asymptotically additive 
in the sense that $S^2_{\tau \setminus \bs{C}} \equiv  S^2_{\tau} - S^2_{\tau \mid \bs{C}} \rightarrow 0$. 
Interestingly, when $f$ is small, which can often happen when conducting survey experiments on large populations, our inference can be asymptotically exact and does not suffer from conservativeness as in usual finite population causal inference. 
These results are all intuitive, and we relegate the technical details to the Supplementary Material.

\subsection{Estimated precision and the corresponding optimal regression adjustment}\label{sec:est_precision}

Section \ref{sec:adj_resem} studied the sampling distribution of a general regression-adjusted estimator 
and demonstrated the optimality of 
$\hat{\tau}(\tilde{\bs{\beta}}, \tilde{\bs{\gamma}})$ in terms of the sampling precision when $\bs{W} \subset \bs{E}$ and $\bs{X} \subset \bs{C}$. 
Below we focus on the estimated precision of a general regression-adjusted estimator. 
Specifically, we say an estimator is $\mathcal{C}$-optimal if its $1-\alpha$ confidence interval multiplied by $\sqrt{n}$ is asymptotically the shortest among a class of estimators for all $\alpha\in (0,1)$, where we use $\mathcal{C}$ to refer to estimated precision. The following theorem shows the optimal 
adjustment in terms of estimated precision. 

\begin{theorem}\label{thm:est_optimal}
    Under ReSEM and Condition \ref{cond:fp_analysis}, 
    $\hat{\tau}(\tilde{\bs{\beta}}, \tilde{\bs{\gamma}})$ has the same estimated precision as $\hat{\tau}(\hat{\bs{\beta}}, \hat{\bs{\gamma}})$ asymptotically, 
    and it is $\mathcal{C}$-optimal among all regression-adjusted estimators of form \eqref{eq:reg}, regardless of whether the analyzer knows the design information at the sampling or assignment stages.
\end{theorem}

Below we give some intuition of Theorem \ref{thm:est_optimal} for two scenarios under which the analyzer knows all or no design information. 
First, when the analyzer has all the design information, 
then $\bs{W} \subset \bs{E}$ and $\bs{X} \subset \bs{C}$, and 
the estimated distribution of a general adjusted estimator $\hat{\tau}(\bs{\beta}, \bs{\gamma})$ has the same weak limit as the convolution of its asymptotic distribution in \eqref{eq:dist_reg_general} and a Gaussian distribution with mean zero and variance $f S^2_{\tau \setminus \bs{C}}$. 
In this case, 
the $\mathcal{C}$-optimality of $\hat{\tau}(\tilde{\bs{\beta}}, \tilde{\bs{\gamma}})$ follows from its $\mathcal{S}$-optimality under ReSEM as in Theorem \ref{thm:optimal}. 
Second, when the analyzer does not have the design information for sampling or assignment, the estimated distribution of $\hat{\tau}(\bs{\beta}, \bs{\gamma})$ has the same weak limit as the convolution of its asymptotic distribution 
under the CRSE and a Gaussian distribution with mean zero and variance $f S^2_{\tau \setminus \bs{C}}$, which is intuitive since without any design information we consider the worst case and pretend that the experiment was conducted as a CRSE. 
In this case, 
the $\mathcal{C}$-optimality of $\hat{\tau}(\tilde{\bs{\beta}}, \tilde{\bs{\gamma}})$ follows from its $\mathcal{S}$-optimality under the CRSE as discussed in Remark \ref{rmk:special_covadj}.

\section{Illustration}\label{sec:illustrate}

\subsection{A simulation study}
We conduct a simulation to illustrate the property of ReSEM. 
We fix the population size at $N = 10^4$
and vary the sample size $n$ from $100$ to $1000$, i.e., the proportion of sampled units $f$ varies from $0.01$ to $0.1$. 
The proportions of treated and control units are fixed at $r_1 = r_0 = 0.5$. 
We generate the potential outcomes and covariates for all $N$ units as i.i.d.\ samples from the following model:
\begin{align*}
    Y(0) & = - \frac{1}{2} \sum_{k=1}^6 C_k  + \delta, 
    \quad Y(1) = Y(0) + \frac{3}{5} \sum_{k=1}^6 C_k,\\
    \bs{W} & = (C_1, C_2), \quad \bs{X} = (C_1, C_2, C_3, C_4), 
    \quad 
    \bs{E} = \bs{W}, 
    \quad \bs{C} = (C_1, C_2, \ldots, C_6),
    \\
    & \quad \ C_1, C_3, C_5 \sim \text{Bernoulli}(0.5), \quad
    C_2, C_4, C_6 \sim \mathcal{N}(0,1), \quad
    \delta \sim \mathcal{N}(0, 0.1^2),
\end{align*}
where $C_1, C_2, \ldots, C_6$ and $\delta$ are mutually independent. 
Once generated, all the potential outcomes and covariates are kept fixed, mimicking the finite population inference. 
Consistent with our notation before, 
$\bs{W}$ denotes the available covariate information at the sampling stage, 
$\bs{X}$ denotes the available covariate information at the treatment assignment stage, 
and 
$\bs{E}$ and 
$\bs{C}$ denote the available covariate information at the analysis stage, 
with $\bs{W} = \bs{E} \subset \bs{X} \subset \bs{C}$. 
We choose $a_T$ and $a_S$ 
based on the corresponding asymptotic acceptance probabilities $p_S = P(\chi^2_{2} \le a_S)$ and $p_T = P(\chi^2_{4} \le a_T)$.
From Remark \ref{rmk:loss_comp_optimal}, 
when $p_S = p_T = 0.01$, 
the improvement on causal effect estimation is at most $4.89\%$ different from the ideal optimal one.

We consider three different designs of the survey experiments: 
(i) ReSEM with $p_S = 0.01$ and $p_T = 1$, 
(ii) ReSEM with $p_S = 1$ and $p_T = 0.01$, 
and 
(iii) 
ReSEM with $p_S = p_T = 0.01$, 
which correspond to rerandomization at only the sampling stage, only the assignment stage and both stages, respectively. 
We consider the difference-in-means estimator $\hat{\tau}$ under these designs, 
and consider additionally the regression-adjusted estimator $\hat{\tau}(\hat{\bs{\beta}}, \hat{\bs{\gamma}})$ under ReSEM with $p_S = p_T = 0.01$.
For each design and each sample size $n$, we generate $10^4$ rerandomizations. 
Figure \ref{fig:hist}(a) shows the histograms of $\sqrt{n}\left(\hat{\tau}-\tau\right)$ under these designs, 
as well as their asymptotic approximations based on \eqref{eq:dist},  when 
$n=800$. 
From Figure \ref{fig:hist}(a), ReSEM with rerandomization at both stages provides the most efficient difference-in-means estimator, and ignoring rerandomization at either stage will lead to efficiency loss. 
Figure \ref{fig:hist}(b) shows the histograms and their asymptotic approximations of the difference-in-means and regression-adjusted estimators under ReSEM with $p_S = p_T = 0.01$. From Figure \ref{fig:hist}(b), regression adjustment can further improve the estimation efficiency by carefully adjusting imbalance of additional covariates. 
Figure \ref{fig:CI} shows the coverage probabilities and average lengths of $95\%$ confidence intervals for $\tau$ under the three rerandomization designs, 
with $n$ varying from $100$ to $1000$, 
using either the difference-in-means or regression-adjusted estimator. 
From Figure \ref{fig:CI}(a), the coverage probabilities 
are all close to the nominal level $95\%$, 
and from Figure \ref{fig:CI}(b), the regression-adjusted estimator has the shortest confidence intervals on average, followed by the difference-in-means estimator under ReSEM with rerandomization at both stages. 
These simulation results further confirm our theory in Sections \ref{sec:resem}--\ref{sec:large_sample_CI}.  
\begin{figure}[htbp]
	\centering
    \begin{subfigure}[htbp]{0.5\textwidth}
        \centering
        \includegraphics[width=.6\textwidth]{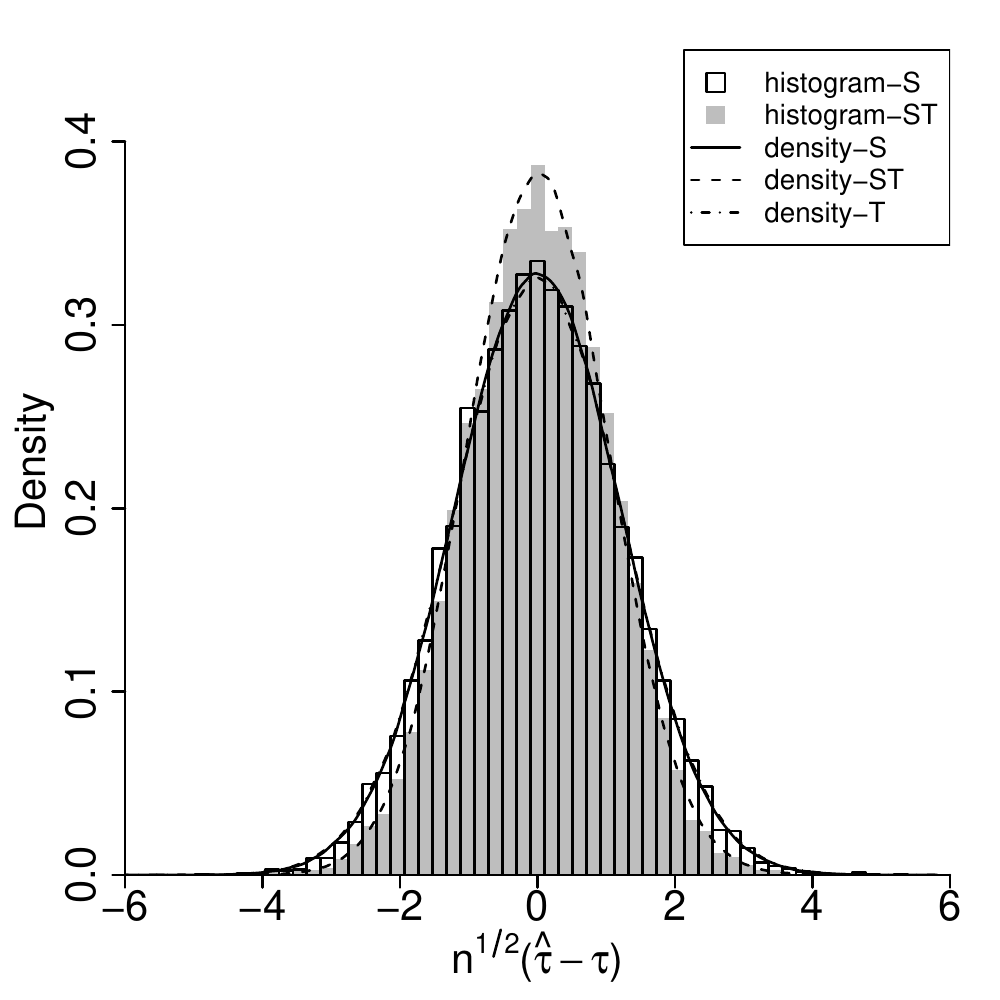}
        \caption{\centering Histograms of $\sqrt{n}( \hat{\tau} - \tau)$ under three rerandomization designs}
    \end{subfigure}%
    \begin{subfigure}[htbp]{0.5\textwidth}
        \centering
        \includegraphics[width=.6\textwidth]{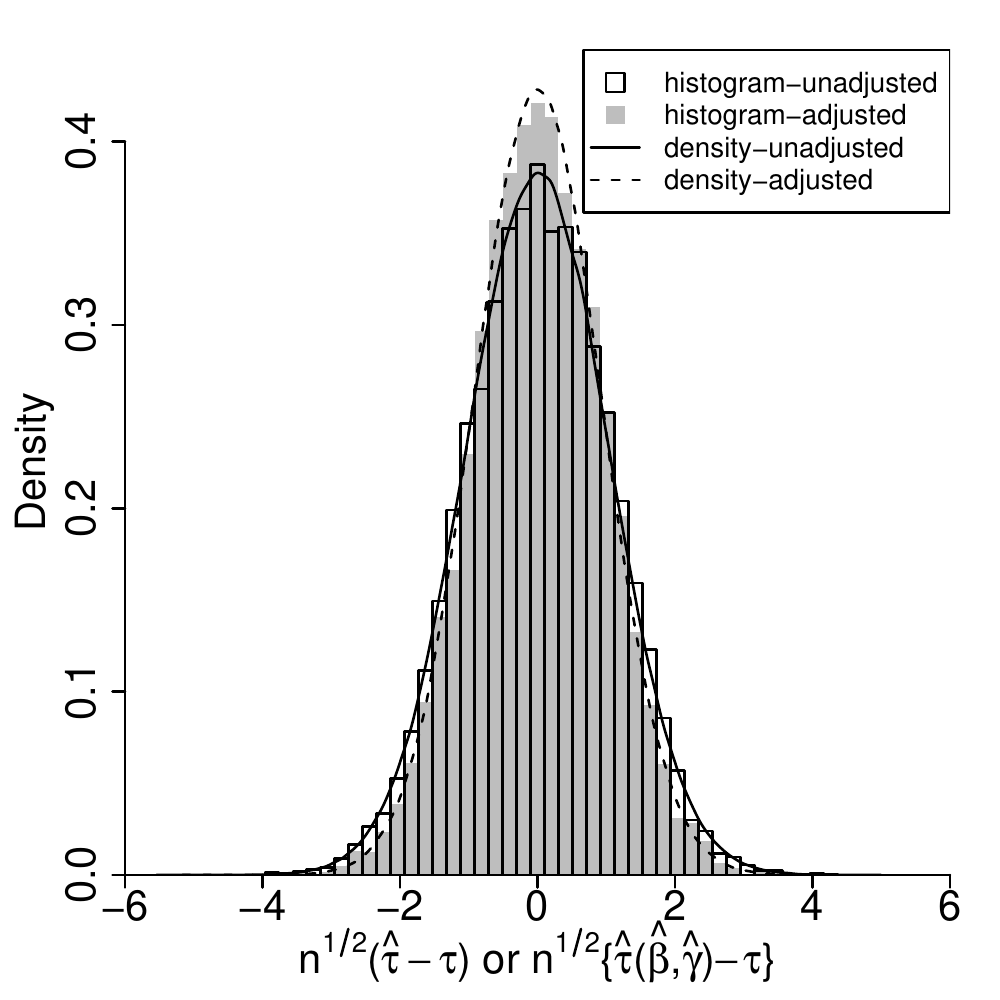}
        \caption{\centering Histograms of $\sqrt{n}( \hat{\tau} - \tau)$ and $\sqrt{n}(\hat{\tau}(\hat{\bs{\beta}}, \hat{\bs{\gamma}})$ $-\tau)$ under ReSEM with $p_S = p_T = 0.01$}
    \end{subfigure}
    \caption{
    Histograms and asymptotic densities of difference-in-means and regression-adjusted estimators under different rerandomization designs. 
    (a) compares difference-in-means estimator under ReSEM with $p_S = 0.01$ and $p_T = 1$ (S), 
    ReSEM with $p_S = 1$ and $p_T = 0.01$ (T)
    and ReSEM with $p_S = p_T = 0.01$ (ST). 
    The histogram for ReSEM with $p_S = 1$ and $p_T = 0.01$ is dropped for clarity. 
    (b) compares unadjusted and regression-adjusted estimators under ReSEM with $p_S = p_T = 0.01$. 
    } 
    \label{fig:hist}
\end{figure}
\begin{figure}[htbp]
    \begin{subfigure}[htbp]{0.5\textwidth}
        \centering
        \includegraphics[width=.6\textwidth]{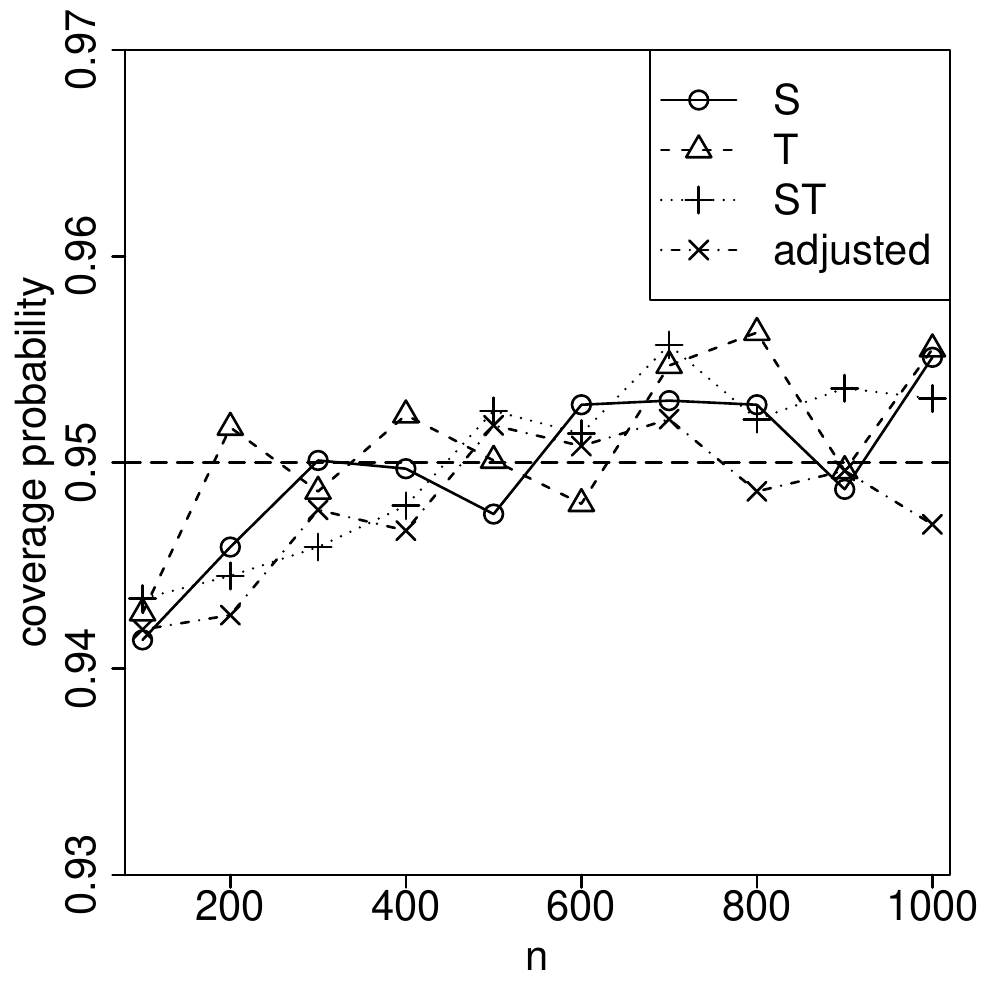}
        \caption{Coverage probability}
    \end{subfigure}%
    \begin{subfigure}[htbp]{0.5\textwidth}
        \centering  
        \includegraphics[width=.6\textwidth]{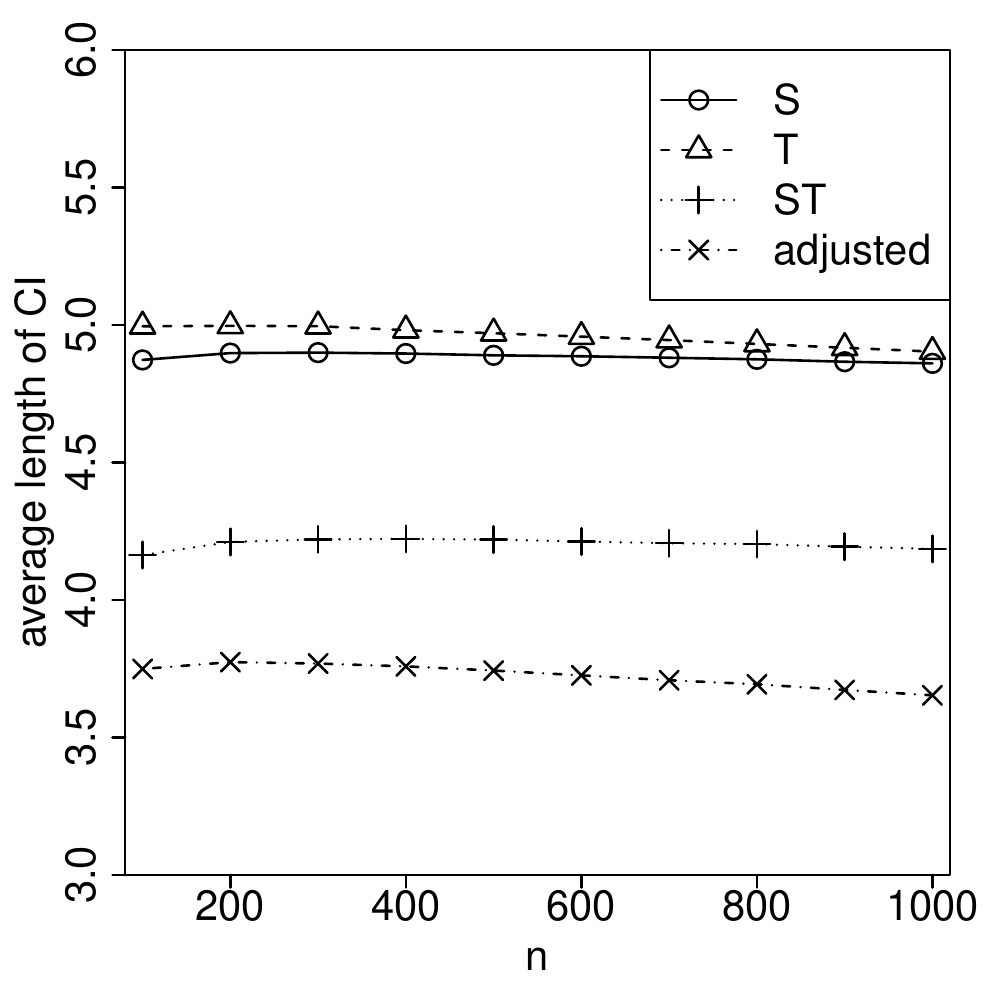}
        \caption{Average length}
    \end{subfigure}
    \caption{
    Coverage probability and average length of the confidence intervals based on the difference-in-means and regression-adjusted estimators under different rerandomization design. 
    (a) shows the empirical coverage probabilities of 95\% confidence intervals using difference-in-means estimator under three rerandomization designs (S, T, ST) and regression-adjusted estimator under ReSEM with $p_S=p_T=0.01$ (adjusted), and (b) shows the corresponding average lengths (multiplied by $\sqrt{n}$) of confidence intervals.
    }
    \label{fig:CI}
\end{figure}

\subsection{An example on election study}\label{sec:election}

To illustrate rerandomization for survey experiments, we use the data from the Cooperative Congressional Election Study \citep[CCES;][]{CCES2014,CCES2020}, which is a national sample survey administrated by YouGov. 
Specifically, the dataset combines surveys conducted by multiple teams, with half of the questionnaire consisting of common content and the other half consisting of team content designed by each individual team\footnote{For a detailed description and guide of the dataset, see \url{https://cces.gov.harvard.edu}.}. 
We consider the team survey of Boston University \citep{CCES2014BOS} from CCES 2014, 
where each participant were asked about their opinion on the federal spending on scientific research. 
Among the total 1000 participants, 
about half 
were provided additional information about current federal research budget: ``each year, just over 1\% of the federal budget is spent on scientific research'', while the remaining units were blind of this information. 
The outcome is the response to the question that whether federal spending on scientific research should be increased, kept the same or decreased, coded as 1, 2 and 3. 
The interest here is how the opinion toward federal spending on scientific research is affected by the information provided on current spending. 

We will use the above survey experiment from \citet{CCES2014BOS} 
to illustrate our rerandomization design. 
The ideal population of interest will be all Americans. Since we do not have the census data, we instead focus on the population consisting of all the participants in the CCES 2014 dataset, for which we have a rich set of covariate information. 
We sample $n=1000$ individuals from the population, and assign $n_1=500$ and $n_0=500$ individuals to treatment and control groups, respectively, where units in the treated group will be provided the information on current federal spending on scientific research. 
We include nine pretreatment covariates, and assume that four covariates are observed at the sampling stage, 
additional two covariates are observed at the treatment assignment stage, 
and 
all nine covariates are observed at the analysis stage; see Table \ref{table:covariates}.
Besides, we assume $\bs{E} = \bs{W}$, i.e., no additional covariates for the whole population are observed in analysis. 
For simplicity, 
we remove individuals with missing outcomes or covariates, resulting in a population of size $N=49452$, 
and we view the outcome as a numerical variable representing how each individual is against the federal spending on scientific research. 

\begin{table}[hbtp]
\centering
\caption{Covariates at sampling, treatment assignment and analysis stages}\label{table:covariates}
\resizebox{0.9\textwidth}{!}{
\begin{tabular}{ccc} 
\toprule
Stage & Notation & (Additional) covariate information
\\
\midrule
\multirow{2}{*}{sampling} & \multirow{2}{*}{$\bs{W}$} & age, gender, race,\\ & & whether the highest level of education is 
college or higher \\ 
\midrule
\multirow{2}{*}{assignment} & \multirow{2}{*}{$\bs{X} \setminus \bs{W}$} &   whether family annual income is less than \$60000,\\ 
 & &  whether the individual thinks economy has gotten worse last year\\ 
\midrule
\multirow{2}{*}{analysis} & \multirow{2}{*}{$\bs{C} \setminus \bs{X}$} & ideology, party identification, \\ & & whether the individual follows news and public affairs most of the time \\ 
\bottomrule
\end{tabular}%
}
\end{table}

To make the simulation more relevant to the real data, 
we fit a linear regression of the observed outcome on the treatment indicator, all covariates in Table \ref{table:covariates} and their interactions, based on the data of 1000 individuals in the team survey, and use the fitted model to generate potential outcomes for all the $N$ individuals. 
To better illustrate the improvement from rerandomization, 
we fix the individual treatment effects but 
shrink the control potential outcomes towards their population mean by a factor of 0.1, 0.2, $\ldots$, 1 to generate 10 different datasets. 
Table \ref{table:R2_simulated} lists the values of the $R^2$ measures in \eqref{eq:R2} and  \eqref{eq:R_A2} for the simulated 10 datasets.
We conduct simulation under three different designs for each generated dataset: 
(i) CRSE, i.e., ReSEM with $p_S = p_T = 1$, 
(ii) ReSEM with $p_S = 1$ and $p_T = 0.001$ 
and 
(iii) ReSEM with $p_S = p_T = 0.001$, 
which corresponds to rerandomization in neither stage, only treatment assignment stage, and both stages. 
From Remark \ref{rmk:loss_comp_optimal}, 
when $p_S = p_T = 0.001$, 
the improvement on causal effect estimation is at most $8.4\%$ different from the ideal optimal one. 
We consider difference-in-means estimator under these three designs, 
and also the regression-adjusted estimator under ReSEM with $p_S = p_T = 0.001$. 
Figure \ref{fig:real_data}(a) and (b) show the empirical variances and average lengths of confidence intervals using the two estimators under the two ReSEMs, 
standardized by the corresponding values for the difference-in-means estimator under the CRSE.  
From Figure \ref{fig:real_data}(a) and (b), 
as the shrinkage factor decreases, under which the control potential outcomes become less heterogeneous while the individual effects are kept fixed,  
the additional gain from rerandomization at the sampling stage (i.e., the gap between the solid and dashed lines, which relates to $R_S^2$) increases, 
while the improvement from regression adjustment (i.e., the dotted line, which relates to $R_E^2 + R_C^2$) decreases. This is consistent with the $R^2$ measures in Table \ref{table:R2_simulated}. 
Figure \ref{fig:real_data}(c) shows the coverage probabilities of the confidence intervals under all cases. 
From Figure \ref{fig:real_data},
compared to the CRSE, 
the precision of the difference-in-means estimator is improved by employing rerandomization at the treatment assignment stage,  
and 
it is further improved by also employing rerandomization at the sampling stage as well as conducting covariate adjustment at the analysis stage.  
Moreover, both rerandomization and regression adjustment reduce the average lengths of the confidence intervals, while still maintaining the coverage probabilities at the nominal level approximately.

\begin{table}[htb]
\centering
\caption{$R^2$ measures  for the simulated 10 datasets}
\label{table:R2_simulated}
\resizebox{0.9\textwidth}{!}{
\begin{tabular}{ccccccccccc}
\toprule
Factor & 0.1 &0.2  & 0.3 & 0.4  & 0.5 & 0.6 &  0.7& 0.8&0.9 & 1\\ \midrule
 $R^2_S=R^2_E$ & 0.2176 & 0.2240 & 0.2076& 0.1769 &0.1434& 0.1140 & 0.0904 & 0.0723 & 0.0586 & 0.0481  \\ \midrule
$R^2_T$  & 0.1657 &0.1429 &0.1614 &0.2085 &0.2643 &0.3157 &0.3580 &0.3913 &0.4172& 0.4373 \\ \midrule
$R^2_C$ & 0.4391& 0.4228 &0.4650 &0.5441 &0.6303 &0.7062 &0.7669& 0.8136 & 0.8490 & 0.8761 \\ \bottomrule
\end{tabular}}%
\end{table}
As a side note, we also consider the greedy pair-switching design recently proposed by \citet{krieger2019nearly}. Specifically, we use the greedy pair-switching algorithm in the assignment stage, and consider both SRS and rejective sampling in the sampling stage. 
Figure \ref{fig:real_data}(a) shows the corresponding empirical variances of the difference-in-means estimator. 
From Figure \ref{fig:real_data}(a), 
the treatment effect estimate is more precise under 
the greedy pair-switching design than under rerandomization, but the improvement is relatively small. 
Compared to the greedy design, rerandomization is  computationally much simpler, and more importantly, it allows large-sample inference for the average treatment effect as discussed throughout the paper, avoiding any constant treatment effect assumption that is typically involved for randomization tests. 
Additionally, employing rerandomization at the sampling stage (i.e., rejective sampling) is also beneficial for the greedy pair-switching design.

\begin{figure}[htb]
    \begin{subfigure}[htbp]{0.33\textwidth}
        \includegraphics[width=\textwidth]{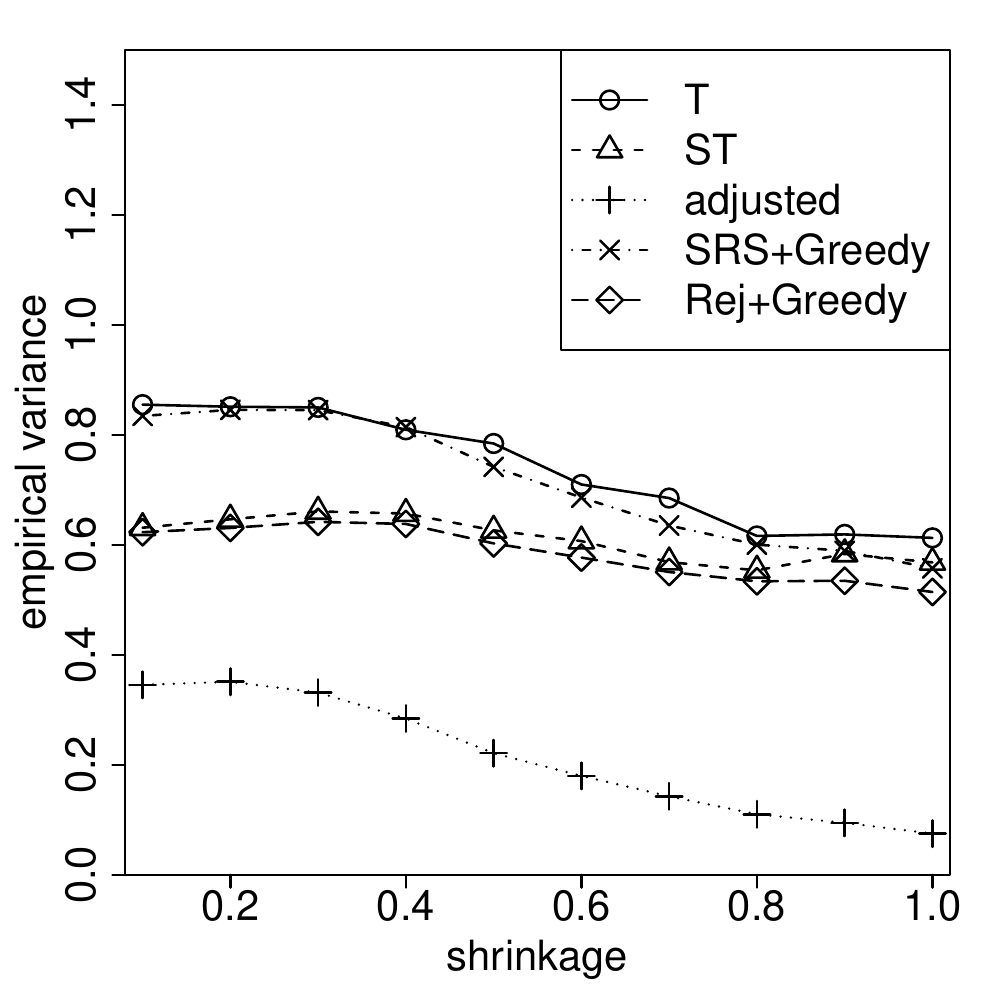}
        \caption{Empirical variance}
    \end{subfigure}%
    \begin{subfigure}[htbp]{0.33\textwidth}
        \includegraphics[width=\textwidth]{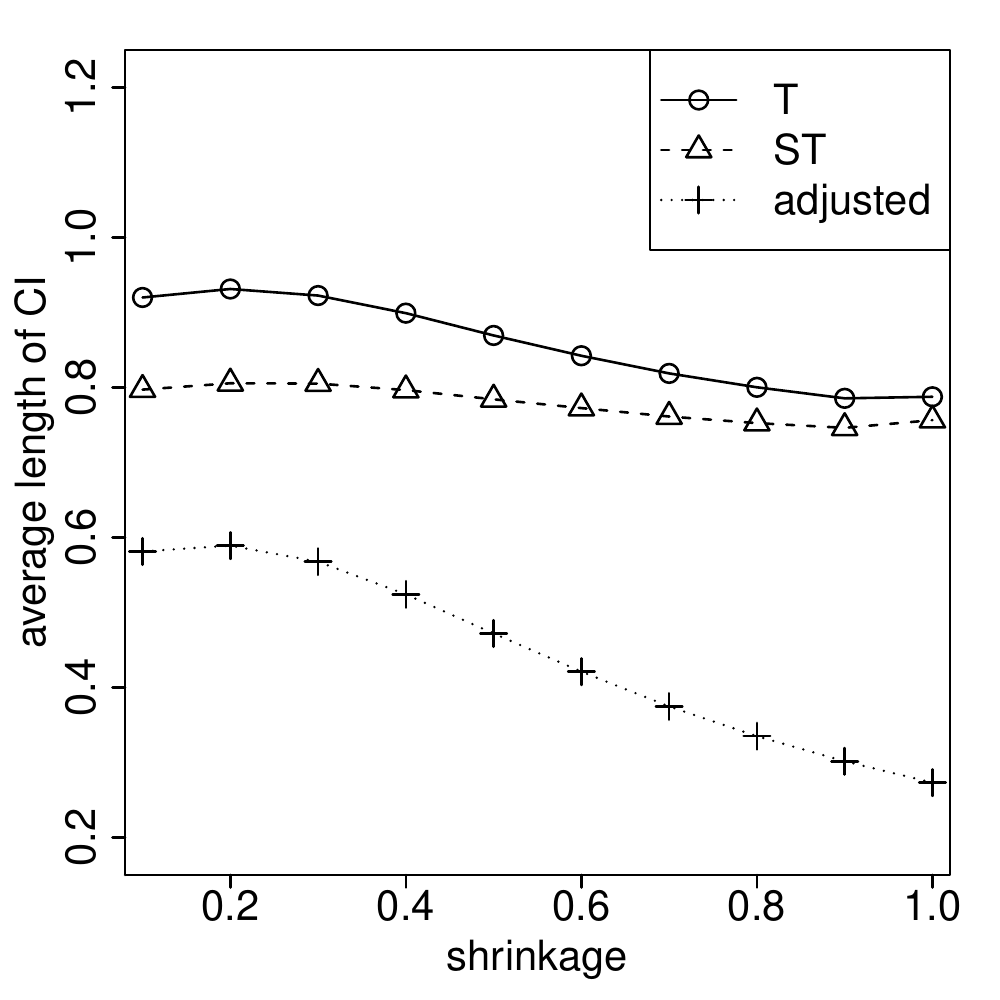}
        \caption{Average length}
    \end{subfigure}%
    \begin{subfigure}[htbp]{0.33\textwidth}
        \includegraphics[width=\textwidth]{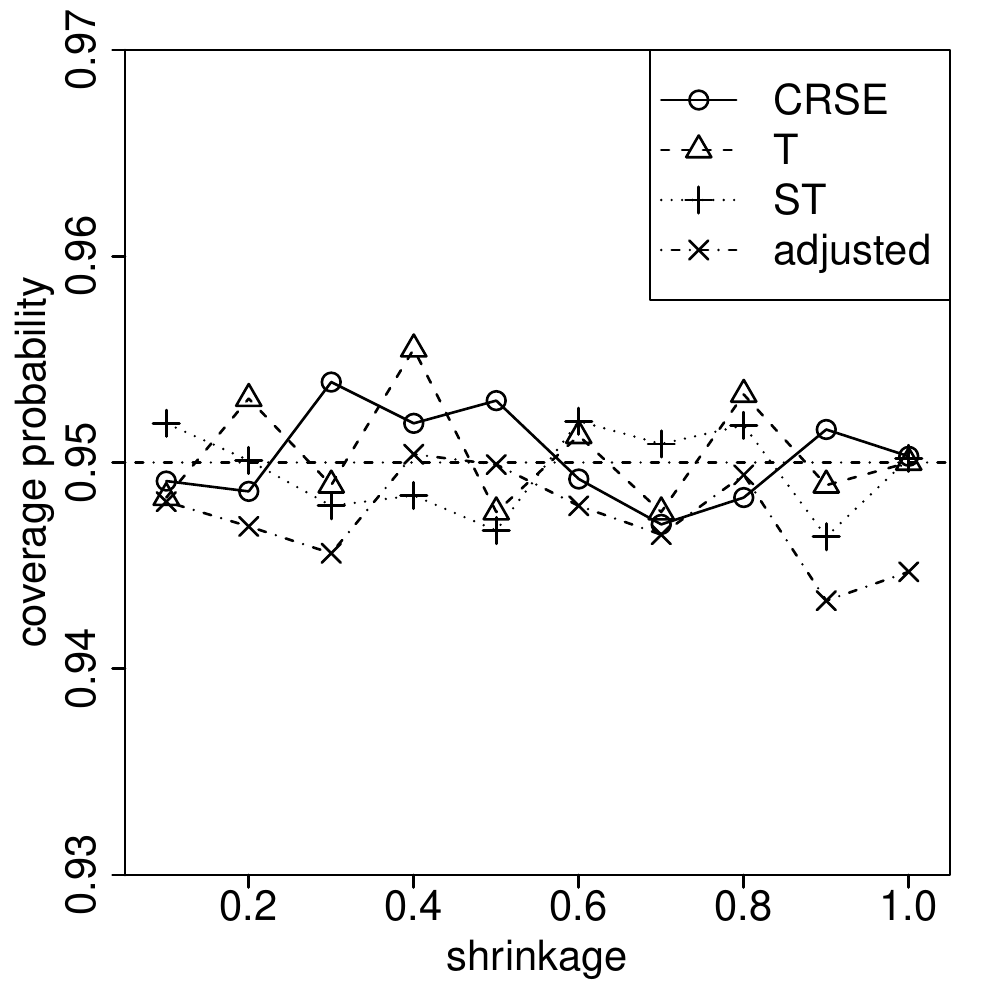}
        \caption{Coverage probability}
    \end{subfigure}
    \caption{
    Comparison of four approaches for design and analysis of the population average treatment effect for the election example: 
    difference-in-means estimator under the CRSE (CRSE) and ReSEM with rerandomization at treatment assignment stage only (T) and at both stages (ST), and regression-adjusted estimator under ReSEM with rerandomization at both stages (adjusted). 
    (a) compares empirical variances of the treatment effect estimators under the four approaches, where each empirical variance is standardized by the corresponding value for the CRSE. 
    (b) compares average lengths of $95\%$ confidence intervals under the four approaches, where each value is standardized by the corresponding value under the CRSE. 
    (c) compares empirical coverage probabilities of the $95\%$ confidence intervals under the four approaches.
    Besides, (a) also includes two additional approaches: 
    one uses SRS at the sampling stage, the other uses rejective sampling (or rerandomization) at the sampling stage, 
    and both of them use 
    greedy pair-switching design at the treatment assignment stage. 
    These two additional approaches are denoted by  SRS+Greedy and Rej+Greedy. 
    }\label{fig:real_data}
\end{figure}

\section{Conclusion}\label{sec:conclude}

We proposed a general two-stage rerandomization for survey experiments, and focused particularly on the covariate balance criteria based on the Mahalanobis distances. 
We studied asymptotic properties of rerandomized survey experiments and constructed large-sample confidence intervals for the population average treatment effect. 
Our results show that rerandomization at both the sampling and treatment assignment stages can improve the treatment effect estimation, and the improvement depends crucially on the association between covariates and potential outcomes as well as individual treatment effects. 
Moreover, adjusting remaining covariate imbalance after the survey experiment can further improve the estimation precision, especially when there is additional covariate information after conducting the experiment. 
We further studied optimal covariate adjustment in terms of both sampling and estimated precision. 

To avoid the paper being too lengthy, 
we relegate several extensions to the Supplementary Material and briefly summarize them below. 
First, 
motivated by \citet{wuding2020}, \citet{zhao2020covariate} and \citet{cohen2020gaussian}, we propose general covariate-adjusted conditional randomization tests for ReSEM that enjoy both finite-sample validity for testing sharp null hypotheses and large-sample validity for testing weak null hypotheses. Second, we study ReSEM with discrete covariates, which corresponds to stratified sampling and blocking. 
Third, we study clustered survey experiments with both sampling and treatment assignment at a cluster level.

\section*{Acknowledgments}
We thank the Associate Editor and two reviewers for insightful and constructive comments.

\bibliographystyle{plainnat}
\bibliography{survey.bib}

\newpage

\setcounter{equation}{0}
\setcounter{section}{0}
\setcounter{figure}{0}
\setcounter{example}{0}
\setcounter{proposition}{0}
\setcounter{corollary}{0}
\setcounter{theorem}{0}
\setcounter{lemma}{0}
\setcounter{table}{0}
\setcounter{condition}{0}
\setcounter{page}{1}
\begin{center}
	\bf \LARGE 
	Supplementary Material 
\end{center}

\renewcommand {\theproposition} {A\arabic{proposition}}
\renewcommand {\theexample} {A\arabic{example}}
\renewcommand {\thefigure} {A\arabic{figure}}
\renewcommand {\thetable} {A\arabic{table}}
\renewcommand {\theequation} {A\arabic{section}.\arabic{equation}}
\renewcommand {\thelemma} {A\arabic{lemma}}
\renewcommand {\thesection} {A\arabic{section}}
\renewcommand {\thetheorem} {A\arabic{theorem}}
\renewcommand {\thecorollary} {A\arabic{corollary}}
\renewcommand {\thecondition} {A\arabic{condition}}

\renewcommand {\thepage} {A\arabic{page}}

Appendix \ref{sec:extension_special} studies extensions to the CRSE, ReM and rejective sampling,  providing the details for 
Remarks \ref{rmk:special_dim} and \ref{rmk:special_covadj}.

Appendix \ref{sec:OLS_interaction} studies the connection between the regression-adjusted estimator $\hat{\tau}(\hat{\bs{\beta}}, \hat{\bs{\gamma}})$ and the least squares estimator from the linear regression of observed outcome on treatment indicator, covariates, and their interaction. 

Appendix \ref{sec:rand_test_resem} studies randomization tests for ReSEM that are not only exactly valid for testing sharp null hypotheses but also asymptotically valid for testing weak null hypotheses. 

Appendix \ref{sec:sts_block} studies rerandomization in survey experiments with discrete covariates, i.e., stratified sampling and blocking. 

Appendix \ref{sec:cluster} studies rerandomization in clustered survey experiments, where the sampling and treatment assignment are conducted at a cluster level instead of the individual level. 

Appendix \ref{app:crse} studies the sampling properties for the CRSE, gives technical details for the comments on Condition \ref{cond:fp}, and proves Proposition \ref{prop:R2}. 

Appendix \ref{sec:equiv_resem_single_two} shows the asymptotic equivalence between the single-stage and two-stage rerandomized survey experiments. 

Appendix \ref{app:resem} studies the sampling properties for ReSEM, proves Theorem \ref{thm:dist} and Corollaries \ref{corollary:PRIASV}--\ref{corollary:QR}, and gives technical details for  Remark \ref{rmk:loss_comp_optimal} and the comments on acceptance probabilities and covariate balance. 

Appendix \ref{app:reg} studies regression adjustment under ReSEM, proves Theorems \ref{thm:dist_reg_general} and \ref{thm:optimal} and Corollaries \ref{cor:dist_reg_general_equ} and \ref{cor:gain_analysis}, and gives technical details for the comments on the optimal adjustment coefficients.

Appendix \ref{sec:proof_special} proves the large-sample properties for the three special cases, the CRSE, ReM and rejective sampling. 

Appendix \ref{app:ci} proves the large-sample validity of variance estimators and confidence intervals under ReSEM, and proves Theorem \ref{thm:plug_in}. 

Appendix \ref{sec:C_optmal_connection_proof} 
proves Theorem \ref{thm:est_optimal}, and connects the regression-adjusted estimator to usual linear regression models with least squares estimates. 

Appendix \ref{sec:proof_rand_test} studies randomization tests for both Fisher's sharp null hypotheses and Neyman's weak null hypotheses, proving Theorems \ref{thm:cond_rand_test} and \ref{thm:frt_weak_null}. 

Appendix \ref{sec:proof_strata} studies the large sample properties for stratified sampling and blocking, proving Theorem \ref{thm:resem_cat}. 

Appendix \ref{sec:proof_cluster} studies the large sample properties for clustered survey experiments.

\section{Extension to special cases}\label{sec:extension_special}

\subsection{Special case: completely randomized survey experiments}\label{sec:CRSE}
If $a_S=\infty$ and $a_T = \infty$, then 
no rerandomization is conducted at either the sampling or treatment assignment stages, 
and consequently 
ReSEM reduces to the CRSE.
Therefore, from Theorem \ref{thm:dist}, 
we derive the asymptotic distribution of $\hat{\tau}$ under the CRSE.

\begin{corollary}\label{cor:crse}
Under Condition \ref{cond:fp} and the CRSE, 
$
\sqrt{n} (\hat\tau - \tau)  \  \dot\sim\  V_{\tau\tau}^{1/2} \cdot \varepsilon.
$
\end{corollary}

\citet{Branson:2019aa} shows that under the CRSE, $\hat{\tau}$ is unbiased for $\tau$ with variance $n^{-1}V_{\tau\tau}$. 
Corollary \ref{cor:crse} supplements their results with a rigorous justification of finite population central limit theorem. 
When $f=1$, in the sense that all units are sampled to enroll the experiment, the CRSE reduces to the usual completely randomized experiment (CRE), 
and $V_{\tau\tau}$ in \eqref{eq:V_tautau} reduces to the usual Neyman variance formula under the CRE \citep{Neyman:1923}. 
When $f \rightarrow 0$ as $N\rightarrow \infty$, in the sense that we sample only a tiny proportion of units to enroll the experiment, 
the asymptotic variance no longer depends on the finite population variance $S^2_{\tau}$ of the individual effects, mimicking the usual variance formula under the infinite superpopulation setting \citep[][Chapter 6]{imbens2015causal}.

Analogously, because the CRSE is essentially ReSEM with $\bs{W} = \bs{X} =\emptyset$ and  $a_S = a_T = \infty$, 
Corollary \ref{cor:dist_reg_general_equ}, \eqref{eq:simp_exp_reg_sampling}, \eqref{eq:simp_exp_reg_assignment} and Theorem \ref{thm:optimal}
imply the asymptotic properties of a general regression-adjusted estimator as well as the optimal regression adjustment under the CRSE. 

\begin{corollary}\label{cor:reg_crse}
	Under Condition \ref{cond:fp_analysis} and the CRSE, 
	for any fixed adjustment coefficients $(\bs{\beta}, \bs{\gamma})$, 
	\begin{align*}%
	& \quad \ \ \sqrt{n} \left\{ \hat{\tau}(\bs{\beta}, \bs{\gamma}) - \tau \right\} 
	\ \dot\sim \ V_{\tau\tau}^{1/2}(\bs{\beta}, \bs{\gamma}) \cdot \varepsilon
	\nonumber
	\\
	& 
    \sim 
	\sqrt{ V_{\tau\tau} (1 - R_E^2 - R_C^2 ) 
	+ (1-f) (\bs{\gamma} - \tilde{\bs{\gamma}})^\top \bs{S}^2_{\bs{E}} (\bs{\gamma} - \tilde{\bs{\gamma}})	
	+  (r_1 r_0)^{-1} (\bs{\beta} - \tilde{\bs{\beta}})^\top \bs{S}^2_{\bs{C}} (\bs{\beta} - \tilde{\bs{\beta}}) 
	}
	\cdot \varepsilon. 
	\end{align*}
	Moreover,  the $\mathcal{S}$-optimal regression-adjusted estimator is attainable at $(\bs{\beta}, \bs{\gamma}) = (\tilde{\bs{\beta}}, \tilde{\bs{\gamma}})$, with the following asymptotic distribution: 
	\begin{align*}%
		\sqrt{n} \big\{ \hat{\tau}( \tilde{\bs{\beta}}, \tilde{\bs{\gamma}}) - \tau \big\} 
		& \ \dot\sim \ 
		 V_{\tau\tau}^{1/2}\sqrt{1 - R_E^2 - R_C^2}
		\cdot \varepsilon. 
	\end{align*}
\end{corollary}

\subsection{Special case: rerandomized treatment-control experiments}\label{sec:rem}

When $f=1$, 
$M_T$ is the same as the unconditional Mahalanobis distance between covariate means under two treatment groups, 
and thus 
ReSEM reduces to rerandomized treatment-control experiments using Mahalanobis distance (ReM) introduced in \citet{Morgan2012}.  
Because 
there is no randomness in the sampling stage and $R_S^2$ reduces to zero when $f=1$, 
by the same logic as Theorem \ref{thm:dist}, we can derive the following asymptotic distribution of $\hat{\tau}$ under ReM, which is a main result in \citet{rerand2018}. 

\begin{corollary}\label{cor:rem}
Under Condition \ref{cond:fp}(ii)--(iv) and ReM,
\begin{align*}
\sqrt{n} (\hat\tau - \tau) \mid   M_T \le a_T  \ & \dot\sim\  V_{\tau\tau}^{1/2} \Big(  \sqrt{1- R_T^2} \cdot \varepsilon
+
\sqrt{R_T^2} \cdot L_{K,a_T} 
\Big),
\end{align*}
where $\varepsilon$ and $L_{K,a_T}$ are mutually independent. 
\end{corollary}

Note that in Condition \ref{cond:fp}(i) we require $f$ to have a limit strictly less than 1, since otherwise $M_S$ may not be well-defined. 
Under ReM with $f=1$, there is essentially no rejective sampling at the first stage and we can thus relax Condition \ref{cond:fp}(i). 
Analogously, 
Corollary \ref{cor:dist_reg_general_equ}, \eqref{eq:simp_exp_reg_sampling}, \eqref{eq:simp_exp_reg_assignment} and Theorem \ref{thm:optimal} imply the following asymptotic properties of regression adjustment under ReM.  

\begin{corollary}\label{cor:rem_reg}
	Under ReM and Condition \ref{cond:fp_analysis} excluding Condition \ref{cond:fp}(i),  
	if $\bs{X} \subset \bs{C}$, i.e., there is more covariate information in analysis than in design, 
	then
	for any fixed adjustment coefficients $(\bs{\beta}, \bs{\gamma})$, 
	\begin{align*}
	& \quad \ \sqrt{n} \left\{ \hat{\tau}(\bs{\beta}, \bs{\gamma}) - \tau \right\} \mid M_T \le a_T 
	\\
	& \dot\sim \ 
	\sqrt{ V_{\tau\tau} (1 - R_C^2) +  (r_1 r_0)^{-1} (\bs{\beta} - \tilde{\bs{\beta}})^\top \bs{S}^2_{\bs{C}\setminus \bs{X}} (\bs{\beta} - \tilde{\bs{\beta}}) }
	\cdot \varepsilon 
	+ 
	\sqrt{ (r_1 r_0)^{-1} (\bs{\beta} - \tilde{\bs{\beta}})^\top \bs{S}^2_{\bs{C} \mid \bs{X}} (\bs{\beta} - \tilde{\bs{\beta}})} \cdot L_{K, a_T}. 
	\end{align*}
	Moreover, the $\mathcal{S}$-optimal regression-adjusted estimator is attainable at $\bs{\beta}  = \tilde{\bs{\beta}}$, with the following asymptotic distribution: 
	\begin{align*}
	\sqrt{n} \left\{ \hat{\tau}( \tilde{\bs{\beta}}, \bs{\gamma}) - \tau \right\} \mid M_T \le a_T 
	\  \dot\sim \ 
	\sqrt{ V_{\tau\tau} (1 - R_C^2) }
	\cdot \varepsilon. 
	\end{align*}
\end{corollary}

It is not surprising that the asymptotic distribution in Corollary \ref{cor:rem_reg} does not depend on the adjustment coefficient $\bs{\gamma}$, 
because $\hat{\bs{\delta}}_{\bs{E}}$ is a constant zero when $f=1$. 
Corollary \ref{cor:rem_reg} recovers the results in \citet[][Corollary 3]{li2019rerandomization}. 

\subsection{Special case: rejective sampling without treatment assignment}\label{sec:rej_samp}

We consider the classical survey sampling setting without treatment assignment at the second stage. 
Let $\{y_1, y_2, \ldots, y_N\}$ be a finite population of interest, 
and we want to estimate the population average $\bar{y} = N^{-1} \sum_{i=1}^{N} y_i$. 
Then for the survey experiment with potential outcomes constructed as  $Y_i(1) = r_1 y_i$ and $Y_i(0) = -r_0 y_i$, 
the average treatment effect $\tau$ reduces to $\bar{y}$, 
and 
the difference-in-means estimator $\hat{\tau}$ reduces to the sample average  $\bar{y}_{\mathcal{S}} = n^{-1} \sum_{i=1}^n Z_i y_i$, an intuitive estimator for the finite population average $\bar{y}$. 
In this special case, 
ReSEM essentially reduces to rejective sampling \citep{Fuller2009}, because the estimator $\bar{y}_{\mathcal{S}}$ no longer depends on the treatment assignment at the second stage. 
From Theorem \ref{thm:dist}, we can immediately derive the asymptotic distribution of the sample average under rejective sampling given the following
regularity condition similar to Condition \ref{cond:fp}. 
Let $S_y^2$ be the finite population variance of the outcome $y$, 
$\bs{S}_{y, \bs{W}} = \bs{S}_{\bs{W},y}^\top$ be the finite population covariance matrix between $y$ and $\bs{W}$, 
and $\rho^2_{y,\bs{W}} = \bs{S}_{y, \bs{W}} (\bs{S}_{\bs{W}}^2)^{-1} \bs{S}_{\bs{W}, y}$ be the squared multiple correlation between $y$ and $\bs{W}$.
\begin{condition}\label{cond:survey}
	As $N \rightarrow \infty$, 
	the sequence of finite populations $\{(y_i, \bs{W}_i): 1\le i \le N\}$'s satisfies 
	\begin{enumerate}[label=(\roman*), topsep=1ex,itemsep=-0.3ex,partopsep=1ex,parsep=1ex]
		\item the proportion $f$ of sampled units has a limit; 
		
		\item the finite population variances $S^2_y$, $\bs{S}_{\bs{W}}^2$
		and covariance
		$
		\bs{S}_{y, \bs{W}}
		$
		have limiting values, 
		and the limit of $\bs{S}_{\bs{W}}^2$ is nonsingular;
		
		\item 	$\max_{1 \le i\le N} (y_i - \bar{y})^2/n \rightarrow 0$, 
		and 
		$\max_{1 \le i\le N}\| \bs{W}_i - \bar{\bs{W}} \|_2^2/n \rightarrow 0. $
	\end{enumerate}
\end{condition}

\begin{corollary}\label{cor:rej_sam}
	Under Condition \ref{cond:survey} and rejective sampling based on SRS with covariate balance criterion $M_S \le a_S$,  
	\begin{align*}
		\sqrt{n} \left( \bar{y}_{\mathcal{S}} - \bar{y} \right) \mid M_S \le a_S
		\ \dot\sim \ 
		\sqrt{ \left(1-f\right) S^2_y }  \cdot
		\left(  \sqrt{1-\rho_{y,\bs{W}}^2 } \cdot \varepsilon
		+ \sqrt{\rho_{y,\bm{W}}^2} \cdot L_{J,a_S}
		\right). 
	\end{align*}
\end{corollary}

We then consider regression-adjusted estimator under rejective sampling. 
Again, we can view rejective sampling as a special case of the survey experiment with potential outcomes constructed as  $Y_i(1) = r_1 y_i$ and $Y_i(0) = -r_0 y_i$.
The regression-adjusted estimator in \eqref{eq:reg} with $\bs{\beta}\equiv \bs{0}$ then reduces to 
$
\hat{\tau}(\bs{0}, \bs{\gamma}) = \bar{y}_{\mathcal{S}} - \bs{\gamma}^\top \hat{\bs{\delta}}_{\bs{W}},
$
a commonly used linearly regression-adjusted estimator in the survey sampling literature. 
From 
Corollary \ref{cor:dist_reg_general_equ}, \eqref{eq:simp_exp_reg_sampling}, \eqref{eq:simp_exp_reg_assignment} and Theorem \ref{thm:optimal}, 
we can immediately derive the following asymptotic properties of regression adjustment under rejective sampling. 
Let $\overline{\bs{\gamma}} = (  \bs{S}^2_{\bs{E}} )^{-1} \bs{S}_{\bs{E}, y}$ be the finite population linear projection coefficient of $y$ on $\bs{E}$, 
and $\rho_{y, \bs{E}}^2$ be the squared multiple correlation between $y$ and $\bs{E}$. 
We introduce the following regularity condition, extending Condition \ref{cond:survey} to include the covariate $\bs{E}$ in analysis. 

\begin{condition}\label{cond:survey_ana}
	Condition \ref{cond:survey} holds, and it still holds with $\bs{W}$ replaced by $\bs{E}$. 
\end{condition}

\begin{corollary}\label{cor:rej_reg}
    Under 
    Condition \ref{cond:survey_ana} and 
    rejective sampling based on SRS with covariate balance criterion $M_S\le a_S$, 
	if $\bs{W} \subset \bs{E}$, 
	then
    \begin{align*}
        \sqrt{n}\left( \bar{y}_{\mathcal{S}} - \bs{\gamma}^\top \hat{\bs{\delta}}_{\bs{E}} - \bar{y} \right) \mid M_S \le a_S
        & \dot\sim \ 
        \sqrt{(1-f)S_y^2(1-\rho^2_{y,\bs{E}})+ (1-f) (\bs{\gamma} - \overline{\bs{\gamma}})^\top \bs{S}^2_{\bs{E} \setminus \bs{W}} (\bs{\gamma} - \overline{\bs{\gamma}}) } \cdot \varepsilon 
        \\
        & \quad \ + 
        \sqrt{(1-f) (\bs{\gamma} - \overline{\bs{\gamma}})^\top \bs{S}^2_{\bs{E}\mid \bs{W}}    (\bs{\gamma} - \overline{\bs{\gamma}}) } \cdot L_{J, a_S}. 
    \end{align*}
    Consequently, the $\mathcal{S}$-optimal regression-adjusted estimator is attainable at $\bs{\gamma} = \overline{\bs{\gamma}}$, with the following asymptotic distribution:
    $$
     \sqrt{n}\left( \bar{y}_{\mathcal{S}} - \overline{\bs{\gamma}}^\top \hat{\bs{\delta}}_{\bs{E}} - \bar{y} \right) \mid M_S \le a_S
        \dot\sim \ 
        \sqrt{(1-f)S_y^2(1-\rho^2_{y,\bs{E}})} \cdot \varepsilon. 
    $$
\end{corollary}

\citet{Fuller2009} studied consistency and asymptotic variance of the regression-adjusted estimator under a general rejective sampling. 
Corollary \ref{cor:rej_reg} supplements 
his results 
with asymptotic distributions (and thus large-sample confidence intervals) and optimality for the regression-adjusted estimator under rejective sampling based on SRS.

\section{Regression with treatment--covariate interaction}\label{sec:OLS_interaction}

In this section 
we study how the regression-adjusted estimator $\hat{\tau}(\hat{\bs{\beta}}, \hat{\bs{\gamma}})$ connects to usual least squares estimators from certain hypothesized regression models. 
Note that we use the covariate $\bs{E}$ to measure the balance between the sampled units and whole population and covariate $\bs{C}$ to measure the balance between the treated and control groups. 
Below we consider several cases depending on the relation between $\bs{E}$ and $\bs{C}$, starting from special cases to the general case. 

First, we consider the case in which $\bs{E} = \emptyset$, i.e., 
there is no adjustment for the covariate balance between sampled units and the population of interest. 
Consequently, the regression-adjusted estimator 
$\hat{\tau}(\hat{\bs{\beta}}, \hat{\bs{\gamma}})$ reduces to $\hat{\tau}(\hat{\bs{\beta}}, \bs{0})$ that no longer involves the adjustment coefficient for $\hat{\bs{\delta}}_{\bs{E}}$. 
From \citet{lin2013} and \citet{fpclt2017}, 
$\hat{\tau}(\hat{\bs{\beta}}, \hat{\bs{\gamma}})$ becomes equivalent to the ordinary least squares (OLS) estimator of the coefficient of the treatment indicator $T$ in the linear regression of $Y_i$ on $T_i$, centered covariate $\bs{C}_i - \bar{\bs{C}}_{\mathcal{S}}$ and their interaction among sampled units in $\mathcal{S}$: 
\begin{align}\label{eq:theta_C_S}
    \hat{\theta}_{\bs{C}_{\mathcal{S}}}
    = 
    \arg\min_{\theta}
    \min_{a, \bs{b}, \bs{c}}
    \sum_{i \in \mathcal{S}}
    \big\{
    Y_i - a - \theta T_i - \bs{b}^\top (\bs{C}_i - \bar{\bs{C}}_{\mathcal{S}})
    - 
    \bs{c}^\top T_i \times  (\bs{C}_i - \bar{\bs{C}}_{\mathcal{S}})
    \big\}^2. 
\end{align}
i.e., $\hat{\tau}(\hat{\bs{\beta}}, \bs{0}) = \hat{\theta}_{\bs{C}_{\mathcal{S}}}$. 
Importantly, 
the covariates $\bs{C}_i$'s are centered at their sample average $\bar{\bs{C}}_{\mathcal{S}}$. 
Below we give some intuition for the regression form in \eqref{eq:theta_C_S} and the importance of the centering of covariates. 
We essentially fit two separate regression lines by OLS for treatment and control groups, respectively, and use the fitted regression lines to impute the potential outcomes for all units in the sample $\mathcal{S}$, only for which the covariate $\bs{C}$ is available.  
We then take the average of the difference between imputed treatment and control potential outcomes over $\mathcal{S}$ as our treatment effect estimate. 
It is not difficult to see that, by the centering of covariates at their sample mean, the resulting average difference is essentially the OLS estimate of the coefficient of $T_i$ in \eqref{eq:theta_C_S}. 

Second, we consider the case in which 
$\bs{E} = \bs{C}$. 
As demonstrated in the Supplementary Material, 
$\hat{\tau}(\hat{\bs{\beta}}, \hat{\bs{\gamma}})$ is equivalent to the OLS estimator of the coefficient of the treatment indicator $T$ in the linear regression of $Y_i$ on $T_i$, centered covariate $\bs{E}_i - \bar{\bs{E}}$ and their interaction among sampled units in $\mathcal{S}$:
\begin{align}\label{eq:theta_E}
    \hat{\theta}_{\bs{E}} 
    = 
    \arg\min_{\theta}
    \min_{a, \bs{b}, \bs{c}}
    \sum_{i \in \mathcal{S}}
    \big\{
    Y_i - a - \theta T_i - \bs{b}^\top (\bs{E}_i - \bar{\bs{E}})
    - 
    \bs{c}^\top T_i \times  (\bs{E}_i - \bar{\bs{E}})
    \big\}^2. 
\end{align}
i.e., $\hat{\tau}(\bs{0}, \hat{\bs{\gamma}}) = \hat{\theta}_{\bs{E}} $. 
Importantly, 
different from \eqref{eq:theta_C_S}, 
the covariates $\bs{E}_i$'s are centered at their population average $\bar{\bs{E}}$, which is 
intuitive given the explanation before for \eqref{eq:theta_C_S}. 
Specifically, we are now able to impute potential outcomes for all $N$ units using the fitted regression lines, since the covariate $\bs{E}$ can be observed for all units. 
The average difference between imputed treatment and control potential outcomes is then the OLS estimate of the coefficient of $T_i$ in \eqref{eq:theta_E}, given that the covariates have been  centered at their population mean.

Third, we consider the general case in which $\bs{E} \subset \bs{C}$.  
Let $\bs{C}^{\res}_i$ be the fitted residual from the linear regression of $\bs{C}_i$ on $\bs{E}_i$ among sampled units in $\mathcal{S}$, 
and $\bar{\bs{C}}^{\res}_{\mathcal{S}}$ be its sample average, which is zero by construction but will be written explicitly to emphasize how we center the covariates.
As demonstrated in the Supplementary Material, 
$\hat{\tau}(\hat{\bs{\beta}}, \hat{\bs{\gamma}})$ is asymptotically equivalent to the OLS estimator of the coefficient of $T$ in the linear regression of $Y_i$ on $T_i$, centered covariates $\bs{C}^{\res}_i - \bar{\bs{C}}^{\res}_{\mathcal{S}}$ and $\bs{E}_i - \bar{\bs{E}}$ 
and their interaction  among sampled units in $\mathcal{S}$: 
\begin{align*}
    \hat{\theta}_{\bs{C}^{\res}_{\mathcal{S}}, \bs{E}}
    & = 
    \arg\min_{\theta}
    \min_{a, \bs{b}, \bs{c}}
    \sum_{i \in \mathcal{S}}
    \left\{
    Y_i - a - \theta T_i - \bs{b}^\top 
    \begin{pmatrix}
        \bs{C}^{\res}_i - \bar{\bs{C}}^{\res}_{\mathcal{S}}
        \\
        \bs{E}_i - \bar{\bs{E}}
    \end{pmatrix}
    - 
    \bs{c}^\top T_i \times 
    \begin{pmatrix}
        \bs{C}^{\res}_i - \bar{\bs{C}}^{\res}_{\mathcal{S}}
        \\
        \bs{E}_i - \bar{\bs{E}}
    \end{pmatrix}
    \right\}^2, 
\end{align*}
i.e., 
$\hat{\tau}(\hat{\bs{\beta}}, \hat{\bs{\gamma}}) = \hat{\theta}_{\bs{C}^{\res}_{\mathcal{S}}, \bs{E}} + o_{\Pr}(n^{-1/2})$. 
This is intuitive given the previous discussion for \eqref{eq:theta_C_S} and \eqref{eq:theta_E}. 
The covariate $\bs{E}$ is observed for all units and can be centered at its population mean, while the covariate $\bs{C}$ is observed only for sampled units and can only be centered at its sample mean. 
Moreover, we only center the part of $\bs{C}$ that cannot be linearly explained by $\bs{E}$ at its sample mean. 

\section{Randomization tests for ReSEM}\label{sec:rand_test_resem}

In Section \ref{sec:large_sample_CI}, we studied \citet{Neyman:1923}'s repeated sampling inference on the average treatment effect. 
Another popular randomization-based inference for treatment effects is the Fisher randomization test \citep{Fisher:1935}, focusing on testing sharp null hypotheses, such as the individual treatment effects are zero or certain constants across all units. 
Randomization tests for rerandomized experiments have been proposed and advocated in \citet{Morgan2012}, \citet{Johansson_long2020}, etc. 
However, it has been less explored for survey experiments. 
Below we will study randomization tests in rerandomized survey experiments, 
which can supplement the inference in Section \ref{sec:large_sample_CI} when large-sample approximation is inaccurate, due to, say, small sample size, heavy-tailed outcomes or too extreme thresholds for rerandomization. 

\subsection{Conditional randomization tests for survey experiments}\label{sec:cond_rand_test}

Unlike usual randomized experiments, 
the null hypothesis that specifies all individual treatment effects, e.g., Fisher's null of no effect for any unit, 
is generally no longer sharp under survey experiments, in the sense that we are not able to impute the potential outcomes for all units using the observed data under the null hypothesis. 
This is because 
neither the treatment nor control potential outcomes are observed for unsampled units. 
Fortunately, this issue can be easily solved by a conditional randomization test that conditions on the sampling vector $\bs{Z}$ (or equivalently the set $\mathcal{S}$ of sampled units). 
In other words, we focus on the sampled units in $\mathcal{S}$ and conduct randomization tests by permuting (or more precisely rerandomizing) only the treatment assignment indicators $\{T_i: i\in \mathcal{S}\}$. 
Not surprisingly, 
to implement the randomization tests, 
we need to know the covariates $\{\bs{X}_i: i \in \mathcal{S}\}$ and the threshold $a_T$ for rerandomization at the treatment assignment stage. 
Throughout this section, we will assume this is true. 

Below we describe in detail the conditional randomization test for a general null hypothesis 
\begin{align}\label{eq:sharp_null}
    H_{\bs{c}}: \tau_i = c_i, \quad 1\le i \le N
\end{align}
that specifies all individual treatment effects, 
where $\bs{c} = (c_1, \ldots, c_N)^\top$. 
For example, when $c_1=c_2=\ldots=c_n=c$ or equivalently $\bs{c} = c \bs{1}$, where $\bs{1}$ here denotes an $N$-dimensional vector with all elements being 1,  the null hypothesis $H_{c\bs{1}}$ reduces to the usual constant effects of size $c$. 
First, we impute the potential outcomes for sampled units in $\mathcal{S}$ based on the observed outcomes and the null hypothesis $H_{\bs{c}}$: 
$\tilde{Y}_i(1) = Y_i + (1-T_i) c_i$ and $\tilde{Y}_i(0) = Y_i - T_i c_i$ for $i \in \mathcal{S}$.
Let $\tilde{\bs{Y}}_{\mathcal{S}}(t)$ be a vector consisting of imputed potential outcomes under treatment arm $t\in \{0,1\}$ for sampled units, 
and 
$\bs{C}_{\mathcal{S}}$ and $\bs{E}_{1:N}$ be matrices consisting of available covariates $\bs{C}$ and $\bs{E}$ in analysis for sampled and all units, respectively. 
Note that the imputed potential outcomes $\tilde{\bs{Y}}_{\mathcal{S}}(0)$ and $\tilde{\bs{Y}}_{\mathcal{S}}(1)$ generally depends on both $\bs{Z}$ and $\bs{T}_{\mathcal{S}}$. However, when the null hypothesis $H_{\bs{c}}$ is true, 
the imputed potential outcomes for the sampled units become the same as the corresponding true potential outcomes, 
no longer depending on $\bs{T}_{\mathcal{S}}$.
Second, we consider a general test statistic of form $g(\bs{T}_{\mathcal{S}}, \tilde{\bs{Y}}_{\mathcal{S}}(0), \tilde{\bs{Y}}_{\mathcal{S}}(1), \bs{C}_{\mathcal{S}}, \bs{E}_{1:N})$, which is a function of the treatment assignment vector, imputed potential outcomes and available covariates. 
The test statistic 
often compares the outcomes of treated and control units (with certain covariate adjustment), e.g., the difference-in-means estimator $\hat{\tau}$ in \eqref{eq:diff_outcome}. 
Third, we impute the randomization distribution of the test statistic under the null hypothesis \eqref{eq:sharp_null}, whose tail probability has the following equivalent forms:
\begin{align}\label{eq:imp_tail_prob}
    G_{\mathcal{S}, \bs{T}}(b)
    & = 
    \Pr\big\{ g( \check{\bs{T}}_{\mathcal{S}}, \tilde{\bs{Y}}_{\mathcal{S}}(0), \tilde{\bs{Y}}_{\mathcal{S}}(1), \bs{C}_{\mathcal{S}}, \bs{E}_{1:N}) \ge b \mid \bs{Z}, \bs{T} \big\}
    \nonumber
    \\
    & = 
    \frac{1}{|\mathcal{A}(\mathcal{S}, \bs{X}, a_T)|}
    \sum_{\bs{t}_{\mathcal{S}} \in \mathcal{A}(\mathcal{S}, \bs{X}, a_T)} 
    \I\big\{ g( \bs{t}_{\mathcal{S}}, \tilde{\bs{Y}}_{\mathcal{S}}(0), \tilde{\bs{Y}}_{\mathcal{S}}(1), \bs{C}_{\mathcal{S}}, \bs{E}_{1:N}) \ge b \big\}, 
\end{align}
where $\check{\bs{T}}_{\mathcal{S}}$ is a random vector independent of $\bs{T}$ given $\bs{Z}$ and satisfies $\check{\bs{T}}_{\mathcal{S}} \mid \bs{Z} \sim \bs{T}_{\mathcal{S}} \mid \bs{Z}$, the set 
$\mathcal{A}(\mathcal{S}, \bs{X}, a_T)$ consists of all acceptable treatment assignments for the sampled units $\mathcal{S}$ under rerandomization using Mahalanobis distance (i.e., ReM) with covariates $\bs{X_i}$'s and threshold $a_T$, 
and $|\mathcal{A}(\mathcal{S}, \bs{X}, a_T)|$ is the cardinality of the set. 
Finally, we calculate the randomization $p$-value, which is the tail probability in \eqref{eq:imp_tail_prob} evaluated at the observed value of the test statistic: 
\begin{align}\label{eq:p_STg}
    p_{\mathcal{S}, \bs{T}, g} 
    & = G_{\mathcal{S}, \bs{T}}
    \big(
    g( \bs{T}_{\mathcal{S}}, \tilde{\bs{Y}}_{\mathcal{S}}(0), \tilde{\bs{Y}}_{\mathcal{S}}(1), \bs{C}_{\mathcal{S}}, \bs{E}_{1:N}) 
    \big). 
\end{align}

The following theorem shows that $p_{\mathcal{S}, \bs{T}, g}$ is a conditionally valid $p$-value for testing $H_{\bs{c}}$, which, by the law of iterated expectation, implies that it is also a marginally valid $p$-value. 

\begin{theorem}\label{thm:cond_rand_test}
Under ReSEM, 
if the null hypothesis $H_{\bs{c}}$ in \eqref{eq:sharp_null} holds, 
then for any test statistic $g(\cdot)$ and any $\alpha\in (0, 1)$, 
the randomization $p$-value $p_{\mathcal{S}, \bs{T}, g}$ in \eqref{eq:p_STg} satisfies 
$
\Pr(p_{\mathcal{S}, \bs{T}, g} \le \alpha \mid \mathcal{S}) \le \alpha
$
and 
$
\Pr(p_{\mathcal{S}, \bs{T}, g} \le \alpha ) \le \alpha. 
$
\end{theorem}

\subsection{Valid randomization tests for weak null hypotheses}

The randomization test described in Section \ref{sec:cond_rand_test} focuses only on null hypotheses of form \eqref{eq:sharp_null} that speculate all individual treatment effects, 
which can be too stringent in practice. 
Below we will 
consider carefully designed
test statistic $g(\cdot)$ such that the resulting $p$-value $p_{\mathcal{S}, \bs{T}, g}$ 
for testing Fisher's null $H_{c\bs{1}}$ of constant effect $c$
is also asymptotically valid for testing Neyman's weak null
hypothesis of average effect $c$: 
\begin{align}\label{eq:weak_null}
    \bar{H}_{c}: \tau = c \ \ \text{or equivalently} \ \ N^{-1} \sum_{i=1}^N \tau_i = c, 
\end{align}
which is usually referred as Neyman's null in contrast to Fisher's null in \eqref{eq:sharp_null}. 
Consequently, the randomization $p$-value $p_{\mathcal{S}, \bs{T}, g}$ enjoys not only exact validity under Fisher's null of form \eqref{eq:sharp_null} but also asymptotic validity under Neyman's null of form \eqref{eq:weak_null}, as studied and advocated by \citet{DingTirthankar2017}, \citet{wuding2020} and \citet{cohen2020gaussian}.

\subsubsection{Intuition based on true and imputed distributions of the test statistic}\label{sec:intuitive_frt}

Suppose we are interested in testing 
Neyman's null $\bar{H}_{c}$ of average effect $c$ for some predetermined constant $c$, and want to utilize the conditional randomization test for Fisher's null $H_{c\bs{1}}$ of constant effects $c$. 
For the randomization test, 
a straightforward choice of the test statistic $g(\cdot)$ is the absolute difference between an estimated and the hypothesized treatment effects, e.g., $|\hat{\tau} - c|$, which, however, can lead to inflated type-I error (even asymptotically) for testing Neyman's null due to the treatment effect heterogeneity. 
Motivated by studentization and prepivoting from \citet{DingTirthankar2017} and \citet{cohen2020gaussian}, 
we construct the test statistic by transforming the absolute difference between treatment effect estimator and the hypothesized effect using the estimated distribution of the effect estimator. For example, we consider the following test statistic for the conditional randomization test 
in Section \ref{sec:cond_rand_test}: 
\begin{align}\label{eq:g_F_diff}
   g(\bs{T}_{\mathcal{S}}, \tilde{\bs{Y}}_{\mathcal{S}}(0), \tilde{\bs{Y}}_{\mathcal{S}}(1), \bs{C}_{\mathcal{S}}, \bs{E}_{1:N}) = 2 \hat{F}\left( \sqrt{n}|\hat{\tau} - c| \right) - 1, 
\end{align}
where $\hat{F}$ denotes the estimated distribution for $\sqrt{n}(\hat{\tau} - c)$ constructed as in Section \ref{sec:estimate_and_CI}.

From Theorem \ref{thm:cond_rand_test}, 
the $p$-value $p_{\mathcal{S}, \bs{T}, g}$ must be exactly valid for testing the Fisher's null $H_{c\bs{1}}$ of constant effect $c$. Below we give some intuition on why the $p$-value $p_{\mathcal{S}, \bs{T}, g}$ with $g$ in \eqref{eq:g_F_diff} can also be asymptotically valid for testing Neyman's null $\bar{H}_{c}$ of average effect $c$. 
Suppose the Neyman's null $\bar{H}_{c}$ is true, i.e., the average effect $\tau = c$. 
To obtain a valid $p$-value for testing the Neyman's null $\bar{H}_{c}$, 
we should compare the observed test statistic in \eqref{eq:g_F_diff} to its {\it true} distribution under ReSEM. 
However, in the conditional randomization test, we compare the observed test statistic to its {\it imputed} distribution as in \eqref{eq:imp_tail_prob}. 
Consequently, to ensure the validity of the randomization test, 
the tail probability of the {\it imputed} distribution need to be larger than or equal to that of the {\it true} distribution,  
at least asymptotically.

We first consider the asymptotic {\it true} distribution of \eqref{eq:g_F_diff}. 
From the previous discussion in Sections \ref{sec:asym_diff} and \ref{sec:estimate_and_CI} on the sampling and estimated distributions, 
the test statistic in \eqref{eq:g_F_diff} converges weakly to a distribution stochastically smaller than or equal to the uniform distribution on $(0,1)$, 
and the difference between them depends on how conservative our estimated distribution is. 
We then consider the asymptotic {\it imputed} distribution of \eqref{eq:g_F_diff}. 
By construction, 
the imputed distribution, with tail probability $G_{\mathcal{S}, \bs{T}}(b)(\cdot)$ in \eqref{eq:imp_tail_prob}, is essentially the distribution of the test statistic \eqref{eq:g_F_diff} under ReM (or equivalently ReSEM with all units enrolled to the experiment, see Remark \ref{rmk:special_dim}) with imputed potential outcomes for the sampled units in $\mathcal{S}$. 
Thus, if the imputed potential outcomes satisfy regularity conditions for finite population asymptotics (which will be true asymptotically as shown in Appendix \ref{sec:proof_rand_test}), 
then, over the random assignments drawn uniformly from $\mathcal{A}(\mathcal{S}, \bs{X}, a_T)$, 
both the difference-in-means estimator and its estimated distribution converge weakly to certain distributions. 
This follows by the same logic as the discussion before for the {\it true} distribution, except now (i) the randomized experiment is constructed by permutation for the purpose of randomization test, (ii) we pretend the imputed potential outcomes as the true potential outcomes and (iii) we pretend there is no sampling stage and consider only the sampled units due to the conditional randomization test. 
More importantly, the two weak limits are the same. 
That is, the estimated distribution, although designed for ReSEM, is consistent for the the asymptotic distribution of $\sqrt{n}(\hat{\tau}-c)$ under ReM with imputed potential outcomes. 
Intuitively, this is because the imputed treatment effects 
are constant across all units, 
under which the rejective sampling under ReSEM does not help improve the estimation precision 
and our generally conservative estimated distribution becomes asymptotically exact. 
Consequently, asymptotically, the imputed distribution of \eqref{eq:g_F_diff} will converge weakly to a uniform distribution on interval $(0,1)$. 

\begin{rmk}
Indeed, with \eqref{eq:g_F_diff}, 
$$1-g(\bs{T}_{\mathcal{S}}, \tilde{\bs{Y}}_{\mathcal{S}}(0), \tilde{\bs{Y}}_{\mathcal{S}}(1), \bs{C}_{\mathcal{S}}, \bs{E}_{1:N}) = \Pr\{ \Unif(0,1) \ge g(\bs{T}_{\mathcal{S}}, \tilde{\bs{Y}}_{\mathcal{S}}(0), \tilde{\bs{Y}}_{\mathcal{S}}(1), \bs{C}_{\mathcal{S}}, \bs{E}_{1:N})\}$$ itself is already a large-sample conservative $p$-value for testing the weak null $\bar{H}_{c}$, 
where $\Unif(0,1)$ denotes a uniform random variable on interval $(0,1)$. 
The randomization $p$-value instead compares the observed value of $g(\bs{T}_{\mathcal{S}}, \tilde{\bs{Y}}_{\mathcal{S}}(0), \tilde{\bs{Y}}_{\mathcal{S}}(1), \bs{C}_{\mathcal{S}}, \bs{E}_{1:N})$ to its imputed distribution obtained from permutation. 
Importantly, 
the randomization $p$-value is not only exactly valid for testing sharp null $H_{c\bs{1}}$ by construction, 
but also asymptotically valid for testing weak null $\bar{H}_{c}$ because the imputed distribution converges weakly to $\Unif(0,1)$ asymptotically as illustrated before. 
Therefore, we can intuitively understand the randomization test as providing a certain finite-sample calibration for the reference distribution of the test statistic $g(\bs{T}_{\mathcal{S}}, \tilde{\bs{Y}}_{\mathcal{S}}(0), \tilde{\bs{Y}}_{\mathcal{S}}(1), \bs{C}_{\mathcal{S}}, \bs{E}_{1:N})$, which provides exact validity under $H_{c\bs{1}}$ and maintains large-sample validity under $\bar{H}_{c\bs{1}}$. 
See \citet{cohen2020gaussian} for related discussion. 
\end{rmk}

\subsubsection{Asymptotically valid covariate-adjusted randomization tests for weak nulls}

To ensure the asymptotic approximation for both the true and imputed distributions of the test statistics, 
we invoke the following regularity condition.  

\begin{condition}\label{cond:fp_est}
    Condition \ref{cond:fp_analysis} holds, and at least one of $S^2_{1\setminus \bs{C}}$ and $S^2_{0\setminus \bs{C}}$ has a positive limit.
\end{condition}

Condition \ref{cond:fp_est} rules out the trivial case where both potential outcomes can be perfectly explained by the covariates.  
which is also needed for the asymptotic conservative inference discussed in Section \ref{sec:estimate_and_CI}. 
The following theorem shows the asymptotic validity of randomization $p$-value for testing weak nulls on average treatment effect using test statistics based on regression-adjusted estimators, 
which extends the previous discussion in Section \ref{sec:intuitive_frt} based on difference-in-means estimator to estimators allowing adjustment for covariate imbalance between treatment and control groups. 

\begin{theorem}\label{thm:frt_weak_null}
For any given $c\in \mathbb{R}$, 
let $p_{\mathcal{S}, \bs{T}, g}$ be a randomization $p$-value for testing Fisher's null of constant effects $c$ using a test statistic of the following form:
\begin{align}\label{eq:frt_general_test_stat}
    g(\bs{T}_{\mathcal{S}}, \tilde{\bs{Y}}_{\mathcal{S}}(0), \tilde{\bs{Y}}_{\mathcal{S}}(1), \bs{C}_{\mathcal{S}}, \bs{E}_{1:N})
    & = 2\hat{F}_{\bs{\beta}, \bs{0}}\left( \left| \hat{\tau}(\bs{\beta}, \bs{0}) - c \right| \right) - 1, 
\end{align}
where $\hat{\tau}(\bs{\beta}, \bs{0})$ is a regression-adjusted estimator as in \eqref{eq:reg}, $\hat{F}_{\bs{\beta}, \bs{0}}$ is the estimated distribution for the asymptotic distribution of $\sqrt{n}\{\hat{\tau}(\bs{\beta}, \bs{0}) - \tau\}$, 
and $\bs{\beta}$ is either any pre-determined fixed adjustment coefficient or the estimated one $\hat{\bs{\beta}}$ defined as in Section \ref{sec:est_beta_gamma_tilde}.  
Then 
under ReSEM and Condition \ref{cond:fp_est}, 
$p_{\mathcal{S}, \bs{T}, g}$ is an asymptotically valid $p$-value for testing the weak null $\bar{H}_c$ of average effect $c$. That is, 
under $\bar{H}_c$, 
$\limsup_{N\rightarrow \infty} \Pr(p_{\mathcal{S}, \bs{T}, g} \le \alpha \mid \text{ReSEM}) \le \alpha$ for any $\alpha\in (0,1)$. 
\end{theorem}

Recently, \citet{zhao2020covariate} studied asymptotically valid covariate-adjusted randomization tests for testing Neyman's null in randomized experiments. 
Theorem \ref{thm:frt_weak_null} generalizes their discussion to survey experiments. 
Note that in Theorem \ref{thm:frt_weak_null}, we do not allow adjustment for the covariate imbalance between sampled units and the whole population, mainly due to the following two reasons (one intuitive and the other more technical). 
First, 
after adjusting the imbalance between sampled units and the whole population, 
the treatment effect estimator is generally conditionally biased (even asymptotically) for $c$ conditional on $\mathcal{S}$ when the Fisher's null $H_{c\bs{1}}$ is true. 
This is 
because $\sqrt{n}\hat{\bs{\delta}}_{\bs{E}}$, which is of order $O_{\Pr}(1)$, becomes a fixed constant conditional on $\mathcal{S}$. 
Therefore, it is more intuitive to use the difference between $c$ and treatment effect estimators of form $\hat{\tau}(\bs{\beta}, \bs{0})$, which are asymptotically unbiased for $c$ conditional on $\mathcal{S}$ under $H_{c\bs{1}}$, to measure the extremeness of the observed data compared to the null hypothesis.
Second, 
if we include the adjustment for $\hat{\bs{\delta}}_{\bs{E}}$, then the imputed distribution of the corresponding test statistic will depend on $\sqrt{n}\hat{\bs{\delta}}_{\bs{E}}$ of order $O_{\Pr}(1)$, 
and 
can no longer be guaranteed to 
dominate 
the true distribution of the test statistic, even asymptotically. 
From the above, including covariate adjustment for the imbalance between sampled units and whole population may invalidate the asymptotic validity of the randomization test for Neyman's nulls on average treatment effect.

Theorem \ref{thm:frt_weak_null} also provides us 
confidence intervals for the true average treatment effect $\tau$ by inverting randomization tests for a sequence of constant treatment effects. 
It is immediate to see that the confidence intervals  are finite-sample exact if the treatment effects are additive and are asymptotically conservative even if the treatment effects are heterogeneous.

\section{Stratified sampling and blocking}\label{sec:sts_block}

In this section, we consider the special case in which $\bs{W} = \bs{X}$ denotes a categorical covariate. 
Without loss of generality, we assume that $\bs{W} = (\I(\tilde{W}=1), \I(\tilde{W}=2), \ldots, \I(\tilde{W}=J))^\top$ is a vector of dummy variables for a categorical covariate $\tilde{W}$ that takes values $1, 2, \ldots, J+1$, for some $J \ge 1$. 
We introduce $\pi_j = N^{-1} \sum_{i=1}^N \I(W_i = j)$ to denote the proportion of units with covariate value $j$ among the whole population of $N$ units, for $1\le j \le J$.

Classical literature \citep{cochran1977, kempthorne1952design} in survey sampling and randomized experiments suggests to stratify the units based on some categorical covariate when  
conducting sampling and treatment assignment, that is, stratified sampling and blocking. 
Specifically, we first stratify the units into $J+1$ strata, each of which consists of units with the same covariates, 
and then conduct CRSE within each stratum $j$, where we use $f_j$ to denote the sampling proportion and $r_{1j}$ and $r_{0j} = 1 - r_{1j}$ to denote the proportion of units assigned to treatment and control, respectively. 
As verified in the Appendix \ref{sec:proof_strata}, the Mahalanobis distances for covariate balance, $M_S$ and $M_T$ in \eqref{eq:M_S} and \eqref{eq:M_T}, are deterministic functions of $f_j$'s and $r_{1j}$'s. 
Moreover, $M_S = M_T = 0$ is equivalent to that $f_1 = f_2 = \ldots = f_{J+1}$ and $r_{11} = r_{12} = \ldots = r_{1,J+1}$. 
Thus, it is not difficult to see that, with categorical covariates, 
ReSEM with the most extreme constraints $M_S = M_T = 0$ is equivalent to stratified sampling and blocking with equal proportions of sampled units as well as treated units across all strata. 
Note that due to integer constraints, it might be impossible to sample or assign equal proportions of units across all strata. 
Indeed, as demonstrated shortly, as long as $(1-f)M_S$ and $M_T$ decreases to zero as $N\rightarrow  \infty$, 
the stratified sampling and blocking can be viewed as ReSEM with perfectly balanced covariates at both sampling and treatment assignment stages. 
We summarize the results in the following theorem. 
Let $n = N \sum_{j=1}^{J+1} \pi_j f_j$ denote the total number of sampled units from all strata, 
and $f = n/N$ be the overall proportion of sampled units. 
We define the squared multiple correlations $R_S^2$ and $R_T^2$ and the variance formula $V_{\tau\tau}$ in the same way as that in \eqref{eq:R2} and \eqref{eq:V_tautau} using the covariate $\bs{W} = \bs{X}$, the vector of dummy variables for the categorical covariate.  

\begin{theorem}\label{thm:resem_cat}
Under stratified sampling and blocking, and 
Condition A5 in Appendix \ref{sec:proof_strata}, 
if both $(1-f) M_S$ and $M_T$ are of order $o(1)$, then, as $N\rightarrow \infty$, the difference-in-means estimator $\hat{\tau}$ in \eqref{eq:diff_outcome} 
has the same weak limit as 
$V_{\tau\tau}^{1/2} \sqrt{1-R_S^2 - R_T^2} \cdot \varepsilon$.  
\end{theorem}

In Theorem \ref{thm:resem_cat}, Condition A5 imposes similar regularity conditions on the finite populations as Condition \ref{cond:fp} and is relegated to Appendix \ref{sec:proof_strata} for conciseness. 
Furthermore, the condition that both $(1-f) M_S$ and $M_T$ are of order $o(1)$ is equivalent to that 
\begin{align*}
    \frac{\max_{1\le j \le J+1}f_j}{\min_{1\le j \le J+1}f_j} - 1 = o(n^{-1/2}), 
    \ \ 
    \frac{\max_{1\le j \le J+1}r_{1j}}{\min_{1\le j \le J+1}r_{1j}} - 1 = o(n^{-1/2}), 
    \ \ 
     \frac{\max_{1\le j \le J+1}r_{0j}}{\min_{1\le j \le J+1}r_{0j}} - 1 = o(n^{-1/2}), 
\end{align*}
in the sense that the proportions of sampled units as well as units assigned to each treatment arm are approximately the same across all strata. 
In practice, these conditions can be easily satisfied when we try to sample and assign equal proportions of units across all strata and use the closest integers when the numbers of sampled or assigned units include fractions; see Appendix \ref{sec:proof_strata} for more details. 

\section{Rerandomization in clustered survey experiments}\label{sec:cluster}

In many applications, a group of units has to be sampled or assigned to the same treatment at the same time, resulting clustered sampling and clustered randomized experiments \citep{raudenbush2007strategies, donner2010design}. In this section, we generalize our previous discussion on survey experiments to clustered survey experiments. 
Assume the $N$ units belongs to $M$ disjoint clusters, and let $G_i \in \{1, 2, \ldots, M\}$ denote the cluster indicator for each unit $i$. 
In a clustered survey experiment, we first sample $m$ clusters from all the $M$ clusters, and then, among the $m$ sampled clusters, we assign $m_1$ to the treatment group and the remaining $m_0 = m-m_1$ to the control group. 
Below we demonstrate that we can view the clustered survey experiment as a special case of the survey experiment we discussed before but with outcomes and covariates aggregated at the cluster level. 

For $1\le l\le M$ and $t=0,1$, let $\tilde{Y}_l(t) = (N/M)^{-1}\sum_{i:G_i = l} Y_i(t)$, $\tilde{\bs{W}}_l = (N/M)^{-1}\sum_{i:G_i = l} \bs{W}_i$ and $\tilde{\bs{X}}_l = (N/M)^{-1}\sum_{i:G_i = l} \bs{X}_i$ be the aggregated potential outcome and covariates for cluster $l$. 
Importantly, we do not aggregate the individual potential outcomes and covariates within each cluster by their averages. 
Instead, we standardize the total outcomes and covariates for each cluster by the same constant, the average cluster size $N/M$; see also \citet{middleton2015unbiased} and \citet{fpclt2017} for related discussion. 
We can verify that the average treatment effect for the aggregated outcomes over all clusters is actually the same as the average treatment effect for the original outcomes over all $N$ units. 
That is, 
$M^{-1} \sum_{l=1}^M\{ \tilde{Y}_l(1) - \tilde{Y}_l(0) \} = 
\tau$. 
Note that under the clustered survey experiment, we are essentially conducting ReSEM at the cluster level. 
Therefore, all the results we derived for ReSEM can be generalized to rerandomized clustered survey experiment, by operating at the cluster level or equivalently viewing each cluster as an individual. 
For example, 
let $\hat{\tilde{\tau}}$ be the difference-in-means of aggregated observed outcomes between treated and control clusters. 
Then, 
under rerandomized clustered survey experiment using Mahalanobis distances and certain regularity conditions, 
$\sqrt{m}(\hat{\tilde{\tau}} - \tau)$ converges weakly to the convolution of a Gaussian and two constrained Gaussian random variables, similar to that in \eqref{eq:dist}. 
For conciseness, we relegate the detailed discussion to Appendix \ref{sec:proof_cluster}.

\setcounter{equation}{0}
\setcounter{section}{0}
\setcounter{figure}{0}
\setcounter{example}{0}
\setcounter{proposition}{0}
\setcounter{corollary}{0}
\setcounter{theorem}{0}
\setcounter{lemma}{0}
\setcounter{table}{0}
\setcounter{condition}{0}

\renewcommand {\theproposition} {B\arabic{proposition}}
\renewcommand {\theexample} {B\arabic{example}}
\renewcommand {\thefigure} {B\arabic{figure}}
\renewcommand {\thetable} {B\arabic{table}}
\renewcommand {\theequation} {B\arabic{section}.\arabic{equation}}
\renewcommand {\thelemma} {B\arabic{lemma}}
\renewcommand {\thesection} {B\arabic{section}}
\renewcommand {\thetheorem} {B\arabic{theorem}}
\renewcommand {\thecorollary} {B\arabic{corollary}}
\renewcommand {\thecondition} {B\arabic{condition}}

\section{Sampling properties under the CRSE}\label{app:crse}

This section includes lemmas and Proposition \ref{prop:R2} for sampling properties of the difference-in-means of outcome and covariates under the CRSE. 
Specifically, 
Section \ref{sec:lemma_crse} includes lemmas that will be used later, 
with the proofs relegated to 
Section \ref{sec:lemma_crse_proof}.
Section \ref{sec:prop_R2} gives the details for the comments on Condition \ref{cond:fp} and the  proof of Proposition \ref{prop:R2}. 

\subsection{Lemmas}\label{sec:lemma_crse}

We first introduce several lemmas for the CRSE. 
Moreover, we emphasize that Lemmas \ref{lemma:clt}, \ref{lemma:s_x2} and \ref{lemma:joint_tilde_CRSE} hold under slightly weaker conditions than Condition \ref{cond:fp}. In particular, they allow $f$ to have limit $1$ as $N\rightarrow \infty$. 
However, the proof for the asymptotic chi-square approximation of $M_S$ in Lemma \ref{lemma:M_ST_CRSE} relies on the condition that the limit of $f$ is less than 1. 

\begin{lemma}\label{lemma:cov_crse}
	Under the CRSE, 
	$\sqrt{n} 
	(\hat\tau -\tau , \hat{\bs{\tau}}_{\bs{X}}^\top, \hat{\bs{\delta}}_{\bs{W}}^\top)^\top $ has mean zero and covariance 
	\begin{align}\label{eq:VV}
	\bs{V} 
	=
	\begin{pmatrix}
	V_{\tau\tau} & \bs{V}_{\tau x} & \bs{V}_{\tau w}\\
	\bs{V}_{x \tau} & \bs{V}_{xx} &\bs{V}_{xw}\\
	\bs{V}_{w \tau} & \bs{V}_{wx} & \bs{V}_{ww}
	\end{pmatrix}
	= 
	\begin{pmatrix}
	r_1^{-1}S^2_{1}+r_0^{-1}S^2_{0} -fS^2_\tau 
	& r_1^{-1}\bs{S}_{1,\bs{X}}+r_0^{-1}\bs{S}_{0,\bs{X}} 
	& \left(1 - f\right)\bs{S}_{\tau,\bs{W}} \\
	r_1^{-1}\bs{S}_{\bs{X},1} +r_0^{-1}\bs{S}_{\bs{X},0}& (r_1r_0)^{-1} \bs{S}^2_{\bs{X}}  & \bs{0} \\
	\left(1 - f\right)\bs{S}_{\bs{W},\tau}& \bs{0} & \left(1  -f \right) \bs{S}^2_{\bs{W}}
	\end{pmatrix}. 
	\end{align}
\end{lemma}

\begin{lemma}\label{lemma:clt}
	Under Condition \ref{cond:fp} and the CRSE, 
	$
	\sqrt{n} 
	(\hat\tau -\tau , \hat{\bs{\tau}}_{\bs{X}}^\top, \hat{\bs{\delta}}_{\bs{W}}^\top)^\top 
	$
	is asymptotically Gaussian 
	with mean zero and covariance matrix $\bs{V}$ in \eqref{eq:VV}, 
	i.e., 
	$
	\sqrt{n} 
	(\hat\tau -\tau , \hat{\bs{\tau}}_{\bs{X}}^\top, \hat{\bs{\delta}}_{\bs{W}}^\top)^\top 
	\dot\sim
	\mathcal{N}(\bs{0}, \bs{V}). 
	$
\end{lemma}

\begin{lemma} \label{lemma:estimate_lemma}
	Let $(Z_1, Z_2, \ldots, Z_N)$ be an indicator vector for a simple random sample of size $n$, i.e., 
	it has probability $\binom{N}{n}^{-1}$ taking value $(z_1, z_2, \ldots, z_N)\in \{0,1\}^N$ if $\sum_{i=1}^N z_i = n$, 
	and zero otherwise. 
	For any finite population $\{(A_i, B_i): i=1,2,\ldots,N\}$, 
	let 
	$\bar{A} = N^{-1} \sum_{i=1}^{N}A_i$ and $\bar{B} = N^{-1}\sum_{i=1}^{N}B_i$ be the finite population averages, 
	$\bar A_1 = n^{-1} \sum_{i=1}^{N} \I (Z_i =1) A_i$
	and 
	$\bar B_1 = n^{-1} \sum_{i=1}^{N} \I (Z_i =1)B_i $ 
	be the sample averages, 
	and 
	$s_{AB} = (n-1)^{-1} \sum_{i=1}^{N} \I (Z_i =1) (A_i -\bar A_1)(B_i -\bar B_1)$ be the sample covariance. 
	Then we have 
	\begin{align*}
	\Var (s_{AB}) \le \frac{4n}{(n-1)^2}\cdot\max_{1 \le j\le N}(A_j-\bar A)^2 \cdot  \frac{1}{N-1} \sum_{i=1}^N (B_i - \bar B)^2. 
	\end{align*}
\end{lemma}

\begin{lemma}\label{lemma:s_x2}
	Under Condition \ref{cond:fp} and the CRSE, 
	$\bs{s}_{\bs{X}}^2 - \bs{S}_{\bs{X}}^2 = o_{\Pr}(1)$. 
\end{lemma}

\begin{lemma}\label{lemma:joint_tilde_CRSE}
    Under Condition \ref{cond:fp} and the CRSE, 
    $
	\sqrt{n} 
	(\hat\tau -\tau , \tilde{\bs{\tau}}_{\bs{X}}^\top, \hat{\bs{\delta}}_{\bs{W}}^\top)^\top \dot\sim
	\mathcal{N}(\bs{0}, \bs{V}), 
	$
	where $\tilde{\bs{\tau}}_{\bs{X}} \equiv (\bs{S}_{\bs{X}}^2)^{1/2}(\bs{s}_{\bs{X}}^2)^{-1/2} \hat{\bs{\tau}}_{\bs{X}}$ and $\bs{V}$ is defined as in \eqref{eq:VV}. 
\end{lemma}

\begin{lemma}\label{lemma:M_ST_CRSE}
Under Condition \ref{cond:fp} and the CRSE, 
$(M_S, M_T) \converged (\chi^2_J, \chi^2_K)$, 
where $\chi^2_J$ and $\chi^2_K$ are two independent chi-square random variables with degrees of freedom $J$ and $K$. 
\end{lemma}

\subsection{Proof of the lemmas}\label{sec:lemma_crse_proof}

\begin{proof}[Proof of Lemma \ref{lemma:cov_crse}]
	From \eqref{eq:srs_cre}, 
	we can view the CRSE as a completely randomized experiment with three groups. 
	The first group consists of sampled units receiving treatment (i.e., $Z_i=T_i=1$), 
	the second group consists of sampled units receiving control (i.e., $Z_i=1$ and $T_i=0$), 
	and 
	the last group consists of unsampled units. 
	We label these three groups as 1, 0 and $-1$. 
	The corresponding group sizes are then $n_1$, $n_0$ and $n_{-1} \equiv N-n$. 
	Moreover, we define a set of pseudo potential  outcome vectors for each unit $i$ as 
	\begin{align}\label{eq:pseudo_po}
	{\bs{U}}_i(1) = 
	\begin{pmatrix}
	Y_i(1) \\ \bs{X}_i   \\ r_1(\bs{W}_i  - \bar{\bs{W}} )
	\end{pmatrix}, 
	\quad
	\bs{U}_i(0) = 
	\begin{pmatrix}
	Y_i(0) \\ \bs{X}_i   \\-  r_0 (\bs{W}_i  - \bar{\bs{W}} )
	\end{pmatrix}, 
	\quad 
	\bs{U}_i(-1) = 
	\begin{pmatrix}
	0 \\ \bs{0}_{K\times 1}  \\ \bs{0}_{J\times 1}
	\end{pmatrix}. 
	\end{align}
	For treatment groups 1 and 0, 
	we use $\bar {\bs{U}}_1 = n_1^{-1}\sum_{i=1}^{N} \I(Z_i = 1, T_i=1) \bs{U}_i(1)$ 
	and 
	$\bar {\bs{U}}_0 = n_0^{-1}\sum_{i=1}^{N} \I(Z_i = 1, T_i=0) \bs{U}_i(0)$ to denote their average outcome vectors, 
	$\bar{Y}_1$ and $\bar{Y}_0$ to denote their average observed outcomes, 
	and $\bar{\bs{W}}_1$ and $\bar{\bs{W}}_0$ to denote their average covariates. 
	Then, by definition, we have 
	\begin{align*}
	\bar{\bs{U}}_1 - \bar{\bs{U}}_0 & = 
	\begin{pmatrix}
	\bar{Y}_1 \\  \bar{\bs{X}}_1    \\ r_1 (\bar{\bs{W}}_1  - \bar{\bs{W}} )
	\end{pmatrix}
	-\begin{pmatrix}
	\bar{Y}_0 \\  \bar{\bs{X}}_0 \\ -r_0 (\bar{\bs{W}}_0  - \bar{\bs{W}} )
	\end{pmatrix} 
	= 
	\begin{pmatrix}
	\bar{Y}_1 - \bar{Y}_0 \\  
	\bar{\bs{X}}_1 - \bar{\bs{X}}_0 \\ 
	\bar{\bs{W}}_{\mathcal{S}}  - \bar{\bs{W}}
	\end{pmatrix}
	= 
	\begin{pmatrix}
	\hat\tau \\ \hat{\bs{\tau}}_{\bs{X}} \\ \hat{\bs{\delta}}_{\bs{W}}
	\end{pmatrix}. 
	\end{align*}
    From \citet[][Theorem 3]{fpclt2017}, 
	under the CRSE, 
	$(\hat\tau, \hat{\bs{\tau}}_{\bs{X}}^\top ,\hat{\bs{\delta}}_{\bs{W}}^\top)^\top = \bar{\bs{U}}_1 - \bar{\bs{U}}_0$ has mean 
	\begin{align*}
	\E\left(
	\begin{pmatrix}
	\hat\tau \\ \hat{\bs{\tau}}_{\bs{X}} \\ \hat{\bs{\delta}}_{\bs{W}}
	\end{pmatrix}
	\right)
	& = 
	\bar{\bs{U}}(1) - \bar{\bs{U}}(0)
	= 
	\begin{pmatrix}
	\bar{Y}(1) \\  \bar{\bs{X}}    \\ r_1 (\bar{\bs{W}}  - \bar{\bs{W}} )
	\end{pmatrix}
	-\begin{pmatrix}
	\bar{Y}(0) \\  \bar{\bs{X}} \\ -r_0 (\bar{\bs{W}}  - \bar{\bs{W}} )
	\end{pmatrix}
	= 
	\begin{pmatrix}
	\tau \\ \bs{0}_{K\times 1}  \\ \bs{0}_{J\times 1}
	\end{pmatrix},
	\end{align*}
	and covariance 
	\begin{align*} 
	\text{Cov}\left(
	\begin{pmatrix}
	\hat\tau \\ \hat{\bs{\tau}}_{\bs{X}} \\ \hat{\bs{\delta}}_{\bs{W}}
	\end{pmatrix}
	\right) 
	& = 
	\frac{1}{n_1} \bs{S}^2_{\bs{U}(1)} + \frac{1}{n_0} \bs{S}^2_{\bs{U}(0)} - \frac{1}{N} \bs{S}^2_{\bs{U}(1) - \bs{U}(0)}\\
	& = 
	\frac{1}{n_1}
	\begin{pmatrix}
	{S}^2_{1} & \bs{S}_{1,\bs{X}} & r_1 \bs{S}_{1, \bs{W}} \\
	\bs{S}_{\bs{X},1}  & \bs{S}^2_{\bs{X}}  & r_1 \bs{S}_{\bs{X}, \bs{W}} \\
	r_1 \bs{S}_{\bs{W},1} & r_1 \bs{S}_{\bs{W}, \bs{X}} & r_1^2 \bs{S}^2_{\bs{W}}
	\end{pmatrix}
	+ 
	\frac{1}{n_0}
	\begin{pmatrix}
	{S}^2_{0} & \bs{S}_{0,\bs{X}} & - r_0 \bs{S}_{0, \bs{W}} \\
	\bs{S}_{\bs{X},0}  & \bs{S}^2_{\bs{X}}  & -r_0 \bs{S}_{\bs{X}, \bs{W}} \\
	-r_0 \bs{S}_{\bs{W},0} & -r_0 \bs{S}_{\bs{W}, \bs{X}} & r_0^2 \bs{S}^2_{\bs{W}}
	\end{pmatrix}\\
	& \quad \ 
	- 
	\frac{1}{N}
	\begin{pmatrix}
	{S}^2_{\tau} & \bs{0} &  \bs{S}_{\tau, \bs{W}} \\
	\bs{0}  & \bs{0}  & \bs{0} \\
	\bs{S}_{\bs{W},\tau} & \bs{0} & \bs{S}^2_{\bs{W}}
	\end{pmatrix}\\
	& =
	\begin{pmatrix}
	\frac{1}{n_1}{S}^2_{1}+\frac{1}{n_0}{S}^2_{0} -\frac{1}{N}{S}^2_\tau & \frac{1}{n_1}\bs{S}_{1,\bs{X}}+\frac{1}{n_0}\bs{S}_{0,\bs{X}} & \left(\frac{1}{n} - \frac{1}{N}\right)\bs{S}_{\tau,\bs{W}} \\
	\frac{1}{n_1}\bs{S}_{\bs{X},1} + \frac{1}{n_0}\bs{S}_{\bs{X},0}& \left(\frac{1}{n_1} +  \frac{1}{n_0}\right) \bs{S}^2_{\bs{X}}  & \bs{0} \\
	\left(\frac{1}{n} - \frac{1}{N}\right)\bs{S}_{\bs{W},\tau}& \bs{0} & \left(\frac{1}{n}  -\frac{1}{N} \right) \bs{S}^2_{\bs{W}}
	\end{pmatrix},
	\end{align*}
	where $\bar{\bs{U}}(1)$ and $\bar{\bs{U}}(0)$ are the finite population averages of the pseudo potential outcomes $\bs{U}(1)$ and $\bs{U}(0)$, 
	and   
	$\bs{S}^2_{\bs{U}(1)}, \bs{S}^2_{\bs{U}(0)}$ and  $\bs{S}^2_{\bs{U}(1) - \bs{U}(0)}$ 
	are the finite population covariance matrices of  
	$\bs{U}(1)$, $\bs{U}(0)$ and $\bs{U}(1)-\bs{U}(0)$. 
	From the above, we can know that Lemma \ref{lemma:cov_crse} holds. 
\end{proof}

\begin{proof}[Proof of Lemma \ref{lemma:clt}]
	Following the proof of Lemma \ref{lemma:cov_crse}, 
	we define pseudo potential outcome vectors the same as in \eqref{eq:pseudo_po}. 
	For any $-1\le q\le 1$ and $1\le l \le 1+K+J$, 
	define 
	\begin{align*}
	m_q(l) & = \max_{1 \le i\le N} \left[ \bs{U}_i(q) - \bar{\bs{U}}(q)\right]_{(l)}^2, 
	\qquad 
	v_q(l) = \frac{1}{N-1}\sum_{i=1}^N[\bs{U}_i(q) -\bar{\bs{U}}(q)]_{(l)}^2, 
	\\ 
	v_\tau(l) & = \frac{1}{N-1}\sum_{i=1}^N[ (\bs{U}_i(1)-\bs{U}_i(0)) -   (\bar{\bs{U}}(1)-\bar{\bs{U}}(0))  ]_{(l)}^2,
	\end{align*}
	where 
	$
	[\bs{u}]_{(l)}
	$
	denotes the $l$th coordinate of a vector $\bs{u}$. 
	By the definition in \eqref{eq:pseudo_po}, 
	\begin{align}\label{eq:max}
	m_q(l) & =  
	\begin{cases}
	\max_{1 \le i\le N} \left[ Y_i(q)- \bar{Y}(q)\right]^2, 
	& \text{if }  l =1 \text{ and } q \in \{0,1\}, \\
	\max_{1 \le i\le N} \left[ \bs{X}_i- \bar{\bs{X}}\right]^2_{(l-1)},
	& \text{if }  2 \le l \le K+1 \text{ and } q \in \{0,1\}, \\
	r_q^2 \max_{1 \le i\le N} \left[ \bs{W}_i- \bar{\bs{W}}\right]^2_{(l-K-1)}, 
	& \text{if }  K+2 \le l \le K+J+1 \text{ and } q \in \{0,1\},\\
	0 & \text{if } q =  -1, 		
	\end{cases}	
	\end{align}
	and for any $1\le l \le K+J+1$, 
	\begin{align*}
	\sum_{q=-1}^{1}n_q^{-1}v_q(l)-N^{-1}v_{\tau}(l) = n^{-1} V_{ll}, 
	\end{align*}	
	where $V_{ll}$ is the $l$th diagonal element of the matrix $\bs{V}$ in \eqref{eq:VV}. 
	For any $-1\le q\le 1$ and $1\le l \le K+J+1$,
	define   
	$$
	G(l,q) = \frac{1}{n_q^2} \frac{m_q(l)}{\sum_{r=-1}^{1}n_r^{-1}v_r(l)-N^{-1}v_{\tau}(l)}
	=  
	\frac{n}{n_q^2} \frac{m_q(l)}{V_{ll}}
	= \frac{m_q(l)}{nr_q^2V_{ll}}. 
	$$ 	
	Then from  \eqref{eq:max}, 
	$G(l,q)$ has the following equivalent form: 
	\begin{align}\label{eq:G}
	G(l,q) = %
	\begin{cases}
	r_q^{-2} V_{ll}^{-1} \cdot n^{-1}\max_{1 \le i\le N} \left[ Y_i(q)- \bar{Y}(q)\right]^2, 
	& \text{if }  l =1 \text{ and } q \in \{0,1\}, \\
	r_q^{-2} V_{ll}^{-1} \cdot n^{-1} \max_{1 \le i\le N} \left[ \bs{X}_i- \bar{\bs{X}}\right]^2_{(l-1)},
	& \text{if }  1 \le l-1 \le K \text{ and } q \in \{0,1\}, \\
	V_{ll}^{-1} \cdot n^{-1}  \max_{1 \le i\le N} \left[ \bs{W}_i- \bar{\bs{W}}\right]^2_{(l-K-1)}, 
	& \text{if }  1 \le l-K-1 \le J \text{ and } q \in \{0,1\},\\
	0 & \text{if } q =  -1. 
	\end{cases}	
	\end{align}
	
	Below we prove the asymptotic Gaussianity of 
	$
	\sqrt{n} 
	(\hat\tau -\tau , \hat{\bs{\tau}}_{\bs{X}}^\top, \hat{\bs{\delta}}_{\bs{W}}^\top)^\top. 
	$
	Under Condition \ref{cond:fp}, $\bs{V}$ in \eqref{eq:VV} must have a limiting value $\bs{V}_{\infty}$ as $N\rightarrow \infty$. 
	We first consider the case where all the diagonal elements of $\bs{V}_\infty$ are positive. 
	In this case, as $N \rightarrow \infty$, 
	$G(l,q)$ in \eqref{eq:G} must converge to zero for all $l$ and $q$, 
	and 
	the correlation matrix of 
	$(\hat\tau , \hat{\bs{\tau}}_{\bs{X}}^\top ,\hat{\bs{\delta}}_{\bs{X}}^\top)^\top$ must converge to 
	$\text{diag} (\bs{V}_{\infty})^{-1/2}\bs{V}_{\infty} \text{diag}(\bs{V}_{\infty})^{-1/2}$, 
	where $\text{diag}(\bs{v})$ denotes a diagonal  matrix whose diagonal elements are the same as that of the square matrix $\bs{v}$. 
	Therefore, 
	from \citet[][Theorem 4]{fpclt2017}, 
	\begin{align*}
	\text{diag}(\bs{V})^{-1/2}  
	\sqrt{n} 
	(\hat\tau -\tau , \hat{\bs{\tau}}_{\bs{X}}^\top, \hat{\bs{\delta}}_{\bs{W}}^\top)^\top
	\converged 
	\mathcal{N}\left(
	\bs{0}, \ 
	\text{diag} (\bs{V}_{\infty})^{-1/2}\bs{V}_{\infty} \text{diag}(\bs{V}_{\infty})^{-1/2}
	\right). 
	\end{align*}
	By Slutsky's theorem, we have
	\begin{align*}
	\sqrt{n} (\hat\tau -\tau, \hat{\bs{\tau}}_{\bs{X}}^\top,\hat{\bs{\delta}}_{\bs{W}}^\top)^\top 
	= 
	\text{diag}(\bs{V})^{1/2} \cdot  
	\text{diag}(\bs{V})^{-1/2}  
	\sqrt{n} 
	(\hat\tau -\tau , \hat{\bs{\tau}}_{\bs{X}}^\top, \hat{\bs{\delta}}_{\bs{W}}^\top)^\top
	\converged \mathcal{N}(\bs{0},\bs{V}_{\infty}).
	\end{align*}
	We then consider the case where some of the diagonal elements of $\bs{V}_\infty$ are zero. 
	Let $\bs{\eta} = \sqrt{n} (\hat\tau -\tau, \hat{\bs{\tau}}_{\bs{X}}^\top ,\hat{\bs{\delta}}_{\bs{W}}^\top)^\top$, 
	and  
	$\mathcal{D}$ be the set of indices corresponding to positive diagonal elements of $\bs{V}_\infty$, 
	i.e., $\bs{V}_{\infty,kk} > 0$ for $k \in \mathcal{D}$ and $\bs{V}_{\infty,kk} = 0$ for $k \notin \mathcal{D}$. 
	Let $\bs{\eta}_{\mathcal{D}}$ be the subvector of $\bs{\eta}$ with indices in $\mathcal{D}$ and 
	$\bs{\eta}_{\mathcal{D}^c}$ be the subvector of $\bs{\eta}$ with indices not in $\mathcal{D}$. 
	Let 
	$\bs{V}_{\infty,  \mathcal{D}\mathcal{D}}$ be the submatrix of $\bs{V}_\infty$ with indices in $\mathcal{D}\times \mathcal{D}$, 
	$\bs{V}_{\infty,  \mathcal{D}^c\mathcal{D}^c}$ be the submatrix of $\bs{V}_\infty$ with indices in $\mathcal{D}^c \times \mathcal{D}^c$, 
	and 
	$\bs{V}_{\infty,  \mathcal{D}^c\mathcal{D}} = \bs{V}_{\infty,  \mathcal{D}\mathcal{D}^c}^\top$ be the submatrix of $\bs{V}_\infty$ with indices in $\mathcal{D}^c \times \mathcal{D}$.  
	By the Cauchy–Schwarz inequality,  for any $k$ and $j$, 
	$|V_{kj}| \le (V_{kk}V_{jj})^{1/2}$. This implies that 
	$V_{\infty, kj}=0$ for $k \notin \mathcal{D}$ and any $j$. Consequently, all elements of matrices 	$\bs{V}_{\infty,  \mathcal{D}^c\mathcal{D}^c}$ and 	$\bs{V}_{\infty,  \mathcal{D}^c\mathcal{D}} = \bs{V}_{\infty,  \mathcal{D}\mathcal{D}^c}^\top$ must be zero. 
	Applying the proof in the first case to the subvector $\bs{\eta}_{\mathcal{D}},$ 
	we can know that  
	$
	\bs{\eta}_{\mathcal{D}} \converged \mathcal{N}(\bs{0}, \bs{V}_{\infty,  \mathcal{D}\mathcal{D}}), 
	$
	where $\bs{V}_{\infty,  \mathcal{D}\mathcal{D}}$ is the submatrix of $\bs{V}_\infty$ with indices in $\mathcal{D}\times \mathcal{D}$. 
	For $k \notin \mathcal{D}$, 
	we have 
	$
	\eta_k = O_{\Pr}\{ \Var^{1/2}(\eta_k)\} = 
	O_{\Pr}( V_{kk}^{1/2}) = o_{\Pr}(1).
	$
	This implies that 
	$
	\bs{\eta}_{\mathcal{D}^c} \converged  \bs{0}.
	$
	Therefore, we must have 
	\begin{align*}
	\begin{pmatrix}
	\bs{\eta}_{\mathcal{D}} \\
	\bs{\eta}_{\mathcal{D}^c}
	\end{pmatrix}
	\converged
	\mathcal{N}
	\left(
	\bs{0}, 
	\begin{pmatrix}
	\bs{V}_{\infty, \mathcal{D}\mathcal{D}} & 
	\bs{0}\\
	\bs{0} & \bs{0}
	\end{pmatrix}
	\right)
	\sim 
	\mathcal{N}
	\left(
	\bs{0}, 
	\begin{pmatrix}
	\bs{V}_{\infty, \mathcal{D} \mathcal{D}} & 
	\bs{V}_{\infty, \mathcal{D} \mathcal{D}^c}\\
	\bs{V}_{\infty, \mathcal{D}^c \mathcal{D}} & \bs{V}_{\infty, \mathcal{D}^c \mathcal{D}^c}
	\end{pmatrix}
	\right).
	\end{align*}
	Equivalently, 
	$
	\bs{\eta} \converged \mathcal{N}(\bs{0}, \bs{V}_{\infty}), 
	$
	i.e., 
	$
	\sqrt{n} (\hat\tau -\tau, \hat{\bs{\tau}}_{\bs{X}}^\top ,\hat{\bs{\delta}}_{\bs{W}}^\top)^\top 
	\converged \mathcal{N}(\bs{0},\bs{V}_{\infty}).
	$
	
	From the above, Lemma \ref{lemma:clt} holds. 
\end{proof}

\begin{proof}[Proof of Lemma \ref{lemma:estimate_lemma}]
	Lemma \ref{lemma:estimate_lemma} follows immediately from the proof of Lemma A15 in the Supplementary Material for \citet{rerand2018}. 
	For completeness, we give a proof here. 
	
	By definition, we can decompose $s_{AB}$ into the following two parts: 
	\begin{align*}
	s_{AB} & = \frac{1}{n-1} \sum_{i:Z_i=1} \left\{(A_i - \bar{A}) - (\bar A_1 - \bar{A})\right\}		
	\left\{
	(B_i - \bar{B})-(\bar B_1 - \bar{B})
	\right\}\\
	& = 
	\frac{n}{n-1} 
	\left\{ \frac{1}{n}
	\sum_{i:Z_i=1} (A_i - \bar{A}) (B_i - \bar{B})
	- (\bar A_1 - \bar{A}) (\bar B_1 - \bar{B})
	\right\}. 
	\end{align*}
	Thus, we can bound the variance of $s_{AB}$ by 
	\begin{align}\label{eq:bound_sab_1}
	\Var (s_{AB}) &= \frac{n^2}{(n-1)^2} \Var\left\{
	\frac{1}{n}
	\sum_{i:Z_i=1} (A_i - \bar{A}) (B_i - \bar{B})
	- (\bar A_1 - \bar{A}) (\bar B_1 - \bar{B})
	\right\}
	\nonumber
	\\ & 
	\le  \frac{2n^2}{(n-1)^2} \left[
	\Var\left\{
	\frac{1}{n}
	\sum_{i:Z_i=1} (A_i - \bar{A}) (B_i - \bar{B})
	\right\} + \Var
	\left\{(\bar  A_1 -\bar A)(\bar B_1 -\bar B)\right\}
	\right].
	\end{align}
	Below we further bound the two terms in \eqref{eq:bound_sab_1}. 
	
	Using the properties of simple random sample, we can bound the first term in \eqref{eq:bound_sab_1} by
	\begin{align*}
	& \quad \ \Var\left\{	\frac{1}{n} \sum_{i:Z_i = 1} (A_i -\bar A)(B_i -\bar B)\right\} 
	\\ & 
	= \left(\frac{1}{n}-\frac{1}{N}\right) \frac{1}{N-1}\sum_{i=1}^N
	\left\{        (A_i -\bar A)(B_i -\bar B) - \frac{1}{N} \sum_{j=1}^N (A_j-\bar A)(B_j -\bar B)              \right\}^2
	\\ &\le   \frac{1}{n}\frac{1}{N-1}\sum_{i=1}^N
	(A_i -\bar A)^2(B_i -\bar B)^2
	\le  \frac{1}{n} \max_{1 \le j\le N}(A_j -\bar A)^2 \cdot\frac{1}{N-1} \sum_{i=1}^N(B_i -\bar B)^2,
	\end{align*}
	and the second term in \eqref{eq:bound_sab_1} by
	\begin{align*}
	& \quad \ \Var \left\{  (\bar{A}_1 -\bar A   ) (\bar{B}_1 -\bar B   )\right\} \\
	& \le
	\E \left\{ (\bar{A}_1 -\bar A   )^2(\bar{B}_1 -\bar B   )^2  \right\}
	\le\max_{1 \le j\le N} (A_j - \bar A)^2 \cdot \E\left\{  ( \bar{B}_1 -\bar B   )^2 \right\}
	= \max_{1 \le j\le N} (A_j- \bar A)^2 \cdot \Var (  \bar{B}_1 )
	\\ & = \max_{1 \le j\le N} (A_j - \bar A)^2 \cdot \left( \frac{1}{n}-\frac{1}{N} \right) \frac{1}{N-1} \ \sum_{i=1}^{N} (B_i - \bar B)^2
	\le  \frac{1}{n} \max_{1 \le i\le N} (A_i - \bar A)^2 \cdot
	\frac{1}{N-1} \sum_{i=1}^{N} (B_i - \bar B)^2.
	\end{align*}
	From the above, we can know that 
	\begin{align*}
	\Var (s_{AB}) 
	& \le 
	\frac{4n}{(n-1)^2}\cdot\max_{1 \le j\le N}(A_j-\bar A)^2 \cdot  \frac{1}{N-1} \sum_{i=1}^N (B_i - \bar B)^2.
	\end{align*}
	Therefore, Lemma \ref{lemma:estimate_lemma} holds. 
\end{proof}

\begin{proof}[Proof of Lemma \ref{lemma:s_x2}]
	For $1\le k \le K$, 
	let  $X_k$ denote the $k$th coordinate of the covariate $\bs{X}$, 
	$X_{ki}$ be the $k$th coordinate of the covariate $\bs{X}_i$ for unit $i$, 
	and $\bar{X}_k$ be the finite population average of the covariate $X_k$. 
	Under the CRSE, by the property of simple random sample,
	for any $1\le k,l\le K$, the sample covariance $s_{X_k, X_l}$ between $X_k$ and $X_l$  for  units in $\mathcal{S}$ is unbiased for their finite population covariance $S_{X_k, X_l}$, i.e., 
	$
	\E(s_{X_k, X_l}) = S_{X_k, X_l}.
	$ 
	Moreover, 
	from Lemma \ref{lemma:estimate_lemma}, 
	the variance of $s_{X_k, X_l}$ can be bounded by 
	\begin{align*}
	\Var(s_{X_k, X_l}) \le \frac{4n}{(n-1)^2}\cdot\max_{1 \le j\le N}(X_{lj}-\bar{X}_l)^2 \cdot  \frac{1}{N-1} \sum_{i=1}^N (X_{ki} - \bar{X}_k)^2,
	\end{align*}
	which converges to zero as  $N\rightarrow \infty$ under Condition \ref{cond:fp}. By Chebyshev's inequality, this implies that
	$
	s_{X_k, X_l} - S_{X_k, X_l} = 
	O_{\Pr}(
	\sqrt{\Var(s_{X_k, X_l})}
	) = o_{\Pr}(1),
	$
	for any $1\le k,l\le K$. 
	Therefore, $\bs{s}_{\bs{X}}^2 - \bs{S}_{\bs{X}}^2 = o_{\Pr}(1)$, i.e., 
	Lemma \ref{lemma:s_x2} holds. 
\end{proof}

\begin{proof}[Proof of Lemma \ref{lemma:joint_tilde_CRSE}]
Lemma \ref{lemma:joint_tilde_CRSE} follows from Lemmas \ref{lemma:clt} and \ref{lemma:s_x2} and Slutsky's theorem. 
\end{proof}

\begin{proof}[Proof of Lemma \ref{lemma:M_ST_CRSE}]
By definition, 
$M_T = \sqrt{n}\tilde{\bs{\tau}}_{\bs{X}}^\top \bs{V}_{xx}^{-1} \sqrt{n}\tilde{\bs{\tau}}_{\bs{X}}$ 
and 
$M_S = \sqrt{n}\hat{\bs{\delta}}_{\bs{W}}^\top \bs{V}_{ww}^{-1} \sqrt{n}\hat{\bs{\delta}}_{\bs{W}}$. 
Therefore, Lemma \ref{lemma:M_ST_CRSE} follows immediately from Lemma \ref{lemma:joint_tilde_CRSE} and continuous mapping theorem. 
\end{proof}

\subsection{Comments on Condition \ref{cond:fp} and Proposition \ref{prop:R2}}\label{sec:prop_R2}

\begin{proof}[{\bf Comments on Condition \ref{cond:fp}}]
    Below we consider Condition \ref{cond:fp}(iii) and (iv) when the potential outcomes and covariates are i.i.d.\ from some distribution. 
    
    First, we consider the case where $f$ has a positive limit, and the distributions for potential outcomes and covariates have more than 2 moments. 
    By the law of large numbers, Condition \ref{cond:fp}(iii) holds with probability one. 
    Let $A_i$ be the variable that can take the value of $Y_i(1)$, $Y_i(0)$,  $X_{ki}$ and $W_{ji}$ for $1\le k \le K$ and $1\le j \le J$, 
    and $\bar{A} = N^{-1} \sum_{i=1}^N A_i$ be the finite population average of $A$. 
    From \citet[][Appendix A1]{fpclt2017}, 
    as $N \rightarrow \infty$, 
    \begin{align*}
        \frac{1}{N}\max_{1\le i \le N}(A_i -\bar{A})^2 \rightarrow 0 \  \text{ almost surely}.
    \end{align*}
    Note that $f = n/N$ has a positive limit as $N\rightarrow \infty$. We must have that, as $N\rightarrow \infty$, \begin{align*}
        \frac{1}{n}\max_{1\le i \le N}(A_i -\bar{A})^2 = \frac{1}{f} \cdot \frac{1}{N}\max_{1\le i \le N}(A_i -\bar{A})^2 \rightarrow 0 \  \text{ almost surely}.
    \end{align*}
    Therefore, Condition \ref{cond:fp}(iv) holds almost surely. 
    
    Second, we consider the case where $f$ has a zero limit, and the distributions for potential outcomes and covariates are all sub-Gaussian. 
    By the property of sub-Gaussian random variables, both potential outcomes and covariates have second moments, and thus, by the law of large numbers, Condition \ref{cond:fp}(iii) holds with probability one. 
    Let $A_i$ be the variable that can take the value of $Y_i(1)$, $Y_i(0)$,  $X_{ki}$ and $W_{ji}$ for $1\le k \le K$ and $1\le j \le J$, $\bar{A} = N^{-1} \sum_{i=1}^N A_i$ be the finite population average of $A$, 
    and 
    $\mu=\E(A_i)$ be the superpopulation mean of $A_i$. 
    Then the maximum distance of $A_i$ from its finite population average can be bounded by 
    \begin{align*}
        \max_{1\le i \le N} \left| A_i - \bar{A} \right| 
        & = \max_{1\le i \le N} \left| \left( A_i - \mu \right)  - \left( \bar{A} - \mu \right) \right| 
        \le \max_{1\le i \le N} \left| A_i - \mu \right| + \left| \bar{A} - \mu \right|\\
        & \le 2\max_{1\le i \le N} \left| A_i - \mu \right|. 
    \end{align*}
    Because $A_i$ is sub-Gaussian, by the property of sub-Gaussian random variables, 
    for any constant $t > 0$,  
    \begin{align*}
        \Pr\left\{  \frac{1}{n}\max_{1\le i \le N}(A_i -\bar{A})^2 \ge t \right\}
        & = 
        \Pr\left(  \max_{1\le i \le N}\left| A_i -\bar{A}\right| \ge \sqrt{n t} \right)
        \le \Pr\left(  \max_{1\le i \le N}\left| A_i -\mu\right| \ge \frac{\sqrt{n t}}{2} \right)\\
        & \le 
        \sum_{i=1}^N \Pr\left(  A_i -\mu \ge \frac{\sqrt{n t}}{2} \right)
        + 
        \sum_{i=1}^N \Pr\left(  -( A_i -\mu) \ge \frac{\sqrt{n t}}{2} \right)
        \\
        & \le 2N \exp\left( - C nt \right),
    \end{align*}
    where $C>0$ is a fixed constant that does not depend on $N$. 
    Because $\log N = o(n)$, 
    there must exist $\underline{N}$ such that when $N \ge \underline{N}$, 
    $\log N/n \le Ct/3$, and thus 
    \begin{align*}
        \Pr\left\{  \frac{1}{n}\max_{1\le i \le N}(A_i -\bar{A})^2 \ge t \right\}
        & \le 
        2\exp\left( \log N- Ctn \right)
        = 2\exp\left\{ - \left( \frac{Ct n}{\log N} - 1 \right) \log N \right\}\\
        & \le 
         2\exp( -2\log N) = \frac{2}{N^2}. 
    \end{align*}
    This implies that 
    \begin{align*}
        & \quad \ \sum_{N=1}^\infty \Pr\left\{  \frac{1}{n}\max_{1\le i \le N}(A_i -\bar{A})^2 \ge t \right\}
        \\
        & \le \sum_{N=1}^{\underline{N}} \Pr\left\{  \frac{1}{n}\max_{1\le i \le N}(A_i -\bar{A})^2 \ge t \right\}
        + \sum_{N=\underline{N}+1}^{\infty} \Pr\left\{  \frac{1}{n}\max_{1\le i \le N}(A_i -\bar{A})^2 \ge t \right\}
        \\
        & \le \underline{N} + \sum_{N=\underline{N}+1}^{\infty} \frac{2}{N^2} < \infty. 
    \end{align*}
    Note that $t$ here can be any positive constant. By the Borel--Cantelli lemma, $n^{-1}\max_{1\le i \le N}(A_i -\bar{A})^2$ converges to zero almost surely. 
    Therefore, Condition \ref{cond:fp}(iv) holds almost surely.
\end{proof}

\begin{proof}[{\bf Proof of Proposition \ref{prop:R2}}]
	First, from Lemma \ref{lemma:cov_crse}, the squared multiple correlation between $\hat{\tau}$ and $\hat{\bs{\delta}}_{\bs{W}}$ is 
	\begin{align*}
	R_S^2 =  \frac{\bs{V}_{\tau w} \bs{V}_{ww}^{-1}\bs{V}_{w\tau}}{ {V}_{\tau\tau} }
	= 
	\frac{
		(1 - f)\bs{\bs{S}}_{\tau,\bs{W}}
		(\bs{\bs{S}}^2_{\bs{W}})^{-1} 
		\bs{\bs{S}}_{\bs{W},\tau} 
	}{
		r_1^{-1}S^2_{1}+r_0^{-1}S^2_{0} -fS^2_\tau
	}
	= 
	\frac{
		(1 - f)\bs{S}^2_{\tau\mid \bs{W}}
	}{
		r_1^{-1}S^2_{1}+r_0^{-1}S^2_{0} -fS^2_\tau
	}.
	\end{align*}
	
	Second, the squared multiple correlation between $\hat{\tau}$ and $\hat{\bs{\tau}}_{\bs{X}}$ is 
	\begin{align}\label{eq:R_t_proof}
	R_T^2 
	& = \frac{\bs{V}_{\tau x} \bs{V}_{x x}^{-1}\bs{V}_{x \tau}}{V_{\tau\tau}}
	= \frac{(r_1^{-1}\bs{S}_{1,\bs{X}} +r_0^{-1}\bs{S}_{0,\bs{X}}) 
		\cdot
		r_0 r_1
		(\bs{S}_{\bs{X}}^2 )^{-1} 
		\cdot 
		(r_1^{-1}\bs{S}_{\bs{X},1} +r_0^{-1}\bs{S}_{\bs{X},0})}{
		r_1^{-1}S^2_{1}+r_0^{-1}S^2_{0} -fS^2_\tau
	}
	\end{align}
	Note that the numerator in \eqref{eq:R_t_proof} has the following equivalent forms: 
	\begin{align*}
	& \quad \ 
	(r_1^{-1}\bs{S}_{1,\bs{X}} +r_0^{-1}\bs{S}_{0,\bs{X}}) 
	\cdot
	r_0 r_1
	(\bs{S}_{\bs{X}}^2 )^{-1} 
	\cdot 
	(r_1^{-1}\bs{S}_{\bs{X},1} +r_0^{-1}\bs{S}_{\bs{X},0})\\
	& = 
	r_0 r_1^{-1} S^2_{1\mid \bs{X}} + r_1 r_0^{-1} S^2_{0\mid \bs{X}} + 
	2\bs{S}_{0,\bs{X}} (\bs{S}_{\bs{X}}^2)^{-1}\bs{S}_{\bs{X},1}\\
	& = 
	(r_0 r_1^{-1}+1) S^2_{1\mid \bs{X}} + (r_1 r_0^{-1}+1) S^2_{0\mid \bs{X}} - 
	\big\{
	S^2_{1\mid \bs{X}} + S^2_{0\mid \bs{X}}
	- 2\bs{S}_{0,\bs{X}} (\bs{S}_{\bs{X}}^2)^{-1}\bs{S}_{\bs{X},1}
	\big\}\\
	& = r_1^{-1} S^2_{1\mid \bs{X}} + r_0^{-1}S^2_{0\mid \bs{X}} -
	(\bs{S}_{1,\bs{X}} - \bs{S}_{0,\bs{X}} )^\top 
	(\bs{S}_{\bs{X}}^2)^{-1}
	(\bs{S}_{1,\bs{X}} - \bs{S}_{0,\bs{X}} )\\
	& = 
	r_1^{-1} S^2_{1\mid \bs{X}} + r_0^{-1}S^2_{0\mid \bs{X}} - 
	S^2_{\tau\mid \bs{X}}. 
	\end{align*}
	Thus, the squared multiple correlation $R_T^2$ has the following equivalent form: 
	\begin{align*}
	R_T^2 & = 
	\frac{
		r_1^{-1} S^2_{1\mid \bs{X}} + r_0^{-1}S^2_{0\mid \bs{X}} - 
		S^2_{\tau\mid \bs{X}}
	}{
		r_1^{-1}S^2_{1}+r_0^{-1}S^2_{0} -fS^2_\tau
	}. 
	\end{align*}
	
	From the above, Proposition \ref{prop:R2} holds. 
\end{proof}

\section{Asymptotic equivalence between single-stage and two-stage rerandomized survey experiments}\label{sec:equiv_resem_single_two}
For descriptive convenience, 
we use $\sresem$ to denote the single-stage rerandomized survey experiment, under which sampling and treatment assignment vectors $(\bs{Z}, \bs{T}_{\mathcal{S}})$ are acceptable if and only if the corresponding Mahalanobis distances satisfy that $M_S\le a_S$ and $M_T\le a_T$, as illustrated in Figure \ref{figure:flowchart_single_stage}. 
For clarification, in the following discussion, we use $\Pr(\cdot)$ exclusively for the probability distribution under the CRSE, 
and use $\Pr(\cdot \mid \text{ReSEM})$ and $\Pr(\cdot \mid \sresem)$ for that under ReSEM and $\sresem$, respectively. 
Here we discuss some subtle issues. 
First, there may be no acceptable sampling and assignment under $\sresem$ or ReSEM. 
Second, under our two-stage ReSEM, it is possible that in the second stage there is no acceptable treatment assignments. For convenience, if these happen, we can define the distribution of the sampling or treatment assignment vectors arbitrarily. As shown in Lemma \ref{lemma:M_ST_CRSE} and demonstrated shortly in Lemma \ref{lemma:resem_prob_2stage}, these events will happen with negligible probability as the sample size goes to infinity under certain regularity conditions.

\usetikzlibrary{positioning}
\begin{figure}[btp] 
\begin{center}
\begin{tikzpicture}[node distance=8mm]
\node (S1) [io, text width=2.5cm] {Collect $\bs{W_i}$ for all $i$};
\node (S2) [process, right=  of S1, text width=4.5cm] { Randomly sample $n$ units, and compute $M_S$};
\node (S3) [io, right=  of S2, text width=2.5cm] { Collect $\bs{X_i}$ for $i \in \mathcal{S}$ };
\node (S5) [decision, aspect=2.6, below= of S2
] {$M_S \le a_S$ and $M_T \le a_T$?};
\node (S6) [process, left=of S5, text width=2.5cm] {Conduct the  experiment };
\node (S4) [process, right= of S5, text width=5.5cm] {  Randomly assign the selected $n$ units into treatment and control, and compute $M_T$};
\draw [arrow] (S1) -- (S2);
\draw [arrow] (S2) -- (S3);
\draw [arrow] (S3) -| (S4);
\draw [arrow] (S4) -- (S5);
\draw [arrow] (S5) -- node [above] {yes}  (S6);
\draw [arrow](S5.north) -- node [left] {no} (S2);
\end{tikzpicture}
\end{center}
\caption{Procedure for conducting the single-stage $\sresem$. %
}\label{figure:flowchart_single_stage} 
\end{figure}

In this section, we will prove the following key theorem for the asymptotic equivalence between the two designs, ReSEM and $\sresem$. 
Let 
\begin{align}\label{eq:total_variation}
    & \quad \text{d}_{\TV}(\sresem, \text{ReSEM}) 
    \nonumber
    \\
    & \equiv
    \sup_{\mathcal{A}\subset \{0,1\}^N \times \{0,1\}^n}
    \left|
    \Pr\big\{ (\bs{Z}, \bs{T}_{\mathcal{S}} )\in \mathcal{A} \mid \sresem \big\}
    - 
    \Pr\big\{ (\bs{Z}, \bs{T}_{\mathcal{S}} )\in \mathcal{A} \mid \text{ReSEM} \big\}
    \right|
\end{align}
denote the total variance distance between the probability distributions of $(\bs{Z}, \bs{T}_{\mathcal{S}})$ under $\sresem$ and ReSEM.
\begin{theorem}\label{thm:equiv_resem_sresem}
Under Condition \ref{cond:fp}, 
$\text{d}_{\TV}(\sresem, \text{ReSEM}) \converge 0$ as $N\rightarrow \infty$.
\end{theorem}

Theorem \ref{thm:equiv_resem_sresem} has important implications. 
Specifically, the asymptotic properties of any estimator, e.g., its asymptotic distribution or consistency, and the asymptotic validity of confidence sets and $p$-values are shared under these two designs.
Therefore, to prove the asymptotic properties for the two-stage ReSEM, it suffices to prove that for the single-stage $\sresem$. 
We summarize some useful results in the following corollary. 

\begin{corollary}\label{cor:equiv_resem_sresem}
    Let $\bs{\theta}(\bs{Z}, \bs{T}_{\mathcal{S}}, \Pi_N)$ and $\Theta(\bs{Z}, \bs{T}_{\mathcal{S}}, \Pi_N)$ be any vector- and set-valued quantities that are uniquely determined by the sampling and treatment assignment indicators and any finite population $\Pi_N$, 
    and $\bs{\theta}_N$ be any vector-valued quantity determined by $\Pi_N$. 
    Let $\sresem$ and ReSEM denote the single- and two-stage rerandomized survey experiments using covariates $\bs{W}$ at the sampling stage and $\bs{X}$ at the treatment assignment stage. 
    Assume that Condition \ref{cond:fp} holds. 
    \begin{enumerate}[label=(\roman*), topsep=1ex,itemsep=-0.3ex,partopsep=1ex,parsep=1ex]
        \item If $\sqrt{n}\{ \bs{\theta}(\bs{Z}, \bs{T}_{\mathcal{S}}, \Pi_N) - \bs{\theta}_N \} \mid \sresem \converged \bs{\xi}$ for some random vector $\bs{\xi}$, 
        then $\sqrt{n}\{ \bs{\theta}(\bs{Z}, \bs{T}_{\mathcal{S}}, \Pi_N) - \bs{\theta}_N \} \mid \text{ReSEM} \converged \bs{\xi}$.
        
        \item If $\bs{\theta}(\bs{Z}, \bs{T}_{\mathcal{S}}, \Pi_N) - \bs{\theta}_N \mid  \sresem \convergep 0$, 
        then $\bs{\theta}(\bs{Z}, \bs{T}_{\mathcal{S}}, \Pi_N) - \bs{\theta}_N \mid  \text{ReSEM} \convergep 0$. 
        
        \item If $\limsup_{N\rightarrow \infty}\Pr\{ \bs{\theta}(\bs{Z}, \bs{T}_{\mathcal{S}}, \Pi_N) \le \alpha \mid  \sresem \} \le \alpha$ for $\alpha \in (0,1)$, then 
        $\limsup_{N\rightarrow \infty}\Pr\{ \bs{\theta}(\bs{Z}, \bs{T}_{\mathcal{S}}, \Pi_N) \le \alpha \mid \text{ReSEM} \} \le \alpha$.
        
        \item If $\liminf_{N\rightarrow \infty} \Pr\{\bs{\theta}_N \in \Theta(\bs{Z}, \bs{T}_{\mathcal{S}}, \Pi_N) \mid \sresem \} \ge 1-\alpha$ for some $\alpha\in(0,1)$, then 
        $\liminf_{N\rightarrow \infty} \Pr\{\bs{\theta}_N \in \Theta(\bs{Z}, \bs{T}_{\mathcal{S}}, \Pi_N) \mid \text{ReSEM} \} \ge 1-\alpha$.
    \end{enumerate}
\end{corollary}

\subsection{Technical lemmas}

\begin{lemma}\label{lemma:total_variation}
Let $\mathcal{M} \subset \{0,1\}^N \times \{0,1\}^n$ denote the set of all possible values of $(\bs{Z}, \bs{T}_{\mathcal{S}})$ under $\sresem$, i.e., $\mathcal{M}$ consists of all acceptable sampling and treatment assignment vectors under $\sresem$. 
The total variation distance between the probability distributions of $(\bs{Z}, \bs{T}_{\mathcal{S}})$ under $\sresem$ and ReSEM, as defined in \eqref{eq:total_variation}, can be bounded by: 
\begin{align*}%
    & \quad \text{d}_{\TV}(\sresem, \text{ReSEM}) 
    \\
    & \equiv
    \sup_{\mathcal{A}\subset \{0,1\}^N \times \{0,1\}^n}
    \left|
    \Pr\big\{ (\bs{Z}, \bs{T}_{\mathcal{S}} )\in \mathcal{A} \mid \sresem \big\}
    - 
    \Pr\big\{ (\bs{Z}, \bs{T}_{\mathcal{S}} )\in \mathcal{A} \mid \text{ReSEM} \big\}
    \right|
    \nonumber
    \\
    & \le 
    \I(\mathcal{M} = \emptyset) 
    + 
    \I(\mathcal{M} \neq \emptyset) \cdot
    \frac{\E \left| \Pr(M_T \le a_T \mid \bs{Z} ) - 
    \Pr(M_T\le a_T \mid M_S\le a_S) \right|}{ \Pr(M_T\le a_T, M_S\le a_S)}.
    \nonumber
\end{align*}
\end{lemma}

\begin{lemma}\label{lemma:resem_prob_2stage}
Consider a survey experiment with covariates $\bs{W}$ at the sampling stage and $\bs{X}$ at the assignment stage. 
Assume that  Condition \ref{cond:fp} holds, $\bs{Z}$ is from rejective sampling with size $n$ and criterion $M_S\le a_S$ for some positive $a_S$, 
and $\bs{T}_{\mathcal{S}}$ is from a CRE with $n_1$ and $n_0$ units assigned to treatment and control, respectively, where $\mathcal{S}$ is the set of sampled units. 
Then 
\begin{enumerate}[label=(\roman*), topsep=1ex,itemsep=-0.3ex,partopsep=1ex,parsep=1ex]
    \item $\bs{s}_{\bs{X}}^2 - \bs{S}_{\bs{X}}^2 = o_{\Pr}(1)$, recalling that $\bs{s}^2_{\bs{X}}$ and $\bs{S}^2_{\bs{X}}$ are the sample and finite population covariance matrices of $\bs{X}$, respectively; 
    
    \item for any positive $a_T$, 
    $
    \Pr(M_T \le a_T \mid \bs{Z}) \convergep \Pr(\chi^2_K \le a_T)
    $
    as $N\rightarrow \infty$. 
\end{enumerate}
\end{lemma}

\subsection{Proofs of the lemmas}

\begin{proof}[Proof of Lemma \ref{lemma:total_variation}]
Let $\mathcal{M}_1$ denote the set of sampling vectors such that the corresponding $M_S\le a_S$, 
and $\mathcal{M}_2(\bs{z}) =\{\bs{t}: (\bs{z}, \bs{t})\in \mathcal{M}\}$ denote the acceptable treatment assignment vector when the sampling vector takes value $\bs{z}$.
Let $\mathcal{M}_1' = \{\bs{z}: (\bs{z}, \bs{t})\in \mathcal{M} \text{ for some } \bs{t} \in \{0,1\}^n \text{ and } \sum_{i=1}^n t_i = n_1 \}$ denote the set of acceptable sampling under $\sresem$. 
Note that $\mathcal{M}_1'$ may be different from $\mathcal{M}_1$. 
We will use $|\mathcal{M}|, |\mathcal{M}_1|$ and $|\mathcal{M}_2(\bs{z})|$ to denote the cardinalities of these sets. 
Note that Lemma \ref{lemma:total_variation} holds obviously when $\mathcal{M} = \emptyset$. 
Below we consider only the case where $\mathcal{M}\ne \emptyset$. 

We first consider the difference between $\sresem$ and ReSEM for sampling and assignment vectors in $\mathcal{M}$. 
For any $(\bs{z}, \bs{t})\in \mathcal{M}$, we have 
\begin{align*}
    \Pr(\bs{Z} = \bs{z}, \bs{T}_{\mathcal{S}} = \bs{t} \mid \sresem)
    & = 
    \frac{1}{|\mathcal{M}|}
    = 
    \frac{1}{\binom{N}{n_1, n_0} \cdot \Pr(M_T\le a_T, M_S \le a_S)}, 
\end{align*}
and 
\begin{align*}
    & \quad \ \Pr(\bs{Z} = \bs{z}, \bs{T}_{\mathcal{S}} = \bs{t} \mid \text{ReSEM})
    \\
    & = \Pr(\bs{Z} = \bs{z} \mid \text{ReSEM}) 
    \cdot
    \Pr(\bs{T}_{\mathcal{S}} = \bs{t} \mid \bs{Z} = \bs{z}, \text{ReSEM})
    = 
    \frac{1}{|\mathcal{M}_1|} \cdot \frac{1}{|\mathcal{M}_2(\bs{z})|}
    \\
    & = 
    \frac{1}{\binom{N}{n} \cdot\Pr(M_S\le a_S)} \cdot
    \frac{1}{\binom{n}{n_1} \cdot \Pr(M_T \le a_T \mid \bs{Z} = \bs{z})}
    \\
    & = 
    \frac{1}{\binom{N}{n_1, n_0} \cdot \Pr(M_S\le a_S) \cdot \Pr(M_T \le a_T \mid \bs{Z} = \bs{z})}. 
\end{align*}
These imply that, for any $(\bs{z}, \bs{t})\in \mathcal{M}$, 
\begin{align*}
    & \quad \ 
    \Pr(\bs{Z} = \bs{z}, \bs{T}_{\mathcal{S}} = \bs{t} \mid \sresem)
    - 
    \Pr(\bs{Z} = \bs{z}, \bs{T}_{\mathcal{S}} = \bs{t} \mid \text{ReSEM})
    \\
    & = 
    \frac{\Pr(M_T \le a_T \mid \bs{Z} = \bs{z}) - 
    \Pr(M_T\le a_T \mid M_S\le a_S)
    }{\binom{N}{n_1, n_0} \cdot \Pr(M_T\le a_T, M_S \le a_S) \cdot \Pr(M_T \le a_T \mid \bs{Z} = \bs{z})}.
\end{align*}

Consequently, we have
\begin{align*}
    & \quad \ 
    \sum_{(\bs{z}, \bs{t})\in \mathcal{M}}
    \left|
    \Pr(\bs{Z} = \bs{z}, \bs{T}_{\mathcal{S}} = \bs{t} \mid \sresem)
    - 
    \Pr(\bs{Z} = \bs{z}, \bs{T}_{\mathcal{S}} = \bs{t} \mid \text{ReSEM})
    \right|
    \\
    & = 
    \sum_{\bs{z}\in \mathcal{M}_1'} |\mathcal{M}_2(\bs{z})| \cdot
    \frac{|\Pr(M_T \le a_T \mid \bs{Z} = \bs{z}) - 
    \Pr(M_T\le a_T \mid M_S\le a_S)|}{\binom{N}{n_1, n_0} \cdot \Pr(M_T\le a_T, M_S \le a_S) \cdot \Pr(M_T \le a_T \mid \bs{Z} = \bs{z})}
    \\
    & = 
    \sum_{\bs{z}\in \mathcal{M}_1'} \binom{n}{n_1} \cdot \Pr(M_T \le a_T \mid \bs{Z} = \bs{z}) \cdot
    \frac{|\Pr(M_T \le a_T \mid \bs{Z} = \bs{z}) - 
    \Pr(M_T\le a_T \mid M_S\le a_S)|}{\binom{N}{n_1, n_0} \cdot \Pr(M_T\le a_T, M_S \le a_S) \cdot \Pr(M_T \le a_T \mid \bs{Z} = \bs{z})}
    \\
    & = 
    \frac{1}{ \Pr(M_T\le a_T, M_S\le a_S)} \cdot
    \frac{1}{\binom{N}{n}}
    \sum_{\bs{z}\in \mathcal{M}_1'} 
    \left|\Pr(M_T \le a_T \mid \bs{Z} = \bs{z}) - 
    \Pr(M_T\le a_T \mid M_S\le a_S)\right|
    \\
    & \le 
    \frac{1}{ \Pr(M_T\le a_T, M_S\le a_S)} \cdot
    \frac{1}{\binom{N}{n}}
    \sum_{\bs{z}\in \{0,1\}^N: \sum z_i = n} 
    \left|\Pr(M_T \le a_T \mid \bs{Z} = \bs{z}) - 
    \Pr(M_T\le a_T \mid M_S\le a_S)\right|
    \\
    & = \frac{1}{ \Pr(M_T\le a_T, M_S\le a_S)} \cdot
    \E \left| \Pr(M_T \le a_T \mid \bs{Z} ) - 
    \Pr(M_T\le a_T \mid M_S\le a_S) \right|. 
\end{align*}

We then consider the difference between $\sresem$ and ReSEM for sampling and assignment vectors not in $\mathcal{M}$. 
The difference can be bounded in the following way:
\begin{align*}
    & \quad \ 
    \sum_{(\bs{z}, \bs{t})\notin \mathcal{M}}
    \left|
    \Pr(\bs{Z} = \bs{z}, \bs{T}_{\mathcal{S}} = \bs{t} \mid \sresem)
    - 
    \Pr(\bs{Z} = \bs{z}, \bs{T}_{\mathcal{S}} = \bs{t} \mid \text{ReSEM})
    \right|
    \\
    & = 
    \sum_{(\bs{z}, \bs{t})\notin \mathcal{M}}
    \Pr(\bs{Z} = \bs{z}, \bs{T}_{\mathcal{S}} = \bs{t} \mid \text{ReSEM})
    = 
    1 - \sum_{(\bs{z}, \bs{t})\in \mathcal{M}}
    \Pr(\bs{Z} = \bs{z}, \bs{T}_{\mathcal{S}} = \bs{t} \mid \text{ReSEM})
    \\
    & = 
    \sum_{(\bs{z}, \bs{t})\in \mathcal{M}}
    \Pr(\bs{Z} = \bs{z}, \bs{T}_{\mathcal{S}} = \bs{t} \mid \sresem) - \sum_{(\bs{z}, \bs{t})\in \mathcal{M}}
    \Pr(\bs{Z} = \bs{z}, \bs{T}_{\mathcal{S}} = \bs{t} \mid \text{ReSEM})
    \\
    & \le 
    \sum_{(\bs{z}, \bs{t})\in \mathcal{M}}
    \left|
    \Pr(\bs{Z} = \bs{z}, \bs{T}_{\mathcal{S}} = \bs{t} \mid \sresem)
    - 
    \Pr(\bs{Z} = \bs{z}, \bs{T}_{\mathcal{S}} = \bs{t} \mid \text{ReSEM})
    \right|. 
\end{align*}

From the above and by the property of total variation distance for discrete measures, we have
\begin{align*}
    & \quad \text{d}_{\TV}(\sresem, \text{ReSEM}) 
    \\
    & = \frac{1}{2} \sum_{(\bs{z},\bs{t})\in \mathcal{M} } 
    \left|
    \Pr(\bs{Z} = \bs{z}, \bs{T}_{\mathcal{S}} = \bs{t} \mid \sresem)
    - 
    \Pr(\bs{Z} = \bs{z}, \bs{T}_{\mathcal{S}} = \bs{t} \mid \text{ReSEM})
    \right|
    \\
    & \quad \ + 
    \frac{1}{2} \sum_{(\bs{z},\bs{t})\notin \mathcal{M} } 
    \left|
    \Pr(\bs{Z} = \bs{z}, \bs{T}_{\mathcal{S}} = \bs{t} \mid \sresem)
    - 
    \Pr(\bs{Z} = \bs{z}, \bs{T}_{\mathcal{S}} = \bs{t} \mid \text{ReSEM})
    \right|
    \\
    & 
    \le \sum_{(\bs{z},\bs{t})\in \mathcal{M} } 
    \left|
    \Pr(\bs{Z} = \bs{z}, \bs{T}_{\mathcal{S}} = \bs{t} \mid \sresem)
    - 
    \Pr(\bs{Z} = \bs{z}, \bs{T}_{\mathcal{S}} = \bs{t} \mid \text{ReSEM})
    \right|
    \\
    & \le 
    \frac{\E | \Pr(M_T \le a_T \mid \bs{Z} ) - 
    \Pr(M_T\le a_T \mid M_S\le a_S) |}{ \Pr(M_T\le a_T, M_S\le a_S)}.
\end{align*}
Therefore, Lemma \ref{lemma:total_variation} holds. 
\end{proof}

\begin{proof}[Proof of Lemma \ref{lemma:resem_prob_2stage}]
First, we prove (i), i.e.,  $\bs{s}^2_{\bs{X}} - \bs{S}^2_{\bs{X}} = o_{\Pr}(1)$ under rejective sampling. 
From Lemma \ref{lemma:M_ST_CRSE} and the proof of Lemma \ref{lemma:s_x2}, we can know that, for any $1\le k,l\le K$, 
\begin{align*}
    \E\{ (s_{X_k, X_l} - S_{X_k, X_l})^2 \mid M_S\le a_S \}
    & 
    \le 
    \Pr(M_S\le a_S)^{-1} \cdot \E\{ (s_{X_k, X_l} - S_{X_k, X_l})^2  \}
    \\
    & = \Pr(M_S\le a_S)^{-1} \cdot \Var(s_{X_k, X_l} ) 
    = o(1). 
\end{align*}
By the Markov inequality, we then have $s_{X_k, X_l} - S_{X_k, X_l}=o_{\Pr}(1)$ under rejective sampling. 
Thus, $\bs{s}^2_{\bs{X}} - \bs{S}^2_{\bs{X}} = o_{\Pr}(1)$ under rejective sampling.

Second, we prove that $n^{-1} \max_{i\in \mathcal{S}} \|\bs{X}_{i} - \bar{\bs{X}}_{\mathcal{S}}\|_2^2 \rightarrow 0$ as $N\rightarrow \infty$. 
For each sampled unit $i\in \mathcal{S}$, its distance from the sample mean can be bounded by 
\begin{align*}
    \|\bs{X}_{i} - \bar{\bs{X}}_{\mathcal{S}}\|_2 
    & \le 
    \|\bs{X}_{i} - \bar{\bs{X}}\|_2 + \|\bar{\bs{X}}_{\mathcal{S}} - \bar{\bs{X}} \|_2 
    \le 2 \max_{1\le j\le N} \|\bs{X}_{j} - \bar{\bs{X}}\|_2, 
\end{align*}
which implies that $n^{-1} \max_{i:Z_i=1} \|\bs{X}_{i} - \bar{\bs{X}}_{\mathcal{S}}\|_2^2 \le 4 \cdot n^{-1} \max_{1\le j\le N} \|\bs{X}_{j} - \bar{\bs{X}}\|_2^2$. 
From Condition \ref{cond:fp}(iv), we must have $n^{-1} \max_{i\in \mathcal{S}} \|\bs{X}_{i} - \bar{\bs{X}}_{\mathcal{S}}\|_2^2 = o(1)$.

Third, we prove (ii), i.e., for any given positive $a_T$, 
$
\Pr(M_T \le a_T \mid \bs{Z}) \convergep \Pr(\chi^2_K \le a_T)
$
as $N\rightarrow \infty$. 
Note that given $\bs{Z}$, $\bs{T}_{\mathcal{S}}$ is from a CRE on the sampled units. 
By the finite population central limit theorem for the CRE (e.g., Lemma \ref{lemma:clt} with $f=1$), 
if,  as $N\rightarrow \infty$, 
the following regularity condition holds almost surely for the sampled units in $\mathcal{S}$: 
\begin{enumerate}[label=(\roman*), topsep=1ex,itemsep=-0.3ex,partopsep=1ex,parsep=1ex]
    \item $r_1$ and $r_0$ have positive limits,
    \item $\bs{s}^2_{\bs{X}}$ has a nonsingular limit, 
    \item $n^{-1} \max_{i\in \mathcal{S}} \|\bs{X}_{i} - \bar{\bs{X}}_{\mathcal{S}}\|_2^2$ converges to zero, 
\end{enumerate}
then we must have $\Pr(M_T\le a_T\mid \bs{Z}) \convergeas \Pr(\chi^2_K \le a_T)$. 
From the discussion before and by the property of convergence in probability \citep[e.g.,][Theorem 2.3.2]{durrett2019probability}, 
we can derive that $\Pr(M_T\le a_T\mid \bs{Z}) \convergep \Pr(\chi^2_K \le a_T)$ as $N\rightarrow \infty$. 

From the above, Lemma \ref{lemma:resem_prob_2stage} holds. 
\end{proof}

\subsection{Proof of the theorem and corollary}

\begin{proof}[\bf Proof of Theorem \ref{thm:equiv_resem_sresem}]
First, from Lemma \ref{lemma:M_ST_CRSE}, 
as $N\rightarrow \infty$, we have 
$\Pr(M_S\le a_S, M_T\le a_T) \converge \Pr(\chi^2_J\le a_S) \Pr(\chi^2_K \le a_T) > 0$ 
and 
$\Pr(M_S\le a_S) \converge \Pr(\chi^2_J\le a_S)$. 
These imply that 
$\Pr(M_T\le a_T \mid M_S\le a_S) \converge \Pr(\chi^2_K \le a_T)$, 
and $\mathcal{M}$ is not an empty set when $N$ is sufficient large. 
Consequently, $\I(\mathcal{M} = \emptyset) \converge 0$ as $N\rightarrow \infty$.

Second, from Lemma \ref{lemma:resem_prob_2stage} with $a_S = \infty$, 
when $\bs{Z}$ is from simple random sampling, 
$\Pr(M_T\le a_T\mid \bs{Z}) \convergep \Pr(\chi^2_K \le a_T)$ as $N\rightarrow \infty$. 
From the discussion before, we then have 
$| \Pr(M_T \le a_T \mid \bs{Z} ) - \Pr(M_T\le a_T \mid M_S\le a_S) | \convergep 0$.
Because $| \Pr(M_T \le a_T \mid \bs{Z} ) - \Pr(M_T\le a_T \mid M_S\le a_S) |$ is upper bounded by 2, 
from \citet[][Theorem 5.5.2]{durrett2019probability}, 
we can know that, 
as $N\rightarrow \infty$, 
\begin{align*}
    \E| \Pr(M_T \le a_T \mid \bs{Z} ) - \Pr(M_T\le a_T \mid M_S\le a_S) | \converge 0
\end{align*}

From the above, using Lemma \ref{lemma:total_variation}, we have, as $N\rightarrow \infty$,  
\begin{align*}
    \text{d}_{\TV}(\sresem, \text{ReSEM}) 
    & \le 
    \I(\mathcal{M} = \emptyset) 
    + 
    \frac{\E \left| \Pr(M_T \le a_T \mid \bs{Z} ) - 
    \Pr(M_T\le a_T \mid M_S\le a_S) \right|}{ \Pr(M_T\le a_T, M_S\le a_S)}
    \\
    & \converge 0. 
\end{align*}
Therefore, Theorem \ref{thm:equiv_resem_sresem} holds. 
\end{proof}

\begin{proof}[\bf Proof of Corollary \ref{cor:equiv_resem_sresem}]
By the property of total variance distance, for any set $\mathcal{A}$,  $c>0$ and $\alpha\in (0,1)$, 
the differences between 
\begin{align}\label{eq:some_prob_sresem}
    & 
   \Pr\big(\sqrt{n}\{ \bs{\theta}(\bs{Z}, \bs{T}_{\mathcal{S}}, \Pi_N) - \bs{\theta}_N \} \in \mathcal{A} \mid \sresem \big),
    & 
   \Pr\big( \| \bs{\theta}(\bs{Z}, \bs{T}_{\mathcal{S}}, \Pi_N) - \bs{\theta}_N \|_2 \ge c \mid \sresem \big), \nonumber
    \\
    & 
    \Pr\big( \bs{\theta}(\bs{Z}, \bs{T}_{\mathcal{S}}, \Pi_N) \le \alpha \mid \sresem \big),
    & 
   \Pr\big( \bs{\theta}_N \in \Theta(\bs{Z}, \bs{T}_{\mathcal{S}}, \Pi_N)  \mid \sresem \big)
\end{align}
and the corresponding 
\begin{align}\label{eq:some_prob_resem}
    & 
   \Pr\big(\sqrt{n}\{ \bs{\theta}(\bs{Z}, \bs{T}_{\mathcal{S}}, \Pi_N) - \bs{\theta}_N \} \in \mathcal{A} \mid \text{ReSEM} \big),
    & 
   \Pr\big( \| \bs{\theta}(\bs{Z}, \bs{T}_{\mathcal{S}}, \Pi_N) - \bs{\theta}_N \|_2 \ge c \mid \text{ReSEM} \big), 
   \nonumber
    \\
    & 
    \Pr\big( \bs{\theta}(\bs{Z}, \bs{T}_{\mathcal{S}}, \Pi_N) \le \alpha \mid \text{ReSEM} \big),
    & 
   \Pr\big( \bs{\theta}_N \in \Theta(\bs{Z}, \bs{T}_{\mathcal{S}}, \Pi_N)  \mid \text{ReSEM} \big)
\end{align}
can be bounded by $\text{d}_{\TV}(\sresem, \text{ReSEM})$, which, by Theorem \ref{thm:equiv_resem_sresem}, converges to zero as $N\rightarrow\infty$. 
Thus, if the limit (or limit inferior or limit superior) of some quantity in \eqref{eq:some_prob_sresem} exist, then the limit (or limit inferior or limit superior) of the corresponding quantity in \eqref{eq:some_prob_resem} must also exist and have the same value. 
We can then derive Corollary \ref{cor:equiv_resem_sresem}. 
\end{proof}

\section{Asymptotic properties of the difference-in-means estimator under ReSEM}\label{app:resem}

This section contains proofs for the asymptotic properties of the difference-in-means estimators under ReSEM. 
For descriptive convenience, 
we introduce $\Var_{\text{a}}(\cdot)$  to denote the variance of the asymptotic distribution of a certain estimator as $N\rightarrow \infty$. 

To prove Theorem \ref{thm:dist}, we first introduce the following lemma. 

\begin{lemma} \label{lemma:LKa}
	Let $L_{K,a} \sim D_1\mid \bs{D}^\top \bs{D}\le a$, where $\bs{D}=(D_1, \ldots, D_K)^\top \sim \mathcal{N}(\bs{0},\bs{I}_K)$. 
	\begin{itemize}
		\item[(i)]  For any
		$K$ dimensional unit vector $\bs{h}$, we have $L_{K,a} \sim \bs{h}^\top  \bs{D}\mid \bs{D}^\top  \bs{D}\le a$.
		\item[(ii)] $\Var(L_{K,a}) = \Pr(\chi^2_{K+2} \le a)/\Pr(\chi^2_{K} \le a) \equiv v_{K,a}$. 
		\item[(iii)] $L_{K,a}$ is symmetric and unimodal around zero. 
	\end{itemize}
\end{lemma}

\begin{proof}[Proof of Lemma \ref{lemma:LKa}]
	Lemma \ref{lemma:LKa}  follows from \citet[][Theorem 3.1]{Morgan2012} and \citet[][Lemma A1 and Proposition 2]{rerand2018}.
\end{proof}

\begin{proof}[{\bf Proof of Theorem \ref{thm:dist}}]
    From Theorem \ref{thm:equiv_resem_sresem}, to prove Theorem \ref{thm:dist}, 
    it suffices to prove that Theorem \ref{thm:dist} holds under $\sresem$. 
    Note that for any estimator, its distribution under the single-stage $\sresem$ is the same as its conditional distribution under the CRSE given that the covariate balance criteria at both the sampling and assignment stages are satisfied (i.e., $M_S\le a_S$ and $M_T\le a_T$). 

	First, from Lemma \ref{lemma:joint_tilde_CRSE} and 
	\citet[][Corollary A]{rerand2018}, we can know that
	\begin{align}\label{eq:cond_dist_1}
		& \quad \ \sqrt{n}(\hat{\tau}-\tau) \mid 
		\sqrt{n}\tilde{\bs{\tau}}_{\bs{X}}^\top \bs{V}_{xx}^{-1} \sqrt{n}\tilde{\bs{\tau}}_{\bs{X}} \le a_T, 
		\sqrt{n}\hat{\bs{\delta}}_{\bs{W}}^\top \bs{V}_{ww}^{-1} \sqrt{n}\hat{\bs{\delta}}_{\bs{W}} \le a_S
		\nonumber
		\\
		& 
		\ \dot\sim \ 
		A \mid 
		\bs{B}_T^\top \bs{V}_{xx}^{-1}\bs{B}_T \le a_T, 
		\bs{B}_S^\top \bs{V}_{ww}^{-1}\bs{B}_S \le a_S,
	\end{align}
	where $(A, \bs{B}_T^\top, \bs{B}_S^\top)^\top \sim \mathcal{N}(\bs{0}, \bs{V})$ with $\bs{V}$ defined as in \eqref{eq:VV}. 
	Note that, by the definitions in \eqref{eq:M_S}, \eqref{eq:M_T} and Lemma \ref{lemma:joint_tilde_CRSE}, 
	$
	\sqrt{n}\hat{\bs{\delta}}_{\bs{W}}^\top \bs{V}_{ww}^{-1} \sqrt{n}\hat{\bs{\delta}}_{\bs{W}} = M_S, 
	$
	and 
	\begin{align*}
		\sqrt{n}\tilde{\bs{\tau}}_{\bs{X}}^\top \bs{V}_{xx}^{-1} \sqrt{n}\tilde{\bs{\tau}}_{\bs{X}}
		& = 
		\left\{ (\bs{S}_{\bs{X}}^2)^{1/2}(\bs{s}_{\bs{X}}^2)^{-1/2} \hat{\bs{\tau}}_{\bs{X}}
		\right\}^\top 
		\left(
		\frac{n}{n_1 n_0} 
		\bs{S}^2_{\bs{X}}
		\right)^{-1} 
		(\bs{S}_{\bs{X}}^2)^{1/2}(\bs{s}_{\bs{X}}^2)^{-1/2} \hat{\bs{\tau}}_{\bs{X}}
		\\
		& = 
		\hat{\bs{\tau}}_{\bs{X}}^\top 
		\left(
		\frac{n}{n_1 n_0} 
		\bs{s}^2_{\bs{X}}
		\right)^{-1}
		\hat{\bs{\tau}}_{\bs{X}}
		= M_T. 
	\end{align*}
	Thus, \eqref{eq:cond_dist_1} is equivalent to 
	\begin{align}\label{eq:cond_dist_3}
		\sqrt{n}(\hat{\tau} - \tau) \mid 
		M_T \le a_T, 
		M_S \le a_S
		\ \dot\sim \ 
		A \mid 
		\bs{B}_T^\top \bs{V}_{xx}^{-1}\bs{B}_T \le a_T, 
		\bs{B}_S^\top \bs{V}_{ww}^{-1}\bs{B}_S \le a_S. 
	\end{align}
	
	Second, define 
	$
	\eta= A - \bs{V}_{\tau x}\bs{V}_{x x}^{-1}\bs{B}_T - \bs{V}_{\tau w}\bs{V}_{ww}^{-1}\bs{B}_S.
	$
	By the property of multivariate Gaussian distribution and from Proposition \ref{prop:R2}, we can derive that $(\eta, \bs{B}_T, \bs{B}_S)$ are mutually independent Gaussian random vectors, and the variance of $\eta$ has the following equivalent forms: 
	\begin{align*}
		\Var(\eta) & = \Var(A) - \Var(\bs{V}_{\tau x}\bs{V}_{x x}^{-1}\bs{B}_T) - \Var(\bs{V}_{\tau w}\bs{V}_{ww}^{-1}\bs{B}_S)
		= V_{\tau\tau} - \bs{V}_{\tau x}\bs{V}_{x x}^{-1}\bs{V}_{x\tau} - \bs{V}_{\tau w}\bs{V}_{ww}^{-1}\bs{V}_{w\tau}\\
		& = 
		V_{\tau\tau}\left( 1 - 
		\frac{\bs{V}_{\tau x}\bs{V}_{x x}^{-1}\bs{V}_{x\tau}}{V_{\tau\tau}} 
		- 
		\frac{\bs{V}_{\tau w}\bs{V}_{ww}^{-1}\bs{V}_{w\tau}}{V_{\tau\tau}}
		\right)\\
		& = V_{\tau\tau} (1-R_T^2 - R_S^2). 
	\end{align*}
	Consequently, $\varepsilon = V_{\tau\tau}^{-1/2} (1-R_T^2 - R_S^2)^{-1/2} \eta$ follows standard Gaussian distribution and is independent of $(\bs{B}_T, \bs{B}_S)$.  
	Define 
	$\bs{h}_T^\top = (V_{\tau\tau}R_T^2)^{-1/2}\bs{V}_{\tau x}\bs{V}_{x x}^{-1/2}$ and 
	$\bs{h}_S^\top = (V_{\tau\tau}R_S^2)^{-1/2}\bs{V}_{\tau w}\bs{V}_{ww}^{-1/2}$. 
	Then 
    by the definitions of $R_S^2$ and $R_T^2$, 
    both $\bs{h}_T$ and $\bs{h}_S$ are unit vectors of length one. 
	Define $\tilde{\bs{B}}_T = \bs{V}_{x x}^{-1/2}\bs{B}_T \sim \mathcal{N}(\bs{0},\bs{I}_K)$ 
	and $\tilde{\bs{B}}_S = \bs{V}_{ww}^{-1/2}\bs{B}_S \sim \mathcal{N}(\bs{0},\bs{I}_J)$. 
	From Lemma \ref{lemma:LKa} and the mutual independence of $(\eta,\bs{B}_T, \bs{B}_S)$, we have 
	\begin{align*}
		& \quad \ 
		A \mid 
		\bs{B}_T^\top \bs{V}_{xx}^{-1}\bs{B}_T \le a_T, 
		\bs{B}_S^\top \bs{V}_{ww}^{-1}\bs{B}_S \le a_S
		\\
		& \sim 
		\eta + \bs{V}_{\tau x}\bs{V}_{x x}^{-1}\bs{B}_T + \bs{V}_{\tau w}\bs{V}_{ww}^{-1}\bs{B}_S
		\mid 
		\bs{B}_T^\top \bs{V}_{xx}^{-1}\bs{B}_T \le a_T, 
		\bs{B}_S^\top \bs{V}_{ww}^{-1}\bs{B}_S \le a_S\\
		& \sim 
		\eta + \bs{V}_{\tau x}\bs{V}_{x x}^{-1/2}\tilde{\bs{B}}_T + \bs{V}_{\tau w}\bs{V}_{ww}^{-1/2}\tilde{\bs{B}}_S
		\mid 
		\tilde{\bs{B}}_T^\top \tilde{\bs{B}}_T \le a_T, 
		\tilde{\bs{B}}_S^\top \tilde{\bs{B}}_S \le a_S\\
		& \sim 
		V_{\tau\tau}^{1/2}\sqrt{1-R_T^2 - R_S^2} \cdot \varepsilon + 
		V_{\tau\tau}^{1/2} \sqrt{R_T^2} \cdot \bs{h}_T^\top \tilde{\bs{B}}_T + 
		V_{\tau\tau}^{1/2} \sqrt{R_S^2} \cdot \bs{h}_S^\top \tilde{\bs{B}}_S
		\mid 
		\tilde{\bs{B}}_T^\top \tilde{\bs{B}}_T \le a_T, 
		\tilde{\bs{B}}_S^\top \tilde{\bs{B}}_S \le a_S\\
		& \sim 
		V_{\tau\tau}^{1/2} \Big( \sqrt{1-R_T^2-R_S^2} \cdot \varepsilon + 
		\sqrt{R_T^2} \cdot L_{K,a_T} + \sqrt{R_S^2} \cdot L_{J,a_S} 
		\Big),
	\end{align*}
	where 
	$L_{K,a_T} \sim  \bs{h}_T^\top \tilde{\bs{B}}_T \mid \tilde{\bs{B}}_T^\top \tilde{\bs{B}}_T \le a_T,$ 
	$L_{J,a_S} \sim \bs{h}_S^\top \tilde{\bs{B}}_S \mid \tilde{\bs{B}}_S^\top \tilde{\bs{B}}_S \le a_S,$
	and $(\varepsilon, L_{K,a_T}, L_{J,a_S})$ are mutually independent. 
	From \eqref{eq:cond_dist_3}, we then have 
	\begin{align*}
		\sqrt{n}(\hat{\tau} - \tau) \mid 
	M_T \le a_T, 
	M_S \le a_S
	\ \dot\sim \ 
	V_{\tau\tau}^{1/2} \Big( \sqrt{1-R_T^2-R_S^2} \cdot \varepsilon + 
	\sqrt{R_T^2} \cdot L_{K,a_T} + \sqrt{R_S^2} \cdot L_{J,a_S} 
	\Big).
	\end{align*}
	
	From the above, Theorem \ref{thm:dist} holds. 
\end{proof}

\begin{proof}[{\bf Comments on the acceptance probabilities under ReSEM}]
    First, from Lemma \ref{lemma:M_ST_CRSE}, under ReSEM, the acceptance probability at the sampling stage is $\Pr(M_S\le a_S) = \Pr(\chi^2_J\le a_S) + o(1)$. 
    Second, under ReSEM, the acceptance probability at the treatment assignment stage is
    $
        \E\{ \Pr(M_T\le a_T\mid \bs{Z}) \mid M_S\le a_S \}. 
    $
    From Lemma \ref{lemma:resem_prob_2stage}(ii),  
    \begin{align*}
        & \quad \ \big| \E\{ \Pr(M_T\le a_T\mid \bs{Z})  \mid M_S\le a_S \} - \Pr(\chi^2_K\le a_T) \big|
        \\
        & \le  \E\big\{ | \Pr(M_T\le a_T\mid \bs{Z}) - \Pr(\chi^2_K\le a_T) | \mid M_S\le a_S \big\}
        \le 
        \frac{\E\big\{ | \Pr(M_T\le a_T\mid \bs{Z}) - \Pr(\chi^2_K\le a_T) | \big\}}{\Pr(M_S\le a_S)}
        \\
        & = o(1). 
    \end{align*}
    Thus, the asymptotic acceptance probability at the assignment stage is $\Pr(\chi^2_K\le a_T)$. 
\end{proof}

\begin{proof}[{\bf Comments on covariate balance under ReSEM}]
	First, note that any covariates, no matter observed or unobserved, can be viewed as pseudo potential outcomes that are unaffected by the treatment. Therefore, 
	by the same logic as the consistency of the difference-in-means estimator from Theorem \ref{thm:dist},   
	all covariates are asymptotically balanced between two treatment groups under ReSEM.
	
	Second, by the same logic as Corollary \ref{cor:rej_sam}, under ReSEM, for any covariate, the sample average is consistent for the population average. 
	Therefore, all covariates are asymptotically balanced between sampled units and the overall population. 
\end{proof}

\begin{proof}[{\bf Proof of Corollary \ref{corollary:PRIASV}}]
	From Corollary \ref{cor:crse}, the asymptotic variance of $\hat{\tau}$ under the CRSE is 
	$
	\Var_{\a}\{\sqrt{n}(\hat{\tau} - \tau)\} = V_{\tau\tau}.
	$
	From Theorem \ref{thm:dist} and Lemma \ref{lemma:LKa}, the asymptotic variance of $\hat{\tau}$ under ReSEM is 
	\begin{align*}
		\Var_{\a}\{\sqrt{n} (\hat\tau - \tau) \mid  \text{ReSEM}  \}
		&= V_{\tau\tau} \left\{ (1-R_S^2 - R_T^2) + R_S^2 v_{J,a_S} + R_T^2 v_{K,a_T} \right\}
		\\& = V_{\tau \tau}\{ 1 -(1-v_{J,a_S})R_S^2-(1- v_{K,a_T})R_T^2 \}.
	\end{align*}
	Therefore, 
	compared to the CRSE, 
	the percentage reduction in asymptotic variance of $\hat{\tau}$ under ReSEM is 
	\begin{align*}
		1 - \frac{\Var_{\a}\{\sqrt{n} (\hat\tau - \tau) \mid  \text{ReSEM}  \}}{\Var_{\a}\{\sqrt{n}(\hat{\tau} - \tau)\}}
		& = 1 - \left\{ 1 -(1-v_{J,a_S})R_S^2-(1- v_{K,a_T})R_T^2 \right\}\\
		& = (1-v_{J,a_S})R_S^2 + (1- v_{K,a_T})R_T^2 \ge 0. 
	\end{align*}
	Therefore, Corollary \ref{corollary:PRIASV} holds. 
\end{proof}

\begin{proof}[{\bf Comments on Remark \ref{rmk:loss_comp_optimal}}]
	From \citet[][Lemma A5]{rerand2018} and  \citet[][Proposition 4]{LDR2020}, we can immediately know that $v_{J,a_S}$ and $v_{K,a_T}$ are nondecreasing in $a_S$ and $a_T$, respectively. 
	Thus, the PRIAV in Corollary \ref{corollary:PRIASV} is nonincreasing in $a_S$ and $a_T$. 
\end{proof}

\begin{proof}[{\bf Proof of Corollary \ref{corollary:QR}}]
	Corollary \ref{corollary:QR} follows immediately from \citet[][Theorem 4]{rerand2018}. 
\end{proof}

\section{Asymptotic properties of the regression-adjusted estimators under ReSEM}\label{app:reg}

\subsection{Lemmas}

\begin{lemma}\label{lemma:proj_coef}
	$\tilde{\bs{\beta}}$ and $\tilde{\bs{\gamma}}$ defined in 
	Section \ref{sec:R2_proj_coef}
	are the linear projection coefficients of $\hat{\tau}$ on $\hat{\bs{\tau}}_{\bs{C}}$ and $\hat{\bs{\delta}}_{\bs{E}}$, respectively, under the CRSE.
\end{lemma}

\begin{lemma}\label{lemma:simp_W_in_E}
$V_{\tau\tau}(\bs{0}, \bs{\gamma})$ has the following decomposition: 
\begin{align*}
    V_{\tau\tau}(\bs{0}, \bs{\gamma}) & = 
    V_{\tau\tau}(1-R^2_E) + (1-f) (\bs{\gamma} - \tilde{\bs{\gamma}})^\top \bs{S}^2_{\bs{E}} (\bs{\gamma} - \tilde{\bs{\gamma}}).
\end{align*}
If $\bs{W} \subset \bs{E}$, then $V_{\tau\tau}(\bs{0}, \bs{\gamma})\{1 - R_S^2(\bs{0}, \bs{\gamma})\}$ and $V_{\tau\tau}(\bs{0}, \bs{\gamma})R_S^2(\bs{0}, \bs{\gamma})$ have the following equivalent forms:
\begin{align*}
    V_{\tau\tau}(\bs{0}, \bs{\gamma})\{1 - R_S^2(\bs{0}, \bs{\gamma})\}
    & = 
    V_{\tau\tau} (1-R_E^2) + 
    (1-f) (\bs{\gamma} - \tilde{\bs{\gamma}})^\top \bs{S}^2_{\bs{E} \setminus \bs{W}} (\bs{\gamma} - \tilde{\bs{\gamma}}), \\
    V_{\tau\tau}(\bs{0}, \bs{\gamma})R_S^2(\bs{0}, \bs{\gamma})
    & = (1-f) (\bs{\gamma} - \tilde{\bs{\gamma}})^\top \bs{S}^2_{\bs{E} \mid \bs{W}} (\bs{\gamma} - \tilde{\bs{\gamma}}). 
\end{align*}
\end{lemma}

\begin{lemma}\label{lemma:simp_X_in_C}
$V_{\tau\tau}(\bs{\beta}, \bs{0})$ has the following decomposition:
\begin{align*}
    V_{\tau\tau}(\bs{\beta}, \bs{0}) & = V_{\tau\tau} (1-R_C^2) +  (r_1 r_0)^{-1} (\bs{\beta} - \tilde{\bs{\beta}})^\top \bs{S}^2_{\bs{C}} (\bs{\beta} - \tilde{\bs{\beta}}).
\end{align*}
If $\bs{X} \subset \bs{C}$, then $V_{\tau\tau}(\bs{\beta}, \bs{0})\{1 - R_T^2(\bs{\beta}, \bs{0}) \}$ and $V_{\tau\tau} R_T^2(\bs{\beta}, \bs{0})$ have the following equivalent forms: 
\begin{align*}
    V_{\tau\tau}(\bs{\beta}, \bs{0})\{1 - R_T^2(\bs{\beta}, \bs{0}) \}
    & = 
    V_{\tau\tau} (1-R_C)^2 +  (r_1 r_0)^{-1} (\bs{\beta} - \tilde{\bs{\beta}})^\top \bs{S}^2_{\bs{C} \setminus \bs{X}} (\bs{\beta} - \tilde{\bs{\beta}}), 
    \\
    V_{\tau\tau} R_T^2(\bs{\beta}, \bs{0})
    & = 
    (r_1 r_0)^{-1} (\bs{\beta} - \tilde{\bs{\beta}})^\top \bs{S}^2_{\bs{C} \mid \bs{X}} (\bs{\beta} - \tilde{\bs{\beta}}),
\end{align*}
\end{lemma}

\begin{lemma}\label{lemma:sum_linear}
	If two independent random variables $\zeta_1$ and $\zeta_2$ are both symmetric and unimodal around zero, 
	then for any constants $c_1$  and $c_2$, 
	$c_1 \zeta_1 + c_2 \zeta_2$ is also symmetric and unimodal around zero. 
\end{lemma}

\begin{lemma}\label{lemma:sum_unimodal}
	Let $\zeta_0, \zeta_1$ and $\zeta_2$ be three mutually independent random variables. If (i) $\zeta_0$ is symmetric and unimodal around zero, 
	(ii) $\zeta_1$  and $\zeta_2$ are symmetric around 0, 
	and 
	(iii) $\Pr( |\zeta_1| \le c) \ge \Pr(|\zeta_2| \le c)$ for any $c\ge0$, 
	then 
	$\Pr( |\zeta_0+\zeta_1| \le c) \ge \Pr(|\zeta_0+\zeta_2| \le c)$ for any $c\ge0$. 
\end{lemma}

\begin{lemma}\label{lemma:linear_com_eps_L}
	For any positive integer $K_1,K_2$ and constants $a_1,a_2$,  
	let $\varepsilon_0 \sim \mathcal{N}(0,1)$, $L_{K_1, a_1} \sim D_1 \mid \bs{D}^\top\bs{D}\le a_1$ and 
	$L_{K_2, a_2} \sim \tilde{D}_1 \mid \tilde{\bs{D}}^\top\tilde{\bs{D}}\le a_2$, 
	where $\bs{D} = (D_1, \ldots, D_{K_1}) \sim \mathcal{N}(\bs{0}, \bs{I}_{K_1})$, $\tilde{\bs{D}} = (\tilde{D}_1, \ldots, \tilde{D}_{K_2}) \sim \mathcal{N}(\bs{0}, \bs{I}_{K_2})$, 
	and $(\varepsilon, L_{K_1, a_1}, L_{K_2, a_2})$ are mutually independent. 
	Then for any nonnegative constants $b_0 \le \overline{b}_0$, $b_1 \le \overline{b}_1$, $b_2 \le \overline{b}_2$, and any $c \ge 0$, 
	\begin{align*}
		\Pr\left(
		\left| b_0 \varepsilon_0 + b_1 L_{K_1,a_1}+b_2 L_{K_2,a_2} \right| \le c
		\right)
		& \ge 
		\Pr\left(
		\left| \overline{b}_0 \varepsilon_0 + \overline{b}_1 L_{K_1,a_1}+\overline{b}_2 L_{K_2,a_2} \right| \le c
		\right).
	\end{align*}
\end{lemma}

\subsection{Proofs of the lemmas}

\begin{proof}[Proof of Lemma \ref{lemma:proj_coef}]
	By the same logic as Lemma \ref{lemma:cov_crse}, $(\hat{\tau}, \hat{\bs{\tau}}_{\bs{C}}^\top, \hat{\bs{\delta}}_{\bs{E}}^\top)^\top$ has the following covariance matrix under the CRSE:
	\begin{align*}
		\bs{\Sigma} 
		=
		\begin{pmatrix}
			\Sigma_{\tau\tau} & \bs{\Sigma}_{\tau c} & \bs{\Sigma}_{\tau e}\\
			\bs{\Sigma}_{c \tau} & \bs{\Sigma}_{cc} & \bs{\Sigma}_{ce}\\
			\bs{\Sigma}_{e \tau} & \bs{\Sigma}_{ec} & \bs{\Sigma}_{ee}
		\end{pmatrix}
		= 
		\begin{pmatrix}
			r_1^{-1}S^2_{1}+r_0^{-1}S^2_{0} -fS^2_\tau 
			& r_1^{-1}\bs{S}_{1,\bs{C}}+r_0^{-1}\bs{S}_{0,\bs{C}} 
			& \left(1 - f\right)\bs{S}_{\tau,\bs{E}} \\
			r_1^{-1}\bs{S}_{\bs{C},1} +r_0^{-1}\bs{S}_{\bs{C},0}& (r_1r_0)^{-1} \bs{S}^2_{\bs{C}}  & \bs{0} \\
			\left(1 - f\right)\bs{S}_{\bs{E},\tau}& \bs{0} & \left(1  -f \right) \bs{S}^2_{\bs{E}}
		\end{pmatrix}. 
	\end{align*}
	Thus, under the CRSE, $\hat{\bs{\tau}}_{\bs{C}}$ and $\hat{\bs{\delta}}_{\bs{E}}$ are uncorrelated, and the linear projection coefficients of $\hat{\tau}$ on $\hat{\bs{\tau}}_{\bs{C}}$ and $\hat{\bs{\delta}}_{\bs{E}}$ are, respectively, 
	\begin{align*}
		\bs{\Sigma}_{cc}^{-1} \bs{\Sigma}_{c \tau} 
		& = 
		\big\{ (r_1r_0)^{-1} \bs{S}^2_{\bs{C}} \big\}^{-1}
		\big(r_1^{-1}\bs{S}_{\bs{C},1}+r_0^{-1}\bs{S}_{\bs{C},0} \big)
		= r_0 \big(\bs{S}_{\bs{C}}^2\big)^{-1} \bs{S}_{\bs{C},1}
		+ 
		r_1 \big(\bs{S}_{\bs{C}}^2\big)^{-1} \bs{S}_{\bs{C},0}
		= r_0 \tilde{\bs{\beta}}_1 + r_1 \tilde{\bs{\beta}}_0 = \tilde{\bs{\beta}}, 
		\nonumber
		\\
		\bs{\Sigma}_{ee}^{-1} \bs{\Sigma}_{e \tau}
		& =  
		\big\{
		(1  - f) \bs{S}^2_{\bs{E}}
		\big\}^{-1} 
		(1 - f)\bs{S}_{\bs{E},\tau}
		= \big( \bs{S}^2_{\bs{E}} \big)^{-1} \bs{S}_{\bs{E}, 1} - \big( \bs{S}^2_{\bs{E}} \big)^{-1} \bs{S}_{\bs{E}, 0}
		= \tilde{\bs{\gamma}}_1 - \tilde{\bs{\gamma}}_0 = \tilde{\bs{\gamma}}. 
	\end{align*}
	Therefore, Lemma \ref{lemma:proj_coef} holds.
\end{proof}

\begin{proof}[Proof of Lemma \ref{lemma:simp_W_in_E}]
First, we simplify $V_{\tau\tau}(\bs{0}, \bs{\gamma})$, which by definition is equivalently the variance of $\sqrt{n}\hat{\tau}(\bs{0}, \bs{\gamma})$ under the CRSE. 
From Lemma \ref{lemma:proj_coef}, 
$\tilde{\bs{\gamma}}$ is the linear projection coefficient of $\hat{\tau}$ on $\hat{\bs{\delta}}_{\bs{E}}$ under the CRSE. 
Thus, 
under the CRSE, 
$\hat{\tau} - \tilde{\bs{\gamma}}^\top \hat{\bs{\delta}}_{\bs{E}}$ is uncorrelated with $\hat{\bs{\delta}}_{\bs{E}}$, 
and, by the same logic as Proposition \ref{prop:R2}, 
its variance is $\Var(\hat{\tau}) (1-R_E^2) = n^{-1} V_{\tau\tau}  (1-R_E^2)$. 
These imply that the variance of $\hat{\tau}(\bs{0}, \bs{\gamma})$ has the following decomposition:
\begin{align*}
    \Var\big\{ \hat{\tau}(\bs{0}, \bs{\gamma}) \big\}
    & = 
    \Var\big\{ \hat{\tau} - \tilde{\bs{\gamma}}^\top \hat{\bs{\delta}}_{\bs{E}} -  (\bs{\gamma} - \tilde{\bs{\gamma}})^\top \hat{\bs{\delta}}_{\bs{E}}  \big\}
    = 
    \Var\big( \hat{\tau} - \tilde{\bs{\gamma}}^\top \hat{\bs{\delta}}_{\bs{E}} \big) 
    +\Var\big\{ (\bs{\gamma} - \tilde{\bs{\gamma}})^\top \hat{\bs{\delta}}_{\bs{E}}  \big\}
    \\
    & = n^{-1} V_{\tau\tau}(1-R_E^2) + (\bs{\gamma} - \tilde{\bs{\gamma}})^\top \Cov(\hat{\bs{\delta}}_{\bs{E}}) (\bs{\gamma} - \tilde{\bs{\gamma}}). 
\end{align*}
By the same logic as Lemma \ref{lemma:cov_crse}, 
$n\Cov(\hat{\bs{\delta}}_{\bs{E}}) = \Cov(\sqrt{n}\hat{\bs{\delta}}_{\bs{E}}) = (1-f) \bs{S}^2_{\bs{E}}$. 
Consequently, $V_{\tau\tau}(\bs{0}, \bs{\gamma})$ has the following equivalent decomposition:
\begin{align*}
    V_{\tau\tau}(\bs{0}, \bs{\gamma}) & = 
    n \Var\big\{ \hat{\tau}(\bs{0}, \bs{\gamma}) \big\}
    = 
    V_{\tau\tau}(1-R^2_E) + (1-f) (\bs{\gamma} - \tilde{\bs{\gamma}})^\top \bs{S}^2_{\bs{E}} (\bs{\gamma} - \tilde{\bs{\gamma}}).
\end{align*}

Second, we simplify $V_{\tau\tau}(\bs{0}, \bs{\gamma})R_S^2(\bs{0}, \bs{\gamma})$ and $V_{\tau\tau}(\bs{0}, \bs{\gamma})\{1 - R_S^2(\bs{0}, \bs{\gamma})\}$ when $\bs{W} \subset \bs{E}$. 
By the same logic as Proposition \ref{prop:R2}, $R_S^2(\bs{0}, \bs{\gamma})$ is equivalently the squared multiple correlation between $\hat{\tau}(\bs{0}, \bs{\gamma})$ and $\hat{\bs{\delta}}_{\bs{W}}$ under the CRSE. This implies that 
\begin{align*}
    n^{-1} V_{\tau\tau}(\bs{0}, \bs{\gamma})R_S^2(\bs{0}, \bs{\gamma})
    & = 
    \Cov\big\{\hat{\tau}(\bs{0}, \bs{\gamma}), \hat{\bs{\delta}}_{\bs{W}}\big\}
    \Cov^{-1}( \hat{\bs{\delta}}_{\bs{W}} )
    \Cov\big\{\hat{\bs{\delta}}_{\bs{W}}, \hat{\tau}(\bs{0}, \bs{\gamma})\big\}. 
\end{align*}
Because $\bs{W} \subset \bs{E}$, $\hat{\bs{\delta}}_{\bs{W}}$ must be a linear function of $\hat{\bs{\delta}}_{\bs{E}}$ and is thus uncorrelated with $\hat{\tau} - \tilde{\bs{\gamma}}^\top \hat{\bs{\delta}}_{\bs{E}}$. 
Consequently, 
\begin{align*}
    \Cov\big\{\hat{\tau}(\bs{0}, \bs{\gamma}), \hat{\bs{\delta}}_{\bs{W}}\big\}
    & 
    = 
    \Cov\big\{ \hat{\tau} - \tilde{\bs{\gamma}}^\top \hat{\bs{\delta}}_{\bs{E}} - (\bs{\gamma} - \tilde{\bs{\gamma}})^\top \hat{\bs{\delta}}_{\bs{E}}, \hat{\bs{\delta}}_{\bs{W}}\big\}
    = 
    -  (\bs{\gamma} - \tilde{\bs{\gamma}})^\top \Cov\big( \hat{\bs{\delta}}_{\bs{E}}, \hat{\bs{\delta}}_{\bs{W}}\big), 
\end{align*}
and $V_{\tau\tau}(\bs{0}, \bs{\gamma})R_S^2(\bs{0}, \bs{\gamma})$ has the following equivalent form: 
\begin{align*}
    V_{\tau\tau}(\bs{0}, \bs{\gamma})R_S^2(\bs{0}, \bs{\gamma})
    & = 
    n (\bs{\gamma} - \tilde{\bs{\gamma}})^\top \Cov\big( \hat{\bs{\delta}}_{\bs{E}}, \hat{\bs{\delta}}_{\bs{W}}\big)
    \Cov^{-1}( \hat{\bs{\delta}}_{\bs{W}} )
    \Cov\big( \hat{\bs{\delta}}_{\bs{W}}, \hat{\bs{\delta}}_{\bs{E}} \big) (\bs{\gamma} - \tilde{\bs{\gamma}})\\
    & = 
    (\bs{\gamma} - \tilde{\bs{\gamma}})^\top \Cov\big( \sqrt{n}\hat{\bs{\delta}}_{\bs{E}}, \sqrt{n}\hat{\bs{\delta}}_{\bs{W}}\big)
    \Cov^{-1}( \sqrt{n}\hat{\bs{\delta}}_{\bs{W}} )
    \Cov\big( \sqrt{n}\hat{\bs{\delta}}_{\bs{W}}, \sqrt{n}\hat{\bs{\delta}}_{\bs{E}} \big) (\bs{\gamma} - \tilde{\bs{\gamma}}).
\end{align*}
By the same logic as Lemma \ref{lemma:cov_crse}, 
$(\sqrt{n}\hat{\bs{\delta}}_{\bs{W}}^\top, \sqrt{n}\hat{\bs{\delta}}_{\bs{E}}^\top)^\top$ has the following covariance under the CRSE:
\begin{align*}
    \begin{pmatrix}
    \Cov(\sqrt{n}\hat{\bs{\delta}}_{\bs{W}}) & \Cov(\sqrt{n}\hat{\bs{\delta}}_{\bs{W}}, \sqrt{n}\hat{\bs{\delta}}_{\bs{E}}) \\ 
    \Cov(\sqrt{n}\hat{\bs{\delta}}_{\bs{E}}, \sqrt{n}\hat{\bs{\delta}}_{\bs{W}}) & \Cov(\sqrt{n}\hat{\bs{\delta}}_{\bs{E}})
    \end{pmatrix}
    = 
    (1-f) 
    \begin{pmatrix}
    \bs{S}^2_{\bs{W}} & \bs{S}_{\bs{W}, \bs{E}}\\
    \bs{S}_{\bs{E}, \bs{W}} & \bs{S}^2_{\bs{E}}
    \end{pmatrix}.
\end{align*}
We can then simplify $V_{\tau\tau}(\bs{0}, \bs{\gamma})R_S^2(\bs{0}, \bs{\gamma})$ as
\begin{align*}
    V_{\tau\tau}(\bs{0}, \bs{\gamma})R_S^2(\bs{0}, \bs{\gamma})
    & = 
    (1-f)(\bs{\gamma} - \tilde{\bs{\gamma}})^\top \bs{S}_{\bs{E}, \bs{W}}
    \big( \bs{S}^2_{\bs{W}} \big)^{-1}
    \bs{S}_{\bs{W}, \bs{E}} (\bs{\gamma} - \tilde{\bs{\gamma}})
    = (1-f)(\bs{\gamma} - \tilde{\bs{\gamma}})^\top \bs{S}_{\bs{E}\mid \bs{W}}^2 (\bs{\gamma} - \tilde{\bs{\gamma}}).
\end{align*}
Consequently, we have 
\begin{align*}
    & \quad \ V_{\tau\tau}(\bs{0}, \bs{\gamma}) \big\{ 1 - R_S^2(\bs{0}, \bs{\gamma}) \big\}
    \\
    & = 
    V_{\tau\tau}(\bs{0}, \bs{\gamma}) - V_{\tau\tau}(\bs{0}, \bs{\gamma})R_S^2(\bs{0}, \bs{\gamma})
    = 
    V_{\tau\tau}(1-R^2_E) + (1-f) (\bs{\gamma} - \tilde{\bs{\gamma}})^\top 
    \left( \bs{S}^2_{\bs{E}} - \bs{S}_{\bs{E}\mid \bs{W}}^2 \right) (\bs{\gamma} - \tilde{\bs{\gamma}})\\
    & = V_{\tau\tau}(1-R^2_E) + (1-f) (\bs{\gamma} - \tilde{\bs{\gamma}})^\top 
    \bs{S}_{\bs{E}\setminus \bs{W}}^2 (\bs{\gamma} - \tilde{\bs{\gamma}}).
\end{align*}

From the above, Lemma \ref{lemma:simp_W_in_E} holds. 
\end{proof}

\begin{proof}[Proof of Lemma \ref{lemma:simp_X_in_C}]
First, we simplify $V_{\tau\tau}(\bs{\beta}, \bs{0})$, which by definition is the variance of $\sqrt{n} \hat{\tau}(\bs{\beta}, \bs{0})$ under the CRSE. 
From Lemma \ref{lemma:proj_coef}, 
$\tilde{\bs{\beta}}$ is the linear projection coefficient of $\hat{\tau}$ on $\hat{\bs{\tau}}_{\bs{C}}$ under the CRSE. 
Thus, 
under the CRSE, 
$\hat{\tau} - \tilde{\bs{\beta}}^\top \hat{\bs{\tau}}_{\bs{C}}$ is uncorrelated with $\hat{\bs{\tau}}_{\bs{C}}$, and, by the same logic as Proposition \ref{prop:R2}, 
its variance is $\Var(\hat{\tau}) (1-R_C^2) = n^{-1} V_{\tau\tau} (1-R_C^2)$. 
These imply that the variance of $\hat{\tau}(\bs{\beta}, \bs{0})$ has the following decomposition:
\begin{align*}
    \Var\big\{ \hat{\tau}(\bs{\beta}, \bs{0}) \big\}
    & = 
    \Var\big\{ \hat{\tau} - \tilde{\bs{\beta}}^\top \hat{\bs{\tau}}_{\bs{C}} 
    - (\bs{\beta} - \tilde{\bs{\beta}})^\top \hat{\bs{\tau}}_{\bs{C}} 
    \big\}
    =
    \Var\big( \hat{\tau} - \tilde{\bs{\beta}}^\top \hat{\bs{\tau}}_{\bs{C}} \big)
    +
    \Var\big\{ (\bs{\beta} - \tilde{\bs{\beta}})^\top \hat{\bs{\tau}}_{\bs{C}} 
    \big\}
    \\
    & 
    = n^{-1} V_{\tau\tau}(1-R_C^2) + 
    (\bs{\beta} - \tilde{\bs{\beta}})^\top
    \Cov(\hat{\bs{\tau}}_{\bs{C}})
    (\bs{\beta} - \tilde{\bs{\beta}}). 
\end{align*}
    By the same logic as Lemma \ref{lemma:cov_crse}, $n\Cov(\hat{\bs{\tau}}_{\bs{C}}) = \Cov(\sqrt{n}\hat{\bs{\tau}}_{\bs{C}}) = (r_1r_0)^{-1} \bs{S}^2_{\bs{C}}$. Consequently, 
    \begin{align*}
        V_{\tau\tau}(\bs{\beta}, \bs{0}) = n \Var\big\{ \hat{\tau}(\bs{\beta}, \bs{0}) \big\}
        = V_{\tau\tau}(1-R_C^2) + 
        (r_1r_0)^{-1} (\bs{\beta} - \tilde{\bs{\beta}})^\top
        \bs{S}^2_{\bs{C}}
        (\bs{\beta} - \tilde{\bs{\beta}}). 
    \end{align*}
    
    Second, we simplify $V_{\tau\tau}(\bs{\beta}, \bs{0}) R_T^2(\bs{\beta}, \bs{0})$ and $V_{\tau\tau}(\bs{\beta}, \bs{0})\{1 - R_T^2(\bs{\beta}, \bs{0}) \}$ when $\bs{X} \subset \bs{C}$. 
    By the same logic as Proposition \ref{prop:R2}, $R_T^2(\bs{\beta}, \bs{0})$ is equivalently the squared multiple correlation between $\hat{\tau}(\bs{\beta}, \bs{0})$ and $\hat{\bs{\tau}}_{\bs{X}}$ under the CRSE. This implies that 
\begin{align*}
    n^{-1} V_{\tau\tau}(\bs{\beta}, \bs{0}) R_T^2(\bs{\beta}, \bs{0})
    & = 
    \Cov\big\{\hat{\tau}(\bs{\beta}, \bs{0}), \hat{\bs{\tau}}_{\bs{X}}\big\}
    \Cov^{-1}( \hat{\bs{\tau}}_{\bs{X}} )
    \Cov\big\{\hat{\bs{\tau}}_{\bs{X}}, \hat{\tau}(\bs{\beta}, \bs{0})\big\}. 
\end{align*}
Because $\bs{X} \subset \bs{C}$, $\hat{\bs{\tau}}_{\bs{X}}$ must be a linear function of $\hat{\bs{\tau}}_{\bs{C}}$ and is thus uncorrelated with $\hat{\tau} - \tilde{\bs{\beta}}^\top \hat{\bs{\tau}}_{\bs{C}}$. 
Consequently, 
\begin{align*}
    \Cov\big\{\hat{\tau}(\bs{\beta}, \bs{0}), \hat{\bs{\tau}}_{\bs{X}}\big\}
    & 
    = 
    \Cov\big\{ \hat{\tau} - \tilde{\bs{\beta}}^\top \hat{\bs{\tau}}_{\bs{C}} - (\bs{\beta} - \tilde{\bs{\beta}})^\top \hat{\bs{\tau}}_{\bs{C}}, \hat{\bs{\tau}}_{\bs{X}}\big\}
    = 
    -  (\bs{\beta} - \tilde{\bs{\beta}})^\top \Cov\big( \hat{\bs{\tau}}_{\bs{C}}, \hat{\bs{\tau}}_{\bs{X}}\big), 
\end{align*}
and $V_{\tau\tau}(\bs{\beta}, \bs{0}) R_T^2(\bs{\beta}, \bs{0})$ has the following equivalent form: 
\begin{align*}
    V_{\tau\tau}(\bs{\beta}, \bs{0}) R_T^2(\bs{\beta}, \bs{0})
    & = 
    n (\bs{\beta} - \tilde{\bs{\beta}})^\top \Cov\big( \hat{\bs{\tau}}_{\bs{C}}, \hat{\bs{\tau}}_{\bs{X}}\big)
    \Cov^{-1}( \hat{\bs{\tau}}_{\bs{X}} )
    \Cov\big( \hat{\bs{\tau}}_{\bs{X}}, \hat{\bs{\tau}}_{\bs{C}} \big) (\bs{\beta} - \tilde{\bs{\beta}})\\
    & = 
    (\bs{\beta} - \tilde{\bs{\beta}})^\top \Cov\big( \sqrt{n}\hat{\bs{\tau}}_{\bs{C}}, \sqrt{n}\hat{\bs{\tau}}_{\bs{X}}\big)
    \Cov^{-1}( \sqrt{n}\hat{\bs{\tau}}_{\bs{X}} )
    \Cov\big( \sqrt{n}\hat{\bs{\tau}}_{\bs{X}}, \sqrt{n}\hat{\bs{\tau}}_{\bs{C}} \big) (\bs{\beta} - \tilde{\bs{\beta}})
\end{align*}
By the same logic as Lemma \ref{lemma:cov_crse}, 
$(\sqrt{n}\hat{\bs{\tau}}_{\bs{X}}^\top, \sqrt{n}\hat{\bs{\tau}}_{\bs{C}}^\top)^\top$ has the following covariance under the CRSE:
\begin{align*}
    \begin{pmatrix}
    \Cov(\sqrt{n}\hat{\bs{\tau}}_{\bs{X}}) & \Cov(\sqrt{n}\hat{\bs{\tau}}_{\bs{X}}, \sqrt{n}\hat{\bs{\tau}}_{\bs{C}}) \\ 
    \Cov(\sqrt{n}\hat{\bs{\tau}}_{\bs{C}}, \sqrt{n}\hat{\bs{\tau}}_{\bs{X}}) & \Cov(\sqrt{n}\hat{\bs{\tau}}_{\bs{C}})
    \end{pmatrix}
    = 
    (r_1r_0)^{-1}
    \begin{pmatrix}
    \bs{S}^2_{\bs{X}} & \bs{S}_{\bs{X}, \bs{C}}\\
    \bs{S}_{\bs{C}, \bs{X}} & \bs{S}^2_{\bs{C}}
    \end{pmatrix}.
\end{align*}
We can then simplify $V_{\tau\tau}(\bs{\beta}, \bs{0}) R_T^2(\bs{\beta}, \bs{0})$ as
\begin{align*}
    V_{\tau\tau}(\bs{\beta}, \bs{0}) R_T^2(\bs{\beta}, \bs{0})
    & = 
    (r_1r_0)^{-1}
    (\bs{\beta} - \tilde{\bs{\beta}})^\top \bs{S}_{\bs{C}, \bs{X}}
    \big( \bs{S}^2_{\bs{X}} \big)^{-1}
    \bs{S}_{\bs{X}, \bs{C}} (\bs{\beta} - \tilde{\bs{\beta}})
    = 
    (r_1r_0)^{-1}
    (\bs{\beta} - \tilde{\bs{\beta}})^\top \bs{S}_{\bs{C}\mid \bs{X}}^2 (\bs{\beta} - \tilde{\bs{\beta}}).
\end{align*}
Consequently, we have
\begin{align*}
    & \quad \ V_{\tau\tau}(\bs{\beta}, \bs{0}) \big\{ 1 - R_S^2(\bs{\beta}, \bs{0}) \big\}
    \\
    & = 
    V_{\tau\tau}(\bs{\beta}, \bs{0}) - V_{\tau\tau}(\bs{\beta}, \bs{0}) R_S^2(\bs{\beta}, \bs{0})
    = 
    V_{\tau\tau}(1-R^2_C) + (r_1r_0)^{-1} (\bs{\beta} - \tilde{\bs{\beta}})^\top 
    \left( \bs{S}^2_{\bs{C}} - \bs{S}_{\bs{C}\mid \bs{X}}^2 \right) (\bs{\beta} - \tilde{\bs{\beta}})\\
    & = V_{\tau\tau}(1-R^2_C) + (r_1r_0)^{-1} (\bs{\beta} - \tilde{\bs{\beta}})^\top 
    \bs{S}_{\bs{C}\setminus \bs{X}}^2 (\bs{\beta} - \tilde{\bs{\beta}}).
\end{align*}

From the above, Lemma \ref{lemma:simp_X_in_C} holds. 
\end{proof}
\begin{proof}[Proof of Lemma \ref{lemma:sum_linear}]
Lemma \ref{lemma:sum_linear} follows immediately from \citet{wintner1936class}. 
\end{proof}
\begin{proof}[Proof of Lemma \ref{lemma:sum_unimodal}]
	Lemma \ref{lemma:sum_unimodal} follows immediately from \citet[][Theorem 7.5]{dharmadhikari1988}. 
\end{proof}
\begin{proof}[Proof of Lemma \ref{lemma:linear_com_eps_L}]
	From Lemma \ref{lemma:LKa},
	$\varepsilon$, $L_{K_1, a_1}$ and $L_{K_2, a_2}$ are all symmetric and unimodal around zero. 
	First, for any $c\ge 0$, 
	because 
	$\Pr( |b_0 \varepsilon_0| \le c) \ge \Pr( |\overline{b}_0 \varepsilon_0| \le c) $ 
	and 
	$b_1 L_{K_1,a_1}+b_2 L_{K_2,a_2}$ is symmetric and unimodal around zero by Lemma \ref{lemma:sum_linear}, 
	from Lemma \ref{lemma:sum_unimodal}, 
	we have 
	\begin{align}\label{eq:prob_bound1}
		\Pr\left(
		\left| b_0 \varepsilon_0 + b_1 L_{K_1,a_1}+b_2 L_{K_2,a_2} \right| \le c
		\right)
		\ge 
		\Pr\left(
		\left| \overline{b}_0 \varepsilon_0 + b_1 L_{K_1,a_1}+b_2 L_{K_2,a_2} \right| \le c
		\right). 
	\end{align}
	Second, for any $c\ge 0$, because 
	$\Pr( |b_1 L_{K_1,a_1}| \le c) \ge \Pr( |\overline{b}_1 L_{K_1,a_1}| \le c)$
	and 
	$
	\overline{b}_0 \varepsilon_0 + b_2 L_{K_2,a_2}
	$ 
	is symmetric and unimodal around zero by Lemma \ref{lemma:sum_linear}, 
	from Lemma \ref{lemma:sum_unimodal}, 
	we have 
	\begin{align}\label{eq:prob_bound2}
	\Pr\left(
	\left| \overline{b}_0 \varepsilon_0 + b_1 L_{K_1,a_1}+b_2 L_{K_2,a_2} \right| \le c
	\right)
	\ge 
	\Pr\left(
	\left| \overline{b}_0 \varepsilon_0 + \overline{b}_1 L_{K_1,a_1}+b_2 L_{K_2,a_2} \right| \le c
	\right). 
	\end{align}
	Third, 
	for any $c\ge 0$, because 
	$\Pr( |b_2 L_{K_2,a_2}| \le c) \ge \Pr( |\overline{b}_2 L_{K_2,a_2}| \le c)$
	and 
	$
	\overline{b}_0 \varepsilon_0 + \overline{b}_1 L_{K_1,a_1}
	$ 
	is symmetric and unimodal around zero by Lemma \ref{lemma:sum_linear}, 
	from Lemma \ref{lemma:sum_unimodal}, 
	we have 
	\begin{align}\label{eq:prob_bound3}
	\Pr\left(
	\left| \overline{b}_0 \varepsilon_0 + \overline{b}_1 L_{K_1,a_1}+b_2 L_{K_2,a_2} \right| \le c
	\right)
	\ge 
	\Pr\left(
	\left| \overline{b}_0 \varepsilon_0 + \overline{b}_1 L_{K_1,a_1}+\overline{b}_2 L_{K_2,a_2} \right| \le c
	\right). 
	\end{align}
	From \eqref{eq:prob_bound1}--\eqref{eq:prob_bound3}, we can then derive that for any $c\ge 0$, 
	\begin{align*}
		\Pr\left(
		\left| b_0 \varepsilon_0 + b_1 L_{K_1,a_1}+b_2 L_{K_2,a_2} \right| \le c
		\right)
		& \ge 
		\Pr\left(
		\left| \overline{b}_0 \varepsilon_0 + \overline{b}_1 L_{K_1,a_1}+\overline{b}_2 L_{K_2,a_2} \right| \le c
		\right).
	\end{align*}
	Therefore, Lemma \ref{lemma:linear_com_eps_L} holds. 
\end{proof}

\subsection{Proofs of the theorems and comments on technical details}

\begin{proof}[\bf Proof of Theorem \ref{thm:dist_reg_general}]
    By definition, the difference-in-means estimator based on observed adjusted outcomes has the following equivalent forms:
    \begin{align*}
        & \quad \ \frac{1}{n_1}\sum_{i\in \mathcal{S}}T_i Y_i(\bs{\beta}, \bs{\gamma})
        - 
        \frac{1}{n_0}\sum_{i\in \mathcal{S}} (1-T_i) Y_i(\bs{\beta}, \bs{\gamma})
        \\
        & = 
        \frac{1}{n_1}\sum_{i\in \mathcal{S}}T_i
        \big\{
        Y_i - \bs{\beta}^\top \bs{C}_i - r_1 \bs{\gamma}^\top  (\bs{E}_i - \bar{\bs{E}})
        \big\}
        - 
        \frac{1}{n_0}\sum_{i\in \mathcal{S}} (1-T_i) 
        \big\{Y_i - \bs{\beta}^\top \bs{C}_i + r_0 \bs{\gamma}^\top (\bs{E}_i - \bar{\bs{E}})\big\}\\
        & = 
        \left\{ \frac{1}{n_1}\sum_{i\in \mathcal{S}}T_i Y_i
        - \frac{1}{n_0}\sum_{i\in \mathcal{S}} (1-T_i) Y_i \right\}
        - \bs{\beta}^\top \left\{ \frac{1}{n_1}\sum_{i\in \mathcal{S}}T_i \bs{C}_i
        - \frac{1}{n_0}\sum_{i\in \mathcal{S}} (1-T_i) \bs{C}_i \right\}\\
        & \quad \ 
        - \bs{\gamma}^\top 
        \left\{
        \frac{1}{n}\sum_{i\in \mathcal{S}}T_i (\bs{E}_i - \bar{\bs{E}})
        + \frac{1}{n}\sum_{i\in \mathcal{S}} (1-T_i) (\bs{E}_i - \bar{\bs{E}})
        \right\}\\
        & = \hat{\tau} 
        - \bs{\beta}^\top \hat{\bs{\tau}}_{\bs{C}} - \bs{\gamma}^\top \hat{\bs{\delta}}_{\bs{E}}, 
    \end{align*}
    which is the same as the regression-adjusted estimator $\hat{\tau}(\bs{\beta}, \bs{\gamma})$ in \eqref{eq:reg}. 
    Moreover, we can verify that, under Condition \ref{cond:fp_analysis}, the adjusted potential outcomes $Y(1; \bs{\beta}, \bs{\gamma})$ and $Y(0; \bs{\beta}, \bs{\gamma})$ and covariates $\bs{X}$ and $\bs{W}$ must also satisfy Condition \ref{cond:fp}. 
    Thus, 
    from Theorem \ref{thm:dist}, 
    the difference-in-means for the adjusted outcomes has the following asymptotic distribution under ReSEM: 
    \begin{align*}
    & \quad \ \sqrt{n} \left\{ \hat{\tau}(\bs{\beta}, \bs{\gamma}) - \tau \right\} \mid  \text{ReSEM} 
    \nonumber
    \\
    & \ \dot\sim\  
    V_{\tau\tau}^{1/2}(\bs{\beta}, \bs{\gamma}) \left(  \sqrt{1-R_S^2(\bs{\beta}, \bs{\gamma}) - R_T^2(\bs{\beta}, \bs{\gamma})} \cdot \varepsilon
    + \sqrt{R_S^2(\bs{\beta}, \bs{\gamma})} \cdot L_{J,a_S}
    +
    \sqrt{R_T^2(\bs{\beta}, \bs{\gamma})} \cdot L_{K,a_T} 
    \right),
    \end{align*}
    where $(\varepsilon, L_{J,a_S}, L_{K,a_T} )$ are defined the same as in Theorem \ref{thm:dist}, and 
    $V_{\tau\tau}(\bs{\beta}, \bs{\gamma})$, $R_S^2(\bs{\beta}, \bs{\gamma})$ and $R_T^2(\bs{\beta}, \bs{\gamma})$ are defined in \eqref{eq:V_tau_reg} and \eqref{eq:R2_reg}. 
    Therefore, Theorem \ref{thm:dist_reg_general} holds. 
\end{proof}

\begin{proof}[\bf Proof of Corollary \ref{cor:dist_reg_general_equ}]
    We first simplify $V_{\tau\tau}(\bs{\beta}, \bs{\gamma})R_T^2(\bs{\beta}, \bs{\gamma})$ and $V_{\tau\tau}(\bs{\beta}, \bs{\gamma})R_S^2(\bs{\beta}, \bs{\gamma})$. 
    From Proposition \ref{prop:R2}, $R_T^2(\bs{\beta}, \bs{\gamma})$ and $R_S^2(\bs{\beta}, \bs{\gamma})$ are equivalently the squared multiple correlations between the difference-in-means estimator based on the adjusted outcome (i.e., $\hat{\tau}(\bs{\beta}, \bs{\gamma})$) and the difference-in-means of covariates $\hat{\bs{\tau}}_{\bs{X}}$ and $\hat{\bs{\delta}}_{\bs{W}}$ under the CRSE. 
    Thus, by definition, we have
    \begin{align}\label{eq:VR2_proof1}
        n^{-1} V_{\tau\tau}(\bs{\beta}, \bs{\gamma})R_T^2(\bs{\beta}, \bs{\gamma})
        & = 
        \Cov\big\{\hat{\tau}(\bs{\beta}, \bs{\gamma}), \hat{\bs{\tau}}_{\bs{X}}\big\}
        \Cov^{-1}( \hat{\bs{\tau}}_{\bs{X}} )
        \Cov\big\{\hat{\bs{\tau}}_{\bs{X}}, \hat{\tau}(\bs{\beta}, \bs{\gamma})\big\}, 
        \nonumber
        \\
        n^{-1} V_{\tau\tau}(\bs{\beta}, \bs{\gamma})R_S^2(\bs{\beta}, \bs{\gamma})
        & = 
        \Cov\big\{\hat{\tau}(\bs{\beta}, \bs{\gamma}), \hat{\bs{\delta}}_{\bs{W}}\big\}
        \Cov^{-1}( \hat{\bs{\delta}}_{\bs{W}} )
        \Cov\big\{\hat{\bs{\delta}}_{\bs{W}}, \hat{\tau}(\bs{\beta}, \bs{\gamma})\big\}.
    \end{align}
    By the same logic as Lemma \ref{lemma:cov_crse}, 
    under the CRSE, 
    $\hat{\bs{\delta}}_{\bs{E}}$ is uncorrelated with $\hat{\bs{\tau}}_{\bs{X}}$, 
    and 
    $\hat{\bs{\tau}}_{\bs{C}}$ is uncorrelated with $\hat{\bs{\delta}}_{\bs{W}}$.  
    From \eqref{eq:reg}, we then have 
    \begin{align}\label{eq:VR2_proof2}
        \Cov\big\{\hat{\tau}(\bs{\beta}, \bs{\gamma}), \hat{\bs{\tau}}_{\bs{X}}\big\}
        & = 
        \Cov\big(\hat{\tau} - \bs{\beta}^\top \hat{\bs{\tau}}_{\bs{C}} 
        -
        \bs{\gamma}^\top \hat{\bs{\delta}}_{\bs{E}}, \hat{\bs{\tau}}_{\bs{X}}\big)
        =
        \Cov\big(\hat{\tau} - \bs{\beta}^\top \hat{\bs{\tau}}_{\bs{C}}, \hat{\bs{\tau}}_{\bs{X}}\big)
        = \Cov\big\{\hat{\tau}(\bs{\beta}, \bs{0}), \hat{\bs{\tau}}_{\bs{X}}\big\}, 
        \nonumber
        \\
        \Cov\big\{\hat{\tau}(\bs{\beta}, \bs{\gamma}), \hat{\bs{\delta}}_{\bs{W}}\big\}
        & = 
        \Cov\big( \hat{\tau} - \bs{\beta}^\top \hat{\bs{\tau}}_{\bs{C}} 
        -
        \bs{\gamma}^\top \hat{\bs{\delta}}_{\bs{E}}, \hat{\bs{\delta}}_{\bs{W}}\big)
        = 
        \Cov\big( \hat{\tau} 
        -
        \bs{\gamma}^\top \hat{\bs{\delta}}_{\bs{E}}, \hat{\bs{\delta}}_{\bs{W}}\big)
        = \Cov\big\{\hat{\tau}(\bs{0}, \bs{\gamma}), \hat{\bs{\delta}}_{\bs{W}}\big\}. 
    \end{align}
    \eqref{eq:VR2_proof1} and \eqref{eq:VR2_proof2} then imply that 
    \begin{align*}
        n^{-1} V_{\tau\tau}(\bs{\beta}, \bs{\gamma})R_T^2(\bs{\beta}, \bs{\gamma})
        & = 
        \Cov\big\{\hat{\tau}(\bs{\beta}, \bs{0}), \hat{\bs{\tau}}_{\bs{X}}\big\}
        \Cov^{-1}( \hat{\bs{\tau}}_{\bs{X}} )
        \Cov\big\{\hat{\bs{\tau}}_{\bs{X}}, \hat{\tau}(\bs{\beta}, \bs{0})\big\} = 
        n^{-1} V_{\tau\tau}(\bs{\beta}, \bs{0})R_T^2(\bs{\beta}, \bs{0}),
        \nonumber
        \\
        n^{-1} V_{\tau\tau}(\bs{\beta}, \bs{\gamma})R_S^2(\bs{\beta}, \bs{\gamma})
        & = 
        \Cov\big\{\hat{\tau}(\bs{0}, \bs{\gamma}), \hat{\bs{\delta}}_{\bs{W}}\big\}
        \Cov^{-1}( \hat{\bs{\delta}}_{\bs{W}} )
        \Cov\big\{\hat{\bs{\delta}}_{\bs{W}}, \hat{\tau}(\bs{0}, \bs{\gamma})\big\} = 
        n^{-1} V_{\tau\tau}(\bs{0}, \bs{\gamma})R_S^2(\bs{0}, \bs{\gamma}), 
    \end{align*}
	i.e., 
	$
	V_{\tau\tau}(\bs{\beta}, \bs{\gamma})R_T^2(\bs{\beta}, \bs{\gamma}) = V_{\tau\tau}(\bs{\beta}, \bs{0})R_T^2(\bs{\beta}, \bs{0})$ 
	and  $V_{\tau\tau}(\bs{\beta}, \bs{\gamma})R_S^2(\bs{\beta}, \bs{\gamma}) = V_{\tau\tau}(\bs{0}, \bs{\gamma})R_S^2(\bs{0}, \bs{\gamma}).$ 
    
    We then give an equivalent decomposition of $V_{\tau\tau}(\bs{\beta}, \bs{\gamma})$. 
    By definition,   
    \begin{align*}
        V_{\tau\tau}(\bs{\beta}, \bs{\gamma})
        & = 
        n\Var\big( \hat{\tau} - \bs{\beta}^\top \hat{\bs{\tau}}_{\bs{C}} - \bs{\gamma}^\top \hat{\bs{\delta}}_{\bs{E}} \big)
        = 
        n\Var\left\{ \big( \hat{\tau} - \bs{\beta}^\top  \hat{\bs{\tau}}_{\bs{C}} \big) + \big(\hat{\tau} - \bs{\gamma}^\top \hat{\bs{\delta}}_{\bs{E}}\big) - \hat{\tau} \right\}\\
        & = V_{\tau\tau}(\bs{\beta}, \bs{0}) + V_{\tau\tau}(\bs{0}, \bs{\gamma}) + V_{\tau\tau} 
        + 2n\Cov\big( \hat{\tau} - \bs{\beta}^\top  \hat{\bs{\tau}}_{\bs{C}}, \hat{\tau} - \bs{\gamma}^\top \hat{\bs{\delta}}_{\bs{E}}\big)
        - 2n\Cov\big( \hat{\tau} - \bs{\beta}^\top  \hat{\bs{\tau}}_{\bs{C}}, \hat{\tau} \big)
        \\
        & \quad \
        - 2n\Cov\big( \hat{\tau} - \bs{\gamma}^\top \hat{\bs{\delta}}_{\bs{E}}, \hat{\tau} \big)
    \end{align*}
    Note that $\Cov( \hat{\bs{\tau}}_{\bs{C}}, \hat{\bs{\delta}}_{\bs{E}}) = \bs{0}$ under the CRSE. We then have 
    \begin{align*}
        & \quad \ \Cov\big( \hat{\tau} - \bs{\beta}^\top  \hat{\bs{\tau}}_{\bs{C}}, \hat{\tau} - \bs{\gamma}^\top \hat{\bs{\delta}}_{\bs{E}}\big)
        - \Cov\big( \hat{\tau} - \bs{\beta}^\top  \hat{\bs{\tau}}_{\bs{C}}, \hat{\tau} \big)
        - \Cov\big( \hat{\tau} - \bs{\gamma}^\top \hat{\bs{\delta}}_{\bs{E}}, \hat{\tau} \big)
        \\
        & = 
        \Cov\big( \hat{\tau} - \bs{\beta}^\top  \hat{\bs{\tau}}_{\bs{C}}, - \bs{\gamma}^\top \hat{\bs{\delta}}_{\bs{E}}\big)
        - \Cov\big( \hat{\tau} - \bs{\gamma}^\top \hat{\bs{\delta}}_{\bs{E}}, \hat{\tau} \big)
        = 
        \Cov\big( \hat{\tau}, - \bs{\gamma}^\top \hat{\bs{\delta}}_{\bs{E}}\big)
        - \Cov\big( \hat{\tau} - \bs{\gamma}^\top \hat{\bs{\delta}}_{\bs{E}}, \hat{\tau} \big)
        \\
        & = \Cov(\hat{\tau}, -\hat{\tau}) = -n^{-1}V_{\tau\tau}. 
    \end{align*}
    This further implies that
    \begin{align}\label{eq:simp_form_V_reg}
         V_{\tau\tau}(\bs{\beta}, \bs{\gamma})
         = V_{\tau\tau}(\bs{\beta}, \bs{0}) + V_{\tau\tau}(\bs{0}, \bs{\gamma}) + V_{\tau\tau} - 2V_{\tau\tau} 
         = V_{\tau\tau}(\bs{\beta}, \bs{0}) + V_{\tau\tau}(\bs{0}, \bs{\gamma}) - V_{\tau\tau}. 
    \end{align}
    
    From the above, the asymptotic distribution in \eqref{eq:dist_reg_general} has the following equivalent forms:
    \begin{align*}
        & \quad \ V_{\tau\tau}^{1/2}(\bs{\beta}, \bs{\gamma}) \left(  \sqrt{1-R_T^2(\bs{\beta}, \bs{\gamma})-R_S^2(\bs{\beta}, \bs{\gamma}) } \cdot \varepsilon
        + \sqrt{R_S^2(\bs{\beta}, \bs{\gamma})} \cdot L_{J,a_S}
        +
        \sqrt{R_T^2(\bs{\beta}, \bs{\gamma})} \cdot L_{K,a_T} 
        \right)
        \\
        & \sim 
        \sqrt{V_{\tau\tau}(\bs{\beta}, \bs{\gamma}) - V_{\tau\tau}(\bs{\beta}, \bs{\gamma}) R_S^2(\bs{\beta}, \bs{\gamma})  - V_{\tau\tau}(\bs{\beta}, \bs{\gamma}) R_T^2(\bs{\beta}, \bs{\gamma}) } \cdot \varepsilon
        \\
        & \quad \ 
        + \sqrt{V_{\tau\tau}(\bs{\beta}, \bs{\gamma})R_S^2(\bs{\beta}, \bs{\gamma})} \cdot L_{J,a_S}
        +
        \sqrt{V_{\tau\tau}(\bs{\beta}, \bs{\gamma})R_T^2(\bs{\beta}, \bs{\gamma})} \cdot L_{K,a_T} 
        \\
        & \sim 
        \sqrt{V_{\tau\tau}(\bs{0}, \bs{\gamma}) + V_{\tau\tau}(\bs{\beta}, \bs{0}) - V_{\tau\tau} - V_{\tau\tau}(\bs{0}, \bs{\gamma}) R_S^2(\bs{0}, \bs{\gamma})  - V_{\tau\tau}(\bs{\beta}, \bs{0}) R_T^2(\bs{\beta}, \bs{0}) } \cdot \varepsilon
        \\
        & \quad \ 
        + \sqrt{V_{\tau\tau}(\bs{0}, \bs{\gamma})R_S^2(\bs{0}, \bs{\gamma})} \cdot L_{J,a_S}
        +
        \sqrt{V_{\tau\tau}(\bs{\beta}, \bs{0})R_T^2(\bs{\beta}, \bs{0})} \cdot L_{K,a_T} 
        \\
        & \sim 
        \sqrt{
        V_{\tau\tau}(\bs{0}, \bs{\gamma})\{1 - R_S^2(\bs{0}, \bs{\gamma})\} + 
        V_{\tau\tau}(\bs{\beta}, \bs{0})\{1 - R_T^2(\bs{\beta}, \bs{0}) \}
        -V_{\tau\tau}
        } \cdot \varepsilon
        \\
        & \quad \ 
        + \sqrt{ V_{\tau\tau}(\bs{0}, \bs{\gamma}) R_S^2(\bs{0}, \bs{\gamma})} \cdot L_{J,a_S}
        +
        \sqrt{V_{\tau\tau}(\bs{\beta}, \bs{0})R_T^2(\bs{\beta}, \bs{0})} \cdot L_{K,a_T}. 
    \end{align*}
    Therefore, Corollary \ref{cor:dist_reg_general_equ} holds. 
\end{proof}

\begin{proof}[\bf Comment on the coefficients $\tilde{\bs{\beta}}$ and $\tilde{\bs{\gamma}}$]
From Lemma \ref{lemma:proj_coef}, $\tilde{\bs{\beta}}$ and $\tilde{\bs{\gamma}}$ 
defined in Section \ref{sec:R2_proj_coef}
are the linear projection coefficients of $\hat{\tau}$ on $\hat{\bs{\tau}}_{\bs{C}}$ and $\hat{\bs{\delta}}_{\bs{E}}$, respectively, under the CRSE.
\end{proof}

\begin{proof}[\bf Comment on the asymptotic distribution of $\hat{\tau}(\bs{\beta}, \bs{\gamma})$ when $\bs{W} \subset \bs{E}$ or $\bs{X} \subset \bs{C}$]
First, we consider the case in which $\bs{W} \subset \bs{E}$. 
From 
Theorem \ref{thm:dist_reg_general}, Corollary \ref{cor:dist_reg_general_equ} and Lemma \ref{lemma:simp_W_in_E}, 
under ReSEM, 
the asymptotic distribution of $\sqrt{n}\{\hat{\tau}(\bs{\beta}, \bs{\gamma})-\tau\}$ has the following equivalent forms:
\begin{align}\label{eq:dist_W_in_E_proof}
    & \quad \ \sqrt{n} \left\{ \hat{\tau}(\bs{\beta}, \bs{\gamma}) - \tau \right\} \mid  \text{ReSEM}
    \nonumber
    \\
    & \ \dot\sim\  
    \sqrt{
    (1-f) (\bs{\gamma} - \tilde{\bs{\gamma}})^\top \bs{S}^2_{\bs{E} \setminus \bs{W}} (\bs{\gamma} - \tilde{\bs{\gamma}})
    + 
    V_{\tau\tau}(\bs{\beta}, \bs{0})\{1 - R_T^2(\bs{\beta}, \bs{0}) \}
    -V_{\tau\tau}R_E^2
    } \cdot \varepsilon
    \nonumber
    \\
    & \quad \ + \sqrt{(1-f) (\bs{\gamma} - \tilde{\bs{\gamma}})^\top \bs{S}^2_{\bs{E} \mid \bs{W}} (\bs{\gamma} - \tilde{\bs{\gamma}})} \cdot L_{J,a_S}
    +
    \sqrt{V_{\tau\tau}(\bs{\beta}, \bs{0})R_T^2(\bs{\beta}, \bs{0})} \cdot L_{K,a_T}.
\end{align}

Second, we consider the case in which $\bs{X} \subset \bs{C}$. 
From 
Theorem \ref{thm:dist_reg_general}, Corollary \ref{cor:dist_reg_general_equ} and Lemma \ref{lemma:simp_X_in_C}, 
under ReSEM, 
the asymptotic distribution of $\sqrt{n}\{\hat{\tau}(\bs{\beta}, \bs{\gamma})-\tau\}$ has the following equivalent forms:
\begin{align}\label{eq:dist_X_in_C_proof}
    & \quad \ \sqrt{n} \left\{ \hat{\tau}(\bs{\beta}, \bs{\gamma}) - \tau \right\} \mid  \text{ReSEM}  
    \nonumber
    \\
    & \ \dot\sim\  
    \sqrt{
    V_{\tau\tau}(\bs{0}, \bs{\gamma})\{1 - R_S^2(\bs{0}, \bs{\gamma})\}
    + (r_1r_0)^{-1} (\bs{\beta} - \tilde{\bs{\beta}})^\top 
    \bs{S}_{\bs{C}\setminus \bs{X}}^2 (\bs{\beta} - \tilde{\bs{\beta}})
    -V_{\tau\tau}R^2_C
    } \cdot \varepsilon
    \nonumber
    \\
    & \quad \ + \sqrt{ V_{\tau\tau}(\bs{0}, \bs{\gamma}) R_S^2(\bs{0}, \bs{\gamma})} \cdot L_{J,a_S}
    +
    \sqrt{(r_1r_0)^{-1}
    (\bs{\beta} - \tilde{\bs{\beta}})^\top \bs{S}_{\bs{C}\mid \bs{X}}^2 (\bs{\beta} - \tilde{\bs{\beta}})} \cdot L_{K,a_T}. 
\end{align}

Third, we consider the case in which $\bs{W} \subset \bs{E}$ and $\bs{X} \subset \bs{C}$. 
From Theorem \ref{thm:dist_reg_general}, 
Corollary \ref{cor:dist_reg_general_equ}
and Lemmas \ref{lemma:simp_W_in_E} and \ref{lemma:simp_X_in_C},
\begin{align}\label{eq:dist_WX_in_EC_proof}
    & \quad \ \sqrt{n} \left\{ \hat{\tau}(\bs{\beta}, \bs{\gamma}) - \tau \right\} \mid  \text{ReSEM}  
    \nonumber
    \\
    & \ \dot\sim\  
    \sqrt{
    V_{\tau\tau} (1-R_E^2-R_C^2) + 
    (1-f) (\bs{\gamma} - \tilde{\bs{\gamma}})^\top \bs{S}^2_{\bs{E} \setminus \bs{W}} (\bs{\gamma} - \tilde{\bs{\gamma}})
    + 
    (r_1 r_0)^{-1} (\bs{\beta} - \tilde{\bs{\beta}})^\top \bs{S}^2_{\bs{C} \setminus \bs{X}} (\bs{\beta} - \tilde{\bs{\beta}})
    } \cdot \varepsilon
    \nonumber
    \\
    & \quad \ + \sqrt{ (1-f) (\bs{\gamma} - \tilde{\bs{\gamma}})^\top \bs{S}^2_{\bs{E} \mid \bs{W}} (\bs{\gamma} - \tilde{\bs{\gamma}}) } \cdot L_{J,a_S}
    +
    \sqrt{(r_1 r_0)^{-1} (\bs{\beta} - \tilde{\bs{\beta}})^\top \bs{S}^2_{\bs{C} \mid \bs{X}} (\bs{\beta} - \tilde{\bs{\beta}})} \cdot L_{K,a_T}. 
\end{align}
\end{proof}

\begin{proof}[\bf Proof of Theorem \ref{thm:optimal}]
    First, we consider the case in which $\bs{W} \subset \bs{E}$. 
    From the equivalent forms of the asymptotic distribution of $\sqrt{n}\{\hat{\tau}(\bs{\beta}, \bs{\gamma})-\tau\}$ in \eqref{eq:dist_W_in_E_proof}, 
    only the coefficients of $\varepsilon$ and $L_{J, a_S}$ depend on $\bs{\gamma}$, and both of them are minimized at $\bs{\gamma} = \tilde{\bs{\gamma}}$. 
    Therefore, from Lemma \ref{lemma:linear_com_eps_L}, 
    for any given $\bs{\beta}$, 
    $\hat{\tau}(\bs{\beta}, \tilde{\bs{\gamma}})$ is $\mathcal{S}$-optimal among all regression-adjusted estimators of form $\hat{\tau}(\bs{\beta}, \bs{\gamma})$. 
    
    Second, we consider the case in which $\bs{X} \subset \bs{C}$. 
    From the equivalent forms of the asymptotic distribution of $\sqrt{n}\{\hat{\tau}(\bs{\beta}, \bs{\gamma})-\tau\}$ in \eqref{eq:dist_X_in_C_proof}, 
    only the coefficients of $\varepsilon$ and $L_{K, a_T}$ depend on $\bs{\beta}$, and both of them are minimized at $\bs{\beta} = \tilde{\bs{\beta}}$. 
    Therefore, from Lemma \ref{lemma:linear_com_eps_L}, 
    for any given $\bs{\gamma}$, 
    $\hat{\tau}(\tilde{\bs{\beta}}, \bs{\gamma})$ is $\mathcal{S}$-optimal among all regression-adjusted estimators of form $\hat{\tau}(\bs{\beta}, \bs{\gamma})$. 
    
    Third, we consider the case in which both $\bs{W} \subset \bs{E}$ and $\bs{X} \subset \bs{C}$. 
    From the equivalent forms of the asymptotic distribution of $\sqrt{n}\{\hat{\tau}(\bs{\beta}, \bs{\gamma})-\tau\}$ in \eqref{eq:dist_WX_in_EC_proof}, 
    the coefficients of $\varepsilon$, $L_{J, a_S}$ and $L_{K, a_T}$ are minimized at $(\bs{\beta}, \bs{\gamma}) = (\tilde{\bs{\beta}}, \tilde{\bs{\gamma}})$. 
    Therefore, from Lemma \ref{lemma:linear_com_eps_L}, 
    $\hat{\tau}(\tilde{\bs{\beta}}, \tilde{\bs{\gamma}})$ is $\mathcal{S}$-optimal among all regression-adjusted estimators of form $\hat{\tau}(\bs{\beta}, \bs{\gamma})$, with the following asymptotic distribution:
    \begin{align*}
    	\sqrt{n} \left\{ \hat{\tau}(\bs{\beta}, \bs{\gamma}) - \tau \right\} \mid  \text{ReSEM} 
    	& \ \dot\sim\  
    	\sqrt{
    		V_{\tau\tau} (1-R_E^2-R_C^2) 
    	} \cdot \varepsilon. 
    \end{align*}

	From the above, Theorem \ref{thm:optimal} holds. 
\end{proof}

\begin{proof}[{\bf Proof of Corollary \ref{cor:gain_analysis}}]
    First, we study the percentage reduction in asymptotic variance. 
	From Lemma \ref{lemma:LKa} and Theorems \ref{thm:dist} and \ref{thm:optimal}, the asymptotic variances of $\hat{\tau}$ and $\hat{\tau}(\bs{\tilde\beta}, \bs{\tilde\gamma})$ 
	under ReSEM are, respetively, 
	\begin{align*}
	    \Var_{\text{a}}\{\sqrt{n} (\hat\tau - \tau) \mid  \text{ReSEM} \}
	    & = V_{\tau \tau}\{ 1 -(1-v_{J,a_S})R_S^2-(1- v_{K,a_T})R_T^2 \}, 
	\end{align*}
	and 
	$$
	\Var_{\text{a}}\Big[ \sqrt{n} \big\{ \hat{\tau}(\bs{\tilde\beta}, \bs{\tilde\gamma} - \tau \big\}
	\mid \text{ReSEM}
	\Big]=
	V_{\tau\tau} ( 1 - R_E^2 - R_C^2 ).
	$$
	These imply that the percentage reduction in asymptotic variance is
	\begin{align*}
	1 - 
	\frac{\Var_{\text{a}}\Big[ \sqrt{n} \big\{ \hat{\tau}(\bs{\tilde\beta}, \bs{\tilde\gamma} - \tau \big\}
	\mid \text{ReSEM}
	\Big]}{\Var_{\text{a}}\{\sqrt{n} (\hat\tau - \tau) \mid  \text{ReSEM} \}}
	& = 1 - 
	\frac{
		V_{\tau\tau} ( 1 - R_E^2 - R_C^2 )
	}{
		V_{\tau \tau}\{ 1 -(1-v_{J,a_S})R_S^2-(1- v_{K,a_T})R_T^2 \}
	}
	\\
	& = 
	\frac{
		R_E^2 -  R_S^2 + R_C^2 - R_T^2  + v_{J, a_S} R_S^2 + v_{K, a_T} R_T^2 
	}{
		1 -(1-v_{J,a_S})R_S^2-(1- v_{K,a_T})R_T^2
	}. 
	\end{align*}
	Because $\bs{W} \subset \bs{E}$ and $\bs{X} \subset \bs{C}$, by definition, we must have 
	$R_E^2 \ge R^2_S$ and $R_C^2 \ge R^2_T$. We can verify that this percentage reduction is nonnegative and nondecreasing in $R_E^2$ and $R_C^2$. 
	
	Second, we study the percentage reduction in length of asymptotic $1-\alpha$ symmetric quantile range.  
	From Theorem \ref{thm:optimal}, the length of asymptotic $1-\alpha$ symmetric quantile range for the optimal adjusted estimator $\sqrt{n}\{\hat{\tau}(\tilde{\bs{\beta}}, \tilde{\bs{\gamma}}) - \tau \}$ under ReSEM is 
	$
	2 V_{\tau\tau}^{1/2}\sqrt{1-R_E^2-R_C^2} \cdot z_{1-\alpha/2}.
	$
	From Theorem \ref{thm:dist}, 
	the length of asymptotic $1-\alpha$ symmetric quantile range for the unadjusted estimator $\sqrt{n}(\hat{\tau}-\tau)$ is $2V_{\tau\tau}^{1/2} \nu_{1-\alpha/2}(R_S^2,R_T^2)$.
	These imply that  the percentage reduction in asymptotic $1-\alpha$ symmetric quantile range is 
	$$	
	1 - \frac{V_{\tau\tau}^{1/2}\sqrt{1-R_E^2-R_C^2} \cdot z_{1-\alpha/2}}{
		V_{\tau\tau}^{1/2} \nu_{1-\alpha/2}(R_S^2,R_T^2)
	}
	= 
	1 - \sqrt{1-R_E^2-R_C^2} \cdot \frac{z_{1-\alpha/2}}	{\nu_{1-\alpha/2}(R_S^2, R_T^2)}
	$$
	By the optimality of the regression-adjusted estimator $\hat{\tau}(\tilde{\bs{\beta}}, \tilde{\bs{\gamma}})$ in Theorem \ref{thm:optimal}, this percentage reduction is nonnegative. Furthermore, we can verify that this percentage reduction is nondecreasing in $R_E^2$ and $R_C^2$. 
	
	From the above, Corollary \ref{cor:gain_analysis} holds. 
\end{proof}

\section{Asymptotic properties for special cases}\label{sec:proof_special}

\subsection{Asymptotic properties for the CRSE}

\begin{proof}[{\bf Proof of Corollary \ref{cor:crse}}]
    Note that the CRSE is essentially ReSEM with $a_S=a_T=\infty$. 
	In Theorem \ref{thm:dist}, let $a_S=a_T=\infty$. 
	Then $L_{J,a_S}$ and $L_{K,a_T}$ become standard Gaussian random variables. Because $(\varepsilon, L_{J,a_S}, L_{K,a_T})$ are mutually independent and all standard Gaussian distributed, 
	we can derive that 
	$ \sqrt{1-R_S^2 - R_T^2} \cdot \varepsilon
	+ \sqrt{R_S^2} \cdot L_{J,a_S}
	+ \sqrt{R_T^2} \cdot L_{K,a_T} 
	\sim \mathcal{N}(0,1).$
	Consequently, 
	from Theorem \ref{thm:dist}, we have 
	$
	\sqrt{n} (\hat\tau - \tau)  \  \dot\sim\  V_{\tau\tau}^{1/2} \cdot \varepsilon.
	$
	Therefore, Corollary \ref{cor:crse} holds. 
\end{proof}

\begin{proof}[{\bf Proof of Corollary \ref{cor:reg_crse}}]
	Note that the CRSE is essentially ReSEM with $\bs{W}=\bs{X} = \emptyset$ and $a_S=a_T=\infty$. 
	In \eqref{eq:dist_WX_in_EC_proof}, let $a_S=a_T=\infty$. 
	Then both $L_{J,a_S}$ and $L_{K,a_T}$ become standard Gaussian random variable. 
	Moreover, because $(\varepsilon, L_{J,\infty}, L_{K,\infty})$ are mutually independent, we have
	\begin{align*}
		& \quad \ \sqrt{n} \left\{ \hat{\tau}(\bs{\beta}, \bs{\gamma}) - \tau \right\} \\
		&\dot\sim
		\sqrt{
			V_{\tau\tau} (1-R_E^2-R_C^2) + 
			(1-f) (\bs{\gamma} - \tilde{\bs{\gamma}})^\top \bs{S}^2_{\bs{E} \setminus \bs{W}} (\bs{\gamma} - \tilde{\bs{\gamma}})
			+ 
			(r_1 r_0)^{-1} (\bs{\beta} - \tilde{\bs{\beta}})^\top \bs{S}^2_{\bs{C} \setminus \bs{X}} (\bs{\beta} - \tilde{\bs{\beta}})
		} \cdot \varepsilon
		\nonumber
		\\
		& \quad \ + \sqrt{ (1-f) (\bs{\gamma} - \tilde{\bs{\gamma}})^\top \bs{S}^2_{\bs{E} \mid \bs{W}} (\bs{\gamma} - \tilde{\bs{\gamma}}) } \cdot L_{J,a_S}
		+
		\sqrt{(r_1 r_0)^{-1} (\bs{\beta} - \tilde{\bs{\beta}})^\top \bs{S}^2_{\bs{C} \mid \bs{X}} (\bs{\beta} - \tilde{\bs{\beta}})} \cdot L_{K,a_T}\\
		& \sim 
		\sqrt{ V_{\tau\tau} (1 - R_E^2 - R_C^2 ) 
			+ (1-f) (\bs{\gamma} - \tilde{\bs{\gamma}})^\top \bs{S}^2_{\bs{E}} (\bs{\gamma} - \tilde{\bs{\gamma}})	
			+  (r_1 r_0)^{-1} (\bs{\beta} - \tilde{\bs{\beta}})^\top \bs{S}^2_{\bs{C}} (\bs{\beta} - \tilde{\bs{\beta}}) 
		}
		\cdot \varepsilon. 
	\end{align*}
	The optimal regression-adjusted estimator is then attainable at 
	$(\bs{\beta}, \bs{\gamma}) = (\tilde{\bs{\beta}}, \tilde{\bs{\gamma}})$, with the following asymptotic distribution: 
	\begin{align}
		\sqrt{n} \big\{ \hat{\tau}( \tilde{\bs{\beta}}, \tilde{\bs{\gamma}}) - \tau \big\} 
		& \ \dot\sim \ 
		V_{\tau\tau}^{1/2}\sqrt{1 - R_E^2 - R_C^2}
		\cdot \varepsilon. 
	\end{align}
	Therefore, Corollary \ref{cor:reg_crse} holds. 
\end{proof}

\subsection{Asymptotic properties for the ReM}

\begin{proof}[{\bf Proof of Corollary \ref{cor:rem}}]
	When $f=1$,  
	all units are sampled to enroll the experiment, and  thus there is essentially no random sampling stage. 
	Recall the discussion at the beginning of Section \ref{sec:lemma_crse}. 
	We can then relax Condition \ref{cond:fp}(i) when studying the asymptotic properties of ReM. 
	By the same logic as Theorem \ref{thm:dist}, we can derive that 
	$
	\sqrt{n} (\hat\tau - \tau) \mid   M_T \le a_T  \ \dot\sim \ V_{\tau\tau}^{1/2} (  \sqrt{1- R_T^2} \cdot \varepsilon
	+
	\sqrt{R_T^2} \cdot L_{K,a_T} 
	),
	$
	i.e., Corollary \ref{cor:rem} holds. 
	Note that when $f=1$, $R_S^2 =0$ from Proposition \ref{prop:R2}. Therefore, the asymptotic distribution for ReM has the same form as that for ReSEM with $f=1$.
\end{proof}

\begin{proof}[{\bf Proof of Corollary \ref{cor:rem_reg}}]
	When $f=1$,
	all units are sampled to enroll the experiment, and  thus there is essentially no random sampling stage. 
	Note that $R_E^2 =0$ from Proposition \ref{prop:R2}. 
    By the same logic as \eqref{eq:dist_WX_in_EC_proof}, we can derive that
	\begin{align*}
		& \quad \ \sqrt{n} \left\{ \hat{\tau}(\bs{\beta}, \bs{\gamma}) - \tau \right\} \mid M_T \le a_T 
		\\
		& \dot\sim \ 
		\sqrt{ V_{\tau\tau} (1 - R_C^2) +  (r_1 r_0)^{-1} (\bs{\beta} - \tilde{\bs{\beta}})^\top \bs{S}^2_{\bs{C}\setminus \bs{X}} (\bs{\beta} - \tilde{\bs{\beta}}) }
		\cdot \varepsilon 
		+ 
		\sqrt{ (r_1 r_0)^{-1} (\bs{\beta} - \tilde{\bs{\beta}})^\top \bs{S}^2_{\bs{C} \mid \bs{X}} (\bs{\beta} - \tilde{\bs{\beta}})} \cdot L_{K, a_T},
	\end{align*}
	From Lemma \ref{lemma:linear_com_eps_L}, 
	the optimal regression-adjusted estimator is attainable at $\bs{\beta}  = \tilde{\bs{\beta}}$, with the following asymptotic distribution: 
	\begin{align*}
		\sqrt{n} \big\{ \hat{\tau}( \tilde{\bs{\beta}}, \bs{\gamma}) - \tau \big\} \mid M_T \le a_T 
		\  \dot\sim \ 
		\sqrt{ V_{\tau\tau} (1 - R_C^2) }
		\cdot \varepsilon. 
	\end{align*}
	Therefore,  Corollary \ref{cor:rem_reg} holds. 
\end{proof}

\subsection{Asymptotic properties for rejective sampling}

\begin{proof}[{\bf Proof of Corollary \ref{cor:rej_sam}}]
	We prove Corollary \ref{cor:rej_sam} using Theorem \ref{thm:dist}, by coupling the rejective sampling with ReSEM under a careful construction of the potential outcomes. 
	Specifically, 
	we construct the potential outcomes as  $Y_i(1) = r_1 y_i$ and $Y_i(0) = -r_0 y_i$, 
	where we choose $r_1 = [n/2]/n$, $r_0 = 1-r_1$, and $[n/2]$ denotes the smallest integer no less than $n/2$.  
	We also choose the sampling indicator vector $\bs{Z}$ to be the same under both rejective sampling and ReSEM, and choose $a_T = \infty$ for ReSEM. 
	We then have $\tau_i = Y_i(1) - Y_i(0) = y_i$, $\tau = N^{-1} \sum_{i=1}^N y_i$, and 
	\begin{align*}
		\hat{\tau} & =  
		\frac{1}{n_1} \sum_{i=1}^N Z_i T_i Y_i  - \frac{1}{n_0} \sum_{i=1}^N Z_i(1-T_i) Y_i
		=
		\frac{1}{nr_1} \sum_{i=1}^N Z_i T_i Y_i(1)  - \frac{1}{nr_0} \sum_{i=1}^N Z_i(1-T_i) Y_i(0)
		\\
		& = \frac{1}{n}\sum_{i=1}^N Z_i T_i y_i + \frac{1}{n} \sum_{i=1}^N Z_i(1-T_i) y_i\\
		& = \frac{1}{n}\sum_{i=1}^N Z_i y_i \equiv \bar{y}_{\mathcal{S}}.  
	\end{align*}
	This implies that 
	\begin{align*}
		\sqrt{n}\left(\hat{\tau} - \tau \right) \mid M_S \le a_S
		& \ \sim \ 
		\sqrt{n}\left(\bar{y}_S - \bar{y} \right) \mid M_S \le a_S. 
	\end{align*}
	From Theorem \ref{thm:dist}, this further implies that, 
	when Condition \ref{cond:fp} with $\bs{X} = \bs{W}$ holds 
	and $a_T = \infty$, 
	\begin{align}\label{eq:diff_rej_proof}
		\sqrt{n}\left(\bar{y}_S - \bar{y} \right) \mid M_S \le a_S  \ & \dot\sim\  V_{\tau\tau}^{1/2} \left(  \sqrt{1-R_S^2 - R_T^2} \cdot \varepsilon
		+ \sqrt{R_S^2} \cdot L_{J,a_S}
		+
		\sqrt{R_T^2} \cdot L_{K,a_T} 
		\right). 
	\end{align}
	
	First, we simplify \eqref{eq:diff_rej_proof}. For the constructed potential outcomes $(Y_i(1), Y_i(0)) = (r_1y_i, -r_0 y_i)$, 
	$V_{\tau\tau}$ and the squared multiple correlations reduce to 
	\begin{align}\label{eq:survey_simplify}
		V_{\tau\tau} & = r_1^{-1}S^2_1 +r_0^{-1}S^2_0 -fS^2_\tau = r_1^{-1} r_1^2 S^2_y +r_0^{-1} r_0^2 S^2_y -fS^2_y = (1-f)S^2_y, 
		\nonumber
		\\
		R_S^2 & = 
		\frac{(1-f)S^2_{\tau \mid  \bs{W}}}{r_1^{-1}S^2_1 +r_0^{-1}S^2_0 -fS^2_\tau}
		= 
		\frac{(1-f)S^2_{y\mid  \bs{W}}}{(1-f)S^2_y}
		= \rho^2_{y,\bs{W}}, 
		\nonumber
		\\
		R_T^2 & = 
		\frac{r_1^{-1}S^2_{1 \mid  \bs{X}}+r_0^{-1}S^2_{0 \mid  \bs{X}} -S^2_{\tau \mid  \bs{X}}}{r_1^{-1}S^2_1 +r_0^{-1}S^2_0 -fS^2_\tau}
		=
		\frac{r_1^{-1} r_1^2 S^2_{y \mid  \bs{X}}+r_0^{-1} r_0^2 S^2_{y \mid  \bs{X}} -S^2_{y \mid  \bs{X}}}{(1-f) S_y^2}
		= 0.
	\end{align}
	Therefore, \eqref{eq:diff_rej_proof} simplifies to 
	\begin{align*}
		\sqrt{n}\left(\bar{y}_S - \bar{y} \right) \mid M_S \le a_S  \ & \dot\sim\  \sqrt{(1-f)S_y^2} \left(  \sqrt{1-\rho^2_{y,\bs{W}}} \cdot \varepsilon
		+ \sqrt{\rho_{y,\bs{W}}^2} \cdot L_{J,a_S}
		\right). 
	\end{align*}
	
	Second, we simplify Condition \ref{cond:fp} with $\bs{X} = \bs{W}$. 
	Note that by construction, $r_1$ and $r_0$ converge to $1/2$ as $N\rightarrow \infty$, i.e., Condition \ref{cond:fp}(ii) holds automatically. This also implies that Condition \ref{cond:fp}(iii) and (iv) are equivalent to that (a) the finite population variances $S^2_y$, $\bs{S}^2_{\bs{W}}$ and covariances $\bs{S}_{y, \bs{W}}$ have limits, and the limit of $\bs{S}^2_{\bs{W}}$ is nonsingular, 
	and 
	(b) 
	$\max_{1\le i \le N}(y_i - \bar{y})^2/n \rightarrow 0$ and 	$\max_{1 \le i\le N}\| \bs{W}_i - \bar{\bs{W}} \|_2^2/n \rightarrow 0.$ 
	Therefore, Condition \ref{cond:fp} with $\bs{X} = \bs{W}$ is equivalent to Condition \ref{cond:survey}.
	
	From the above, Corollary \ref{cor:rej_sam} holds under Condition \ref{cond:survey}. 
\end{proof}

\begin{proof}[{\bf Proof of Corollary \ref{cor:rej_reg}}]
	Similar to the proof of Corollary \ref{cor:rej_sam}, we prove Corollary \ref{cor:rej_reg} using \eqref{eq:dist_WX_in_EC_proof}, by coupling rejective sampling with ReSEM. 
	We construct potential outcomes as $Y_i(1) = r_1 y_i$ and $Y_i(0) =  - r_0 y_i$ with $r_1 = [n/2]/n$ and $r_0 = 1 - r_0$. 
	Let $\bs{C} = \bs{E}$, $\bs{X} = \bs{W}$ and $a_T = \infty$. 
	By the same logic as the proof of Corollary \ref{cor:rej_sam} and from \eqref{eq:dist_W_in_E_proof}, 
	when Condition \ref{cond:fp_analysis} holds, 
	\begin{align}\label{eq:rej_reg_proof}
		& \quad \ 
		\sqrt{n} \left( \hat{\tau} - \bs{\gamma}^\top \hat{\bs{\delta}}_{\bs{E}} - \bar{y} \right) \mid M_S \le a_S 
		\ \sim \ 
		\sqrt{n} \left\{ \hat{\tau}(\bs{0}, \bs{\gamma}) - \tau \right\} \mid M_S \le a_S
		\nonumber
		\\
		& \dot\sim \ 
		\sqrt{ V_{\tau\tau} (1 - R_E^2)
			+ (1-f) (\bs{\gamma} - \tilde{\bs{\gamma}})^\top \bs{S}^2_{\bs{E} \setminus \bs{W}} (\bs{\gamma} - \tilde{\bs{\gamma}})
		}
		\cdot \varepsilon 
		+ 
		\sqrt{ (1-f) (\bs{\gamma} - \tilde{\bs{\gamma}})^\top \bs{S}^2_{\bs{E} \mid \bs{W}} (\bs{\gamma} - \tilde{\bs{\gamma}}) }
		\cdot L_{J, a_S}. 
	\end{align}
	
	First, we simplify \eqref{eq:rej_reg_proof}. 
	By the construction of potential outcomes and by the same logic as  \eqref{eq:survey_simplify}, 
	$V_{\tau\tau} = (1-f)S^2_y$ and $R^2_E = \rho^2_{y,\bs{E}}$, 
	and
	$
		\tilde{\bs{\gamma}}
		= 
		r_1 (  \bs{S}^2_{\bs{E}} )^{-1} \bs{S}_{\bs{E}, y} - (-r_0) \cdot  (  \bs{S}^2_{\bs{E}} )^{-1} \bs{S}_{\bs{E}, y}
		= 
		(  \bs{S}^2_{\bs{E}} )^{-1} \bs{S}_{\bs{E}, y}
		= 
		\overline{\bs{\gamma}}, 
	$
	where $\bs{S}_{\bs{E}, y}$ is the finite population covariance between $\bs{E}$ and $y$. 
	Therefore, \eqref{eq:rej_reg_proof} simplifies to 
	\begin{align}\label{eq:rej_reg_proof_2}
		\sqrt{n} \left( \hat{\tau} - \bs{\gamma}^\top \hat{\bs{\delta}}_{\bs{W}} - \bar{y} \right) \mid M_S \le a_S 
		& \ \dot\sim \ 
		\sqrt{ (1-f) S^2_y (1 - \rho_{y,\bs{E}}^2) + (1-f) (\bs{\gamma} - \overline{\bs{\gamma}})^\top \bs{S}^2_{\bs{E} \setminus \bs{W}} (\bs{\gamma} - \overline{\bs{\gamma}}) }
		\cdot \varepsilon 
		\nonumber
		\\
		& \quad \ + 
		\sqrt{ (1-f) (\bs{\gamma} - \overline{\bs{\gamma}})^\top \bs{S}^2_{\bs{E} \mid \bs{W}} (\bs{\gamma} - \overline{\bs{\gamma}}) }
		\cdot L_{J, a_S}.
	\end{align}

	Second, we simplify Condition \ref{cond:fp_analysis} with $\bs{C} = \bs{E}$ and $\bs{X} = \bs{W}$. By the same logic as the proof of Corollary \ref{cor:rej_sam}, 
	Condition \ref{cond:fp_analysis} with $\bs{C} = \bs{E}$ and $\bs{X} = \bs{W}$ is equivalent to Condition \ref{cond:survey_ana}. 
	
	Third, we derive the optimal regression adjustment under rejective sampling. 
	From \eqref{eq:rej_reg_proof_2} and Lemma \ref{lemma:linear_com_eps_L}, we can immediately know that the optimal regression-adjusted estimator is attainable at $\bs{\gamma} = \overline{\bs{\gamma}}$, with asymptotic distribution 
	\begin{align}%
		\sqrt{n} \big( \hat{\tau} - \overline{\bs{\gamma}}^\top \hat{\bs{\delta}}_{\bs{W}} - \bar{y} \big) \mid M_S \le a_S 
		\ \dot\sim \ 
		\sqrt{ (1-f) S^2_y (1 - \rho_{y,\bs{E}}^2) }
		\cdot \varepsilon. 
	\end{align}
	
	From the above, Corollary \ref{cor:rej_reg} holds. 
\end{proof}

\section{Variance estimation and confidence intervals under ReSEM}\label{app:ci}

\subsection{Lemmas}
\begin{lemma}\label{lemma:estimate}
    Let $A$ and $B$ be any finite population quantities that can be the treatment or control potential outcome, or any coordinate of covariates $\bs{W}, \bs{X}, \bs{E}$ or $\bs{C}$. 
	Under Condition \ref{cond:fp_analysis} and 
	$\sresem$, 
	for $t=0,1$ and units under treatment arm $t$,  
	the sample covariance between $A$ and $B$ is consistent for the finite population covariance between $A$ and $B$. 
\end{lemma}

\begin{lemma}\label{lemma:diff_average}
	Let $A$ be any finite population quantity that can be the treatment or control potential outcome, or any coordinate of the covariates $\bs{W}, \bs{X}, \bs{E}$ or $\bs{C}$, 
	$\bar{A}$ be its finite population average, 
	$\bar{A}_{\mathcal{S}}$ be its sample average, 
	and $\bar{A}_1$ and $\bar{A}_0$ be its averages in treatment and control groups. 
	Under 
    $\sresem$ 
    and Condition \ref{cond:fp_analysis}, 
	\begin{align*}
		\bar{A}_1 - \bar{A} = O_{\Pr}(n^{-1/2}), \ \ 
		\bar{A}_0 - \bar{A} = O_{\Pr}(n^{-1/2}), \ \ 
		\bar{A}_{\mathcal{S}} - \bar{A} = O_{\Pr}(n^{-1/2}), \ \ 
		\bar{A}_1 - \bar{A}_0 = O_{\Pr}(n^{-1/2}). 
	\end{align*}
\end{lemma}

\begin{lemma}\label{lemma:sotch_R2}
	For any positive integer $K_1,K_2$ and constants $a_1,a_2$,  
	let $\varepsilon_0 \sim \mathcal{N}(0,1)$, $L_{K_1, a_1} \sim D_1 \mid \bs{D}^\top\bs{D}\le a$ and 
	$L_{K_2, a_2} \sim \tilde{D}_1 \mid \tilde{\bs{D}}^\top\tilde{\bs{D}}\le a$, 
	where $\bs{D} = (D_1, \ldots, D_{K_1}) \sim \mathcal{N}(\bs{0}, \bs{I}_{K_1})$, $\tilde{\bs{D}} = (\tilde{D}_1, \ldots, \tilde{D}_{K_2}) \sim \mathcal{N}(\bs{0}, \bs{I}_{K_2})$, 
	and $(\varepsilon, L_{K_1, a_1}, L_{K_2, a_2})$ are mutually independent. 
	If $\tilde{R}_1^2 \le R_1^2$ and $\tilde{R}_2^2 \le R_2^2,$
	then for any $V\ge 0$ and $c\ge0$, 
	\begin{align*}
		& \quad \ \Pr\left(
		\left| 
		V^{1/2}\left(\sqrt{1 - \tilde{R}_1^2 - \tilde{R}_2^2} \ \varepsilon
		+\sqrt{\tilde{R}_1^2}\ L_{K_1, a_1} + \sqrt{\tilde{R}_2^2}\ L_{K_2, a_2}
		\right)
		\right| \le c
		\right)
		\\
		& \le 
		\Pr\left(
		\left| 
		V^{1/2}\left(\sqrt{1 - R_1^2 - R_2^2} \ \varepsilon
		+\sqrt{R_1^2}\ L_{K_1, a_1} + \sqrt{R_2^2}\ L_{K_2, a_2}
		\right)
		\right| \le c
		\right).
	\end{align*}
\end{lemma}

\begin{lemma}\label{lemma:sotch_est}
	For any positive integer $K_1,K_2$ and constants $a_1,a_2$,  
	let $\varepsilon \sim \mathcal{N}(0,1)$, $L_{K_1, a_1} \sim D_1 \mid \bs{D}^\top\bs{D}\le a_1$ and 
	$L_{K_2, a_2} \sim \tilde{D}_1 \mid \tilde{\bs{D}}^\top\tilde{\bs{D}}\le a_2$, 
	where $\bs{D} = (D_1, \ldots, D_{K_1}) \sim \mathcal{N}(\bs{0}, \bs{I}_{K_1})$, $\tilde{\bs{D}} = (\tilde{D}_1, \ldots, \tilde{D}_{K_2}) \sim \mathcal{N}(\bs{0}, \bs{I}_{K_2})$, 
	and $(\varepsilon, L_{K_1, a_1}, L_{K_2, a_2})$ are mutually independent. 
	If
	\begin{align*}
		\tilde{V} \ge V,
		\quad
		\tilde{V}\tilde{R}_1^2 \le VR_1^2, 
		\quad 
		\tilde{V}\tilde{R}_2^2 \le VR_2^2, 
	\end{align*}
	then for any $c\ge 0$, 
	\begin{align*}
		& \quad \ \Pr\left(
		\left| 
		\tilde{V}^{1/2}\left(\sqrt{1 - \tilde{R}_1^2 - \tilde{R}_2^2} \ \varepsilon
		+\sqrt{\tilde{R}_1^2}\ L_{K_1, a_1} + \sqrt{\tilde{R}_2^2}\ L_{K_2, a_2}
		\right)
		\right| \le c
		\right)
		\\
		& \le 
		\Pr\left(
		\left| 
		V^{1/2}\left(\sqrt{1 - R_1^2 - R_2^2} \ \varepsilon
		+\sqrt{R_1^2}\ L_{K_1, a_1} + \sqrt{R_2^2}\ L_{K_2, a_2}
		\right)
		\right| \le c
		\right).
	\end{align*}
\end{lemma}

\begin{lemma}\label{lemma:continuous_cdf}
	Let $a_1$ and $a_2$ be two positive constants, $K_1$ and $K_2$ be two positive integers, 
	$\varepsilon\sim \mathcal{N}(0,1)$ be a standard Gaussian random variable, 
	$L_{K_1, a_1} \sim D_1 \mid \bs{D}^\top\bs{D} \le a_1$ 
	and 
	$L_{K_2, a_2} \sim \tilde{D}_1 \mid \tilde{\bs{D}}^\top\tilde{\bs{D}} \le a_2$ be two constrained Gaussian random variables with 
	$\bs{D} = (D_1, \ldots, D_{K_1}) \sim \mathcal{N}(\bs{0}, \bs{I}_{K_1})$ 
	and 
	$\tilde{\bs{D}} = (\tilde{D}_1, \ldots, \tilde{D}_{K_2}) \sim \mathcal{N}(\bs{0}, \bs{I}_{K_2})$, 
	and $(\varepsilon, L_{K_1, a_1}, L_{K_2, a_2})$ be mutually independent. 
	Let $\Psi_{V, R_1^2, R_2^2}(\cdot)$ be the distribution function of 
	$$
	V^{1/2} \left( \sqrt{1-R_1^2-R_2^2} \cdot \varepsilon + \sqrt{R_1^2} \cdot L_{K_1, a_1} + \sqrt{R_2^2} \cdot L_{K_2, a_2} \right), 
	$$
	and $\Psi^{-1}_{V, R_1^2, R_2^2}(\cdot)$ be the corresponding quantile function. 
	Then $\Psi_{V, R_1^2, R_2^2}(y)$ is continuous in $(V, R_1^2, R_2^2, y) \in 
	(0, \infty) \times [0, 1] \times [0, 1] \times \mathbb{R}$, 
	and $\Psi^{-1}_{V, R_1^2, R_2^2}(p)$ is continuous in $(V, R_1^2, R_2^2, p) \in 
	(0, \infty) \times [0, 1] \times [0, 1] \times (0,1)$.  
\end{lemma}

\begin{lemma}\label{lemma:conserv_inf}
	Suppose $\hat{V}_N, \hat{R}_{1,N}^2$ and $\hat{R}_{2,N}^2$ are consistent estimators for $\tilde{V}_N, \tilde{R}_{1,N}^2\in [0,1]$ and $\tilde{R}_{2,N}^2\in [0,1]$, in the sense that as $N\rightarrow \infty$, 
	$$
	\hat{V}_N-\tilde{V}_N = o_{\Pr}(1), \quad \hat{R}_{1,N}^2-\tilde{R}_{1,N}^2 = o_{\Pr}(1), \quad \hat{R}_{2,N}^2-\tilde{R}_{2,N}^2 = o_{\Pr}(1), 
	$$ 
	and the quantities 
	$\tilde{V}_N, \tilde{R}_{1,N}^2, \tilde{R}_{2,N}^2$, 
	$V_N, R_{1,N}^2$ and $R_{2,N}^2$ have limits as $N\rightarrow \infty$, denoted by $\tilde{V}_\infty, \tilde{R}_{1,\infty}^2, \tilde{R}_{2,\infty}^2$, 
	$V_\infty, R_{1,\infty}^2$ and $R_{2,\infty}^2$. 
	Let $\psi_N$ be a random variable converging weakly to the following distribution: 
	\begin{align*}
		\psi_N \ \dot\sim \  
		V_N^{1/2}\left(\sqrt{1 - R_{1,N}^2 - R_{2,N}^2} \ \varepsilon
		+\sqrt{R_{1,N}^2}\ L_{K_1, a_1} + \sqrt{R_{2,N}^2}\ L_{K_2, a_2}
		\right). 
	\end{align*}
	and $\hat{\psi}_N$ be a random variable having the following distribution:
	\begin{align*}
		\hat{\psi}_N \sim  
		\hat{V}_N^{1/2}\left(\sqrt{1 - \hat{R}_{1,N}^2 - \hat{R}_{2,N}^2} \ \varepsilon
		+\sqrt{\hat{R}_{1,N}^2}\ L_{K_1, a_1} + \sqrt{\hat{R}_{2,N}^2}\ L_{K_2, a_2}
		\right). 
	\end{align*}
	With a slight abuse of notation, here we view $\hat{\psi}_N$ as a random variable whose distribution is determined by  $\hat{V}_N, \hat{R}_{1,N}^2$ and $\hat{R}_{2,N}^2$.
	If $\tilde{V}_{\infty} > 0$, 
	and 
	\begin{align*}
		\tilde{V}_\infty \ge V_\infty,
		\quad
		\tilde{V}_\infty \tilde{R}_{1,\infty}^2 \le V_\infty R_{1, \infty}^2, 
		\quad 
		\tilde{V}_\infty \tilde{R}_{2,\infty}^2 \le V_\infty R_{2, \infty}^2, 
	\end{align*}
	then, as $N\rightarrow \infty$, 
	the probability limit of the variance of $\hat{\psi}_N$ is greater than or equal to the asymptotic variance of $\psi_N$, 
	and, for any $\alpha\in(0,1)$, the probability that the $1-\alpha$ symmetric quantile range of the distribution of $\hat{\psi}_N$ covers $\psi_N$ converges to a limit greater than or equal to $1-\alpha$. 
	Moreover, both of them have equality hold 
    when 
	$(\tilde{V}_\infty, \tilde{R}_{1,\infty}^2, \tilde{R}_{2,\infty}^2) = (V_\infty, R_{1,\infty}^2, R_{2,\infty}^2)$.
\end{lemma}

\subsection{Proofs of the lemmas}

\begin{proof}[Proof of Lemma \ref{lemma:estimate}]
	Let $A$ and $B$ be any finite population quantities that can be the treatment or control potential outcome, or any coordinate of covariates $\bs{W}, \bs{X}, \bs{E}$ or $\bs{C}$. 
	Let $\bar{A} =N^{-1} \sum_{i=1}^{N}A_i$ and $\bar{B} =  N^{-1} \sum_{i=1}^{N}B_i$ be the finite population averages of $A$ and $B$, 
	and 
	$S_{AB} = (N-1)^{-1} \sum_{i=1}^N (A_i-\bar{A}) (B_i - \bar{B})$ be the finite population covariance between $A$ and $B$.
	For $t=0,1$, let $s_{AB}(t)$ be the sample covariance between $A$ and $B$ in the treatment group $t$. 
	
	First, we study the sampling property of $s_{AB}(t)$ under the CRSE. 
	Note that under the CRSE, the units in treatment group $t$ is essentially a simple random sample of size $n_t$ from the  finite population of $N$ units. 
	By the property of simple random sampling, we can know that $s_{AB}(t)$ is unbiased for $S_{AB}$. 
	Moreover, from Lemma \ref{lemma:estimate_lemma}, the variance of $s_{AB}(t)$ under the CRSE is bounded by
	\begin{align}\label{eq:var_SAB_t}
		\Var\{s_{AB}(t)\} & \le \frac{4n_t}{(n_t-1)^2}\cdot \max_{1 \le j\le N}(A_j-\bar A)^2 \cdot  \frac{1}{N-1} \sum_{i=1}^N (B_i - \bar B)^2
		\nonumber
		\\
		& = \frac{4n_t^2}{(n_t-1)^2}\cdot \frac{n}{n_t}\cdot \frac{1}{n}\max_{1 \le j\le N}(A_j-\bar A)^2 \cdot  \frac{1}{N-1} \sum_{i=1}^N (B_i - \bar B)^2
	\end{align}
	which must converge to zero under Condition \ref{cond:fp_analysis}.

	Second, we study the sampling property of $s_{AB}(t)$ under $\sresem$. 
	By the law of total expectation, we can know that under $\sresem$, 
	\begin{align}\label{eq:s_AB_t_resem}
		\E \big[  \big\{s_{AB}(t)-S_{AB}\big\}^2 \mid
		M_T \le a_T, M_S \le a_S  \big.  \big]
		&\le P  \big( M_T \le a_T, M_S \le a_S \big) ^{-1}
		\E\big[ \{s_{AB}(t)-S_{AB}\}^2   \big]   \nonumber
		\\&=  P  \big(M_T \le a_T, M_S \le a_S \big) ^{-1}  \Var\{s_{AB}(t)\} . 
	\end{align}
	From Lemma \ref{lemma:M_ST_CRSE} and \eqref{eq:var_SAB_t}, \eqref{eq:s_AB_t_resem} must converge to zero as $N\rightarrow \infty.$ 
	By the Markov inequality, this implies that $s_{AB}(t) - S_{AB} = o_{\Pr}(1)$ under $\sresem$.
	
	From the above, Lemma \ref{lemma:estimate} holds. 
\end{proof}

\begin{proof}[Proof of Lemma \ref{lemma:diff_average}]
	Define $\hat{\tau}_A = \bar{A}_1 - \bar{A}_0$ and $\hat{\delta}_A = \bar{A}_{\mathcal{S}} - \bar{A}$. 
	Then by definition, we can verify that $\bar{A}_1 - \bar{A}_{\mathcal{S}} = r_0 \hat{\tau}_A$ and $\bar{A}_0 - \bar{A}_{\mathcal{S}} = -r_1 \hat{\tau}_A$. 
	Thus, to prove Lemma \ref{lemma:diff_average}, 
	it suffices to show that both $\hat{\tau}_A$ and $\hat{\delta}_A$ are of order $O_{\Pr}(n^{-1/2})$. 
	By the same logic as Lemma \ref{lemma:clt}, under Condition \ref{cond:fp_analysis} and the CRSE, 
	$
	\sqrt{n} 
	(\hat\tau -\tau , \hat{\bs{\tau}}_{\bs{X}}^\top, \hat{\tau}_{A}, \hat{\bs{\delta}}_{\bs{W}}^\top, \hat{\delta}_{A})^\top 
	\dot\sim
	(H, \bs{B}_T^\top, B_1, \bs{B}_S^\top, B_{2} )^\top, 
	$
	where $(H, \bs{B}_T^\top, B_1, \bs{B}_S^\top, B_{2} )^\top$ follows a multivariate Gaussian distribution with mean zero and covariance matrix
	the sampling covariance of $
	\sqrt{n} 
	(\hat\tau -\tau , \hat{\bs{\tau}}_{\bs{X}}^\top, \hat{\tau}_{A}, \hat{\bs{\delta}}_{\bs{W}}^\top, \hat{\delta}_{A})^\top $ under the CRSE. 
	By Lemma \ref{lemma:s_x2} and Slutsky's theorem, 
	\begin{align*}
		\sqrt{n}
		\begin{pmatrix}
			\hat\tau -\tau \\
			(\bs{S}_{\bs{X}}^2)^{1/2}(\bs{s}^2_{\bs{X}})^{-1/2} \hat{\bs{\tau}}_{\bs{X}}\\
			\hat{\tau}_{A}\\
			\hat{\bs{\delta}}_{\bs{W}}\\ 
			\hat{\delta}_{A}
		\end{pmatrix}
		\ \dot\sim \  
		\begin{pmatrix}
			H\\
			\bs{B}_T\\
			\bs{B}_1\\
			\bs{B}_S\\
			\bs{B}_2
		\end{pmatrix}
	\end{align*}
	From \citet[][Corollary A1]{rerand2018}, we then have 
	\begin{align*}
		\sqrt{n}
		\begin{pmatrix}
			\hat\tau -\tau \\
			(\bs{S}_{\bs{X}}^2)^{1/2}(\bs{s}^2_{\bs{X}})^{-1/2} \hat{\bs{\tau}}_{\bs{X}}\\
			\hat{\tau}_{A}\\
			\hat{\bs{\delta}}_{\bs{W}}\\ 
			\hat{\delta}_{A}^\top
		\end{pmatrix}
		\mid M_T \le a_T, M_S \le a_S, 
		\ \dot\sim \ 
		\begin{pmatrix}
			H\\
			\bs{B}_T\\
			B_1\\
			\bs{B}_S\\
			B_2
		\end{pmatrix}
		\mid \bs{B}_T^\top \bs{V}_{\bs{xx}}^{-1} \bs{B}_T \le a_T, 
		\bs{B}_S^\top \bs{V}_{\bs{ww}}^{-1} \bs{B}_S \le a_S.  
	\end{align*}
	This immediately implies that 
	$\hat{\tau}_{A} = O_{\Pr}(n^{-1/2})$ and $\hat{\delta}_{A} = O_{\Pr}(n^{-1/2})$ under $\sresem$. 
	Therefore, Lemma \ref{lemma:diff_average} holds. 
\end{proof}

\begin{proof}[Proof of Lemma \ref{lemma:sotch_R2}]
	Lemma \ref{lemma:sotch_R2} follows immediately from \citet[][Lemma A10]{rerand2018}. 
\end{proof}

\begin{proof}[Proof of Lemma \ref{lemma:sotch_est}]
	If $V \le \tilde{V} = 0$, Lemma \ref{lemma:sotch_est} holds obviously. Below we consider only the case where $\tilde{V}>0$. 
	Because $\tilde{R}_1^2 \le VR_1^2/\tilde{V}$ and $\tilde{R}_2^2 \le VR_2^2/\tilde{V}$, 
	Lemma \ref{lemma:sotch_R2} implies that, for any $c\ge 0$, 
	\begin{align*}
		& \quad \ \Pr\left(
		\left| 
		\tilde{V}^{1/2}\left(\sqrt{1 - \tilde{R}_1^2 - \tilde{R}_2^2} \ \varepsilon
		+\sqrt{\tilde{R}_1^2}\ L_{K_1, a_1} + \sqrt{\tilde{R}_2^2}\ L_{K_2, a_2}
		\right)
		\right| \le c
		\right)
		\\
		& \le 
		\Pr\left(
		\left| 
		\tilde{V}^{1/2} \left(\sqrt{1 - VR_1^2/\tilde{V} - VR_2^2/\tilde{V}} \ \varepsilon
		+\sqrt{VR_1^2/\tilde{V}}\ L_{K_1, a_1} + \sqrt{VR_2^2/\tilde{V}}\ L_{K_2, a_2}
		\right)
		\right| \le c
		\right)\\
		& = 
		\Pr\left(
		\left| 
		\sqrt{\tilde{V} - VR_1^2 - VR_2^2} \ \varepsilon
		+\sqrt{VR_1^2}\ L_{K_1, a_1} + \sqrt{VR_2^2}\ L_{K_2, a_2}
		\right| \le c
		\right). 
	\end{align*}
	Because $\tilde{V} \ge V$, from Lemma \ref{lemma:linear_com_eps_L}, we then have
	\begin{align*}
		& \quad \ \Pr\left(
		\left| 
		\tilde{V}^{1/2}\left(\sqrt{1 - \tilde{R}_1^2 - \tilde{R}_2^2} \ \varepsilon
		+\sqrt{\tilde{R}_1^2}\ L_{K_1, a_1} + \sqrt{\tilde{R}_2^2}\ L_{K_2, a_2}
		\right)
		\right| \le c
		\right)
		\\
		& \le 
		\Pr\left(
		\left| 
		\sqrt{\tilde{V} - VR_1^2 - VR_2^2} \ \varepsilon
		+\sqrt{VR_1^2}\ L_{K_1, a_1} + \sqrt{VR_2^2}\ L_{K_2, a_2}
		\right| \le c
		\right)
		\\
		& \le 
		\Pr\left(
		\left| 
		\sqrt{V - VR_1^2 - VR_2^2} \ \varepsilon
		+\sqrt{VR_1^2}\ L_{K_1, a_1} + \sqrt{VR_2^2}\ L_{K_2, a_2}
		\right| \le c
		\right)\\
		& = \Pr\left(
		\left| 
		V^{1/2}\left(\sqrt{1 - R_1^2 - R_2^2} \ \varepsilon
		+\sqrt{R_1^2}\ L_{K_1, a_1} + \sqrt{R_2^2}\ L_{K_2, a_2}
		\right)
		\right| \le c
		\right).
	\end{align*}
	Therefore, Lemma \ref{lemma:sotch_est} holds. 
\end{proof}

\begin{proof}[Proof of Lemma \ref{lemma:continuous_cdf}]
	First, we prove that $\Psi_{V, R_1^2, R_2^2}(y)$ is continuous in $(V, R_1^2, R_2^2, y) \in (0, \infty) \times [0, 1] \times [0, 1] \times \mathbb{R}$. 
	Consider any $(V, R_1^2, R_2^2, y) \in (0, \infty) \times [0, 1] \times [0, 1] \times \mathbb{R}$
	and any sequence $\{(V_j, R_{1j}^2, R_{2j}^2, y_j)\}_{j=1}^\infty$
	in $(0, \infty) \times [0, 1] \times [0, 1] \times \mathbb{R}$ that converges to $(V, R_1^2, R_2^2, y)$. 
	We can derive that, as $j \rightarrow \infty$, 
	\begin{align*}
		& \quad \ \ \ \  V_j^{1/2} \left( \sqrt{1-R_{1j}^2-R_{2j}^2} \cdot \varepsilon + \sqrt{R_{1j}^2} \cdot L_{K_1, a_1} + \sqrt{R_{2j}^2} \cdot L_{K_2, a_2} \right) - y_j\\
		& 
		\convergeas 
		V^{1/2} \left( \sqrt{1-R_1^2-R_2^2} \cdot \varepsilon + \sqrt{R_1^2} \cdot L_{K_1, a_1} + \sqrt{R_2^2} \cdot L_{K_2, a_2} \right) - y. 
	\end{align*}
	This immediately implies that as $j\rightarrow\infty$, 
	\begin{align*}
		& \quad \ \ \ \  \Pr\left\{ V_j^{1/2} \left( \sqrt{1-R_{1j}^2-R_{2j}^2} \cdot \varepsilon + \sqrt{R_{1j}^2} \cdot L_{K_1, a_1} + \sqrt{R_{2j}^2} \cdot L_{K_2, a_2} \right) - y_j \le 0 \right\}
		\\
		& \rightarrow 
		\Pr\left( V^{1/2} \left( \sqrt{1-R_1^2-R_2^2} \cdot \varepsilon + \sqrt{R_1^2} \cdot L_{K_1, a_1} + \sqrt{R_2^2} \cdot L_{K_2, a_2} \right) - y \le 0 \right).
	\end{align*}
	Equivalently, as $j\rightarrow\infty$, 
	$
	\Psi_{V_j, R_{1j}^2, R_{2j}^2}(y_j) \rightarrow \Psi_{V, R_1^2, R_2^2}(y). 
	$
	Therefore, $\Psi_{V, R_1^2, R_2^2}(y)$ must be continuous in $(V, R_1^2, R_2^2, y) \in (0, \infty) \times [0, 1] \times [0, 1] \times \mathbb{R}$. 
	
	Second, we prove that $\Psi^{-1}_{V, R_1^2, R_2^2}(y)$ is continuous in $(V, R_1^2, R_2^2, p) \in (0, \infty) \times [0, 1] \times [0, 1] \times (0,1)$. 
	Consider any $(V, R_1^2, R_2^2, p) \in (0, \infty) \times [0, 1] \times [0, 1] \times (0,1)$
	and any sequence $\{(V_j, R_{1j}^2, R_{2j}^2, p_j)\}_{j=1}^\infty$
	in $(0, \infty) \times [0, 1] \times [0, 1] \times (0,1)$ that converges to $(V, R_1^2, R_2^2, p)$. 
	Let 
	$\xi = \Psi^{-1}_{V, R_1^2, R_2^2}(p)$ and $\xi_j = \Psi^{-1}_{V_j, R_{1j}^2, R_{2j}^2}(p_j)$ for all $j\ge 1$. 
	Below we prove that $\xi_j$ converges to $\xi$ by contradiction. 
	When $j$ is sufficiently large, $V_j$ must be positive, and  $\Psi^{-1}_{V_j, R_{1j}^2, R_{2j}^2}(\cdot)$ is the quantile funtion for a continuous distribution, which implies that 
	$\Psi_{V_j, R_{1j}^2, R_{2j}^2}(\xi_j) = p_j$. Similarly, $\Psi_{V, R_1^2, R_2^2}(\xi) = p$.  
	If $\xi_j$ does not converge to $\xi$,  then there must exist a subsequence of $\xi_j$ that converge to a point $\tilde{\xi} \ne \xi$. 
	From the first part of the proof, we can know that along this subsequence, $\Psi_{V_j, R_{1j}^2, R_{2j}^2}(\xi_j)$ must converge to $\Psi_{V, R_{1}^2, R_{2}^2}(\tilde{\xi})$,
	which is different from 
	$\Psi_{V, R_{1}^2, R_{2}^2}(\xi) = p$ since $\Psi_{V, R_{1}^2, R_{2}^2}(\cdot)$ is strictly increasing at $\xi$.
	However, this contradicts with the fact that $\Psi_{V_j, R_{1j}^2, R_{2j}^2}(\xi_j)=p_j$ converges to $p$ as $j\rightarrow \infty$. 
	Therefore, we must have $\xi_j = \Psi^{-1}_{V_j, R_{1j}^2, R_{2j}^2}(p_j)$ converges to $\xi = \Psi^{-1}_{V, R_1^2, R_2^2}(p)$ as $j \rightarrow \infty$. 
	
	From the above, Lemma \ref{lemma:continuous_cdf} holds. 
\end{proof}

\begin{proof}[Proof of Lemma \ref{lemma:conserv_inf}]
	We first consider the variance of $\hat{\psi}_N$. 
	Let $v_1$ and $v_2$ denote the variances of $L_{K_1, a_1}$ and $L_{K_2, a_2}$. 
	Then the variance of $\hat{\psi}_N$ satisfies that 
	\begin{align*}
		\Var(\hat{\psi}_N)
		& = 
		\hat{V}_N\big\{ 
		1 - (1-v_1)\hat{R}_{1,N}^2 - (1-v_2)\hat{R}_{2,N}^2
		\big\}
		= \tilde{V}_N - (1-v_1)\tilde{V}_N\tilde{R}_{1,N}^2 - (1-v_2)\tilde{V}_N\tilde{R}_{2, N}^2 + o_{\Pr}(1)\\
		& 
		\ge 
		V_N - (1-v_1)V_N R_{1,N}^2 - (1-v_2)V_N R_{2, N}^2 + o_{\Pr}(1)
		= 
		\Var_{\text{a}}(\psi_N) + o_{\Pr}(1). 
	\end{align*}
	Thus, as $N\rightarrow \infty$, 
	the probability limit of $\Var(\hat{\psi}_N)$ is greater than or equal to the asymptotic variance of $\psi_N$.

	We then consider the symmetric quantile ranges of the distribution of $\hat{\psi}_N$. 
	For any $\alpha\in (0,1)$, let $\xi_{\alpha,N}$ be the $(1-\alpha/2)$th quantile of the distribution of  $\hat{\psi}_N$. 
	From Lemma \ref{lemma:continuous_cdf} and the conditions in Lemma \ref{lemma:conserv_inf}, 
	by the continuous mapping theorem, $\xi_{\alpha,N}$ must converge in probability to 	
	the $(1-\alpha/2)$th quantile (denoted by $\xi_{\alpha,\infty}$) of the following distribution
		\begin{align*}%
			\tilde{\psi}_{\infty} \sim
			\tilde{V}_{\infty}^{1/2}\left(\sqrt{1 - \tilde{R}_{1,\infty}^2 - \tilde{R}_{2,\infty}^2} \ \varepsilon
			+\sqrt{\tilde{R}_{1,\infty}^2}\ L_{K_1, a_1} + \sqrt{\tilde{R}_{2,\infty}^2}\ L_{K_2, a_2}
			\right). 
		\end{align*}
	Let $\psi_{\infty}$ be a random variable following the asymptotic distribution of $\psi_N$, i.e., 
	\begin{align*}
		\psi_{\infty} \sim 
		V_{\infty}^{1/2}\left(\sqrt{1 - R_{1,\infty}^2 - R_{2,\infty}^2} \ \varepsilon
		+\sqrt{R_{1,\infty}^2}\ L_{K_1, a_1} + \sqrt{R_{2,\infty}^2}\ L_{K_2, a_2}
		\right). 
	\end{align*}
	Because $\tilde{V}_{\infty}>0$, $\xi_{\alpha,\infty}$ must be positive. 
	By Slutsky's theorem, we have 
	$
	\psi_{N}/\xi_{\alpha,N} \converged \psi_{\infty}/\xi_{\alpha,\infty}. 
	$
	This further implies that 
	\begin{align}\label{eq:cov_prob_proof_lemma}
		\Pr(|\psi_N| \le \xi_{\alpha,N})
		\rightarrow 
		\Pr(|\psi_\infty| \le \xi_{\alpha,\infty})
		\ge 
		\Pr(|\tilde{\psi}_\infty| \le \xi_{\alpha,\infty}) = 1-\alpha, 
	\end{align}
	where the last inequality holds due to Lemma \ref{lemma:sotch_est}, 
	and the last equality holds because $\tilde{\psi}_\infty$ is a continuous random variable. 
	
	From the above derivation, 
	we can verify that both the variance and quantile ranges of $\hat{\psi}_N$ become asymptotically exact for $\psi_N$ when $(\tilde{V}_\infty, \tilde{R}_{1,\infty}^2, \tilde{R}_{2,\infty}^2) = (V_\infty, R_{1,\infty}^2, R_{2,\infty}^2)$. 
	Therefore, Lemma \ref{lemma:conserv_inf} holds. 
\end{proof}

\subsection{Estimation of the regression adjustment coefficients}

\begin{proof}[\bf Proof of Theorem \ref{thm:plug_in}]
    From Corollary \ref{cor:equiv_resem_sresem}, 
    it suffices to prove that Theorem \ref{thm:plug_in} holds under $\sresem$.

	Recall that $\bs{s}^2_{\bs{C}}(t)$ and $\bs{s}^2_{\bs{E}}(t)$ are the sample covariance matrices of the $\bs{C}_i$'s and $\bs{E}_i$'s in treatment group $t$, respectively, for $t=0,1$. 
	Then by definition, we have 
	$\hat{\bs{\beta}}_t=\{s_{\bs{C}}^{2}(t)\}^{-1} s_{\bs{C},t}$ and $
	\hat{\bs{\gamma}}_t  =\{s_{\bs{E}}^{2}(t)\}^{-1}s_{\bs{E},t}.
	$
	From Lemmas \ref{lemma:estimate}
	and by the definition of $\tilde{\bs{\beta}}_t$ and $\tilde{\bs{\gamma}}_t$ in 
	Section \ref{sec:R2_proj_coef}, 
	we can derive that under $\sresem$,
	$\hat{\bs{\beta}}_t - \tilde{\bs{\beta}}_t = o_{\Pr}(1)$ 
	and 
	$\hat{\bs{\gamma}}_t - \tilde{\bs{\gamma}}_t = o_{\Pr}(1)$ for $t=0,1$. 
	By definition, 
	these further imply that under $\sresem$,
	$\hat{\bs{\beta}} - \tilde{\bs{\beta}} = o_{\Pr}(1)$ and $\hat{\bs{\gamma}} - \tilde{\bs{\gamma}} = o_{\Pr}(1)$.
	
	We then prove that $\sqrt{n}\{\hat{\tau}(\hat{\bs{\beta}}, \hat{\bs{\gamma}}) - \tau\}$ has the same asymptotic distribution as $\sqrt{n}\{
	\hat{\tau}(\tilde{\bs{\beta}}, \tilde{\bs{\gamma}}) - \tau
	\}$ under $\sresem$. 
	By definition and from Lemma \ref{lemma:diff_average}, 
	\begin{align*}
		\hat{\tau}(\hat{\bs{\beta}}, \hat{\bs{\gamma}})
		-
		\hat{\tau}(\tilde{\bs{\beta}}, \tilde{\bs{\gamma}}) & 
		= 
		(\tilde{\bs{\beta}}-\hat{\bs{\beta}})^\top \hat{\bs{\tau}}_{\bs{C}} 
		+
		(\tilde{\bs{\gamma}}-\hat{\bs{\gamma}})^\top \hat{\bs{\delta}}_{\bs{W}}
		= o_{\Pr}(n^{-1/2}).
	\end{align*}
	Consequently, under $\sresem$,
	$$
	\sqrt{n}\{\hat{\tau}(\hat{\bs{\beta}}, \hat{\bs{\gamma}}) - \tau\} = 
	\sqrt{n}\{
	\hat{\tau}(\tilde{\bs{\beta}}, \tilde{\bs{\gamma}}) - \tau
	\}+\sqrt{n}\{	\hat{\tau}(\hat{\bs{\beta}}, \hat{\bs{\gamma}})
	-
	\hat{\tau}(\tilde{\bs{\beta}}, \tilde{\bs{\gamma}})\} 
	= 
	\sqrt{n}\{
	\hat{\tau}(\tilde{\bs{\beta}}, \tilde{\bs{\gamma}}) - \tau
	\} + o_{\Pr}(1). 
	$$
	By Slutsky's theorem, 
	$\sqrt{n}\{\hat{\tau}(\hat{\bs{\beta}}, \hat{\bs{\gamma}}) - \tau\}$ and  $\sqrt{n}\{
	\hat{\tau}(\tilde{\bs{\beta}}, \tilde{\bs{\gamma}}) - \tau
	\}$ have the same asymptotic distribution under $\sresem$. 
	
	From the above, Theorem \ref{thm:plug_in} holds. 
\end{proof}

\subsection{Inference based on regression-adjusted estimator with fixed $(\bs{\beta}, \bs{\gamma})$}

\begin{proposition}\label{prop:est_reg_resem_fix_coef}
	Under ReSEM and Condition \ref{cond:fp_analysis}, 
	as $N\rightarrow \infty$, 
	if the analyzer knows all the design information, then 
	the estimators in \eqref{eq:VR2_estimate} satisfy 
	\begin{align*}
		\hat V_{\tau \tau}(\bs{\beta}, \bs{\gamma})
		& = 
		r_1^{-1}S^2_{1}(\bs{\beta}, \bs{\gamma}) + r_0^{-1}S^2_{0}(\bs{\beta}, \bs{\gamma})
		- 
		f \bs{S}^2_{\tau \mid \bs{C}}(\bs{\beta}, \bs{\gamma})
		+ o_{\Pr}(1), 
		\\
		\hat V_{\tau \tau}(\bs{\beta}, \bs{\gamma}) \hat{R}^2_S(\bs{\beta}, \bs{\gamma})
		& = 
		(1-f) S^2_{\tau\mid \bs{W}} (\bs{\beta}, \bs{\gamma}) + o_{\Pr}(1), 
		\\
		\hat V_{\tau \tau}(\bs{\beta}, \bs{\gamma}) \hat{R}^2_T(\bs{\beta}, \bs{\gamma})
		& = r_1^{-1} S^2_{1\mid  \bs{X}}(\bs{\beta}, \bs{\gamma}) +r_0^{-1} S^2_{0\mid  \bs{X}}(\bs{\beta}, \bs{\gamma}) - S^2_{\tau\mid \bs{X}}(\bs{\beta}, \bs{\gamma}) + o_{\Pr}(1). 
	\end{align*}
\end{proposition}

\begin{proof}[\bf Proof of Proposition \ref{prop:est_reg_resem_fix_coef}]
    From Theorem \ref{thm:equiv_resem_sresem}, it suffices to prove that Proposition \ref{prop:est_reg_resem_fix_coef} holds under $\sresem$. 
	By definition and from Lemma \ref{lemma:estimate},
	under $\sresem$, 
	for $t=0, 1$, 
	\begin{align*}
		s_t^2(\bs{\beta}, \bs{\gamma})
		& = 
		s_t^2 + \bs{\beta}^\top s_{\bs{C}}^2(t) \bs{\beta} 
		+ 
		r_t^2 \bs{\gamma}^\top s_{\bs{E}}^2(t) \bs{\gamma}
		- 2 \bs{\beta}^\top \bs{s}_{\bs{C},t} - 2 (-1)^{t-1} r_t \bs{\gamma}^\top \bs{s}_{\bs{E},t} 
		+ 2 (-1)^{t-1} r_t \bs{\beta}^\top \bs{s}_{\bs{C}, \bs{E}}(t) \bs{\gamma}
		\\
		& = 
		S_t^2 + \bs{\beta}^\top S_{\bs{C}}^2 \bs{\beta} 
		+ 
		r_t^2 \bs{\gamma}^\top S_{\bs{E}}^2 \bs{\gamma}
		- 2 \bs{\beta}^\top \bs{S}_{\bs{C},t} - 2 (-1)^{t-1} r_t \bs{\gamma}^\top \bs{S}_{\bs{E},t} 
		+ 2 (-1)^{t-1} r_t \bs{\beta}^\top \bs{S}_{\bs{C}, \bs{E}} \bs{\gamma} + o_{\Pr}(1)
		\\
		& = 
		S_t^2(\bs{\beta}, \bs{\gamma}) + o_{\Pr}(1), 
	\end{align*}
	and 
	\begin{align*}
		\bs{s}_{t, \bs{C}}(\bs{\beta}, \bs{\gamma}) \bs{s}_{\bs{C}}^{-1}(t)
		& = 
		\big\{ \bs{s}_{t, \bs{C}} - \bs{\beta}^\top \bs{s}^2_{\bs{C}}(t) - (-1)^{t-1} r_t \bs{\gamma}^\top \bs{s}_{\bs{E}, \bs{C}}(t) \big\} \cdot 
		\big\{ \bs{s}_{\bs{C}}^{2}(t) \big\}^{-1/2}
		\\
		& = 
		\big\{ \bs{S}_{t, \bs{C}} - \bs{\beta}^\top \bs{S}^2_{\bs{C}} - (-1)^{t-1} r_t \bs{\gamma}^\top \bs{S}_{\bs{E}, \bs{C}} \big\} \cdot 
		\big( \bs{S}_{\bs{C}}^{2} \big)^{-1/2} + o_{\Pr}(1)\\
		& = \bs{S}_{t, \bs{C}}(\bs{\beta}, \bs{\gamma}) \big( \bs{S}_{\bs{C}}^{2} \big)^{-1/2} + o_{\Pr}(1),  
	\end{align*}
	which further implies that 
	\begin{align*}
		s_{\tau \mid \bs{C}}^2 (\bs{\beta}, \bs{\gamma})  & = \big\| \bs{s}_{1,\bs{C}}(\bs{\beta}, \bs{\gamma}) \cdot \bs{s}^{-1}_{\bs{C}}(1) - \bs{s}_{0,\bs{C}}(\bs{\beta}, \bs{\gamma}) \cdot \bs{s}^{-1}_{\bs{C}}(0) \big\|_2^2
		\\
		& = 
		\big\| \bs{S}_{1, \bs{C}}(\bs{\beta}, \bs{\gamma}) \big( \bs{S}_{\bs{C}}^{2} \big)^{-1/2} - \bs{S}_{0, \bs{C}}(\bs{\beta}, \bs{\gamma}) \big( \bs{S}_{\bs{C}}^{2} \big)^{-1/2} \big\|_2^2 + o_{\Pr}(1)\\
		& 
		= \bs{S}_{\tau\mid \bs{C}}^2 (\bs{\beta}, \bs{\gamma}) + o_{\Pr}(1).
	\end{align*}
	By the same logic, 
	$
	\bs{s}_{\tau\mid \bs{W}}^2 (\bs{\beta}, \bs{\gamma}) = \bs{S}_{\tau\mid \bs{W}}^2 (\bs{\beta}, \bs{\gamma}) + o_{\Pr}(1)$ 
	and 
	$\bs{s}_{\tau\mid \bs{X}}^2 (\bs{\beta}, \bs{\gamma}) = \bs{S}_{\tau\mid \bs{X}}^2 (\bs{\beta}, \bs{\gamma}) + o_{\Pr}(1), 
	$
	and for $t=0,1$, 
	\begin{align*}
		s^2_{t\mid \bs{X}}(\bs{\beta}, \bs{\gamma})
		& = 
		\bs{s}_{t,\bs{X}}(\bs{\beta}, \bs{\gamma}) \cdot \big\{ \bs{s}^2_{\bs{X}}(t) \big\}^{-1} \bs{s}_{\bs{X}, t}(\bs{\beta}, \bs{\gamma})
		= \bs{S}_{t,\bs{X}}(\bs{\beta}, \bs{\gamma}) \cdot \big( \bs{S}^2_{\bs{X}} \big)^{-1} \bs{S}_{\bs{X}, t}(\bs{\beta}, \bs{\gamma}) + o_{\Pr}(1)
		\\
		& = \bs{S}_{t\mid \bs{X}} (\bs{\beta}, \bs{\gamma}) + o_{\Pr}(1). 
	\end{align*}
	From the above and by definition, we can immediately derive Proposition \ref{prop:est_reg_resem_fix_coef}. 
\end{proof}

\begin{proposition}\label{prop:est_fix_coef_unknown_design}
	For any $(\bs{\beta}, \bs{\gamma})$,  define 
	\begin{align*}
		\tilde{V}_{\tau \tau}(\bs{\beta}, \bs{\gamma})
		& = 
		r_1^{-1}S^2_{1}(\bs{\beta}, \bs{\gamma}) + r_0^{-1}S^2_{0}(\bs{\beta}, \bs{\gamma})
		- 
		f \bs{S}^2_{\tau \mid \bs{C}}(\bs{\beta}, \bs{\gamma})
		= 
		V_{\tau \tau}(\bs{\beta}, \bs{\gamma}) + f S^2_{\tau \setminus \bs{C}}(\bs{\beta}, \bs{\gamma}),
	\\
		\tilde{R}^2_S(\bs{\beta}, \bs{\gamma})
		& = 
		\begin{cases}
			\frac{(1-f) S^2_{\tau\mid \bs{W}} (\bs{\beta}, \bs{\gamma})}{\tilde{V}_{\tau \tau}(\bs{\beta}, \bs{\gamma})}
			= 
			\frac{V_{\tau \tau}(\bs{\beta}, \bs{\gamma}) R^2_S(\bs{\beta}, \bs{\gamma})}{\tilde{V}_{\tau \tau}(\bs{\beta}, \bs{\gamma})}, & \text{if both $\bs{W}$ and $a_S$ are known}, 
			\\
			0, & \text{otherwise}, 
		\end{cases}
	\\
		\tilde{R}^2_T(\bs{\beta}, \bs{\gamma})
		& = 
		\begin{cases}
			\frac{r_1^{-1} S^2_{1\mid  \bs{X}}(\bs{\beta}, \bs{\gamma}) +r_0^{-1} S^2_{0\mid  \bs{X}}(\bs{\beta}, \bs{\gamma}) - S^2_{\tau\mid \bs{X}}(\bs{\beta}, \bs{\gamma})}{\tilde{V}_{\tau \tau}(\bs{\beta}, \bs{\gamma})} 
			= \frac{V_{\tau \tau}(\bs{\beta}, \bs{\gamma}) R^2_T(\bs{\beta}, \bs{\gamma})}{\tilde{V}_{\tau \tau}(\bs{\beta}, \bs{\gamma})}, 
			& \text{if both $\bs{X}$ and $a_T$ are known}, \\
			0, & \text{otherwise}. 
		\end{cases}
	\end{align*}
	\begin{itemize}
		\item[(i)] $\tilde{V}_{\tau \tau}(\bs{\beta}, \bs{\gamma})$ has the equivalent form: 
		$\tilde{V}_{\tau \tau}(\bs{\beta}, \bs{\gamma}) = V_{\tau \tau}(\bs{\beta}, \bs{\gamma}) + f S^2_{\tau \setminus \bs{C}}$; 
		\item[(ii)] Under ReSEM and Condition \ref{cond:fp_analysis}, 
		\begin{equation*}%
			\hat V_{\tau \tau}(\bs{\beta}, \bs{\gamma}) = \tilde{V}_{\tau \tau}(\bs{\beta}, \bs{\gamma}) + o_{\Pr}(1), 
			\ \ 
			\hat{R}^2_S(\bs{\beta}, \bs{\gamma}) = \tilde{R}^2_S(\bs{\beta}, \bs{\gamma}) + o_{\Pr}(1), 
			\ \  
			\hat{R}^2_T(\bs{\beta}, \bs{\gamma}) = \tilde{R}^2_T(\bs{\beta}, \bs{\gamma}) + o_{\Pr}(1),
		\end{equation*}
		\item[(iii)] If further 
		Condition \ref{cond:fp_est} holds, 
        then $\tilde{V}_{\tau \tau}(\bs{\beta}, \bs{\gamma})$ must have a positive limit. 
	\end{itemize}
\end{proposition}

\begin{proof}[\bf Proof of Proposition \ref{prop:est_fix_coef_unknown_design}]
	We first prove the equivalent form of $\tilde{V}_{\tau \tau}(\bs{\beta}, \bs{\gamma})$. 
	By definition, the adjusted individual effect for each unit $i$ has the following equivalent forms: 
	\begin{align*}
		\tau_i(\bs{\beta}, \bs{\gamma}) & = Y_i(1; \bs{\beta}, \bs{\gamma}) - Y_i(0; \bs{\beta}, \bs{\gamma}) = 
		\{ Y_i(1) - \bs{\beta}^\top \bs{C}_i - r_1 \bs{\gamma}^\top (\bs{E}_i - \bar{\bs{E}}) \} - 
		\{ Y_i(0) - \bs{\beta}^\top \bs{C}_i + r_0 \bs{\gamma}^\top (\bs{E}_i - \bar{\bs{E}}) \}\\
		& = \tau_i - \bs{\gamma}^\top (\bs{E}_i - \bar{\bs{E}}). 
	\end{align*}
	Because $\bs{E} \subset \bs{C}$, the residual from the linear projection of the adjusted individual effect $\tau_i(\bs{\beta}, \bs{\gamma})$ on $\bs{C}_i$ must be the same as that from the linear projection of the original individual effect $\tau_i$ on $\bs{C}_i$. 
	Consequently, $S^2_{\tau \setminus \bs{C}}(\bs{\beta}, \bs{\gamma})  = S^2_{\tau \setminus \bs{C}}$, and thus $\tilde{V}_{\tau \tau}(\bs{\beta}, \bs{\gamma})$ has the following equivalent forms: 
	\begin{align*}
		\tilde{V}_{\tau \tau}(\bs{\beta}, \bs{\gamma})
		& 
		= 
		V_{\tau \tau}(\bs{\beta}, \bs{\gamma}) + f S^2_{\tau \setminus \bs{C}}(\bs{\beta}, \bs{\gamma})
		= 
		V_{\tau \tau}(\bs{\beta}, \bs{\gamma}) + f S^2_{\tau \setminus \bs{C}}. 
	\end{align*}
	
	We then prove (ii) in Proposition \ref{prop:est_fix_coef_unknown_design}. Indeed, it follows immediately from the proof of Proposition \ref{prop:est_reg_resem_fix_coef}  and the construction of our estimators when there lacks some design information. 
	
	Finally, we consider the limit of $\tilde{V}_{\tau \tau}(\bs{\beta}, \bs{\gamma})$. 
	From \eqref{eq:simp_form_V_reg} in the proof of Corollary \ref{cor:dist_reg_general_equ} and Lemmas \ref{lemma:simp_W_in_E} and \ref{lemma:simp_X_in_C}, 
	$V_{\tau \tau}(\bs{\beta}, \bs{\gamma})$ has the following equivalent forms:
	\begin{align*}
		V_{\tau \tau}(\bs{\beta}, \bs{\gamma})
		& = V_{\tau\tau}(\bs{\beta}, \bs{0}) + V_{\tau\tau}(\bs{0}, \bs{\gamma}) - V_{\tau\tau}\\
		& = V_{\tau\tau} (1-R_E^2 - R_C^2 )+  (r_1 r_0)^{-1} (\bs{\beta} - \tilde{\bs{\beta}})^\top \bs{S}^2_{\bs{C}} (\bs{\beta} - \tilde{\bs{\beta}}) +  (1-f) (\bs{\gamma} - \tilde{\bs{\gamma}})^\top \bs{S}^2_{\bs{E}} (\bs{\gamma} - \tilde{\bs{\gamma}})\\
		& \ge V_{\tau\tau} (1-R_E^2 - R_C^2 ). 
	\end{align*}
	By the definitons of $R_E^2$ and $R_C^2$ in \eqref{eq:R_A2}, the lower bound of $V_{\tau \tau}(\bs{\beta}, \bs{\gamma})$ has the following equivalent forms:
	\begin{align*}
		V_{\tau\tau} (1 - R_E^2 -R_C^2)
		& = 
		r_1^{-1}S^2_1 +r_0^{-1}S^2_0 -fS^2_\tau - (1-f)S^2_{\tau \mid  \bs{E}} - r_1^{-1}S^2_{1 \mid  \bs{C}}-r_0^{-1}S^2_{0 \mid  \bs{C}} +S^2_{\tau \mid  \bs{C}}
		\\
		& = 
		r_1^{-1}S^2_{1\setminus \bs{C}} +r_0^{-1}S^2_{0\setminus \bs{C}} -f(S^2_{\tau\mid \bs{C}} + S^2_{\tau\setminus \bs{C}}) - (1-f)S^2_{\tau \mid  \bs{E}} +S^2_{\tau \mid  \bs{C}}
		\\
		& = r_1^{-1}S^2_{1\setminus \bs{C}} +r_0^{-1}S^2_{0\setminus \bs{C}} -fS^2_{\tau\setminus\bs{C}} + (1-f) 
		\big( S^2_{\tau \mid  \bs{C}} - S^2_{\tau \mid  \bs{E}}\big)
		\\
		& \ge r_1^{-1}S^2_{1\setminus \bs{C}} +r_0^{-1}S^2_{0\setminus \bs{C}} -fS^2_{\tau\setminus\bs{C}}, 
	\end{align*}
	where the last inequality holds because $\bs{E} \subset \bs{C}$. 
	These imply that $\tilde{V}_{\tau \tau}(\bs{\beta}, \bs{\gamma})$ can be bounded by 
	\begin{align*}
		\tilde{V}_{\tau \tau}(\bs{\beta}, \bs{\gamma})
		& = V_{\tau \tau}(\bs{\beta}, \bs{\gamma}) + f S^2_{\tau \setminus \bs{C}}
		\ge r_1^{-1}S^2_{1\setminus \bs{C}} +r_0^{-1}S^2_{0\setminus \bs{C}}, 
	\end{align*}
	 which must have a positive limit as $N\rightarrow \infty$ under Condition \ref{cond:fp_est}. 
	 Therefore, $\tilde{V}_{\tau \tau}(\bs{\beta}, \bs{\gamma})$ must have a positive limit. 
	 
	 From the above, Proposition \ref{prop:est_fix_coef_unknown_design} holds. 
\end{proof}

\begin{proof}[\bf Proof for conservative inference based on adjusted estimator $\hat{\tau}(\bs{\beta}, \bs{\gamma})$]
	By definition and Proposition \ref{prop:est_fix_coef_unknown_design}(i), 
	\begin{align*}
		\tilde{V}_{\tau \tau}(\bs{\beta}, \bs{\gamma})
		& = 
		V_{\tau \tau}(\bs{\beta}, \bs{\gamma}) + f S^2_{\tau \setminus \bs{C}} \ge V_{\tau \tau}(\bs{\beta}, \bs{\gamma}), 
		\\
		\tilde{V}_{\tau \tau}(\bs{\beta}, \bs{\gamma}) \tilde{R}^2_S(\bs{\beta}, \bs{\gamma})
		& = 
		\begin{cases}
			V_{\tau \tau}(\bs{\beta}, \bs{\gamma}) R^2_S(\bs{\beta}, \bs{\gamma}), & \text{if both $\bs{W}$ and $a_S$ are known}
			\\
			0, & \text{otherwise}
		\end{cases}
		\le V_{\tau \tau}(\bs{\beta}, \bs{\gamma}) R^2_S(\bs{\beta}, \bs{\gamma}), 
		\\
		\tilde{V}_{\tau \tau}(\bs{\beta}, \bs{\gamma}) \tilde{R}^2_T(\bs{\beta}, \bs{\gamma}) 
		& = 
		\begin{cases}
			V_{\tau \tau}(\bs{\beta}, \bs{\gamma}) R^2_T(\bs{\beta}, \bs{\gamma}), 
			& \text{if both $\bs{X}$ and $a_T$ are known} \\
			0, & \text{otherwise}
		\end{cases}
		\le V_{\tau \tau}(\bs{\beta}, \bs{\gamma}) R^2_T(\bs{\beta}, \bs{\gamma}). 
	\end{align*}
	From Proposition \ref{prop:est_fix_coef_unknown_design} and Lemma \ref{lemma:conserv_inf}, we can immediately know that both the variance estimator and confidence intervals based on the estimated distribution are asymptotically conservative. 
	Furthermore, when the design information is known, 
	both the variance estimator and confidence intervals beocme asymptotically exact when $\tilde{V}_{\tau \tau}(\bs{\beta}, \bs{\gamma}) - V_{\tau \tau}(\bs{\beta}, \bs{\gamma}) = o(1)$ or equivalently $f S^2_{\tau \setminus \bs{C}} =o(1)$, which holds when the proportion of sampled units $f = o(1)$ or the adjusted treatment effects are asymptotically additive in the sense that $S^2_{\tau \setminus \bs{C}} = o(1)$. 
\end{proof}

\begin{rmk*}
		Here we give an additional remark on our estimators in Section \ref{sec:estimate_and_CI}. 
		Specifically, we can replace the terms  $\bs{s}^2_{\bs{X}}(1)$ and $\bs{s}^2_{\bs{X}}(0)$ by $\bs{s}^2_{\bs{X}}$ (and analogously for those terms involving covariates $\bs{W}$ and $\bs{C}$), 
		and the resulting variance estimator and confidence intervals will still be asymptotically conservative. The reason is that all the three terms, $\bs{s}^2_{\bs{X}}(1)$,  $\bs{s}^2_{\bs{X}}(0)$ and $\bs{s}^2_{\bs{X}}$ are equal to $\bs{S}^2_{\bs{X}}+o_{\Pr}(1)$, as implied by Lemma \ref{lemma:estimate} and later Lemma  \ref{lemma:cov_sampled_units_resem}. 
		However, 
		for $t=0,1$, 
		$\bs{s}^2_{\bs{X}}(t)$ may be preferred since it is correlated with $\bs{s}_{t, \bs{X}}(\bs{\beta}, \bs{\gamma})$, and 
		the estimator $\bs{s}_{t, \bs{X}}(\bs{\beta}, \bs{\gamma})\bs{s}^{-1}_{\bs{X}}(t)$, compared to $\bs{s}_{t, \bs{X}}(\bs{\beta}, \bs{\gamma})\bs{s}^{-1}_{\bs{X}}$, 
		can lead to gains in precision. 
		This is related to ratio estimators in survey sampling \citep{cochran1977}. See also \citet{decompose2019} for 
        related discussion. 
\end{rmk*}

\subsection{Inference based on regression-adjusted estimator with estimated $(\hat{\bs{\beta}}, \hat{\bs{\gamma}})$}

\begin{proposition}\label{prop:est_reg_resem_est_coef}
	Under ReSEM and Condition \ref{cond:fp_analysis}, 
	as $N\rightarrow \infty$, 
	if the analyzer knows all the design information, then 
	the estimators in \eqref{eq:VR2_estimate} with estimated coefficients $(\hat{\bs{\beta}}, \hat{\bs{\gamma}})$ satisfy 
	\begin{align*}
		\hat V_{\tau \tau}(\hat{\bs{\beta}}, \hat{\bs{\gamma}})
		& = 
		r_1^{-1}S^2_{1}(\tilde{\bs{\beta}}, \tilde{\bs{\gamma}}) + r_0^{-1}S^2_{0}(\tilde{\bs{\beta}}, \tilde{\bs{\gamma}})
		- 
		f \bs{S}^2_{\tau \mid \bs{C}}(\tilde{\bs{\beta}}, \tilde{\bs{\gamma}})
		+ o_{\Pr}(1), 
		\\
		\hat V_{\tau \tau}(\hat{\bs{\beta}}, \hat{\bs{\gamma}}) \hat{R}^2_S(\hat{\bs{\beta}}, \hat{\bs{\gamma}})
		& = 
		(1-f) S^2_{\tau\mid \bs{W}} (\tilde{\bs{\beta}}, \tilde{\bs{\gamma}}) + o_{\Pr}(1), 
		\\
		\hat V_{\tau \tau}(\hat{\bs{\beta}}, \hat{\bs{\gamma}}) \hat{R}^2_T(\hat{\bs{\beta}}, \hat{\bs{\gamma}})
		& = r_1^{-1} S^2_{1\mid  \bs{X}}(\tilde{\bs{\beta}}, \tilde{\bs{\gamma}}) +r_0^{-1} S^2_{0\mid  \bs{X}}(\tilde{\bs{\beta}}, \tilde{\bs{\gamma}}) - S^2_{\tau\mid \bs{X}}(\tilde{\bs{\beta}}, \tilde{\bs{\gamma}}) + o_{\Pr}(1). 
	\end{align*}
\end{proposition}

\begin{proof}[\bf Proof of Proposition \ref{prop:est_reg_resem_est_coef}]
    From Theorem \ref{thm:equiv_resem_sresem}, it suffices to prove that Proposition \ref{prop:est_reg_resem_est_coef} holds under $\sresem$. 
    This follows by the same logic as the proof of Proposition \ref{prop:est_reg_resem_fix_coef}, noting that $\hat{\bs{\beta}} = \tilde{\bs{\beta}} + o_{\Pr}(1)$ and $\hat{\bs{\gamma}} = \tilde{\bs{\gamma}} + o_{\Pr}(1)$ under $\sresem$ from the proof of Theorem \ref{thm:plug_in}. 
\end{proof}

\begin{proposition}\label{prop:est_est_coef_unknown_design}
	Define $\tilde{V}_{\tau \tau}(\bs{\beta}, \bs{\gamma})$, $\tilde{R}^2_S(\bs{\beta}, \bs{\gamma})$ and $\tilde{R}^2_T(\bs{\beta}, \bs{\gamma})$ the same as in Proposition \ref{prop:est_fix_coef_unknown_design}. 
	Under Under ReSEM and Conditions \ref{cond:fp_est}, 
	$\tilde{V}_{\tau \tau}(\tilde{\bs{\beta}}, \tilde{\bs{\gamma}})$ has a positive limit, and 
	\begin{align*}
		\hat V_{\tau \tau}(\hat{\bs{\beta}}, \hat{\bs{\gamma}}) = \tilde{V}_{\tau \tau}(\tilde{\bs{\beta}}, \tilde{\bs{\gamma}}) + o_{\Pr}(1), 
		\ \ 
		\hat{R}^2_S(\hat{\bs{\beta}}, \hat{\bs{\gamma}}) = \tilde{R}^2_S(\tilde{\bs{\beta}}, \tilde{\bs{\gamma}}) + o_{\Pr}(1), 
		\ \ 
		\hat{R}^2_T(\hat{\bs{\beta}}, \hat{\bs{\gamma}})  = \tilde{R}^2_T(\tilde{\bs{\beta}}, \tilde{\bs{\gamma}}) + o_{\Pr}(1). 
	\end{align*}
\end{proposition}

\begin{proof}[\bf Proof of Proposition \ref{prop:est_est_coef_unknown_design}]
	Proposition \ref{prop:est_est_coef_unknown_design} follows immediately from 
	Proposition \ref{prop:est_fix_coef_unknown_design}(i) and (iii)
	and 
	Proposition \ref{prop:est_reg_resem_est_coef}. 
\end{proof}

\begin{proof}[\bf Proof for conservative inference based on adjusted estimator $\hat{\tau}(\hat{\bs{\beta}}, \hat{\bs{\gamma}})$]
	The proof is almost the same as that for conservative inference based on adjusted estimator $\hat{\tau}(\bs{\beta}, \bs{\gamma})$ with $(\bs{\beta}, \bs{\gamma}) = (\tilde{\bs{\beta}}, \tilde{\bs{\gamma}})$. 
	Thus, we omit the proof here for conciseness. 
\end{proof}

\section{$\mathcal{C}$-optimal adjustment and connection to regression models}\label{sec:C_optmal_connection_proof}

\subsection{Technical lemmas}

\begin{lemma}\label{lemma:least_square_general}
For any outcome $y$ and covariate vector $\bs{x}$, 
the least squares solution 
\begin{align*}
    (\hat{a}, \hat{\bs{b}}, \hat{\bs{c}}, \hat{\theta}) & = \argmin_{a, \bs{b}, \bs{c}, \bs{\theta}}\sum_{i \in \mathcal{S}}
    \big\{
    y_i - a - \theta T_i - \bs{b}^\top \bs{x}_i
    - 
    \bs{c}^\top T_i \times  \bs{x}_i
    \big\}^2
\end{align*}    
has the following form: 
$\hat{a} = \bar{y}_0 - \hat{\bs{b}}_0^\top \bar{\bs{x}}_0$, 
$\hat{\bs{b}} = \hat{\bs{b}}_0$, 
$\hat{\bs{c}} = \hat{\bs{b}}_1 - \hat{\bs{b}}_0$, 
and 
$
    \hat{\theta} = 
    (\bar{y}_1 - \hat{\bs{b}}_1^\top \bar{\bs{x}}_1)
    -
    (\bar{y}_0 - \hat{\bs{b}}_0^\top \bar{\bs{x}}_0), 
$
where $\bar{y}_t$ and $\bar{\bs{x}}_t$ denote the average outcome and covariate vector for units in $\{i \in \mathcal{S}: T_i=t\}$,
and $\hat{\bs{b}}_t$ denotes the least squares coefficient of the outcome on covariates for units in $\{i \in \mathcal{S}: T_i=t\}$. 
\end{lemma}

\begin{lemma}\label{lemma:cov_sampled_units_resem}
    Let $A$ and $B$ be any finite population quantities that can be the treatment or control potential outcome, or any coordinate of covariates $\bs{W}, \bs{X}, \bs{E}$ or $\bs{C}$. 
	Under Condition \ref{cond:fp_analysis} and $\sresem$, 
    the sample covariance between $A$ and $B$ for sampled units is consistent for the finite population covariance between $A$ and $B$. 
\end{lemma}

\begin{lemma}\label{lemma:samp_cov_C_res}
Under $\sresem$, 
let $\bs{C}^{\res}_i$ be the fitted residual from the linear regression of $\bs{C}_i$ on $\bs{E}_i$ for sampled units $i\in \mathcal{S}$. 
If Condition \ref{cond:fp_analysis} holds, then among units under treatment arm $t$, 
the sample covariance between covariates $\bs{C}^{\res}$ and $\bs{E}$ and observed outcome $Y$ has the following probability limit: 
\begin{align*}
    \begin{pmatrix}
    \bs{s}_{\bs{C}^\res, t}\\
    \bs{s}_{\bs{E}, t}
    \end{pmatrix}
    & = 
    \begin{pmatrix}
    \bs{S}_{\bs{C}, t} 
    - \bs{S}_{\bs{C}, \bs{E}} ( \bs{S}^{2}_{\bs{E}} )^{-1} \bs{S}_{\bs{E}, t}
    \\
    \bs{S}_{\bs{E}, t}
    \end{pmatrix}
    + o_{\Pr}(1), 
\end{align*}
and the sample covariance matrix of the covariates $\bs{C}^{\res}_i$ and $\bs{E}$ has the following probability limit:
\begin{align*}
    \begin{pmatrix}
    \bs{s}_{\bs{C}^\res}^2(t) & \bs{s}_{\bs{C}^\res, \bs{E}}(t)\\
    \bs{s}_{\bs{E}, \bs{C}^\res}(t) & \bs{s}_{\bs{E}}^2(t)
    \end{pmatrix}
    & = 
    \begin{pmatrix}
    \bs{S}_{\bs{C} \setminus \bs{E}}^2 & \bs{0}\\
    \bs{0}  & \bs{S}_{\bs{E}}^2
    \end{pmatrix}
    + o_{\Pr}(1). 
\end{align*}
\end{lemma}

\subsection{Proofs of the lemmas}
\begin{proof}[Proof of Lemma \ref{lemma:least_square_general}]
    Lemma \ref{lemma:least_square_general} follows immediately from some algebra. 
\end{proof}

\begin{proof}[Proof of Lemma \ref{lemma:cov_sampled_units_resem}]
    Let $A$ and $B$ be any finite population quantities that can be the treatment or control potential outcome, or any coordinate of covariates $\bs{W}, \bs{X}, \bs{E}$ or $\bs{C}$. 
    Let $\bar{A} =N^{-1} \sum_{i=1}^{N}A_i$ and $\bar{B} =  N^{-1} \sum_{i=1}^{N}B_i$ be the finite population averages of $A$ and $B$, 
    and 
    $S_{AB} = (N-1)^{-1} \sum_{i=1}^N (A_i-\bar{A}) (B_i - \bar{B})$ be the finite population covariance between $A$ and $B$.
    Let $s_{AB}$ be the sample covariance between $A$ and $B$ for sampled units.

    We first consider the property of $s_{AB}$ under the CRSE. 
    By the property of simple random sampling and from Lemma  \ref{lemma:estimate_lemma}, $s_{AB}$ is unbiased for $S_{AB}$ and its variance can be bounded by 
    \begin{align*}
    \Var (s_{AB}) \le \frac{4n^2}{(n-1)^2}\cdot \frac{1}{n}\max_{1 \le j\le N}(A_j-\bar A)^2 \cdot  \frac{1}{N-1} \sum_{i=1}^N (B_i - \bar B)^2, 
	\end{align*}
	which must converge to zero as $N\rightarrow \infty$ under Condition \ref{cond:fp_analysis}. 
	
	We then consider the property of $s_{AB}$ under $\sresem$. 
	By the law of total expectation, 
	\begin{align*}
		\E \left\{  \left( s_{AB}-S_{AB} \right)^2 \mid
		M_T \le a_T, M_S \le a_S \right\}
		&\le 
		\frac{\E\left\{  \left( s_{AB}-S_{AB} \right)^2   \right\} }{\Pr( M_T \le a_T, M_S \le a_S )}
		=  
		\frac{\Var(s_{AB})}{\Pr\big( M_T \le a_T, M_S \le a_S \big)}.
	\end{align*}
	From the discussion before and Lemma \ref{lemma:M_ST_CRSE}, we must have 
	$\E\{  \left( s_{AB}-S_{AB} \right)^2 \mid
	M_T \le a_T, M_S \le a_S \}= o(1)$. 
	By the Markov inequality, we then have $s_{AB}-S_{AB}  = o_{\Pr}(1)$ under $\sresem$. 
	
	From the above, Lemma \ref{lemma:cov_sampled_units_resem} holds. 
\end{proof}

\begin{proof}[Proof of Lemma \ref{lemma:samp_cov_C_res}]
By definition, 
$
\bs{C}_i^\res 
= 
\bs{C}_i - \bar{\bs{C}}_{\mathcal{S}}
- \bs{s}_{\bs{C}, \bs{E}} ( \bs{s}^{2}_{\bs{E}} )^{-1} ( \bs{E}_i - \bar{\bs{E}}_{\mathcal{S}} ). 
$
From Lemmas \ref{lemma:estimate} and \ref{lemma:cov_sampled_units_resem}, 
under $\sresem$, 
among units under treatment arm $t$, 
the sample covariances between covariates and outcome satisfy that 
$\bs{s}_{\bs{E}, t} = \bs{S}_{\bs{E}, t} + o_{\Pr}(1)$ and 
\begin{equation*}
    \bs{s}_{\bs{C}^\res, t}
    = \bs{s}_{\bs{C}, t} 
    - \bs{s}_{\bs{C}, \bs{E}} ( \bs{s}^{2}_{\bs{E}} )^{-1} \bs{s}_{\bs{E}, t}
    = 
    \bs{S}_{\bs{C}, t} 
    - \bs{S}_{\bs{C}, \bs{E}} ( \bs{S}^{2}_{\bs{E}} )^{-1} \bs{S}_{\bs{E}, t} + o_{\Pr}(1), 
\end{equation*}
and the sample covariances for covariates satisfy that 
$\bs{s}_{\bs{E}}^2(t) = \bs{S}_{\bs{E}}^2 + o_{\Pr}(1)$, 
\begin{align*}
    \bs{s}_{\bs{C}^\res}^2(t)
    & = 
    \bs{s}_{\bs{C}}^2(t) + 
    \bs{s}_{\bs{C}, \bs{E}} ( \bs{s}^{2}_{\bs{E}} )^{-1} \bs{s}^2_{\bs{E}}(t)
    ( \bs{s}^{2}_{\bs{E}} )^{-1}\bs{s}_{\bs{E}, \bs{C}}
    - \bs{s}_{\bs{C}, \bs{E}}(t) ( \bs{s}^{2}_{\bs{E}} )^{-1}\bs{s}_{\bs{E}, \bs{C}}
    - \bs{s}_{\bs{C}, \bs{E}} ( \bs{s}^{2}_{\bs{E}} )^{-1} \bs{s}_{\bs{E}, \bs{C}}(t)\\
    & = 
    \bs{S}_{\bs{C}}^2 + 
    \bs{S}_{\bs{C}, \bs{E}} ( \bs{S}^{2}_{\bs{E}} )^{-1} \bs{S}^2_{\bs{E}}
    ( \bs{S}^{2}_{\bs{E}} )^{-1}\bs{S}_{\bs{E}, \bs{C}}
    - \bs{S}_{\bs{C}, \bs{E}} ( \bs{S}^{2}_{\bs{E}} )^{-1}\bs{S}_{\bs{E}, \bs{C}}
    - \bs{S}_{\bs{C}, \bs{E}} ( \bs{S}^{2}_{\bs{E}} )^{-1} \bs{S}_{\bs{E}, \bs{C}} + o_{\Pr}(1)
    \\
    & = \bs{S}_{\bs{C}}^2 - \bs{S}_{\bs{C}, \bs{E}} ( \bs{S}^{2}_{\bs{E}} )^{-1} \bs{S}_{\bs{E}, \bs{C}} + o_{\Pr}(1)
    = \bs{S}_{\bs{C} \setminus \bs{E}}^2 + o_{\Pr}(1), 
\end{align*}
and 
\begin{equation*}
    \bs{s}_{\bs{C}^\res, \bs{E}}(t)
    = \bs{s}_{\bs{C}, \bs{E}}(t) - \bs{s}_{\bs{C}, \bs{E}} ( \bs{s}^{2}_{\bs{E}} )^{-1} \bs{s}_{\bs{E}}^2(t)
    = 
    \bs{S}_{\bs{C}, \bs{E}} - \bs{S}_{\bs{C}, \bs{E}} ( \bs{S}^{2}_{\bs{E}} )^{-1} \bs{S}_{\bs{E}}^2 + o_{\Pr}(1)
    = o_{\Pr}(1). 
\end{equation*}
From the above, Lemma \ref{lemma:samp_cov_C_res} holds. 
\end{proof}

\subsection{Proof for regression adjustment with optimal estimated precision}

\begin{proof}[\bf Proof of Theorem \ref{thm:est_optimal}]
	From Propositions \ref{prop:est_fix_coef_unknown_design} and \ref{prop:est_est_coef_unknown_design}, it is obvious that $\hat{\tau}(\tilde{\bs{\beta}}, \tilde{\bs{\gamma}})$ and $\hat{\tau}(\hat{\bs{\beta}}, \hat{\bs{\gamma}})$ have the same estimated distribution asymptotically. 
	Below we prove the $\mathcal{C}$-optimality of $\hat{\tau}(\tilde{\bs{\beta}}, \tilde{\bs{\gamma}})$. 
	From the discussion in Section \ref{app:ci}, 
	the length of the $1-\alpha$ confidence interval multiplied by $\sqrt{n}$ will converge in probability to the $1-\alpha$ symmetric quantile range of the asymptotic estimated distribution. Therefore, to prove the $\mathcal{C}$-optimality of the regression-adjusted estimator $\hat{\tau}(\tilde{\bs{\beta}}, \tilde{\bs{\gamma}})$, it suffices to prove that the asymptotic estimated distribution of $\hat{\tau}(\tilde{\bs{\beta}}, \tilde{\bs{\gamma}})$ has the shortest symmetric quantile ranges.

	From Proposition \ref{prop:est_fix_coef_unknown_design}, 
	the asymptotic estimated distribution of the regression-adjusted estimator $\hat{\tau}(\bs{\beta}, \bs{\gamma})$ has the same weak limit as 
	\begin{align}\label{eq:est_dist_gen}
		& \quad \ \tilde{V}_{\tau\tau}^{1/2}(\bs{\beta}, \bs{\gamma}) \Big(  \sqrt{1-\tilde{R}_S^2(\bs{\beta}, \bs{\gamma}) - \tilde{R}_T^2(\bs{\beta}, \bs{\gamma})} \cdot \varepsilon
		+ \sqrt{\tilde{R}_S^2(\bs{\beta}, \bs{\gamma})} \cdot L_{J,a_S}
		+
		\sqrt{\tilde{R}_T^2(\bs{\beta}, \bs{\gamma})} \cdot L_{K,a_T} 
		\Big)
		\nonumber
		\\
		& \sim 
		\sqrt{fS^2_{\tau\setminus \bs{C}} + V_{\tau\tau}(\bs{\beta}, \bs{\gamma}) - \I_{S} V_{\tau\tau}^{1/2}(\bs{\beta}, \bs{\gamma}) R_S^2(\bs{\beta}, \bs{\gamma}) - \I_{T} V_{\tau\tau}^{1/2}(\bs{\beta}, \bs{\gamma}) R_T^2(\bs{\beta}, \bs{\gamma}) } \cdot \varepsilon 
		\nonumber
		\\
		& \quad \ 
		+ \sqrt{\I_{S} V_{\tau\tau}^{1/2}(\bs{\beta}, \bs{\gamma}) R_S^2(\bs{\beta}, \bs{\gamma})} \cdot L_{J,a_S}
		+ \sqrt{\I_{T} V_{\tau\tau}^{1/2}(\bs{\beta}, \bs{\gamma}) R_T^2(\bs{\beta}, \bs{\gamma})} \cdot L_{K,a_T}, 
	\end{align}
	where $\I_{S}$ and $\I_{T}$ are indicators for whether the analyzer knows the design information in sampling and treatment assignment stages, respectively. 
	The asymptotic estimated distribution \eqref{eq:est_dist_gen} is essentially the convolution of a Gaussian distribution $\mathcal{N}(0, fS^2_{\tau\setminus \bs{C}})$ and the asymptotic sampling distribution of the regression-adjusted estimator $\hat{\tau}(\bs{\beta}, \bs{\gamma})$ under ReSEM with the following rerandomization criteria: 
	(i) covariate $\bs{W}$ and threshold $a_S$ for sampling stage if $\I_{S} = 1$, and covariate $\emptyset$ and threshold $\infty$ for sampling stage if $\I_{S} = 0$; 
	(ii) covariate $\bs{X}$ and threshold $a_T$ for treatment assignment stage if $\I_{T} = 1$, and covariate $\emptyset$ and threshold $\infty$ for treatment assignment stage if $\I_{T} = 0$. 
	In either of these cases, 
    the covariates in design will be subsets of that in analysis.
	By the optimality of regression-adjusted estimator $\hat{\tau}(\tilde{\bs{\beta}}, \tilde{\bs{\gamma}})$ in Theorem \ref{thm:optimal}(iii), the unimodality of Gaussian distribution and Lemma \ref{lemma:sum_unimodal}, we can know that the asymptotic estimated distribution of $\hat{\tau}(\bs{\beta}, \bs{\gamma})$ achieves the shortest symmetric quantile ranges at $(\bs{\beta}, \bs{\gamma}) = (\tilde{\bs{\beta}}, \tilde{\bs{\gamma}})$. 

	From the above, Theorem \ref{thm:est_optimal} holds. 
\end{proof}

\subsection{Proof for regression with treatment-covariate interaction}

\begin{proof}[\bf Proof for the special case in which $\bs{E} = \emptyset$]
When $\bs{E} = \emptyset$, by definition, 
$\hat{\tau}(\hat{\bs{\beta}}, \hat{\bs{\gamma}})$ reduces to 
$\hat{\tau} - \hat{\bs{\beta}}^\top \hat{\bs{\tau}}_{\bs{C}}$. 
Besides, by definition, 
we can derive that 
$\bar{\bs{C}}_1 - \bar{\bs{C}}_{\mathcal{S}} = r_0 \hat{\bs{\tau}}_{\bs{C}}$ 
and 
$\bar{\bs{C}}_0 - \bar{\bs{C}}_{\mathcal{S}} = -r_1 \hat{\bs{\tau}}_{\bs{C}}$. 
Consequently, we have 
\begin{align}\label{eq:tau_beta_tau_X}
    \hat{\tau} - \hat{\bs{\beta}}^\top \hat{\bs{\tau}}_{\bs{C}}
    & = 
    \bar{Y}_1 
    - \bar{Y_0} - (r_0\hat{\bs{\beta}}_1 + r_1 \hat{\bs{\beta}}_0)^\top\hat{\bs{\tau}}_{\bs{C}}
    = 
    \bar{Y}_1 - \hat{\bs{\beta}}_1^\top r_0 \hat{\bs{\tau}}_{\bs{C}}
    - 
    \bar{Y}_0 - \hat{\bs{\beta}}_0^\top r_1 \hat{\bs{\tau}}_{\bs{C}}
    \nonumber
    \\
    & = 
    \bar{Y}_1 - \hat{\bs{\beta}}_1^\top (\bar{\bs{C}}_1 - \bar{\bs{C}}_{\mathcal{S}})
    - 
    \bar{Y}_0 + \hat{\bs{\beta}}_0^\top (\bar{\bs{C}}_0 - \bar{\bs{C}}_{\mathcal{S}}). 
\end{align}
From Lemma \ref{lemma:least_square_general}, $\hat{\tau}(\hat{\bs{\beta}}, \hat{\bs{\gamma}})$ is the same as $\hat{\theta}_{\bs{C}_{\mathcal{S}}}$ defined in \eqref{eq:theta_C_S} in this special case. 
\end{proof}

\begin{proof}[\bf Proof for the special case in which $\bs{E} = \bs{C}$]
By definition, $\hat{\bs{\delta}}_{\bs{E}}$ has the following equivalent forms: 
\begin{align*}
    \hat{\bs{\delta}}_{\bs{E}}
    & = 
    \bar{\bs{E}}_{\mathcal{S}} - \bar{\bs{E}}
    = 
    -(\bar{\bs{E}}_1 - \bar{\bs{E}}_{\mathcal{S}}) + (\bar{\bs{E}}_1 - \bar{\bs{E}})
    = 
    -(\bar{\bs{E}}_0 - \bar{\bs{E}}_{\mathcal{S}}) + (\bar{\bs{E}}_0 - \bar{\bs{E}})
\end{align*}
This implies that 
\begin{align*}%
    \hat{\bs{\gamma}} \hat{\bs{\delta}}_{\bs{E}}
    & = 
    \hat{\bs{\gamma}}_1^\top \hat{\bs{\delta}}_{\bs{E}}
    - 
    \hat{\bs{\gamma}}_0^\top \hat{\bs{\delta}}_{\bs{E}}
    = 
    - \hat{\bs{\gamma}}_1^\top (\bar{\bs{E}}_1 - \bar{\bs{E}}_{\mathcal{S}}) 
    + \hat{\bs{\gamma}}_1^\top(\bar{\bs{E}}_1 - \bar{\bs{E}})
    + \hat{\bs{\gamma}}_0^\top(\bar{\bs{E}}_0 - \bar{\bs{E}}_{\mathcal{S}})
    - \hat{\bs{\gamma}}_0^\top(\bar{\bs{E}}_0 - \bar{\bs{E}})
    \nonumber
    \\
    & = 
    \big\{ \bar{Y}_1 - \hat{\bs{\gamma}}_1^\top (\bar{\bs{E}}_1-\bar{\bs{E}}_{\mathcal{S}})  - \bar{Y}_0 + \hat{\bs{\gamma}}_0^\top(\bar{\bs{E}}_0 - \bar{\bs{E}}_{\mathcal{S}}) \big\}
    - 
    \big\{ \bar{Y}_1 - \hat{\bs{\gamma}}_1^\top (\bar{\bs{E}}_1-\bar{\bs{E}})  - \bar{Y}_0 + \hat{\bs{\gamma}}_0^\top(\bar{\bs{E}}_0 - \bar{\bs{E}}) \big\}
    \nonumber
    \\
    & = 
    \hat{\theta}_{\bs{E}_{\mathcal{S}}} - \hat{\theta}_{\bs{E}}, 
\end{align*}
where $\hat{\theta}_{\bs{E}_{\mathcal{S}}}$ is defined analogously as in \eqref{eq:theta_C_S} and  
$\hat{\theta}_{\bs{E}}$ is defined in \eqref{eq:theta_E}. 
From \eqref{eq:tau_beta_tau_X}, 
the regression-adjusted estimator then has the following equivalent forms: 
\begin{align*}
    \hat{\tau}(\hat{\bs{\beta}}, \hat{\bs{\gamma}})
    & = 
    \left( \hat{\tau} - \hat{\bs{\beta}}^\top \hat{\bs{\tau}}_{\bs{C}} \right)
    - \hat{\bs{\gamma}} \hat{\bs{\delta}}_{\bs{E}}
    = \hat{\theta}_{\bs{C}_{\mathcal{S}}} - \big( \hat{\theta}_{\bs{E}_{\mathcal{S}}} - \hat{\theta}_{\bs{E}} \big). 
\end{align*}
Because $\bs{E} = \bs{C}$ here, the regression-adjusted estimator further reduces to 
$\hat{\tau}(\hat{\bs{\beta}}, \hat{\bs{\gamma}}) = \hat{\theta}_{\bs{E}}$. 
\end{proof}

\begin{proof}[\bf Proof for the general case in which $\bs{E} \subset \bs{C}$]

From Theorem \ref{thm:equiv_resem_sresem}, it suffices to prove that $\hat{\theta}_{\tilde{\bs{C}}_\mathcal{S}^\res, \bs{E}} -  \hat{\tau}(\hat{\bs{\beta}}, \hat{\bs{\gamma}}) = o_{\Pr}(n^{-1/2})$ under $\sresem$. 
Without loss of generality, 
we assume $\bs{C} = (\bs{E}^\top, \tilde{\bs{C}}^\top)^\top$
with 
$\tilde{\bs{C}}$ 
being the subvector of $\bs{C}$ that 
cannot be linearly represented by $\bs{E}$. 
Specifically, 
the finite population covariance of the residual from the linear projection of $\tilde{\bs{C}}$ on $\bs{E}$, 
$\bs{S}^2_{\tilde{\bs{C}} \setminus \bs{E}}$, has a nonsingular limit as $N\rightarrow\infty$. 
Let $\tilde{\bs{C}}_i^\res$ be the residual from the linear projection of $\tilde{\bs{C}}_i$ on $\bs{E}_i$ among sampled units $i\in \mathcal{S}$. 
By the property of linear projection, the average of $\tilde{\bs{C}}^\res$ over sampled units in $\mathcal{S}$ must be zero. 
From Lemma \ref{lemma:samp_cov_C_res}, the least squares coefficient of the outcome $Y_i$ on covariates $\tilde{\bs{C}}_i^\res$ and $\bs{E}_i$ for units under treatment arm $t$ is 
\begin{align}\label{eq:ols_coef_c_res_E}
    \begin{pmatrix}
        \hat{\bs{\zeta}}_{\tilde{\bs{C}}^\res,t}\\
        \hat{\bs{\zeta}}_{\bs{E},t}
    \end{pmatrix}
    & = 
    \begin{pmatrix}
    \bs{s}_{\tilde{\bs{C}}^\res}^2(t) & \bs{s}_{\tilde{\bs{C}}^\res, \bs{E}}(t)\\
    \bs{s}_{\bs{E}, \tilde{\bs{C}}^\res}(t) & \bs{s}_{\bs{E}}^2(t)
    \end{pmatrix}^{-1} 
    \begin{pmatrix}
    \bs{s}_{\tilde{\bs{C}}^\res, t}\\
    \bs{s}_{\bs{E}, t}
    \end{pmatrix}
    = 
    \begin{pmatrix}
    \bs{S}_{\tilde{\bs{C}} \setminus \bs{E}}^2 & \bs{0}\\
    \bs{0}  & \bs{S}_{\bs{E}}^2
    \end{pmatrix}^{-1} 
    \begin{pmatrix}
    \bs{S}_{\tilde{\bs{C}}, t} 
    - \bs{S}_{\tilde{\bs{C}}, \bs{E}} ( \bs{S}^{2}_{\bs{E}} )^{-1} \bs{S}_{\bs{E}, t}
    \\
    \bs{S}_{\bs{E}, t}
    \end{pmatrix}
    + o_{\Pr}(1)
    \nonumber
    \\
    & = 
    \begin{pmatrix}
        \tilde{\bs{\zeta}}_{\tilde{\bs{C}}^\res,t}\\
        \tilde{\bs{\gamma}}_{t}
    \end{pmatrix}
    + o_{\Pr}(1), 
\end{align}
where $\tilde{\bs{\zeta}}_{\tilde{\bs{C}}^\res,t}$ is essentially the linear projection coefficient of $Y(t)$ on the residual from the linear projection of $\tilde{\bs{C}}$ on $\bs{E}$, 
and $\tilde{\bs{\gamma}}_{t}$ is the linear projection coefficient of $Y(t)$ on $\bs{E}$ as defined in Section \ref{sec:R2_proj_coef}.

We first consider equivalent forms of $\hat{\tau}(\hat{\bs{\beta}}, \hat{\bs{\gamma}})$. 
Note that by construction,  $\bs{C}_i - \bar{\bs{C}}_{\mathcal{S}}$ and $(\bs{E}_i - \bar{\bs{E}}_{\mathcal{S}}^\top, (\tilde{\bs{C}}_i^\res)^\top)^\top$ are linear transformations of each other for sampled units in $\mathcal{S}$. 
By the same logic as \eqref{eq:tau_beta_tau_X}, we have
\begin{align*}
    \hat{\tau} - \hat{\bs{\beta}}^\top \hat{\bs{\tau}}_{\bs{C}}
    & = 
    \bar{Y}_1 - \hat{\bs{\beta}}_1^\top (\bar{\bs{C}}_1 - \bar{\bs{C}}_{\mathcal{S}})
    - 
    \bar{Y}_0 + \hat{\bs{\beta}}_0^\top r_1 (\bar{\bs{C}}_0 - \bar{\bs{C}}_{\mathcal{S}})\\
    & = 
    \bar{Y}_1 - \hat{\bs{\zeta}}_{\tilde{\bs{C}}^\res,1}^\top \bar{\tilde{\bs{C}}}^\res_1 - \hat{\bs{\zeta}}_{\bs{E},1}^\top (\bar{\bs{E}}_1 - \bar{\bs{E}}_{\mathcal{S}})
    - \bar{Y}_0 + 
    \hat{\bs{\zeta}}_{\tilde{\bs{C}}^\res,0}^\top \bar{\tilde{\bs{C}}}^\res_0
    + 
    \hat{\bs{\zeta}}_{\bs{E},0}^\top (\bar{\bs{E}}_0 - \bar{\bs{E}}_{\mathcal{S}}). 
\end{align*}
This immediately implies that $\hat{\tau}(\hat{\bs{\beta}}, \hat{\bs{\gamma}}) = \hat{\tau} - \hat{\bs{\beta}}^\top \hat{\bs{\tau}}_{\bs{C}}
    -\hat{\bs{\gamma}} \hat{\bs{\delta}}_{\bs{E}}$ has the following equivalent form:
\begin{align*}
    \hat{\tau}(\hat{\bs{\beta}}, \hat{\bs{\gamma}})
    & 
    = \bar{Y}_1 - \hat{\bs{\zeta}}_{\tilde{\bs{C}}^\res,1}^\top \bar{\tilde{\bs{C}}}^\res_1 - \hat{\bs{\zeta}}_{\bs{E},1}^\top (\bar{\bs{E}}_1 - \bar{\bs{E}}_{\mathcal{S}})
    - \bar{Y}_0 + 
    \hat{\bs{\zeta}}_{\tilde{\bs{C}}^\res,0}^\top \bar{\tilde{\bs{C}}}^\res_0
    + 
    \hat{\bs{\zeta}}_{\bs{E},0}^\top (\bar{\bs{E}}_0 - \bar{\bs{E}}_{\mathcal{S}}) - \hat{\bs{\gamma}}_1^\top \hat{\bs{\delta}}_{\bs{E}} + \hat{\bs{\gamma}}_0^\top \hat{\bs{\delta}}_{\bs{E}}. 
\end{align*}
We then consider equivalent forms of $\hat{\theta}_{\tilde{\bs{C}}_\mathcal{S}^\res, \bs{E}}$. 
From Lemma \ref{lemma:least_square_general}, 
\begin{align*}
    \hat{\theta}_{\tilde{\bs{C}}_\mathcal{S}^\res, \bs{E}}
    & = 
    \bar{Y}_1 - \hat{\bs{\zeta}}_{\tilde{\bs{C}}^\res,1}^\top \bar{\tilde{\bs{C}}}^\res_1 - \hat{\bs{\zeta}}_{\bs{E},1}^\top (\bar{\bs{E}}_1 - \bar{\bs{E}})
    - \bar{Y}_0 + 
    \hat{\bs{\zeta}}_{\tilde{\bs{C}}^\res,0}^\top \bar{\tilde{\bs{C}}}^\res_0
    + 
    \hat{\bs{\zeta}}_{\bs{E},0}^\top (\bar{\bs{E}}_0 - \bar{\bs{E}}). 
\end{align*}
From the above, we can know that 
\begin{align*}
    \hat{\tau}(\hat{\bs{\beta}}, \hat{\bs{\gamma}}) - \hat{\theta}_{\tilde{\bs{C}}_\mathcal{S}^\res, \bs{E}}
    & = \hat{\bs{\zeta}}_{\bs{E},1}^\top (\bar{\bs{E}}_{\mathcal{S}} - \bar{\bs{E}}) - 
    \hat{\bs{\zeta}}_{\bs{E},0}^\top (\bar{\bs{E}}_{\mathcal{S}} - \bar{\bs{E}}) - \hat{\bs{\gamma}}_1^\top \hat{\bs{\delta}}_{\bs{E}} + \hat{\bs{\gamma}}_0^\top \hat{\bs{\delta}}_{\bs{E}}
    \\
    & = 
    \big( \hat{\bs{\zeta}}_{\bs{E},1} - \hat{\bs{\gamma}}_1 \big)^\top \hat{\bs{\delta}}_{\bs{E}}
    - 
    \big( \hat{\bs{\zeta}}_{\bs{E},0} - \hat{\bs{\gamma}}_0 \big)^\top \hat{\bs{\delta}}_{\bs{E}}. 
\end{align*}
From Lemma \ref{lemma:diff_average}, \eqref{eq:ols_coef_c_res_E} and the proof of Theorem \ref{thm:plug_in}, we can know that
$\hat{\bs{\delta}}_{\bs{E}} = O_{\Pr}(n^{-1/2})$,  
$\hat{\bs{\zeta}}_{\bs{E},1} - \hat{\bs{\gamma}}_1 = o_{\Pr}(1)$ and $\hat{\bs{\zeta}}_{\bs{E},0} - \hat{\bs{\gamma}}_0 = o_{\Pr}(1)$. 
Therefore, 
$\hat{\tau}(\hat{\bs{\beta}}, \hat{\bs{\gamma}}) - \hat{\theta}_{\tilde{\bs{C}}_\mathcal{S}^\res, \bs{E}} = o_{\Pr}(n^{-1/2})$. 
\end{proof}

\section{Randomization test for ReSEM}\label{sec:proof_rand_test}

Because conditional on $\mathcal{S}$, 
ReSEM reduces to a usual randomized experiment for the sampled units. 
Theorem \ref{thm:cond_rand_test} then follows immediately from the usual justification of Fisher randomization test. 
Besides, it is not hard to see that Theorem \ref{thm:cond_rand_test} also holds under $\sresem$.
Below we focus on the proof of Theorem \ref{thm:frt_weak_null}. 

\subsection{Technical lemmas}

\begin{lemma}\label{lemma:imp_outcome_fp_cond_resem}
    Consider $\sresem$ with pre-determined positive thresholds $a_S$ and $a_T$ for a finite population of size $N$. 
    For the sampled units in $\mathcal{S} = \{i: Z_i=1, 1\le i \le N\}$, 
    define $\tilde{Y}_i(1) = Y_i + (1-T_i) \tau$ and $\tilde{Y}_i(0) = Y_i - T_i \tau$, 
    where $\tau = \bar{Y}(1) - \bar{Y}(0)$ is 
    the true average treatment effect. 
    For the sampled units in $\mathcal{S}$, 
    let 
    \begin{align*}
        \bar{\tilde{Y}}_{\mathcal{S}}(t)
        = 
        n^{-1} \sum_{i \in \mathcal{S}} \tilde{Y}_i(t)
        \ \ \text{and} \ \ 
        \tilde{s}_t^2 = (n-1)^{-1} \sum_{i\in \mathcal{S}} \left\{ \tilde{Y}_i(t) - \bar{\tilde{Y}}_{\mathcal{S}}(t) \right\}^2, 
        \qquad 
        (t=0, 1)
    \end{align*}
    be the sample mean and sample variance of the potential outcome $\tilde{Y}_i(t)$, 
    $\bar{\bs{C}}_{\mathcal{S}}$ and $\bs{s}^2_{\bs{C}}$ be the sample mean and sample covariance matrix of covariate $\bs{C}$, 
    $\tilde{\bs{s}}_{\bs{C}, t}$ be the sample covariance matrix between $\bs{C}$ and $\tilde{Y}(t)$, 
    and 
    $\tilde{s}^2_{t\setminus \bs{C}} = \tilde{s}^2_t- \tilde{\bs{s}}_{t,\bs{C}}(\bs{s}^2_{\bs{C}})^{-1} \tilde{\bs{s}}_{\bs{C},t}$ be the sample variance of the residual from the linear projection of $\tilde{Y}(t)$ on $\bs{C}$. 
    If Condition \ref{cond:fp_est} holds, then under $\sresem$, 
    as $N\rightarrow \infty$, we have 
    \begin{itemize}
        \item[(i)] $\tilde{s}_1^2 = \tilde{s}_0^2$ converges in probability to a finite limit, 
        $\bs{s}^2_{\bs{C}}$ converges in probability to a nonsingular finite limit, and 
        $\tilde{\bs{s}}_{\bs{C}, 1} = \tilde{\bs{s}}_{\bs{C}, 0}$ converges in probability to a finite limit; 
        \item[(ii)] 
        $\max_{i \in \mathcal{S}} |\tilde{Y}_i(t) - \bar{\tilde{Y}}_{\mathcal{S}}(t)|^2/n \rightarrow 0$ for $t=0,1$ 
        and 
         $\max_{i \in \mathcal{S}} \|\bs{C}_i - \bar{\bs{C}}_{\mathcal{S}} \|_2^2/n \rightarrow 0$; 
         \item[(iii)] $\tilde{s}^2_{1\setminus \bs{C}} = \tilde{s}^2_{0\setminus \bs{C}} $ converges in probability to a positive finite limit. 
    \end{itemize}
\end{lemma}

\begin{lemma}\label{lemma:est_cdf_trans_converge}
    Suppose $\hat{V}_N, \hat{R}_{1N}^2$ and $\hat{R}_{2N}^2$ are consistent estimators for $V_N>0, R_{1N}^2\in [0,1]$ and $R_{2N}^2\in [0,1]$, in the sense that as $N\rightarrow \infty$, 
    $$
    \hat{V}_N-V_N = o_{\Pr}(1), \quad \hat{R}_{1N}^2-R_{1N}^2 = o_{\Pr}(1), \quad \hat{R}_{2N}^2-R_{2N}^2 = o_{\Pr}(1), 
    $$ 
    $V_N, R_{1N}^2$ and $R_{2N}^2$ converge to $V_{\infty}>0$, $R_{1\infty}^2\in [0,1]$ and $R_{2\infty}^2\in [0,1]$ as $N\rightarrow \infty$, 
    and 
    $\xi_N$ is a random variable converging weakly to $\xi_{\infty}$ as $N\rightarrow \infty$. 
    Define $\Psi_{V, R_1^2, R_2^2}(\cdot)$ the same as in Lemma \ref{lemma:continuous_cdf}, with predetermined fixed positive thresholds $a_1$ and $a_2$ and positive integers $K_1$ and $K_2$. 
    Then, as $N\rightarrow \infty$, 
    $
        \Psi_{\hat{V}_N, \hat{R}_{1N}^2, \hat{R}_{2N}^2}(\xi_N) 
        \converged 
        \Psi_{V_{\infty}, R_{1\infty}^2, R_{2\infty}^2}(\xi_{\infty}). 
    $
\end{lemma}

\begin{lemma}\label{lemma:stoch_uniform}
    Define $\Psi_{V, R_1^2, R_2^2}(\cdot)$ the same as in Lemma \ref{lemma:continuous_cdf}, with predetermined fixed positive thresholds $a_1$ and $a_2$ and positive integers $K_1$ and $K_2$, 
    and $\psi_{V, R_1^2, R_2^2}$ as a random variable with distribution function $\Psi_{V, R_1^2, R_2^2}(\cdot)$. 
    If two triples $(\tilde{V}, \tilde{R}_1^2, \tilde{R}_2^2)$ and $(V, R_1^2, R_2^2)$ in $[0, \infty) \times [0, 1] \times [0, 1]$ satisfy that $\tilde{V} > 0$, and 
    \begin{align}\label{eq:stoch_uniform_cond}
        \tilde{V} \ge V,
        \quad
        \tilde{V}\tilde{R}_1^2 \le VR_1^2, 
        \quad 
        \tilde{V}\tilde{R}_2^2 \le VR_2^2, 
    \end{align}
    then 
    the random variable $2\Psi_{\tilde{V}, \tilde{R}_1^2, \tilde{R}_2^2}(|\psi_{V, R_1^2, R_2^2}|)-1$ must be stochastically smaller than or equal to 
    $\Unif(0,1)$. 
    Moreover, $2\Psi_{\tilde{V}, \tilde{R}_1^2, \tilde{R}_2^2}(|\psi_{V, R_1^2, R_2^2}|)-1$ becomes uniformly distributed on interval $(0,1)$ when the inequalities in \eqref{eq:stoch_uniform_cond} hold with equality. 
\end{lemma}

\subsection{Proofs of the lemmas}

\begin{proof}[Proof of Lemma \ref{lemma:imp_outcome_fp_cond_resem}]
    First, we prove that  $\tilde{s}_1^2 = \tilde{s}_0^2$ converges in probability to a finite limit as $N\rightarrow \infty$. 
    Note that $\tilde{Y}_i(1) = \tilde{Y}_i(0)+\tau$ for all $i\in \mathcal{S}$.  
    $\tilde{s}_1^2$ and $\tilde{s}_0^2$ must be equal, and thus it suffices to prove that $\tilde{s}_1^2$ converges in probability to a finite limit. 
    By definition, $\bar{\tilde{Y}}_{\mathcal{S}}(1)$ has the following equivalent forms:
    \begin{align*}
        \bar{\tilde{Y}}_{\mathcal{S}}(1) & = 
        \frac{1}{n} \sum_{i\in \mathcal{S}} \tilde{Y}_i(1)
        = \frac{1}{n} \sum_{i\in \mathcal{S}} T_i Y_i + \frac{1}{n} \sum_{i\in \mathcal{S}} (1-T_i)(Y_i+\tau)
        = r_1 \bar{Y}_1 + r_0 \bar{Y}_0 + r_0 \tau. 
    \end{align*}
    Thus, the difference between $\tilde{Y}_i(1)$ and $\bar{\tilde{Y}}_{\mathcal{S}}(1)$ for $i\in \mathcal{S}$ has the following equivalent forms:
    \begin{align}\label{eq:centered_Y_tilde}
        \tilde{Y}_i(1) - \bar{\tilde{Y}}_{\mathcal{S}}(1)
        & = 
        T_i \left\{ Y_i - \left( r_1 \bar{Y}_1 + r_0 \bar{Y}_0 + r_0 \tau \right) \right\}
        + 
        (1-T_i) \left\{ Y_i + \tau - \left( r_1 \bar{Y}_1 + r_0 \bar{Y}_0 + r_0 \tau \right) \right\}
        \nonumber
        \\
        & = 
        T_i \left\{ Y_i - \bar{Y}_1 + r_0 \left( \bar{Y}_1 - \bar{Y}_0 - \tau \right) \right\}
        + 
        (1-T_i) \left\{ Y_i - \bar{Y}_0 - r_1 \left(  \bar{Y}_1 - \bar{Y}_0 - \tau  \right) \right\}. 
    \end{align}
    This implies that the sample variance of $\tilde{Y}_i(1)$ has the following equivalent forms: 
    \begin{align*}
        (n-1)\tilde{s}^2_1 
        & 
        = 
        \sum_{i\in \mathcal{S}} T_i \left\{ Y_i - \bar{Y}_1 + r_0 \left( \bar{Y}_1 - \bar{Y}_0 - \tau \right) \right\}^2 
        + 
        \sum_{i\in \mathcal{S}} (1-T_i) \left\{ Y_i - \bar{Y}_0 - r_1 \left(  \bar{Y}_1 - \bar{Y}_0 - \tau  \right) \right\}^2 
        \\
        & = 
        \sum_{i\in \mathcal{S}} T_i \left( Y_i - \bar{Y}_1\right)^2 
        + 
        n_1 r_0^2 \left( \bar{Y}_1 - \bar{Y}_0 - \tau \right)^2
        + 
        \sum_{i\in \mathcal{S}} (1-T_i) \left( Y_i - \bar{Y}_0 \right)^2
        + 
        n_0 r_1^2 \left(  \bar{Y}_1 - \bar{Y}_0 - \tau  \right)^2\\
        & = 
        (n_1-1) s^2_1 + (n_0-1) s^2_0 + 
        n r_1 r_0 \left(\bar{Y}_1 - \bar{Y}_0 - \tau \right)^2. 
    \end{align*}
    From Lemmas \ref{lemma:estimate} and \ref{lemma:diff_average}, 
    under $\sresem$, 
    $s_t^2 - S_t^2 \convergep 0$ and $\bar{Y}_t - \bar{Y}(t) \convergep 0$ for $t=0,1$ as $N\rightarrow \infty$. 
    These then imply that 
    \begin{align*}
        \tilde{s}_1^2 - (r_1 S_1^2 + r_0 S_0^2)
        & = 
        \frac{n_1-1}{n-1} s^2_1 - r_1S_1^2 + \frac{n_0-1}{n-1} s^2_0 - r_0 S_0^2 + 
        \frac{n}{n-1} r_1 r_0 \left(\bar{Y}_1 - \bar{Y}_0 - \tau \right)^2\\
        & = 
        \frac{n_1-1}{n-1} \left( s^2_1 - S^2_1 \right) - \frac{1-r_1}{n-1} S_1^2 + \frac{n_0-1}{n-1} \left( s^2_0 - S^2_0 \right) - \frac{1-r_0}{n-1} S_0^2 \\
        & \quad \ + 
        \frac{n}{n-1} r_1 r_0 \left\{ \bar{Y}_1 - \bar{Y}(1) - \bar{Y}_0 + \bar{Y}(0) \right\}^2\\
        & \convergep 0. 
    \end{align*}
    Therefore, $\tilde{s}_1^2$ must converge in probability to the limit of $(r_1 S_1^2 + r_0 S_0^2)$ as $N\rightarrow \infty$.  
    
    Second, from Lemma \ref{lemma:cov_sampled_units_resem}, 
    $\bs{s}^2_{\bs{C}} - \bs{S}^2_{\bs{C}} \convergep 0$ as $N \rightarrow \infty$ under $\sresem$. 
    Consequently, from Condition \ref{cond:fp_est}, 
    $\bs{s}^2_{\bs{C}}$ converges in probability to the limit of $\bs{S}^2_{\bs{C}}$ as $N\rightarrow \infty$, which is nonsingular. 
    
    Third, we prove that $\tilde{\bs{s}}_{\bs{C}, 1} = \tilde{\bs{s}}_{\bs{C}, 0}$ converges in probability to a finite limit. 
    Because 
    $\tilde{Y}_i(1) = \tilde{Y}_i(0)+\tau$ for all $i\in \mathcal{S}$, 
    $\tilde{\bs{s}}_{\bs{C}, 1}$ and $\tilde{\bs{s}}_{\bs{C}, 0}$ must be equal. Thus, it suffices to prove that $\tilde{\bs{s}}_{\bs{C}, 1}$ converges in probability to a finite limit. 
    By definition and from \eqref{eq:centered_Y_tilde}, 
    \begin{align*}
        \tilde{s}_{\bs{C}, 1} & = \frac{1}{n-1}\sum_{i\in \mathcal{S}} ( \bs{C}_i - \bar{\bs{C}}_{\mathcal{S}} ) \{ \tilde{Y}_i(1) - \bar{\tilde{Y}}_{\mathcal{S}}(1) \}
        \\
        & = 
        \frac{1}{n-1}\sum_{i\in \mathcal{S}} T_i ( \bs{C}_i - \bar{\bs{C}}_1 + \bar{\bs{C}}_1 -  \bar{\bs{C}}_{\mathcal{S}} ) \left\{ Y_i - \bar{Y}_1 + r_0 \left( \bar{Y}_1 - \bar{Y}_0 - \tau \right) \right\}
        \\
        & \quad \ + 
        \frac{1}{n-1} \sum_{i\in \mathcal{S}} (1-T_i) ( \bs{C}_i - \bar{\bs{C}}_0 + \bar{\bs{C}}_0 -  \bar{\bs{C}}_{\mathcal{S}} ) \left\{ Y_i - \bar{Y}_0 - r_1 \left(  \bar{Y}_1 - \bar{Y}_0 - \tau  \right) \right\}
        \\
        & = 
        \frac{n_1-1}{n-1} \bs{s}_{\bs{C}, 1}
        + 
        \frac{n_0-1}{n-1} \bs{s}_{\bs{C}, 0}
        + 
        \frac{n_1 r_0}{n-1} (\bar{\bs{C}}_1 -  \bar{\bs{C}}_{\mathcal{S}})  \left( \bar{Y}_1 - \bar{Y}_0 - \tau \right)
        - 
        \frac{n_0r_1}{n-1} (\bar{\bs{C}}_0 -  \bar{\bs{C}}_{\mathcal{S}}) \left(  \bar{Y}_1 - \bar{Y}_0 - \tau  \right) \\
        & = 
        \frac{n_1-1}{n-1} \bs{s}_{\bs{C}, 1}
        + 
        \frac{n_0-1}{n-1} \bs{s}_{\bs{C}, 0}
        + 
        \frac{n}{n-1} r_1r_0 (\bar{\bs{C}}_1 -  \bar{\bs{C}}_0)  \left( \bar{Y}_1 - \bar{Y}_0 - \tau \right).
    \end{align*}
    From Lemmas \ref{lemma:estimate} and \ref{lemma:diff_average}, 
    under $\sresem$, 
    $\bs{s}_{\bs{C}, t} - \bs{S}_{\bs{C}, t} \convergep 0$, 
    $\bar{\bs{C}}_t - \bar{\bs{C}} \convergep 0$ 
    and 
    $\bar{Y}_t - Y(t) \convergep 0$ for $t=0,1$ as $N\rightarrow \infty$. 
    These then imply that 
    \begin{align*}%
        \tilde{s}_{\bs{C}, 1} - (r_1 \bs{S}_{\bs{C}, 1} + r_0 \bs{S}_{\bs{C}, 0})
        &\  = 
        \frac{n_1-1}{n-1} \left( \bs{s}_{\bs{C}, 1} - \bs{S}_{\bs{C}, 1} \right)
        - 
        \frac{1-r_1}{n-1} \bs{S}_{\bs{C}, 1}
        + 
        \frac{n_0-1}{n-1} \left( \bs{s}_{\bs{C}, 0} - \bs{S}_{\bs{C}, 0} \right)
        - 
        \frac{1-r_0}{n-1} \bs{S}_{\bs{C}, 0}
        \nonumber
        \\
        & \ \quad \ + 
        \frac{n}{n-1} r_1r_0 (\bar{\bs{C}}_1 - \bar{\bs{C}} -  \bar{\bs{C}}_0 + \bar{\bs{C}})  \left( \bar{Y}_1 - \bar{Y}(1) - \bar{Y}_0 + \bar{Y}(0) \right)
        \nonumber
        \\
        & \convergep 0. 
    \end{align*}
    Therefore, $\tilde{s}_{\bs{C}, 1}$ must converge in probability to the limit of $r_1 \bs{S}_{\bs{C}, 1} + r_0 \bs{S}_{\bs{C}, 0}$. 
    
    Fourth, we prove that $\max_{i \in \mathcal{S}} |\tilde{Y}_i(t) - \bar{\tilde{Y}}_{\mathcal{S}}(t)|^2/n \rightarrow 0$ for $t=0,1$ 
    and 
    $\max_{i \in \mathcal{S}} \|\bs{C}_i - \bar{\bs{C}}_{\mathcal{S}} \|_2^2/n \rightarrow 0$
    as $N\rightarrow \infty$. 
    By definition and from \eqref{eq:centered_Y_tilde}, 
    \begin{align*}
        \left| \tilde{Y}_i(1) - \bar{\tilde{Y}}_{\mathcal{S}}(1) \right|
        & = 
        T_i \left| Y_i(1) - \bar{Y}_1 + r_0 \left( \bar{Y}_1 - \bar{Y}_0 - \tau \right) \right|
        + 
        (1-T_i) \left| Y_i(0) - \bar{Y}_0 - r_1 \left(  \bar{Y}_1 - \bar{Y}_0 - \tau \right) \right|
        \\
        & \le 
        \left| Y_i(1) - \bar{Y}_1 \right|
        + 
        \left| Y_i(0) - \bar{Y}_0 \right| + \left|  \bar{Y}_1 - \bar{Y}_0 - \tau  \right|\\
        & \le \left| Y_i(1) - \bar{Y}_1 \right|
        + 
        \left| Y_i(0) - \bar{Y}_0 \right| + \left|  \bar{Y}_1 - \bar{Y}(1) \right| + \left| \bar{Y}_0 - \bar{Y}(0) \right|. 
    \end{align*}
    Note that for $t=0,1$, 
    \begin{align*}
        \left| \bar{Y}_t - \bar{Y}(t) \right| & = 
        \left| \frac{1}{n_t}\sum_{j:Z_j = 1, T_j=t} Y_j(t) - \bar{Y}(t) \right|
        \le 
        \frac{1}{n_t} \sum_{j:Z_j = 1, T_j=t} \left| Y_j(t) - \bar{Y}(t) \right|
        \le 
        \max_{1\le j \le N} \left| Y_j(t) - \bar{Y}(t) \right|, 
    \end{align*}
    and 
    \begin{align*}
        \left| Y_i(t) - \bar{Y}_t \right|
        & \le 
        \left| Y_i(t) - \bar{Y}(t) \right| + \left| \bar{Y}_t - \bar{Y}(t)  \right|
        \le 2\max_{1\le j \le N} \left| Y_j(t) - \bar{Y}(t) \right|. 
    \end{align*}
    We can then bound the maximum absolute difference between $\tilde{Y}_i(1)$ and  $\bar{\tilde{Y}}_{\mathcal{S}}(1)$ by 
    \begin{align*}
        \max_{i\in \mathcal{S}}\left| \tilde{Y}_i(1) - \bar{\tilde{Y}}_{\mathcal{S}}(1) \right|
        & \le 
        3 \max_{1\le j \le N} \left| Y_j(1) - \bar{Y}(1) \right|
        + 
        3 \max_{1\le j \le N} \left| Y_j(0) - \bar{Y}(0) \right|. 
    \end{align*}
    From 
    Condition \ref{cond:fp_est}, as $N\rightarrow \infty$,  
    \begin{align*}
        \frac{1}{n^{1/2}} \max_{i\in \mathcal{S}}\left| \tilde{Y}_i(1) - \bar{\tilde{Y}}_{\mathcal{S}}(1) \right|
        & \  \le 
        \frac{3}{n^{1/2}} \max_{1\le i \le N} \left| Y_i(1) - \bar{Y}(1) \right|
        + 
        \frac{3}{n^{1/2}}  \max_{1\le i \le N} \left| Y_i(0) - \bar{Y}(0) \right| 
        \\
        & \converge 0.
    \end{align*}
    This immediately implies that 
    $n^{-1} \max_{i\in \mathcal{S}} \{ \tilde{Y}_i(1) - \bar{\tilde{Y}}_{\mathcal{S}}(1) \}^2 \converge 0$. 
    Because $\tilde{Y}_i(1) - \bar{\tilde{Y}}_{\mathcal{S}}(1) = \tilde{Y}_i(0) - \bar{\tilde{Y}}_{\mathcal{S}}(0)$ for all $i\in \mathcal{S}$, we also have 
    $n^{-1} \max_{i\in \mathcal{S}}\{ \tilde{Y}_i(0) - \bar{\tilde{Y}}_{\mathcal{S}}(0) \}^2 \converge 0$ as $N\rightarrow \infty$. 
    Note that for any $i \in \mathcal{S}$, 
    \begin{align*}
        \|\bs{C}_i - \bar{\bs{C}}_{\mathcal{S}} \|_2
        & = 
        \left\| \bs{C}_i - \bar{\bs{C}} - n^{-1}\sum_{j\in \mathcal{S}} \left( \bs{C}_j - \bar{\bs{C}} \right) \right\|_2
        \le 
        \left\| \bs{C}_i - \bar{\bs{C}} \right\|_2
        + 
        n^{-1} \sum_{j\in \mathcal{S}} \left\| \bs{C}_j - \bar{\bs{C}}\right\|_2
        \le 
        2 \max_{1\le j \le N} \left\| \bs{C}_j - \bar{\bs{C}} \right\|_2. 
    \end{align*}
    Thus, as $N\rightarrow \infty$, we must have 
    \begin{align*}
        \frac{1}{n} \max_{i \in \mathcal{S}}\|\bs{C}_i - \bar{\bs{C}}_{\mathcal{S}} \|_2^2
        & \le 
        \frac{4}{n} \max_{1\le j \le N} \left\| \bs{C}_j - \bar{\bs{C}} \right\|_2^2
        \rightarrow 0. 
    \end{align*}
    
    Fifth, we prove that $\tilde{s}^2_{1\setminus \bs{C}} = \tilde{s}^2_{0\setminus \bs{C}}$
    converges in probability to a positive finite limit. 
    Because 
    $\tilde{Y}_i(1) = \tilde{Y}_i(0)+\tau$ for all $i\in \mathcal{S}$, 
    $\tilde{s}^2_{1\setminus \bs{C}}$ and $\tilde{s}^2_{0\setminus \bs{C}}$ must be equal. Thus, it suffices to prove that $\tilde{s}^2_{1\setminus \bs{C}}$ converges in probability to a positive finite limit. 
    From the discussion before and Lemma \ref{lemma:cov_sampled_units_resem}, as $N\rightarrow \infty$, $\tilde{s}^2_{1\setminus \bs{C}} =\tilde{s}^2_1 -  \tilde{\bs{s}}_{1,\bs{C}}(\bs{s}^2_{\bs{C}})^{-1} \tilde{\bs{s}}_{\bs{C},1}$  converges in probability to the limit of 
    \begin{align*}
    	& \quad \ (r_1 S_1^2 + r_0 S_0^2) - (r_1 \bs{S}_{1, \bs{C}} + r_0 \bs{S}_{0, \bs{C}}) \left(\bs{S}^2_{\bs{C}}\right)^{-1} (r_1 \bs{S}_{\bs{C}, 1} + r_0 \bs{S}_{\bs{C}, 0})
    	\\
    	& = r_1 S_{1\setminus \bs{C}}^2 + r_0 S_{0\setminus \bs{C}}^2 + r_1 S_{1\mid \bs{C}}^2 + r_0 S_{0\mid \bs{C}}^2
    	-r_1^2 S_{1\mid \bs{C}}^2 - r_0^2 S_{0\mid \bs{C}}^2 - 2 r_1r_0 \bs{S}_{1, \bs{C}} \left(\bs{S}^2_{\bs{C}}\right)^{-1}\bs{S}_{\bs{C}, 0}
    	\\
    	& = r_1 S_{1\setminus \bs{C}}^2 + r_0 S_{0\setminus \bs{C}}^2 + r_1r_0 \left\{ S_{1\mid \bs{C}}^2  + S_{0\mid \bs{C}}^2 - 2 \bs{S}_{1, \bs{C}} \left(\bs{S}^2_{\bs{C}}\right)^{-1}\bs{S}_{\bs{C}, 0}  \right\}\\
    	& = r_1 S_{1\setminus \bs{C}}^2 + r_0 S_{0\setminus \bs{C}}^2 + r_1r_0 (\bs{S}_{1, \bs{C}} - \bs{S}_{0, \bs{C}}) \left(\bs{S}^2_{\bs{C}}\right)^{-1} (\bs{S}_{\bs{C}, 1} - \bs{S}_{\bs{C}, 0})\\
    	& = r_1 S_{1\setminus \bs{C}}^2 + r_0 S_{0\setminus \bs{C}}^2 + r_1r_0S_{\tau\mid \bs{C}}^2, 
    \end{align*}
	which must be positive under Condition \ref{cond:fp_est}. 
    
    From the above, Lemma \ref{lemma:imp_outcome_fp_cond_resem} holds. 
\end{proof}

\begin{proof}[Proof of Lemma \ref{lemma:est_cdf_trans_converge}]
    From the conditions in Lemma \ref{lemma:est_cdf_trans_converge}, we can know that 
    $
        (\hat{V}_N, \hat{R}_{1N}^2,  \hat{R}_{2N}^2, \xi_N)
        \converged
        (V_{\infty}, R_{1\infty}^2, R_{2\infty}^2, \xi_{\infty}). 
    $
    From Lemma \ref{lemma:continuous_cdf} and the continuous mapping theorem, 
    $
        \Psi_{\hat{V}_N, \hat{R}_{1N}^2, \hat{R}_{2N}^2}(\xi_N) 
        \converged 
        \Psi_{V_{\infty}, R_{1\infty}^2, R_{2\infty}^2}(\xi_{\infty}), 
    $
    i.e., Lemma \ref{lemma:est_cdf_trans_converge} holds. 
\end{proof}

\begin{proof}[Proof of Lemma \ref{lemma:stoch_uniform}]
	For any $\alpha\in (0,1)$, 
	let $\xi_{\alpha} = \Psi_{\tilde{V}, \tilde{R}_1^2, \tilde{R}_2^2}^{-1}((\alpha+1)/2)$ denote the $(\alpha+1)/2$th quantile of the random variable $\psi_{\tilde{V}, \tilde{R}_1^2, \tilde{R}_2^2}$, 
	where $\Psi_{\tilde{V}, \tilde{R}_1^2, \tilde{R}_2^2}^{-1}(\cdot)$ denotes the quantile function. 
	Because $\tilde{V}>0$, 
	$\psi_{\tilde{V}, \tilde{R}_1^2, \tilde{R}_2^2}$ must be a continuous random variable, and $\xi_{\alpha} > 0$. 
	For any $\alpha \in (0, 1)$, by the property of quantile function, we have 
	\begin{align*}
		& \quad \ \Pr\big\{
		2\Psi_{\tilde{V}, \tilde{R}_1^2, \tilde{R}_2^2}(|\psi_{V, R_1^2, R_2^2}|)-1 \ge \alpha
		\big\}
		\\
		& = 
		\Pr\big\{
		\Psi_{\tilde{V}, \tilde{R}_1^2, \tilde{R}_2^2}(|\psi_{V, R_1^2, R_2^2}|) \ge (\alpha+1)/2
		\big\}
		= 
		\Pr\big\{
		|\psi_{V, R_1^2, R_2^2}| \ge \xi_{\alpha} 
		\big\}
		= 
		1 - 
		\lim_{t \rightarrow \xi_{\alpha}-}
		\Pr\big\{
		|\psi_{V, R_1^2, R_2^2}| \le t
		\big\}
		\\
		& 
		\le 
		1 - 
		\lim_{t \rightarrow \xi_{\alpha}-}
		\Pr\big\{
		|\psi_{\tilde{V}, \tilde{R}_1^2, \tilde{R}_2^2}| \le t
		\big\}
		= 
		\Pr\big\{
		|\psi_{\tilde{V}, \tilde{R}_1^2, \tilde{R}_2^2}| \ge \xi_{\alpha} 
		\big\}
		= 2 
		\Pr\big(
		\psi_{\tilde{V}, \tilde{R}_1^2, \tilde{R}_2^2} \ge \xi_{\alpha} 
		\big)
		=2 \big\{ 1 - (\alpha+1)/2 \big\} 
		\\
		& = 1-\alpha, 
	\end{align*}
	where 
	the last inequality holds due to Lemma \ref{lemma:sotch_est} and 
	the second last equality holds because $\psi_{\tilde{V}, \tilde{R}_1^2, \tilde{R}_2^2}$ is a continuous random variable. 
	This immediately implies that $2\Psi_{\tilde{V}, \tilde{R}_1^2, \tilde{R}_2^2}(|\psi_{V, R_1^2, R_2^2}|)-1$ is stochastically smaller than or equal to $\Unif(0,1)$.
	If the inequalities in \eqref{eq:stoch_uniform_cond} hold with equality, from the above derivation, it is immediate to see that $2\Psi_{\tilde{V}, \tilde{R}_1^2, \tilde{R}_2^2}(|\psi_{V, R_1^2, R_2^2}|)-1 \sim \Unif(0,1)$.  
	Therefore, Lemma \ref{lemma:stoch_uniform} holds. 
\end{proof}

\subsection{Proof of Theorem \ref{thm:frt_weak_null}}
From Theorem \ref{thm:equiv_resem_sresem}, it suffices to prove that Theorem \ref{thm:frt_weak_null} holds under $\sresem$. 
For descriptive convenience, 
we introduce $\mathcal{I} = \{\tilde{\bs{Y}}_{\mathcal{S}}(0), \tilde{\bs{Y}}_{\mathcal{S}}(1), \bs{C}_{\mathcal{S}}, \bs{E}_{1:N}\}$ to denote the set consisting of imputed potential outcomes and available covariates, 
and write the test statistic $g(\bs{T}_{\mathcal{S}}, \tilde{\bs{Y}}_{\mathcal{S}}(0), \tilde{\bs{Y}}_{\mathcal{S}}(1), \bs{C}_{\mathcal{S}}, \bs{E}_{1:N})$ simply as 
$g(\bs{T}_{\mathcal{S}}, \mathcal{I})$.

\begin{proof}[\bf True sampling distribution of the test statistic]
	For any fixed coefficient $\bs{\beta}$, from Lemma \ref{lemma:est_cdf_trans_converge} and the proofs of Theorem \ref{thm:dist_reg_general} and Proposition  \ref{prop:est_fix_coef_unknown_design},  
	under $\sresem$ and Neyman's null $\bar{H}_c$ in \eqref{eq:weak_null}, 
	\begin{align*}
		g(\bs{T}_{\mathcal{S}}, \mathcal{I}) 
		& = 2\hat{F}_{\bs{\beta}, \bs{0}}\left( \sqrt{n}\left| \hat{\tau}(\bs{\beta}, \bs{0}) - c \right| \right) - 1
		= 
		2\hat{F}_{\bs{\beta}, \bs{0}}\left( \sqrt{n}\left| \hat{\tau}(\bs{\beta}, \bs{0}) - \tau \right| \right) - 1
		\\
		& = 
		2 \Psi_{\hat V_{\tau \tau}(\bs{\beta}, \bs{0}), \hat{R}^2_T(\bs{\beta}, \bs{0}), \hat{R}^2_S(\bs{\beta}, \bs{0})}
		\big( \sqrt{n}\left| \hat{\tau}(\bs{\beta}, \bs{0}) - \tau \right| \big) - 1
		\\
		& 
		\ \dot\sim \ 
		2 \Psi_{\tilde{V}_{\tau \tau}(\bs{\beta}, \bs{0}), \tilde{R}^2_T(\bs{\beta}, \bs{0}), \tilde{R}^2_S(\bs{\beta}, \bs{0})}
		\Big( \big|\psi_{V_{\tau \tau}(\bs{\beta}, \bs{0}), R^2_T(\bs{\beta}, \bs{0}), R^2_S(\bs{\beta}, \bs{0})}\big| \Big) - 1, 
	\end{align*}
	where $\tilde{V}_{\tau \tau}(\bs{\beta}, \bs{0}), \tilde{R}^2_T(\bs{\beta}, \bs{0})$ and $\tilde{R}^2_S(\bs{\beta}, \bs{0})$ are defined as in Proposition  \ref{prop:est_fix_coef_unknown_design}, and $\psi_{V_{\tau \tau}(\bs{\beta}, \bs{0}), R^2_T(\bs{\beta}, \bs{0}), R^2_S(\bs{\beta}, \bs{0})}$ is a random variable following the asymptotic distribution of $\hat{\tau}(\bs{\beta}, \bs{0})$ as in \eqref{eq:dist_reg_general}. 

	From Lemma \ref{lemma:stoch_uniform}, the true asymptotic sampling distribution of the test statistic $g(\bs{T}_{\mathcal{S}}, \mathcal{I})$ is stochastically smaller than or equal to $\Unif(0,1)$. 
	Equivalently, 
	the true asymptotic sampling distribution of $1-g(\bs{T}_{\mathcal{S}}, \mathcal{I})$ is stochastically larger than or equal to $\Unif(0,1)$. 
\end{proof}

\begin{proof}[\bf Imputed randomization distribution of the test statistic]
	Let $\check{\bs{T}}_S$ be a random vector uniformly distributed in $\mathcal{A}(\mathcal{S}, \bs{X}, a_T)$ given $\bs{Z}$ and $\bs{T}$. 
	Recall that $\hat{\tau}(\bs{\beta}, \bs{0})$ and $\hat{F}_{\bs{\beta}, \bs{0}}(\cdot)$ are the regression-adjusted estimator and its estimated distribution calculated using the observed data, which can be viewed as functions of the observed assignment vector $\bs{T}_{\mathcal{S}}$ and the information set $\mathcal{I}$. 
	We introduce $\check{\tau}(\bs{\beta}, \bs{0})$ and $\check{F}_{\bs{\beta}, \bs{0}}(\cdot)$ to denote the corresponding values calculated using $\check{\bs{T}}_{\mathcal{S}}$ and $\mathcal{I}$. 
	Let 
	$
	F_{\mathcal{S}, \bs{T}}(b)
	= \Pr\{ g(\check{\bs{T}}_{\mathcal{S}}, \mathcal{I})
	\le b \mid \mathcal{S}, \bs{T} \}
	$
	with 
	$
	g(\check{\bs{T}}_{\mathcal{S}}, \mathcal{I}) = 
	2 \check{F}_{\bs{\beta}, \bs{0}}
	( 
	\sqrt{n} | \check{\tau}(\bs{\beta}, \bs{0}) - c |
	) - 1. 
	$
	By definition and construction, 
	to obtain the imputed randomization distribution 
	$F_{\mathcal{S}, \bs{T}}(\cdot)$,
	we are essentially conducting ReM with covariate $\bs{X}$ and threshold $a_T$ for the $n$ sampled units, 
	where we use the imputed potential outcomes $\tilde{Y}_i(1)$ and $\tilde{Y}_i(0)$ as the true ones. 
	Below we introduce several finite population quantities. 
	For any $\bs{\beta} > 0$, 
	let $\tilde{Y}_i(t;\bs{\beta}, \bs{0}) = \tilde{Y_i}(t) - \bs{\beta}^\top \bs{C}_i$ be the adjusted imputed potential outcomes, 
	and 
	$\tilde{\tau}_i(\bs{\beta}, \bs{0}) = \tilde{Y}_i(1;\bs{\beta}, \bs{0}) - \tilde{Y}_i(0;\bs{\beta}, \bs{0}) = \tilde{Y}_i(1) - \tilde{Y}_i(0)$ be the adjusted imputed treatment effects, 
	which must be constant $c$ for all $i \in \mathcal{S}$ since the potential outcomes are imputed under the null $H_{c\bs{1}}$, i.e., 
	\begin{align}\label{eq:imp_adj_effect_constant}
		\tilde{\tau}_i(\bs{\beta}, \bs{0}) = c, \text{ for all } i \in \mathcal{S}. 
	\end{align}
	Define $\tilde{s}_1^2(\bs{\beta}, \bs{0})$, $\tilde{s}_0^2(\bs{\beta}, \bs{0})$ and $\tilde{s}^2_{\tau} (\bs{\beta}, \bs{0})$ as the sample variances of adjusted imputed potential outcomes and individual effect for units in $\mathcal{S}$, 
	and $\tilde{s}_{1\mid \bs{X}}^2(\bs{\beta}, \bs{0})$, $\tilde{s}_{0\mid \bs{X}}^2(\bs{\beta}, \bs{0})$, $\tilde{s}^2_{\tau\mid \bs{X}} (\bs{\beta}, \bs{0})$
	and $\tilde{s}^2_{\tau\mid \bs{W}} (\bs{\beta}, \bs{0})$ 
	as the sample variances of their linear projections on covariates. 
	Due to \eqref{eq:imp_adj_effect_constant}, we have 
	$\tilde{s}^2_{\tau} (\bs{\beta}, \bs{0}) = \tilde{s}^2_{\tau\mid \bs{X}} (\bs{\beta}, \bs{0}) = \tilde{s}^2_{\tau\mid \bs{W}} (\bs{\beta}, \bs{0}) =0$.

	Below we consider a fixed sequence of $(\bs{Z}, \bs{T})$ such that regularity conditions in Lemma \ref{lemma:imp_outcome_fp_cond_resem}(i)--(iii) hold exactly (i.e., the convergence in probability there is replaced by usual convergence). 
	Note that we can always view ReM as $\sresem$ with all units sampled to enroll the experiment. 
	By the same logic as the proof of Theorem \ref{thm:dist_reg_general}, we can immediately know that 
	\begin{align}\label{eq:asymp_perm}
		\sqrt{n} \left( \check{\tau}(\bs{\beta}, \bs{0}) - c \right) \mid \bs{Z}, \bs{T} 
		\ \ \dot\sim \ \ 
		\sqrt{\tilde{v}(\bs{\beta}, \bs{0})}
		\left(
		\sqrt{1 - \tilde{r}^2(\bs{\beta}, \bs{0})} \cdot \varepsilon + \sqrt{\tilde{r}^2 (\bs{\beta}, \bs{0})} \cdot L_{K,a_T}
		\right), 
	\end{align}
	where $\tilde{v}(\bs{\beta}, \bs{0})$ and $\tilde{r}^2_{T}(\bs{\beta}, \bs{0})$ are defined similarly as that in \eqref{eq:V_tau_reg} and \eqref{eq:R2_reg} but using the adjusted imputed potential outcomes for the units in $\mathcal{S}$ and have the following simplified forms due to \eqref{eq:imp_adj_effect_constant}: 
	\begin{align*}
		\tilde{v}(\bs{\beta}, \bs{0}) & = r_1^{-1} \tilde{s}_1^2(\bs{\beta}, \bs{0}) + 
		r_0^{-1} \tilde{s}_0^2(\bs{\beta}, \bs{0}) 
		- \tilde{s}^2_{\tau} (\bs{\beta}, \bs{0}) = r_1^{-1} \tilde{s}_1^2(\bs{\beta}, \bs{0}) + 
		r_0^{-1} \tilde{s}_0^2(\bs{\beta}, \bs{0}) 
		= (r_1r_0)^{-1} \tilde{s}_1^2(\bs{\beta}, \bs{0}), \\
		\tilde{r}^2(\bs{\beta}, \bs{0}) & = 
		\frac{r_1^{-1} \tilde{s}^2_{1\mid \bs{X}}(\bs{\beta}, \bs{0}) + r_0^{-1} \tilde{s}^2_{0\mid \bs{X}}(\bs{\beta}, \bs{0}) - \tilde{s}^2_{\tau\mid \bs{X}}(\bs{\beta}, \bs{0})}{
			r_1^{-1} \tilde{s}_1^2(\bs{\beta}, \bs{0}) + 
			r_1^{-1} \tilde{s}_0^2(\bs{\beta}, \bs{0}) 
			- \tilde{s}^2_{\tau} (\bs{\beta}, \bs{0})
		}
		= 
		\frac{r_1^{-1} \tilde{s}^2_{1\mid \bs{X}}(\bs{\beta}, \bs{0}) + r_0^{-1} \tilde{s}^2_{0\mid \bs{X}}(\bs{\beta}, \bs{0})}{
			r_1^{-1} \tilde{s}_1^2(\bs{\beta}, \bs{0}) + 
			r_1^{-1} \tilde{s}_0^2(\bs{\beta}, \bs{0}) 
		}
		= \tilde{s}^2_{1\mid \bs{X}}(\bs{\beta}, \bs{0})/\tilde{s}^2_{1}(\bs{\beta}, \bs{0}). 
	\end{align*}
	By the same logic as the proof Proposition \ref{prop:est_fix_coef_unknown_design} and due to \eqref{eq:imp_adj_effect_constant}, we can then know that the estimation of variance formula and $R^2$ measures based on $\check{T}_{\mathcal{S}}$ and $\mathcal{I}$ satisfies 
	\begin{align*}
		\check V_{\tau \tau}(\bs{\beta}, \bs{\gamma}) &=
		r_1^{-1}\tilde{s}^2_{1}(\bs{\beta}, \bs{0}) + r_0^{-1} \tilde{s}^2_{0}(\bs{\beta}, \bs{0}) - 
		f \tilde{s}^2_{\tau\mid \bs{C}}(\bs{\beta}, \bs{0}) + o_{\Pr}(1)
		= \tilde{v}(\bs{\beta}, \bs{0}) + o_{\Pr}(1), 
		\\
		\check{R}^2_S(\bs{\beta}, \bs{\gamma}) & =
		\begin{cases}
			(1-f) \tilde{v}^{-1} (\bs{\beta}, \bs{0})
			\tilde{s}^2_{\tau\mid \bs{W}}(\bs{\beta}, \bs{0}) + o_{\Pr}(1)
			= o_{\Pr}(1),  & \text{if both $\bs{W}$ and $a_S$ are known}, \\
			0, & \text{otherwise}, 
		\end{cases}
		\\
		\check{R}^2_T(\bs{\beta}, \bs{\gamma}) &=
		\tilde{v}^{-1} (\bs{\beta}, \bs{0}) \left\{ r_1^{-1}\tilde{s}^2_{1\mid  \bs{X}}(\bs{\beta}, \bs{0}) +r_0^{-1}\tilde{s}^2_{0\mid  \bs{X}}(\bs{\beta}, \bs{0}) -
		\tilde{s}^2_{\tau\mid \bs{X}}(\bs{\beta}, \bs{0})
		\right\} + o_{\Pr}(1)
		= \tilde{r}^2(\bs{\beta}, \bs{0}) + o_{\Pr}(1). 
	\end{align*}
	Therefore, 
	from Lemma \ref{lemma:est_cdf_trans_converge},  
	the imputed distribution of the test statistic has the following asymptotic distribution: 
	\begin{align*}
		\check{F}_{\bs{\beta}, \bs{0}}\left(\sqrt{n}\left| \check{\tau}(\bs{\beta}, \bs{0}) - c \right|\right) 
		\mid \bs{Z}, \bs{T}
		\ \dot\sim \  
		\Psi_{\tilde{v}(\bs{\beta}, \bs{0}), \tilde{r}^2(\bs{\beta}, \bs{0})}\left( |
		\psi_{\tilde{v}(\bs{\beta}, \bs{0}), \tilde{r}^2(\bs{\beta}, \bs{0})}| \right)
		\sim \Unif(0,1)
	\end{align*}
	where $\Psi_{\tilde{v}(\bs{\beta}, \bs{0}), \tilde{r}^2(\bs{\beta}, \bs{0})}(\cdot)$ is the distribution function of the asymptotic distribution on the right hand side of \eqref{eq:asymp_perm}, 
	$\psi_{\tilde{v}(\bs{\beta}, \bs{0}), \tilde{r}^2(\bs{\beta}, \bs{0})}$ is a random variable with distribution function $\Psi_{\tilde{v}(\bs{\beta}, \bs{0}), \tilde{r}^2(\bs{\beta}, \bs{0})}(\cdot)$, 
	and 
	the last $\sim$ sign follows from Lemma \ref{lemma:stoch_uniform}.  
	Consequently, for any $\alpha\in (0,1)$, $F_{\mathcal{S}, \bs{T}}^{-1}(1-\alpha) \converge 1-\alpha$, 
	where $F_{\mathcal{S}, \bs{T}}^{-1}(\cdot)$
    is the quantile function for the imputed randomization distribution of the test statistic. 
	
	From the above and Lemma \ref{lemma:imp_outcome_fp_cond_resem}, as well as the property of convergence in probability \citep[e.g.,][Theorem 2.3.2]{durrett2019probability}, 
	we can derive that under $\sresem$, for any $\alpha\in (0,1)$, $F_{\mathcal{S}, \bs{T}}^{-1}(1-\alpha) \convergep 1-\alpha$.
\end{proof}

\begin{proof}[\bf The asymptotic validity of the randomization $p$-value]
    To prove the asymptotic validity under ReSEM, it suffices to prove that under $\sresem$. 
	We prove only the case with fixed $\bs{\beta}$. The proof for the case with estimated $\hat{\bs{\beta}}$ is very similar and is thus omitted. 
	
    By definition,
	$
	G_{\mathcal{S}, \bs{T}}(c) \ge 1 - F_{\mathcal{S}, \bs{T}}(c)
	$
	for any $c \in \mathbb{R}$.
	Thus, for any $\alpha\in (0,1)$, 
	\begin{align*}
		\Pr\big( p_{\mathcal{S}, \bs{T}, g} \le \alpha \mid \sresem \big)
		& = 
		\Pr\big\{ G_{\mathcal{S}, \bs{T}}\big( g(\bs{T}_{\mathcal{S}}, \mathcal{I}) \big) \le \alpha \mid \sresem \big\}
		\le 
		\Pr\big\{ 1 - F_{\mathcal{S}, \bs{T}}\big( g(\bs{T}_{\mathcal{S}}, \mathcal{I}) \big) \le \alpha \mid \sresem \big\} 
		\\
		& = \Pr\big\{ F_{\mathcal{S}, \bs{T}}\big( g(\bs{T}_{\mathcal{S}}, \mathcal{I}) \big) \ge 1-\alpha \mid \sresem \big\} 
		= \Pr\big\{ g(\bs{T}_{\mathcal{S}}, \mathcal{I}) \ge F_{\mathcal{S}, \bs{T}}^{-1}(1-\alpha) \mid \sresem \big\}.
	\end{align*}
	From the above discussion on the imputed randomization distribution, we can know that 
    $
		F_{\mathcal{S}, \bs{T}}^{-1}(1-\alpha) \convergep 1-\alpha 
	$
	under $\sresem$. 
	Consequently, for any constant $\delta > 0$,
	\begin{align}\label{eq:bound_rand_pval_proof1}
		& \quad \ \Pr\big( p_{\mathcal{S}, \bs{T}, g} \le \alpha \mid \sresem \big)
		\nonumber
		\\
		& \le \Pr\big\{ g(\bs{T}_{\mathcal{S}}, \mathcal{I}) \ge F_{\mathcal{S}, \bs{T}}^{-1}(1-\alpha), \ 
		\big| F_{\mathcal{S}, \bs{T}}^{-1}(1-\alpha) - (1-\alpha)\big| \le \delta
		\mid \sresem \big\}
		\nonumber
		\\
		& \quad \ 
		+ 
		\Pr\big\{ g(\bs{T}_{\mathcal{S}}, \mathcal{I}) \ge F_{\mathcal{S}, \bs{T}}^{-1}(1-\alpha),  \ \big| F_{\mathcal{S}, \bs{T}}^{-1}(1-\alpha) - (1-\alpha)\big| > \delta \mid \sresem \big\}
		\nonumber
		\\
		& \le \Pr\big\{ g(\bs{T}_{\mathcal{S}}, \mathcal{I}) \ge 1-\alpha -\delta \mid \sresem \big\}
		+ 
		\Pr\big\{ \big| F_{\mathcal{S}, \bs{T}}^{-1}(1-\alpha) - (1-\alpha)\big| > \delta \mid \sresem \big\}. 
	\end{align}
	From the previous discussion on the true sampling distribution, 
	for any $\alpha\in (0,1)$, 
	\begin{align*}
		\limsup_{N\rightarrow \infty} \Pr(g(\bs{T}_{\mathcal{S}}, \mathcal{I})  \ge 1-\alpha \mid \sresem)
		= 
		\limsup_{N\rightarrow \infty} \Pr(1-g(\bs{T}_{\mathcal{S}}, \mathcal{I})  \le \alpha \mid \sresem)
		& \le \alpha. 
	\end{align*}
	Thus, letting $N\rightarrow \infty$ in \eqref{eq:bound_rand_pval_proof1}, we have 
	\begin{align*}
		& \quad \ \limsup_{N\rightarrow \infty}\Pr\big( p_{\mathcal{S}, \bs{T}, g} \le \alpha \mid \sresem \big)
		\\
		& \le \limsup_{N\rightarrow \infty}\Pr\big\{ g(\bs{T}_{\mathcal{S}}, \mathcal{I}) \ge 1-\alpha -\delta \mid \sresem \big\}
		+ 
		\lim_{N\rightarrow \infty}\Pr\big\{ \big| F_{\mathcal{S}, \bs{T}}^{-1}(1-\alpha) - 1-\alpha\big| > \delta \mid \sresem \big\}\\
		& \le \alpha + \delta.
	\end{align*}
	Because the above inequality holds for any $\delta>0,$ we must have 
	$
		\limsup_{N\rightarrow \infty}\Pr\big( p_{\mathcal{S}, \bs{T}, g} \le \alpha \mid \sresem \big)
		\le \alpha.
	$
	Therefore, $p_{\mathcal{S}, \bs{T}, g}$ is an asymptotically valid $p$-value for testing Neyman's null $\bar{H}_c$ under $\sresem$, i.e., Theorem \ref{thm:frt_weak_null} holds under $\sresem$.  
\end{proof}

\section{Stratified sampling and blocking}\label{sec:proof_strata}

In this section, we focus on survey experiments based on stratified sampling and blocked treatment assignment. 
Below we first introduce several notations that will be used throughout this section. 
Similar to the discussion for CRSE and ReSEM, for the finite population of $N$ units, 
we introduce finite population averages $\bar{Y}(1), \bar{Y}(0)$ and $\bar{\bs{W}}$, 
and finite population variances and covariances $S_1^2, S_0^2, S^2_{\tau}, \bs{S}^2_{\bs{W}}, S^2_{z\setminus \bs{W}}$ and $S^2_{\tau \setminus \bs{W}}$. 
Recall that $n = N\sum_{j=1}^{J+1} \pi_j f_j$ denotes the total number of sampled units from all strata, and $f = n/N$ denotes the overall proportion of sampled units. 
For $z=0,1$, 
we introduce $n_z = N \sum_{j=1}^{J+1} N \pi_j f_j r_{zj}$ to denote the total number of units under treatment arm $z$ from all strata, 
and $r_z = n_z/n$ to denote the overall proportion of units under treatment arm $z$ among all sampled units. 
Moreover, we define $R_S^2$ and $R_T^2$ by the same formulas as in \eqref{eq:R2}. 

\subsection{Regularity conditions for finite population asymptotics}

To derive the asymptotic distribution of the difference-in-means estimator under stratified sampling and blocking, we invoke the following regularity conditions. 
For units in each stratum $j$ ($1\le j \le J+1$), 
let $S_{1j}^{2}, S_{0j}^{2}$ and $S_{\tau j}^{2}$ be the finite population variances of treatment potential outcome, control potential outcome and individual treatment effect, 
and $\bar{Y}_j(1)$ and $\bar{Y}_j(0)$ be the average treatment and control potential outcomes.  

\begin{condition}\label{cond:categorical}
As $N \rightarrow \infty$, the sequence of finite populations satisfies that, for $1\le j \le J+1$, 
\begin{enumerate}[label=(\roman*)]
\item $f_j$ has a limit, and $r_{1z}$ and $r_{0j}$ have positive limits;
\item the finite population variances $S_{1j}^{2}, S_{0j}^{2}, S_{\tau j}^{2}$ have limiting values;
\item $\max _{i: \tilde{W}_i = j}\{Y_{i}(z)-\bar{Y}_j(z)\}^{2} / (N\pi_j f_j) \rightarrow 0$ for $z = 0, 1$; 
\item $\pi_j$ has a positive limit for all $j$;
\item $\bar{Y}_j(z) = O(1)$ for $z = 0, 1$. 
\end{enumerate}
\end{condition}

Condition \ref{cond:categorical}(i)-(iii) are analogous to Condition \ref{cond:fp}(i)--(iv). 
Condition \ref{cond:categorical}(iv) requires that there is non-negligible proportion of units within each strata, which is equivalent to that the finite population covariance matrix $\bs{S}^2_{\bs{W}}$ has a nonsingular limit, as proved below, 
and 
(v) imposes mild conditions on the average potential outcomes within each stratum. 

\begin{proof}[\bf Comment on Condition \ref{cond:categorical}]
Below we prove that Condition \ref{cond:categorical}(iv) is equivalent to that the finite population covariance matrix $\bs{S}^2_{\bs{W}}$ has a nonsingular limit. 
It suffices to prove that for the finite population of size $N$, $\bs{S}^2_{\bs{W}}$ is nonsingular if and only if $\pi_j >0$ for all $1\le j \le J+1$. 
By definition, $\bs{S}^2_{\bs{W}}$ is positive semi-definite and has the following equivalent forms:
\begin{align*}
   \bs{S}^2_{\bs{W}} 
   & = \frac{1}{N-1} \sum_{i=1}^N (\bs{W}_i - \bar{\bs{W}}) (\bs{W}_i - \bar{\bs{W}})^\top
   = 
   \frac{1}{N-1} \sum_{i=1}^N \bs{W}_i \bs{W}_i^\top 
   - \frac{N}{N-1} \sum_{i=1}^N \bar{\bs{W}} \bar{\bs{W}}^\top\\
   & = 
   \frac{N}{N-1} 
   \left\{
   \begin{pmatrix}
       \pi_1 & 0 & \cdots & 0 \\
       0 & \pi_2 & \cdots & 0 \\
       \vdots & \vdots & \ddots & \vdots\\
       0 & 0 & \cdots & \pi_{J}
   \end{pmatrix}
   - 
   \begin{pmatrix}
   \pi_1\\
   \pi_2\\
   \vdots
   \\
   \pi_J
   \end{pmatrix}
   \begin{pmatrix}
   \pi_1 & \pi_2 & \cdots & \pi_J
   \end{pmatrix}
   \right\}. 
\end{align*}

First, suppose that $\pi_{j_0} = 0$ for some $1\le j_0 \le J+1$. If $1\le j_0 \le J$, then $\bs{S}^2_{\bs{W}}$ has a zero diagonal element and thus must be singular. 
Otherwise, $\pi_{J+1} = 0$, and $\sum_{j=1}^J \pi_J = 1$, under which $\bs{1}_J^\top \bs{S}^2_{\bs{W}} \bs{1}_J = 0$, where $\bs{1}_J$ is a $J$ dimensional vector with all of its elements being 1. 
This implies that $\bs{S}^2_{\bs{W}}$ is singular.  

Second, suppose that $\pi_j > 0$ for all $1\le j \le J+1$. 
Then for any $J$ dimensional vector $\bs{c}\ne 0$, we have 
\begin{align*}
    \bs{c}^\top \frac{N-1}{N} \bs{S}^2_{\bs{W}} \cdot \bs{c}
    & = 
    \sum_{j=1}^J \pi_j c_j^2 - \left( \sum_{j=1}^J \pi_j c_j \right)^2 
    = \pi_{J+1} \cdot \sum_{j=1}^J \pi_j c_j^2 + 
    \left( \sum_{j=1}^J \pi_j \right)  \cdot \left( \sum_{j=1}^J \pi_j c_j^2 \right) - \left( \sum_{j=1}^J \pi_j c_j \right)^2 \\
    & \ge \pi_{J+1} \cdot \sum_{j=1}^J \pi_j c_j^2 > 0, 
\end{align*}
where the second last inequality holds due to Cauchy–Schwarz inequality, and the last inequality holds because $\pi_j > 0$ for all $1\le j \le J+1$ and $\bs{c} \ne 0$. 

From the above, we can know that $\bs{S}^2_{\bs{W}}$ is nonsingular if and only if $\pi_j >0$ for all $1\le j \le J+1$. 
Therefore, Condition \ref{cond:categorical}(iv) is equivalent to that $\bs{S}^2_{\bs{W}}$ has a nonsingular limit. 
\end{proof}

\subsection{Technical lemmas}

The following lemmas will be utilized for deriving the asymptotic properties of survey experiments with stratified sampling and blocking. 
For descriptive convenience, 
we introduce $f_{\max} = \max_{1\le j \le J+1} f_j$ and $f_{\min} = \min_{1\le j \le J+1} f_j$, 
and define analogously 
$\pi_{\max}, \pi_{\min}$, $r_{1, \max}, r_{1, \min}$ and $r_{0,\max}, r_{0, \min}$. 

\begin{lemma}\label{lemma:prop_weight_diff}
    Let $(a_1, a_2, \ldots, a_{J+1})$ and $(b_1, b_2, \ldots, b_{J+1})$ be two sets of nonnegative constants. 
    Define 
    \begin{align*}
        \Delta_j = \frac{a_j b_j}{ \sum_{k=1}^{J+1} a_k b_k} - 
        \frac{a_j}{ \sum_{k=1}^{J+1} a_k}, \qquad (j= 1, 2, \ldots, J+1), 
    \end{align*}
    Then 
    \begin{align*}
        \left( \frac{a_{\min}}{\sum_{k=1}^{J+1}a_k} \right)^2 
        \left( 1 -  \frac{b_{\min}}{b_{\max}} \right)
        \le 
        \max_{1\le j \le J} |\Delta_j| \le \max_{1\le j \le J+1} |\Delta_j| \le  \frac{b_{\max}}{b_{\min}} - 1, 
    \end{align*}
    where $a_{\min} = \min_{1\le j \le J+1} a_j$, 
    $b_{\min} = \min_{1\le j \le J+1} b_j$
    and 
    $b_{\max} = \max_{1\le j \le J+1} b_j$. 
\end{lemma}

\begin{lemma}\label{lemma:M_S_cat}
Under stratified sampling and blocking, 
\begin{align*}
     \pi_{\min}^2\left( 1 - \frac{f_{\min}}{f_{\max}} \right) 
    \le 
    \big\| \hat{\bs{\delta}}_{\bs{W}} \big\|_{\infty} 
    \le 
    \frac{f_{\max}}{f_{\min}} - 1
\end{align*}
and 
\begin{align*}
    \sqrt{ \frac{\lambda_{\min}(\bs{S}^2_{\bs{W}})}{J}  \left( \frac{1}{n} - \frac{1}{N} \right) M_S}
    \le 
    \big\| \hat{\bs{\delta}}_{\bs{W}} \big\|_{\infty}
    \le 
    \sqrt{ \lambda_{\max} (\bs{S}^2_{\bs{W}}) \left( \frac{1}{n} - \frac{1}{N} \right) M_S} \ , 
\end{align*}
where $\lambda_{\min}(\bs{S}^2_{\bs{W}})$ and $\lambda_{\max} (\bs{S}^2_{\bs{W}})$ denote the smallest and largest eigenvalues of $\bs{S}^2_{\bs{W}}$. 
\end{lemma}

\begin{lemma}\label{lemma:M_T_cat}
Under stratified sampling and blocking, for $z=0,1$, 
\begin{align*}
    \left( \frac{\pi_{\min} f_{\min}}{\sum_{k=1}^{J+1}\pi_k f_k} \right)^2 
    \left( 1 -  \frac{r_{z,\min}}{r_{z,\max}} \right)
    \le \big\| \hat{\bs{\tau}}_{\bs{W}}\big\|_{\infty} 
    \le 
    \frac{1}{r_{1-z, \min}} \left( \frac{r_{z, \max}}{r_{z, \min}} - 1 \right), 
\end{align*}
and 
\begin{align*}
    \sqrt{ \frac{\lambda_{\min}(\bs{s}_{\bs{W}}^2)}{J} \frac{M_T}{n} }
    \le \big\| \hat{\bs{\tau}}_{\bs{W}}\big\|_{\infty} 
    \le 
    \sqrt{\frac{\lambda_{\max}(\bs{s}_{\bs{W}}^2)}{r_{1, \min}r_{0, \min}} \frac{M_T}{n}}, 
\end{align*}
where 
$\bs{s}^2_{\bs{W}}$ is the sample covariance matrix of the covariate $\bs{W}$ for sampled units, 
and 
$\lambda_{\min}(\bs{s}^2_{\bs{W}})$ and $\lambda_{\max} (\bs{s}^2_{\bs{W}})$ denote the smallest and largest eigenvalues of $\bs{s}^2_{\bs{W}}$. 
\end{lemma}

\begin{lemma}\label{lemma:samp_cov_cat}
If Condition \ref{cond:categorical} holds and $\lim_{N\rightarrow \infty}f_{\max}/f_{\min} \rightarrow 1$, then 
$$
\lim_{N\rightarrow \infty} \lambda_{\min}(\bs{S}^2_{\bs{W}}) \le 
\liminf_{N\rightarrow \infty} \lambda_{\min}(\bs{s}^2_{\bs{W}}) \le \limsup_{N\rightarrow \infty} \lambda_{\max}(\bs{s}^2_{\bs{W}} ) \le 2 \lim_{N\rightarrow \infty} \lambda_{\max}(\bs{S}^2_{\bs{W}}).
$$
\end{lemma}

\begin{lemma}\label{lemma:cat_M_ST_cong_zero}
Under Condition \ref{cond:categorical} and stratified sampling and blocking, 
both $(1-f)M_S$ and $M_T$ are of order $o(1)$ if and only if 
\begin{align}\label{eq:ratio_converge}
    \frac{f_{\max}}{f_{\min}} - 1 =  o(n^{-1/2}), 
    \ \ 
    \frac{r_{1,\max}}{r_{1,\min}} - 1 =  o(n^{-1/2}), 
    \ \ 
    \frac{r_{0,\max}}{r_{0,\min}} - 1 = o(n^{-1/2}). 
\end{align}
\end{lemma}

\subsection{Proofs of the lemmas}

\begin{proof}[Proof of Lemma \ref{lemma:prop_weight_diff}]
By definition, 
for $1 \le j \le J+1$, 
$\Delta_j$ has the following equivalent forms:
\begin{align}\label{eq:Delta_j}
    \Delta_j = \frac{a_j b_j}{ \sum_{k=1}^{J+1} a_k b_k} - 
    \frac{a_j}{ \sum_{k=1}^{J+1} a_k}
    = 
    \frac{
    \sum_{k=1}^{J+1} a_k a_j b_j - \sum_{k=1}^{J+1} a_k b_k a_j
    }{\sum_{k=1}^{J+1} a_k b_k \cdot \sum_{k=1}^{J+1} a_k}
    = 
    \frac{
    \sum_{k=1}^{J+1} a_k a_j ( b_j - b_k)
    }{\sum_{k=1}^{J+1} a_k b_k \cdot \sum_{k=1}^{J+1} a_k}.
\end{align}   

First, we consider the upper bound of $\max_{1\le j \le J+1} |\Delta_j|$. 
From \eqref{eq:Delta_j}, for $1 \le j \le J+1$, we have 
\begin{align*}
    |\Delta_j| \le 
    \frac{
    \sum_{k=1}^{J+1} a_k a_j |b_j - b_k|
    }{\sum_{k=1}^{J+1} a_k b_k \cdot \sum_{k=1}^{J+1} a_k}
    \le 
    \frac{a_j (b_{\max} - b_{\min}) \sum_{k=1}^{J+1} a_k}{
    \sum_{k=1}^{J+1} a_k b_{\min} \cdot \sum_{k=1}^{J+1} a_k}
    = 
    \frac{a_j}{\sum_{k=1}^{J+1} a_k} \left( \frac{b_{\max}}{b_{\min}} - 1 \right). 
\end{align*}
Thus, 
\begin{align*}
    \max_{1\le j \le J+1}|\Delta_j| \le \frac{\max_{1\le j \le J+1}a_j}{\sum_{k=1}^{J+1} a_k} \left( \frac{b_{\max}}{b_{\min}} - 1 \right)
    \le \frac{b_{\max}}{b_{\min}} - 1. 
\end{align*}

Second, we consider the lower bound of $\max_{1\le j \le J} |\Delta_j|$. 
Below we consider two cases, separately. 
If $b_{J+1} \ne b_{\max}$, then there must exist $1\le j_0\le J$ such that $b_{j_0} = b_{\max}$. 
From \eqref{eq:Delta_j}, we then have 
\begin{align*}
    \max_{1\le j \le J} |\Delta_j| \ge 
    |\Delta_{j_0}| = 
    \frac{
    \sum_{k=1}^{J+1} a_k a_{j_0} ( b_{j_0} - b_k)
    }{\sum_{k=1}^{J+1} a_k b_k \cdot \sum_{k=1}^{J+1} a_k}
    \ge 
    \frac{
    a_{\min}^2 ( b_{\max} - b_{\min})
    }{\sum_{k=1}^{J+1} a_k b_{\max} \cdot \sum_{k=1}^{J+1} a_k}
    = 
    \left( \frac{a_{\min}}{\sum_{k=1}^{J+1}a_k} \right)^2 
    \left( 1 - \frac{b_{\min}}{b_{\max}} \right). 
\end{align*}
Otherwise, $b_{J+1} = b_{\max}$, and thus there must exist $1\le j_0 \le J$ such that $b_{j_0} = b_{\min}$. From \eqref{eq:Delta_j}, we then have 
\begin{align*}
    \max_{1\le j \le J} |\Delta_j| \ge 
    |\Delta_{j_0}| = 
    \frac{
    \sum_{k=1}^{J+1} a_k a_{j_0} ( b_k - b_{j_0})
    }{\sum_{k=1}^{J+1} a_k b_k \cdot \sum_{k=1}^{J+1} a_k}
    \ge 
    \frac{
    a_{\min}^2 ( b_{\max} - b_{\min})
    }{\sum_{k=1}^{J+1} a_k b_{\max} \cdot \sum_{k=1}^{J+1} a_k}
    = 
    \left( \frac{a_{\min}}{\sum_{k=1}^{J+1}a_k} \right)^2 
    \left( 1 - \frac{b_{\min}}{b_{\max}} \right). 
\end{align*}

From the above, Lemma \ref{lemma:prop_weight_diff} holds. 
\end{proof}

\begin{proof}[Proof of Lemma \ref{lemma:M_S_cat}]
First, we bound $\| \hat{\bs{\delta}}_{\bs{W}} \|_{\infty}$ using the ratio between $f_{\max}$ and $f_{\min}$. 
By definition, for $1\le j \le J$, the $j$th coordinate of $\hat{\bs{\delta}}_{\bs{W}}$ has the following equivalent forms: 
\begin{align*}
    [\hat{\bs{\delta}}_{\bs{W}}]_{(j)} & = 
    [\bar{\bs{W}}_{\mathcal{S}}]_{(j)} - [\bar{\bs{W}}]_{(j)}
    = \frac{\pi_j f_j}{\sum_{k=1}^{J+1} \pi_k f_k} - \pi_j
    = \frac{\pi_j f_j}{\sum_{k=1}^{J+1} \pi_k f_k} - \frac{\pi_j}{\sum_{k=1}^{J+1} \pi_k}. 
\end{align*}
Applying Lemma \ref{lemma:prop_weight_diff} with $a_j = \pi_j$ and $b_j = f_j$, we then have 
\begin{align*}
    \big\| \hat{\bs{\delta}}_{\bs{W}} \big\|_{\infty} 
    & = \max_{1\le j \le J} [\hat{\bs{\delta}}_{\bs{W}}]_{(j)}
    \ge 
    \left( \frac{\pi_{\min}}{\sum_{k=1}^{J+1}\pi_k} \right)^2 
    \left( 1 -  \frac{f_{\min}}{f_{\max}} \right)
    = \pi_{\min}^2 \left( 1 -  \frac{f_{\min}}{f_{\max}} \right), 
\end{align*}
and 
\begin{align*}
    \big\| \hat{\bs{\delta}}_{\bs{W}} \big\|_{\infty} 
    & = \max_{1\le j \le J} [\hat{\bs{\delta}}_{\bs{W}}]_{(j)}
    \le \frac{f_{\max}}{f_{\min}} - 1. 
\end{align*}

Second, we bound $\| \hat{\bs{\delta}}_{\bs{W}} \|_{\infty}$ using the Mahalanobis distance $M_S$. 
By definition, $(n^{-1} - N^{-1}) M_S = \hat{\bs{\delta}}_{\bs{W}}^\top (\bs{S}^2_{\bs{W}})^{-1} \hat{\bs{\delta}}_{\bs{W}}$, and thus it can be bounded below by 
\begin{align*}
    \left( \frac{1}{n} - \frac{1}{N} \right) M_S & =  \hat{\bs{\delta}}_{\bs{W}}^\top \left( \bs{S}^2_{\bs{W}} \right)^{-1} \hat{\bs{\delta}}_{\bs{W}}
    \ge \lambda_{\max}^{-1} (\bs{S}^2_{\bs{W}}) \big\|\hat{\bs{\delta}}_{\bs{W}} \big\|_2^2 \ge \lambda_{\max}^{-1} (\bs{S}^2_{\bs{W}}) \big\|\hat{\bs{\delta}}_{\bs{W}} \big\|_{\infty}^2
\end{align*}
and bounded above by 
\begin{align*}
    \left( \frac{1}{n} - \frac{1}{N} \right) M_S & =  \hat{\bs{\delta}}_{\bs{W}}^\top \left( \bs{S}^2_{\bs{W}} \right)^{-1} \hat{\bs{\delta}}_{\bs{W}}
    \le \lambda_{\min}^{-1} (\bs{S}^2_{\bs{W}}) \big\|\hat{\bs{\delta}}_{\bs{W}} \big\|_2^2 \le J\lambda_{\min}^{-1} (\bs{S}^2_{\bs{W}}) \big\|\hat{\bs{\delta}}_{\bs{W}} \big\|_{\infty}^2. 
\end{align*}

From the above, we can immediately derive Lemma \ref{lemma:M_S_cat}. 
\end{proof}

\begin{proof}[Proof of Lemma \ref{lemma:M_T_cat}]
First, we bound $\| \hat{\bs{\tau}}_{\bs{W}}\|_{\infty}$ using the ratio between $r_{z,\min}$ and $r_{z,\max}$ for $z=0,1$. 
By definition, for $z=0,1$ and $1\le j \le J$, 
\begin{align*}
    [ \bar{\bs{W}}_z - \bar{\bs{W}}_{\mathcal{S}}]_{(j)} 
    & = 
    \frac{\pi_j f_j r_{zj}}{\sum_{k=1}^{J+1} \pi_k f_k r_{zk}} 
    - 
    \frac{\pi_j f_j}{\sum_{k=1}^{J+1} \pi_k f_k}. 
\end{align*}
Applying Lemma \ref{lemma:prop_weight_diff} with $a_j = \pi_j f_j$ and $b_j = r_{zj}$, we have 
\begin{align*}
    \big\|\bar{\bs{W}}_z - \bar{\bs{W}}_{\mathcal{S}}\big\|_{\infty} 
    \ge 
    \left( \frac{\min_{1\le j \le J+1}\pi_j f_j}{\sum_{k=1}^{J+1}\pi_k f_k} \right)^2 
    \left( 1 -  \frac{r_{z,\min}}{r_{z,\max}} \right)
    \ge 
    \left( \frac{\pi_{\min} f_{\min}}{\sum_{k=1}^{J+1}\pi_k f_k} \right)^2 
    \left( 1 -  \frac{r_{z,\min}}{r_{z,\max}} \right), 
\end{align*}
and 
\begin{align*}
    \big\|\bar{\bs{W}}_z - \bar{\bs{W}}_{\mathcal{S}}\big\|_{\infty} 
    \le 
    \frac{r_{z, \max}}{r_{z, \min}} - 1. 
\end{align*}
By some algebra, we can verify that 
$\bar{\bs{W}}_1 - \bar{\bs{W}}_{\mathcal{S}} = r_0 \hat{\bs{\tau}}_{\bs{W}}$ and $\bar{\bs{W}}_0 - \bar{\bs{W}}_{\mathcal{S}} = -r_1 \hat{\bs{\tau}}_{\bs{W}}$. 
Besides, 
\begin{align*}
    r_z = 
    \frac{\sum_{j=1}^{J+1} \pi_j f_j r_{zj}}{\sum_{j=1}^{J+1} \pi_j f_j} \ge r_{z, \min}, \qquad (z=0,1). 
\end{align*}
These immediately imply that, for $z=0,1$, 
\begin{align*}
    \big\| \hat{\bs{\tau}}_{\bs{W}}\big\|_{\infty} 
    = \frac{1}{r_{1-z}}\big\| \bar{\bs{W}}_z - \bar{\bs{W}}_{\mathcal{S}}  \big\|_{\infty} 
    \ge 
    \left( \frac{\pi_{\min} f_{\min}}{\sum_{k=1}^{J+1}\pi_k f_k} \right)^2 
    \left( 1 -  \frac{r_{z,\min}}{r_{z,\max}} \right), 
\end{align*}
and 
\begin{align*}
    \big\| \hat{\bs{\tau}}_{\bs{W}}\big\|_{\infty} 
    = \frac{1}{r_{1-z}}\big\| \bar{\bs{W}}_z - \bar{\bs{W}}_{\mathcal{S}}  \big\|_{\infty} 
    \le 
    \frac{1}{r_{1-z, \min}} \left( \frac{r_{z, \max}}{r_{z, \min}} - 1 \right). 
\end{align*}

Second, we bound $\hat{\bs{\tau}}_{\bs{W}}$ using the Mahalanobis distance $M_T$. 
By definition, $M_T/(nr_1r_0) = \hat{\bs{\tau}}_{\bs{W}}^\top (\bs{s}_{\bs{W}}^2)^{-1} \hat{\bs{\tau}}_{\bs{W}}$. 
This immediately implies that 
\begin{align*}
    M_T/(nr_1r_0) & = \hat{\bs{\tau}}_{\bs{W}}^\top (\bs{s}_{\bs{W}}^2)^{-1} \hat{\bs{\tau}}_{\bs{W}} 
    \ge \lambda_{\max}^{-1}(\bs{s}_{\bs{W}}^2) \big\| \hat{\bs{\tau}}_{\bs{W}} \big\|_2^2 
    \ge \lambda_{\max}^{-1}(\bs{s}_{\bs{W}}^2) \big\| \hat{\bs{\tau}}_{\bs{W}} \big\|_{\infty}^2, 
\end{align*}
and 
\begin{align*}
    M_T/(nr_1r_0) & = \hat{\bs{\tau}}_{\bs{W}}^\top (\bs{s}_{\bs{W}}^2)^{-1} \hat{\bs{\tau}}_{\bs{W}} 
    \le \lambda_{\min}^{-1}(\bs{s}_{\bs{W}}^2) \big\| \hat{\bs{\tau}}_{\bs{W}} \big\|_2^2 
    \le J \lambda_{\min}^{-1}(\bs{s}_{\bs{W}}^2) \big\| \hat{\bs{\tau}}_{\bs{W}} \big\|_{\infty}^2. 
\end{align*}
Consequently, we have 
\begin{align*}
    \big\| \hat{\bs{\tau}}_{\bs{W}} \big\|_{\infty}^2 \le  \lambda_{\max}(\bs{s}_{\bs{W}}^2) \frac{M_T}{nr_1r_0}
    \le \lambda_{\max}(\bs{s}_{\bs{W}}^2) \frac{M_T}{nr_{1, \min}r_{0, \min}}, 
\end{align*}
and 
\begin{align*}
    \big\| \hat{\bs{\tau}}_{\bs{W}} \big\|_{\infty}^2 \ge \frac{\lambda_{\min}(\bs{s}_{\bs{W}}^2)}{J} \frac{M_T}{nr_1r_0} 
    \ge 
    \frac{\lambda_{\min}(\bs{s}_{\bs{W}}^2)}{J} \frac{M_T}{n}.  
\end{align*}

From the above, we can immediately derive Lemma \ref{lemma:M_T_cat}. 
\end{proof}

\begin{proof}[Proof of Lemma \ref{lemma:samp_cov_cat}]
From Lemma \ref{lemma:M_S_cat}, when $f_{\max}/f_{\min} \rightarrow 1$, we have $\|\hat{\bs{\delta}}_{\bs{W}}\|_{\infty} \rightarrow 0$. 
This implies that $\bar{\bs{W}}_{\mathcal{S}} - \bar{\bs{W}} = o(1)$, 
and thus the sample covariance matrix of $\bs{W}$ satisfies that 
\begin{align*}
    \bs{s}^2_{\bs{W}} 
    & = \frac{1}{n-1} \sum_{i \in \mathcal{S}} \bs{W}_i \bs{W}_i^\top - \frac{n}{n-1} \bar{\bs{W}}_{\mathcal{S}} \bar{\bs{W}}_{\mathcal{S}}^\top 
    = 
    \frac{n}{n-1} 
    \left\{ 
    \text{diag}(\bar{\bs{W}}_{\mathcal{S}}) - \bar{\bs{W}}_{\mathcal{S}} \bar{\bs{W}}_{\mathcal{S}}^\top \right\}\\
    & = 
    \frac{n}{n-1} 
    \left\{ 
    \text{diag}(\bar{\bs{W}}) - \bar{\bs{W}} \bar{\bs{W}}^\top \right\} + o(1)
    = 
    \frac{n}{n-1} \frac{N-1}{N} \bs{S}_{\bs{W}}^2 + o(1).
\end{align*}  
This implies that $(n-1)/n\cdot \bs{s}_{\bs{W}}^2 = \bs{S}_{\bs{W}}^2 + o(1)$, and thus 
\begin{align*}
    \limsup_{N\rightarrow \infty} \lambda_{\max}(\bs{s}^2_{\bs{W}} ) \le 
    2 \limsup_{N\rightarrow \infty} \frac{n-1}{n}\lambda_{\max}(\bs{s}^2_{\bs{W}} ) = 2 \lim_{N\rightarrow \infty} \lambda_{\max}(\bs{S}^2_{\bs{W}}), 
\end{align*}
and 
\begin{align*}
    \liminf_{N\rightarrow \infty} \lambda_{\min}(\bs{s}^2_{\bs{W}} ) \ge 
    \liminf_{N\rightarrow \infty} \frac{n-1}{n} \lambda_{\min}(\bs{s}^2_{\bs{W}} )
    = \lim_{N\rightarrow \infty} \lambda_{\min}(\bs{S}^2_{\bs{W}}).
\end{align*}
Therefore, Lemma \ref{lemma:samp_cov_cat} holds. 
\end{proof}

\begin{proof}[Proof of Lemma \ref{lemma:cat_M_ST_cong_zero}]
First, we assume that both $(1-f)M_S$ and $M_T$ are of order $o(1)$. 
From Lemma \ref{lemma:M_S_cat}, 
\begin{align*}
    1 - \frac{f_{\min}}{f_{\max}}
    \le 
    \frac{1}{\pi_{\min}^2}
    \big\| \hat{\bs{\delta}}_{\bs{W}} \big\|_{\infty} 
    \le 
    \frac{1}{\pi_{\min}^2} \sqrt{ \lambda_{\max} (\bs{S}^2_{\bs{W}}) \frac{1}{n} \left( 1 - f \right) M_S}
    = o(n^{-1/2}), 
\end{align*}
which implies that 
\begin{align}\label{eq:f_max_min}
    \frac{f_{\max}}{f_{\min}} - 1
    & = 
    \frac{1 - f_{\min}/f_{\max}}{f_{\min}/f_{\max}} = o(n^{-1/2}). 
\end{align}
From Lemmas \ref{lemma:M_T_cat} and \ref{lemma:samp_cov_cat}, for $z=0,1$, 
\begin{align*}
    1 -  \frac{r_{z,\min}}{r_{z,\max}}
    \le 
    \left( \frac{\sum_{k=1}^{J+1}\pi_k f_k}{\pi_{\min} f_{\min}} \right)^2 
    \big\| \hat{\bs{\tau}}_{\bs{W}}\big\|_{\infty} 
    \le 
    \left( \frac{\sum_{k=1}^{J+1}\pi_k f_k}{\pi_{\min} f_{\min}} \right)^2 
    \sqrt{\frac{\lambda_{\max}(\bs{s}_{\bs{W}}^2)}{r_{1, \min}r_{0, \min}} \frac{M_T}{n}}
    = o(n^{-1/2}), 
\end{align*}
which, by the same logic as \eqref{eq:f_max_min}, implies that $r_{z,\max}/r_{z,\min} - 1 = o(n^{-1/2})$. 

Second, we assume that \eqref{eq:ratio_converge} holds. 
From Lemma \ref{lemma:M_S_cat}, 
\begin{align*}
    (1-f) M_S \le \frac{J}{\lambda_{\min}(\bs{S}^2_{\bs{W}})} n \big\| \hat{\bs{\delta}}_{\bs{W}} \big\|_{\infty}^2
    \le \frac{J}{\lambda_{\min}(\bs{S}^2_{\bs{W}})} \cdot n \left( \frac{f_{\max}}{f_{\min}} - 1 \right)^2 
    = o(1), 
\end{align*}
and from Lemmas \ref{lemma:M_T_cat} and \ref{lemma:samp_cov_cat}, 
\begin{align*}
    M_T \le \frac{J}{\lambda_{\min}(\bs{s}_{\bs{W}}^2)} n \big\| \hat{\bs{\tau}}_{\bs{W}}\big\|_{\infty}^2 
    \le 
    \frac{J}{\lambda_{\min}(\bs{s}_{\bs{W}}^2)}
    \frac{n}{r_{1-z, \min}^2} \left( \frac{r_{z, \max}}{r_{z, \min}} - 1 \right)^2
    = o(1). 
\end{align*}

From the above, Lemma \ref{lemma:cat_M_ST_cong_zero} holds. 
\end{proof}

\subsection{Asymptotics under stratified sampling and blocking}

First, we derive the asymptotic distribution of the difference-in-means estimator under stratified sampling and blocking, as shown in the following proposition. 

\begin{proposition}\label{prop:str_crse_clt}
Under the stratified sampling and blocking, if Condition \ref{cond:categorical} holds and both $(1-f)M_S$ and $M_T$ are of order $o(1)$, then 
\begin{align*}
    \sqrt{n}\left( \hat{\tau} - \tau \right) 
    \ \dot\sim \ 
    \sqrt{
    \sum_{j=1}^{J+1}\pi_j V_{\tau\tau,j}
    }
    \cdot \varepsilon, 
\end{align*}
where $\varepsilon\sim \mathcal{N}(0,1)$, and
$
V_{\tau\tau,j} = r_{1j}^{-1} S^2_{1j} + r_{0j}^{-1} S^2_{0j} - f_j S^2_{\tau j}
$
for $1\le j \le J+1$. 
\end{proposition}
\begin{proof}[\bf Proof of Proposition \ref{prop:str_crse_clt}]
For each stratum $j$ ($1\le j \le J+1$), 
let $\hat{Y}_{1j}$ and $\hat{Y}_{0j}$ denote the average observed outcomes for units in treatment and control groups, respectively, 
and 
$\hat{\tau}_j = \hat{Y}_{1j} - \hat{Y}_{0j}$ denote the corresponding difference-in-means estimator. 
Then, 
by definition, 
the difference-in-means estimator has the following equivalent forms:
\begin{align}\label{eq:tau_hat_cat}
    \hat{\tau} 
    & =
    \frac{1}{n_1} \sum_{i=1}^{N} Z_{i} T_{i} Y_{i}-\frac{1}{n_{0}} \sum_{i=1}^{N} Z_{i}\left(1-T_{i}\right) Y_{i} 
    =\frac{1}{n_{1}} \sum_{j=1}^{J+1} \sum_{i: \tilde{W}_{i}=j} Z_{i} T_{i} Y_{i}-\frac{1}{n_{0}} \sum_{j=1}^{J+1} \sum_{i: \tilde{W}_{i}=j} Z_{i}\left(1-T_{i}\right) Y_{i} 
    \nonumber
    \\
    & = 
    \sum_{j=1}^{J+1}\left\{\frac{N \pi_j f_j r_{1j}}{n_1} \frac{1}{N \pi_j f_j r_{1j}} \sum_{i: \tilde{W}_{i}=j} Z_{i} T_{i} Y_{i} - 
    \frac{N \pi_j f_j r_{0j}}{n_0} \frac{1}{N \pi_j f_j r_{0j}} \sum_{i: \tilde{W}_{i}=j} Z_{i}\left(1-T_{i}\right) Y_{i}\right\}
    \nonumber
    \\
    & = 
    \sum_{j=1}^{J+1}\left( \frac{N \pi_j f_j r_{1j}}{n_1} \hat{Y}_{1j} - 
    \frac{N \pi_j f_j r_{0j}}{n_0} \hat{Y}_{0j}
    \right)
    \nonumber
    \\
    & = 
    \sum_{j=1}^{J+1}\pi_j  \left( \hat{Y}_{1j} - 
    \hat{Y}_{0j}
    \right) 
    + \sum_{j=1}^{J+1} \left( \frac{N \pi_j f_j r_{1j}}{n_1} - \pi_j \right) \hat{Y}_{1j} 
    - \sum_{j=1}^{J+1} \left( \frac{N \pi_j f_j r_{0j}}{n_0} - \pi_j \right) \hat{Y}_{0j}
    \nonumber
    \\
    & = 
    \sum_{j=1}^{J+1}\pi_j  \hat{\tau}_{j}
    + \sum_{j=1}^{J+1} \left( \frac{N \pi_j f_j r_{1j}}{n_1} - \pi_j \right) \hat{Y}_{1j} 
    - \sum_{j=1}^{J+1} \left( \frac{N \pi_j f_j r_{0j}}{n_0} - \pi_j \right) \hat{Y}_{0j}. 
\end{align}
Below we consider the three terms in \eqref{eq:tau_hat_cat}, separately.

For the first term, note that under the stratified sampling and blocking, we are essentially conducting stratified CRSE within each stratum. 
Thus, from Corollary \ref{cor:crse}, under Condition \ref{cond:categorical}, we have 
\begin{align*}
\sqrt{N \pi_j f_j} (\hat\tau_j - \tau_j)  \  \dot\sim\  V_{\tau\tau,j}^{1/2} \cdot \varepsilon_j, \qquad (1\le j \le J+1). 
\end{align*}
where $(\varepsilon_1, \varepsilon_2, \ldots, \varepsilon_{J+1})$ are mutually independent standard Gaussian random variables. 
From Lemmas \ref{lemma:prop_weight_diff} and \ref{lemma:cat_M_ST_cong_zero}, when
both $(1-f)M_S$ and $M_T$ are of order $o(1)$, we have 
\begin{align*}%
    \max_{1\le j \le J+1}\left| \frac{N \pi_j f_j}{n} - \pi_j \right|
    = 
    \max_{1\le j \le J+1}\left| \frac{\pi_jf_j}{\sum_{k=1}^{J+1} \pi_k f_k} - \pi_j \right|
    \le 
    \frac{f_{\max}}{f_{\min}} - 1 = o(1). 
\end{align*}
By Slutsky's theorem and the mutual independence of sampling and treatment assignment across all strata, these imply that 
\begin{align*}
    \sqrt{n}\left( \sum_{j=1}^{J+1}\pi_j  \hat{\tau}_{j} - \tau  \right)
    & = 
    \sqrt{n}
    \sum_{j=1}^{J+1}\pi_j ( \hat{\tau}_{j} - \tau_j ) 
    = 
    \sum_{j=1}^{J+1}\pi_j \sqrt{ \frac{n}{N \pi_j f_j} } \sqrt{N \pi_j f_j} ( \hat{\tau}_{j} - \tau_j ) 
    \\
    & \dot\sim 
    \sum_{j=1}^{J+1}\pi_j \sqrt{ \frac{1}{\pi_j} } V_{\tau\tau,j}^{1/2} \cdot \varepsilon_j
    \sim 
    \sqrt{
    \sum_{j=1}^{J+1}\pi_j V_{\tau\tau,j}
    }
    \cdot \varepsilon. 
\end{align*}

For the second and third terms, 
from Lemmas \ref{lemma:prop_weight_diff} and \ref{lemma:cat_M_ST_cong_zero}, when
both $(1-f)M_S$ and $M_T$ are of order $o(1)$, we have, for $z=0,1$,
\begin{align}\label{eq:f_j_pi_j_r_zj}
    \max_{1\le j\le J+1}\left| \frac{N \pi_j f_j r_{zj}}{n_z} - \pi_j \right|
    & = 
    \max_{1\le j \le J+1}
    \left| 
    \frac{\pi_j f_j r_{zj}}{\sum_{k=1}^{J+1} \pi_k f_k r_{zk}} - \pi_j
    \right| 
    \le 
    \frac{\max_{1\le j \le J+1} f_j r_{zj}}{\min_{1\le j \le J+1} f_j r_{zj}} - 1
    \le 
    \frac{f_{\max} r_{z, \max}}{f_{\min} r_{z,\min}} - 1
    \nonumber
    \\
    & = 
    \{1 + o(n^{-1/2})\} \cdot \{1 + o(n^{-1/2})\} - 1
    =
    o(n^{-1/2}). 
\end{align}
Note that the set of units under treatment arm $z$ for each stratum is essentially a simple random sample of size $N\pi_j f_j r_{zj}$ from all the units in that stratum.  
By the property of simple random sampling and Chebyshev's inequality, 
under Condition \ref{cond:categorical}, 
for $z=0,1$
\begin{align*}
    \hat{Y}_{zj} - \bar{Y}_{j}(z)
    & = 
    O_{\Pr}\left( 
    \sqrt{
    \frac{ 1 - f_jr_{zj} }{N\pi_jf_j r_{zj}}
    S^2_{1j}
    }
    \right)
    = 
    O_{\Pr}\left( 
    \sqrt{
    \frac{N}{N-1}\frac{\max _{i: \tilde{W}_i = j}\{Y_{i}(z)-\bar{Y}_j(z)\}^{2}}{N\pi_jf_j}
    }
    \right)
    = o_{\Pr}(1). 
\end{align*}
Consequently, from Condition \ref{cond:categorical} and \eqref{eq:f_j_pi_j_r_zj}, for $z=0,1$, 
\begin{align*}
    \sqrt{n} \sum_{j=1}^{J+1} \left( \frac{N \pi_j f_j r_{zj}}{n_z} - \pi_j \right) \hat{Y}_{zj} 
    & = 
    \sum_{j=1}^{J+1} \sqrt{n} \left( \frac{N \pi_j f_j r_{zj}}{n_z} - \pi_j \right) 
    \left( \bar{Y}_{j}(z) + o_{\Pr}(1) \right)
    \\
    & = \sum_{j=1}^{J+1} o(1)\cdot\left( O(1) + o_{\Pr}(1) \right)
    = o_{\Pr}(1). 
\end{align*}

From the above and by Slutsky's theorem, we have 
\begin{align*}
    \sqrt{n}\left( \hat{\tau} - \tau \right) & = 
    \sqrt{n}\left( \sum_{j=1}^{J+1}\pi_j  \hat{\tau}_{j} - \tau \right)
    + \sqrt{n}\sum_{j=1}^{J+1} \left( \frac{N \pi_j f_j r_{1j}}{n_1} - \pi_j \right) \hat{Y}_{1j} 
    - \sqrt{n} \sum_{j=1}^{J+1} \left( \frac{N \pi_j f_j r_{0j}}{n_0} - \pi_j \right) \hat{Y}_{0j}
    \\
    & \dot\sim 
    \sqrt{
    \sum_{j=1}^{J+1}\pi_j V_{\tau\tau,j}
    }
    \cdot \varepsilon.
\end{align*}
Therefore, Proposition \ref{prop:str_crse_clt} holds.     
\end{proof}

Second, we prove Theorem \ref{thm:resem_cat} using Proposition \ref{prop:str_crse_clt}. 

\begin{proof}[\bf Proof of Theorem \ref{thm:resem_cat}]
By some algebra, we can verify that the residual from the linear projection of $Y(z)$ on $\bs{W}$ is $Y(z) - \sum_{j=1}^{J+1} \mathbb{1}\{\tilde{W} = j\} \bar{Y}_{j}(z)$, for $z=0,1$. 
This implies that, under Condition \ref{cond:categorical}, 
the corresponding finite population variance has the following equivalent forms: 
\begin{align*}
    S^2_{z\setminus \bs{W} } & = \frac{1}{N-1} \sum_{i=1}^{N} \left\{
    Y_i(z) - \sum_{j=1}^{J+1}  \mathbb{1}\{\tilde{W}_i = j\} \bar{Y}_{j}(z) \right\}^2 
    = \frac{1}{N-1} \sum_{j=1}^{J+1}   \sum_{i:\tilde{W}_i = j} \left\{ Y_i(z) - \bar{Y}_{j}(z) \right\}^2
    \\& = \sum_{j=1}^{J+1} \frac{N \pi_j - 1}{N-1} S^2_{zj}
    =\sum_{j=1}^{J+1} \pi_j S^2_{zj} - \frac{1}{N-1} \sum_{j=1}^{J+1}\left( 1- \pi_j \right)S^2_{z j} 
    \\& = \sum_{j=1}^{J+1} \pi_j S^2_{z j} + O(N^{-1}), 
    \qquad (z=0, 1). 
\end{align*}
By the same logic, we can derive that 
$S^2_{\tau\setminus \bs{W}} = \sum_{j=1}^{J+1} \pi_j S^2_{\tau j} + O(N^{-1})$. 
Note that, for $1\le j \le J+1$ and $z=0,1$, 
\begin{align*}
    \frac{f_j}{f} = 
    \frac{f_j}{\sum_{k=1}^{J+1} \pi_k f_k} 
    \in 
    \left[\frac{f_{\min}}{f_{\max}},\  \frac{f_{\max}}{f_{\min}}\right], \quad
    \frac{r_{zj}}{r_z} = 
    \frac{r_{zj} \sum_{k=1}^{J+1} \pi_k f_k}{\sum_{k=1}^{J+1} \pi_k f_k r_{zk}}
    \in 
    \left[\frac{r_{z,\min}}{r_{z, \max}},\  \frac{r_{z,\max}}{r_{z, \min}}\right]. 
\end{align*}
From Lemma \ref{lemma:cat_M_ST_cong_zero}, when both $(1-f)M_S$ and $M_T$ are of order $o(1)$, \eqref{eq:ratio_converge} holds and thus 
$f_j/f = 1 + o(1)$ and $r_{zj}/r_z = 1+o(1)$ for $1\le j \le J+1$ and $z=0,1$. 

From the above, the asymptotic variance of $\sqrt{n}( \hat{\tau} - \tau)$ in Proposition \ref{prop:str_crse_clt} has the following equivalent forms:
\begin{align*}
    \sum_{j=1}^{J+1} \pi_j V_{\tau\tau,j}
    & 
    = \sum_{j=1}^{J+1}\pi_j 
    \left\{ 
    r_{1j}^{-1} S_{1j}^{2}+r_{0j}^{-1} S_{0j}^{2}-f_j S_{\tau j}^{2}
    \right\} 
    \\
    & = 
    r_1^{-1}
    \sum_{j=1}^{J+1}\pi_j \frac{r_1}{r_{1j}} S_{1j}^{2}+ r_0^{-1} \sum_{j=1}^{J+1}\pi_j \frac{r_0}{r_{0j}} S_{0j}^{2}-
    f\sum_{j=1}^{J+1}\pi_j \frac{f_j}{f} S_{\tau j}^{2} \\
    & = 
    r_1^{-1}
    \sum_{j=1}^{J+1}\pi_j S_{1j}^{2}+ r_0^{-1} \sum_{j=1}^{J+1}\pi_j S_{0j}^{2}-
    f\sum_{j=1}^{J+1}\pi_j S_{\tau j}^{2} + o(1)
    \\
    & = r_1^{-1} S^2_{1\setminus \bs{W}} + r_0^{-1} S^2_{0\setminus \bs{W}} - f S^2_{\tau \setminus \bs{W}} + o(1). 
\end{align*}
By the definition of the squared multiple correlation $R_S^2$ and $R_T^2$ and the variance formula $V_{\tau\tau}$, 
\begin{align*}
    V_{\tau\tau}(1- R_S^2 - R_T^2)
    & = 
    \left( r_1^{-1}S_1^2 + r_0^{-1}S_0^2 -f S_{\tau}^2\right) - (1-f) S^2_{\tau \mid \bs{W}} - 
    \left( r_1^{-1} S^2_{1\mid \bs{W}} + r_0^{-1} S^2_{0\mid \bs{W}} - S^2_{\tau\mid \bs{W}} \right)
    \\
    & = 
    r_1^{-1} \left( S_1^2 - S^2_{1\mid \bs{W}} \right) + 
    r_0^{-1} \left( S_0^2 - S^2_{0\mid \bs{W}} \right) - f\left( S_{\tau}^2 -  S^2_{\tau \mid \bs{W}} \right)\\
    & = 
    r_1^{-1} S^2_{1\setminus \bs{W}} + r_0^{-1} S^2_{0\setminus \bs{W}} - f S^2_{\tau \setminus \bs{W}}. 
\end{align*}
Thus, we have 
\begin{align*}
    \sum_{j=1}^{J+1} \pi_j V_{\tau\tau,j} 
    & = V_{\tau\tau}(1- R_S^2 - R_T^2) + o(1). 
\end{align*}
From Proposition \ref{prop:str_crse_clt} and Slutsky's theorem, the asymptotic distribution of $\sqrt{n}(\hat{\tau}-\tau)$ then has the following equivalent forms: 
\begin{align*}
    \sqrt{n}\left( \hat{\tau} - \tau \right) 
    \ \dot\sim \ 
    \sqrt{
    \sum_{j=1}^{J+1}\pi_j V_{\tau\tau,j}
    }
    \cdot \varepsilon
    \ \dot\sim \
    \sqrt{V_{\tau\tau}(1- R_S^2 - R_T^2)} \cdot \varepsilon. 
\end{align*}
Therefore, Theorem \ref{thm:resem_cat} holds. 
\end{proof}

Third, we comment on the equivalent conditions of $(1-f) M_S = o(1)$ and $M_T=o(1)$.

\begin{proof}[\bf Comment on equivalent conditions of $(1-f) M_S = o(1)$ and $M_T=o(1)$]
From Lemma \ref{lemma:cat_M_ST_cong_zero}, 
it is immediate that 
$(1-f) M_S$ and $M_T$ are of order $o(1)$ if and only if 
\begin{align*}
    \frac{\max_{1\le j \le J+1}f_j}{\min_{1\le j \le J+1}f_j} - 1 = o(n^{-1/2}), 
    \ \ 
    \frac{\max_{1\le j \le J+1}r_{1j}}{\min_{1\le j \le J+1}r_{1j}} - 1 = o(n^{-1/2}), 
    \ \ 
     \frac{\max_{1\le j \le J+1}r_{0j}}{\min_{1\le j \le J+1}r_{0j}} - 1 = o(n^{-1/2}).
\end{align*}
\end{proof}

\section{Clustered survey experiments}\label{sec:proof_cluster}

\subsection{Aggregated potential outcomes, sampling and treatment assignment}
In this section, we study rerandomization for clustered survey experiments. 
Recall that there are in total $M$ clusters, and for each cluster $1\le l \le M$, 
$\tilde{Y}_l(1) = (N/M)^{-1}\sum_{i:G_i = l} Y_i(1)$ and  $\tilde{Y}_l(0) = (N/M)^{-1}\sum_{i:G_i = l} Y_i(0)$ denote the aggregated treatment and control potential outcomes, 
and $\tilde{\bs{W}}_l = (N/M)^{-1}\sum_{i:G_i = l} \bs{W}_i$ and $\tilde{\bs{X}}_l = (N/M)^{-1}\sum_{i:G_i = l} \bs{X}_i$ denote the aggregated covariates at the sampling and treatment assignment stages, 
where $\tilde{\bs{W}}_l\in \mathbb{R}^J$ and $\tilde{\bs{X}}_l \in \mathbb{R}^K$. 
The average treatment effect at the cluster level is then \begin{align*}
    \tilde{\tau} & \equiv
    \frac{1}{M} \sum_{l=1}^M \left\{ \tilde{Y}_l(1) -  \tilde{Y}_l(0) \right\} = 
    \frac{1}{M} \cdot \frac{M}{N} \sum_{i=1}^N \left\{ Y_i(1) -  Y_i(0) \right\}
    = \frac{1}{N} \sum_{i=1}^N \left\{ Y_i(1) -  Y_i(0) \right\}
    = \tau, 
\end{align*}
the same as 
the average treatment effect at the individual level. 
We introduce $\tilde{Z}_l$ and $\tilde{T}_l$ to denote the sampling and treatment assignment indicators for cluster $l$. 
Specifically, $\tilde{Z}_l=1$ if and only if units in cluster $l$ are sampled to enroll the experiment, 
and for sampled cluster $l$, $\tilde{T}_l$ equals 1 if the units in the cluster receive treatment and 0 otherwise. 
Then, for a sampled cluster $l$ with $\tilde{Z}_l=1$, its observed aggregated outcome is $\tilde{Y}_l = \tilde{T}_l \tilde{Y}_l(1) + (1 - \tilde{T}_l) \tilde{Y}_l(0)$, one of the two aggregated potential outcomes. 

For descriptive convenience, we introduce some finite population quantities at the cluster level. 
Let $\bar{\tilde{Y}}(1) = M^{-1} \sum_{l=1}^M  \tilde{Y}_l(1)$ and $\bar{\tilde{Y}}(0) = M^{-1} \sum_{l=1}^M  \tilde{Y}_l(0)$ be the average aggregated potential outcomes, 
$\bar{\tilde{\bs{W}}} = M^{-1} \sum_{l=1}^M \tilde{\bs{W}}_l$ and $\bar{\tilde{\bs{X}}} = M^{-1} \sum_{l=1}^M \tilde{\bs{X}}_l$ be the average aggregated covariates, 
and $\tilde{S}^2_1, \tilde{S}^2_0, \tilde{S}^2_\tau, \tilde{\bs{S}}^2_{\bs{W}}, \tilde{\bs{S}}^2_{\bs{X}}$, 
$\tilde{\bs{S}}_{1, \bs{W}}, \tilde{\bs{S}}_{0, \bs{W}}, \tilde{\bs{S}}_{1, \bs{X}}, \tilde{\bs{S}}_{0, \bs{X}}$ be
the finite population variances and covariances for aggregated potential outcomes, treatment effect and covariates. 
Analogously, we define the finite population variances of the linear projections of the aggregated potential outcomes and treatment effect on aggregated covariates $\tilde{S}^2_{1 \mid \bs{X}}$, $\tilde{S}^2_{0 \mid \bs{X}}$,  $\tilde{S}^2_{\tau \mid \bs{X}}$ and $\tilde{S}^2_{\tau \mid \bs{W}}$.

\subsection{Covariate balance criteria and rerandomized survey experiment}
Below we consider rerandomization for clustered survey experiments. 
Analogous to the discussion in Section \ref{sec:balance}, we introduce the difference-in-means of aggregated covariates between sampled clusters and all clusters and that between treated clusters and control clusters: 
\begin{align*}
    \hat{\bs{\delta}}_{\tilde{\bs{W}}}
    & = 
    \frac{1}{m} \sum_{l=1}^M \tilde{Z}_l \tilde{\bs{W}}_l - 
    \frac{1}{M} \sum_{l=1}^M \tilde{\bs{W}}_l, 
    \quad
    \hat{\bs{\tau}}_{\tilde{\bs{X}}} = 
    \frac{1}{m_1} \sum_{l=1}^M \tilde{Z}_l \tilde{T}_l \tilde{\bs{X}}_l - 
    \frac{1}{m_0} \sum_{l=1}^M \tilde{Z}_l ( 1-\tilde{T}_l ) \tilde{\bs{X}}_l, 
\end{align*}
and use the corresponding Mahalanobis distances to measure the covariate balance for clustered sampling and clustered treatment assignment: 
\begin{align*}
    \tilde{M}_S \equiv \hat{\bs{\delta}}_{\tilde{\bs{W}}}^\top
    \left\{ \left( \frac{1}{m} - \frac{1}{M}  \right) \tilde{\bs{S}}^2_{\bs{W}} \right\}^{-1}
    \hat{\bs{\delta}}_{\tilde{\bs{W}}}, 
    \quad 
    \tilde{M}_T \equiv \hat{\bs{\tau}}_{\tilde{\bs{X}}}^\top 
    \left( \frac{m}{m_1m_0} \tilde{\bs{s}}^2_{\bs{X}} \right)^{-1}
    \hat{\bs{\tau}}_{\tilde{\bs{X}}}, 
\end{align*}
where $\tilde{\bs{s}}^2_{\bs{X}}$ is the sample covariance matrix of the aggregated covariate $\tilde{\bs{X}}$ among sampled clusters. 
Rerandomized clustered survey experiment using Mahalanobis distances (ReCSEM) then has the following steps, where $\tilde{a}_S$ and $\tilde{a}_T$ are two predetermined positive thresholds:
\begin{itemize}
    \item[(1)] Draw a simple random sample of $m$ clusters from the in total $M$ clusters.
    \item[(2)] 
    Compute the Mahalanobis distance $\tilde{M}_S$. 
    If $\tilde{M}_S \le \tilde{a}_S$, continue to step (3); otherwise, return to step (2). 
    \item[(3)] Completely randomize the $m$ clusters into treatment and control groups, which contain $m_1$ and $m_0$ clusters respectively. 
    \item[(4)] Compute the Mahalanobis distance $\tilde{M}_T$. 
    If $\tilde{M}_T \le \tilde{a}_T$, continue to step (5); otherwise, return to step (3).
    \item[(5)] Conduct the experiment using the accepted assignment from step (4). 
\end{itemize}

\subsection{Asymptotic properties for rerandomized clustered survey experiment}
Note that under rerandomized clustered survey experiment, we are essentially conducting ReSEM at the cluster level. 
Therefore, all the results we derived for ReSEM can be generalized to ReCSEM once we operate at the cluster level or equivalently view each cluster as an individual. 
For conciseness, below we discuss only the asymptotic distribution of the difference-in-means estimator $\hat{\tilde{\tau}}$ under ReCSEM. 

To conduct finite population asymptotic analysis, analogous to Condition \ref{cond:fp}, we invoke the following regularity condition along the sequence of finite populations for aggregated potential outcomes and covariates at the cluster level. 
Let $\tilde{f} = m/M$ be the proportion of sampled clusters, and $\tilde{r}_1 = m_1/m$ and $\tilde{r}_0 = m_0/m$ be the proportions of sampled clusters assigned to treatment and control. 

\begin{condition}\label{cond:cluster_fp}
	As $N \rightarrow \infty$, the total number of clusters $M$ goes to infinity, and 
	the sequence of finite populations satisfies 
	\begin{enumerate}[label=(\roman*), topsep=1ex,itemsep=-0.3ex,partopsep=1ex,parsep=1ex]
		\item the proportion $\tilde{f}$ of sampled clusters has a limit in $[0, 1)$; 
		
		\item the proportions $\tilde{r}_1$ and $\tilde{r}_0$ of units assigned to treatment and control have positive limits;  
		
		\item the finite population variances $\tilde{S}^2_1, \tilde{S}^2_0, \tilde{S}^2_\tau$ 
		and covariances
		$
		\tilde{\bs{S}}_{\bs{W}}^2, \tilde{\bs{S}}_{1, \bs{W}}, \tilde{\bs{S}}_{0, \bs{W}}
		$
		$
		\tilde{\bs{S}}_{\bs{X}}^2, \tilde{\bs{S}}_{1, \bs{X}}, \tilde{\bs{S}}_{0, \bs{X}}
		$
		at the cluster level 
		have limiting values, 
		and the limits of $\tilde{\bs{S}}_{\bs{W}}^2$ 
		and $\tilde{\bs{S}}_{\bs{X}}^2$ are
		nonsingular;
		
		\item
		for $t\in \{0, 1\}$, 
		\begin{align*}
		    \frac{1}{m} \max_{1 \le l\le M} \{ \tilde{Y}_l(t) -\bar{\tilde{Y}}(t) \}^2 \rightarrow 0, \ \ 
		    \frac{1}{m} \max_{1 \le l\le M}\| \tilde{\bs{W}}_l - \bar{\tilde{\bs{W}}} \|_2^2
		    \rightarrow 0, \ \ 
		    \frac{1}{m} \max_{1 \le l\le M}\| \tilde{\bs{X}}_l - \bar{\tilde{\bs{X}}} \|_2^2
		    \rightarrow 0. 
		\end{align*}
	\end{enumerate}
\end{condition}

By the same logic as Theorem \ref{thm:dist}, we can then derive the asymptotic distribution of the difference-in-means estimator at the cluster level $\hat{\tilde{\tau}}$ under ReCSEM. 
Analogous to \eqref{eq:R2} and \eqref{eq:V_tautau}, define 
\begin{align*}
	\tilde{R}_S^2 & = 
	\frac{(1-\tilde{f})\tilde{S}^2_{\tau \mid  \bs{W}}}{\tilde{r}_1^{-1}\tilde{S}^2_1 +\tilde{r}_0^{-1}\tilde{S}^2_0 -\tilde{f} \tilde{S}^2_\tau},
	\ \  
	\tilde{R}_T^2 = 
	\frac{\tilde{r}_1^{-1}\tilde{S}^2_{1 \mid  \bs{X}}+\tilde{r}_0^{-1}\tilde{S}^2_{0 \mid  \bs{X}} -\tilde{S}^2_{\tau \mid  \bs{X}}}{\tilde{r}_1^{-1}\tilde{S}^2_1 +\tilde{r}_0^{-1}\tilde{S}^2_0 -\tilde{f}\tilde{S}^2_\tau}, 
	\ \ 
	\text{and}
	\ \ 
	\tilde{V}_{\tau \tau} = \tilde{r}_1^{-1}\tilde{S}_1^2 + \tilde{r}_0^{-1}\tilde{S}_0^2 -\tilde{f} \tilde{S}_{\tau}^2. 
\end{align*}
Analogous to $L_{J,a_S}$ and $L_{K,a_T}$, 
we define 
constrained Gaussian random variables $L_{J,\tilde{a}_S}$ and $L_{K,\tilde{a}_T} $ with thresholds $a_S$ and $a_T$ replaced by $\tilde{a}_S$ and $\tilde{a}_T$. 
Let $\varepsilon\sim \mathcal{N}(0,1)$ be a standard Gaussian random variable. 

\begin{theorem}
Under Condition \ref{cond:cluster_fp} and ReCSEM, 
\begin{align*}
\sqrt{m} (\hat{\tilde{\tau}} - \tau) \mid  \text{ReCSEM}
\ & \dot\sim\  
\tilde{V}_{\tau\tau}^{1/2} \left(  \sqrt{1-\tilde{R}_S^2 - \tilde{R}_T^2} \cdot \varepsilon
+ \sqrt{\tilde{R}_S^2} \cdot L_{J,\tilde{a}_S}
+
\sqrt{\tilde{R}_T^2} \cdot L_{K,\tilde{a}_T} 
\right),
\end{align*}
where 
$(\varepsilon, L_{J,\tilde{a}_S}, L_{K,\tilde{a}_T})$ are mutually independent.
\end{theorem}

\end{document}